\documentclass{article}
\usepackage{rotating}

\usepackage{amsmath,amssymb,amsthm}
\usepackage{geometry}
\newcommand{\Z}{\mathbb{Z}}
\newcommand{\R}{\mathbb{R}}
\newcommand{\N}{\mathbb{N}}
\newcommand{\E}{\mathbb{E}}
\newcommand{\C}{\mathbb{C}}

\newtheorem{thm}{Theorem}
\newtheorem{lem}[thm]{Lemma}
\newtheorem{cor}[thm]{Corollary}
\newtheorem{prop}[thm]{Proposition}
\newtheorem{remark}[thm]{Remark}

\begin{document}

\title{Renormalized powers of
Ornstein-Uhlenbeck processes and\\
well-posedness
of stochastic Ginzburg-Landau 
equations}

\author{Weinan E$^1$, Arnulf Jentzen$^2$
and Hao Shen$^3$
\bigskip
\\
\small{$^1 $BICMR and School of Mathematical Sciences,  Peking University, Beijing, China, 100870;}
\\
\small{Department of Mathematics and 
Program in Applied and Computational Mathematics,}
\\
\small{Princeton University, Princeton, NJ 08544-1000, USA,
e-mail: weinan@math.princeton.edu}
\smallskip
\\
\small{$^2$Seminar for Applied 
Mathematics,
Swiss Federal Institute of 
Technology, Zurich,}\\
\small{8092 Z\"urich, Switzerland, e-mail: arnulf.jentzen@sam.math.ethz.ch;
Program in Applied}
\\
\small{and Computational Mathematics,
Princeton University,
Princeton, NJ 08544-1000, USA}
\smallskip
\\
\small{$^3 $Program in Applied and Computational Mathematics,
Princeton University,}
\\
\small{Princeton, NJ 08544-1000, USA,
e-mail: hshen@princeton.edu}
\\
}

\maketitle

\tableofcontents

\begin{abstract}
This article analyzes 
well-definedness 
and regularity of renormalized
powers of Ornstein-Uhlenbeck processes
and uses this analysis 
to establish 
local
existence, uniqueness and regularity
of strong solutions of 
stochastic Ginzburg-Landau equations
with polynomial nonlinearities 
in two space dimensions
and with quadratic nonlinearities
in three space dimensions.
\end{abstract}

\section{Introduction}

The first part of 
this article 
(see Section~\ref{sec:RenormalizedPowers} below)
investigates well-definedness
and regularity of suitable 
renormalized
powers of 
Ornstein-Uhlenbeck processes.
More formally, let
$ 
  \left( 
    \Omega, \mathcal{F}, \mathbb{P}
  \right)
$
be a probability space,
let $ d \in \N := \{ 1, 2, \dots \} $,
$ n \in \{ 2, 3, 4, \dots \} $
and let
$ ( W_t )_{ t \in \R } $
be a two-sided cylindrical
$ I $-Wiener process
on the $ \R $-Hilbert space
$ L^2( [0, 2 \pi]^d, \R ) $
of equivalence classes of Lebesgue
square integrable functions
from $ [0, 2 \pi]^d $ to $ \R $.
Moreover, 
let
$
  C_{ \mathcal{P} }( [0, 2\pi]^d, \R ) 
$
be the space of periodic continuous
functions from
$ [0, 2 \pi]^d $ to $ \R $,
let
$
  A
  \colon
  D( A )
  \subset 
  C_{ \mathcal{P} }( [0, 2\pi]^d, \R ) 
  \to
  C_{ \mathcal{P} }( [0, 2\pi]^d, \R ) 
$
be the Laplacian with periodic
boundary conditions 
on 
$
  C_{ \mathcal{P} }( [0, 2\pi]^d, \R ) 
$
minus the identity operator
(see \eqref{eq:def_Laplacian} below for details)
and consider the stationary solution
$ V_t = \int_{ - \infty }^t e^{ A (t - s) } \, dW_s $,
$ t \in \R $,
of the SPDE
\begin{equation}
\label{eq:SPDE_intro}
  d V_t
= 
  A V_t 
  \, dt
  + dW_t 
\end{equation}
for $ t \in \R $.
Note that
the process 
$ V_t $, $ t \in \R $,
does in the case $ d \geq 2 $ 
$ \mathbb{P} $-almost surely 
not take values 
in a function space anymore
but 
in
$
  D( ( - A )^{ ( 2 - d ) / 4 - \varepsilon } )
$
(see, for instance,
Da Prato \& Zabczyk~\cite{dz92}).
Nonetheless, powers of $ V $
are well defined in a suitable sense
in the case $ d = 2 $.
Indeed, $ n $-th renormalized power of $ V $,
that is, the stochastic process
$
  : \! ( V_t )^n \! :
$,
$ t \in \R $,
is well defined and its regularity is
analyzed in the case $ d = 2 $
in Lemma~3.2 in Da Prato
\& Debussche~\cite{DaPratoDebussche2003}
(see, e.g., also \cite{simon_functional_1979,glimm_quantum_1987,DaPratoTubaro2007} 
for further details
on the definition of the $ n $-th
renormalized power).
Proposition~\ref{prop:regularity_Wick} in this 
article extends the regularity statement
of this result and also establish 
well definedness of
$
  : \! ( V_t )^2 \! :
$,
$ t \in \R $,
in the case $ d = 3 $.
Moreover, if $ d = 3 $, $ n \geq 3 $
or if $ d \geq 4 $, then
$
  : \! ( V_t )^n \! :
$,
$ t \in \R $,
can not be defined anymore
(see Section~7.1 in Da Prato \&
Tubaro~\cite{DaPratoTubaro2007}
in the case $ d = n = 3 $
and Lemma~\ref{lem:limitation} 
below in the general case).
Although 
$
  : \! ( V_t )^3 \! :
$,
$ t \in \R $,
does not make
sense in the case $ d = 3 $,
we establish in 
Proposition~\ref{prop:regularity_WickAver}
and Lemma~\ref{lem:limitationAWP}
below that the processes
$
  \int_{ t_0 }^t : \! ( V_s )^n \! : \, ds
$,
$ t \in [ t_0, \infty ) $,
$ t_0 \in \R $,
(which we refer as \emph{averaged Wick powers})
are well defined if and only if
$
  \frac{ n + 1 }{ n - 1 }
  > 
  \frac{ d }{ 2 }
$
(i.e.,
if and only if
$ d \in \{ 1, 2 \} $
or
($
  d = 3 
$
and
$
  n \in \{ 2, 3, 4 \}
$)
or 
($
  d \in \{ 4, 5 \}
$
and
$ n = 2 $)).
The integral thus mollifies the renormalized
power in a suitable sense and allows us to
define 
$
  \int_{ t_0 }^t : \! ( V_s )^3 \! : \, ds
$,
$ t \in [ t_0, \infty ) $,
$ t_0 \in \R $,
even in the case $ d = 3 $.
Another possibility to extend
the definition of
$
  : \! ( V_t )^n \! : 
$,
$ t \in \R $,
is to consider the process
$
  \int_{ - \infty }^t e^{ A (t - s) } 
  : \! ( V_s )^n \! : \, ds
$,
$ t \in \R $,
which we refer as
\emph{convolutional Wick power}.
Proposition~\ref{prop:regularity_WickConv}
and Lemma~\ref{lem:limitationCWP}
prove that
$
  \int_{ - \infty }^t e^{ A (t - s) } 
  : \! ( V_s )^n \! : \, ds
$,
$ t \in \R $,
is (as in the case of averaged
Wick powers) well defined if and 
only if
$
  \frac{ n + 1 }{ n - 1 }
  > 
  \frac{ d }{ 2 }
$.
Proposition~\ref{prop:regularity_WickConv}
also proves that
convolutional Wick powers
enjoy more regularity properties
than averaged Wick powers
constructed in 
Proposition~\ref{prop:regularity_WickAver}.
Our analysis of convolutional Wick powers
is inspired by a Walsh-expansion
for the KPZ equation
in the fundamental recent 
article Hairer~\cite{Hairer2012}.
For details on the results on 
Wick power, averaged Wick powers
and convolutional Wick powers the
reader is referred to the summary
in Subsection~\ref{sec:summary}
below.

The above outlined results on the
well-definedness and regularity of
renormalized powers of
$ V $ are used 
in the second part of 
this article 
(see Section~\ref{sec:SPDEs} below)
to analyzes 
strong solutions of 
stochastic Ginzburg-Landau
equations
with polynomial nonlinearities.
More formally,
let 
$ 
  \eta,
  \kappa_0, \kappa_1, \dots, \kappa_n 
  \in \R 
$,
let 
$
  x_0  
  \in
  D( ( - A )^{ \eta } )
$
and consider a solution
process
$ ( X_t )_{ t \in [0, \infty) } $ 
of the SPDE
\begin{equation}
\label{eq:SPDE_intro}
  d X_t
= 
  \left[
    A X_t
    +
    : \!
    \left(
      \sum\nolimits_{ i = 0 }^n
      \kappa_i \left( X_t \right)^i
    \right)
    \! :
  \right]
  dt
  + dW_t 
\end{equation}
for $ t \in [0,\infty) $
with the initial condition
$   
  X_0 
  = 
  x_0  
$
and where the expression
$
    : \!
    \left(
      \sum_{ i = 0 }^n
      \kappa_i ( X_t )^i
    \right)
    \! :
$
is a suitable renormalization
of the term
$
  \sum_{ i = 0 }^n
  \kappa_i ( X_t )^i
$
for $ t \in [0, \infty) $
(see Subsections~\ref{sec:two_dim}
and \ref{sec:three_dim} 
below for further details).
The parameter $ \eta \in \R $ 
thus measures the regularity of the
initial value.
SPDEs of the 
form~\eqref{eq:SPDE_intro}
have a strong connection
to models from quantum field
theory; see
\cite{ParisiWu1981}.
Local and global existence, uniqueness
and regularity of solutions 
of SPDEs of the form
\eqref{eq:SPDE_intro}
(and suitable mollified
versions of \eqref{eq:SPDE_intro}
respectively)
have been intensively 
studied
in the last two decades;
see, e.g., the monograph
\cite{dz92} and the references
mentioned therein 
for the one-dimensional case 
$ d = 1 $
and see
\cite{Jona-lasinioMitter1985,BorkarChariMitter1988,AlbeverioRoeckner1991,dz92,DaPratoZabczyk1996,GatarekGoldys1996,LiskevichRockner1998,MikuleviciusRozovskii1999,DaPratoDebussche2003}
for the more
subtle two-dimensional 
case $ d = 2 $.
In this article we are mainly
interested in strong solutions
of \eqref{eq:SPDE_intro}
and we therefore review
results for strong solutions
of \eqref{eq:SPDE_intro}
in a bit more detail 
in the following.

%
%
%
%
%
%
%
%
%
%

\emph{In the case
$ d = 1 $}, global
existence, uniqueness and
regularity of strong 
solutions follows, e.g., 
from
Section~7.2 in
Da Prato \& 
Zabczyk~\cite{dz92}
if $ n $ is odd and if
$ \kappa_n < 0 $.
In the case $ d = 1 $ the
expression 
$
  \sum_{ i = 0 }^n
  \kappa_i ( X_t )^i
$
appearing in \eqref{eq:SPDE_intro}
is well defined and it is
not necessary to replace it
by its renormalization
$
    : \!
    \left(
      \sum_{ i = 0 }^n
      \kappa_i ( X_t )^i
    \right)
    \! :
$
for $ t \in [0, \infty) $.
Moreover, note that
the solution process
$ ( X_t )_{ t \in [0, \infty ) } $
of the SPDE~\eqref{eq:SPDE_intro}
satisfies
$
  \mathbb{P}\big[
    X_t \in 
    D( ( - A )^{
      1 / 4 - \varepsilon
    } )
    \cup \{ \infty \}
  \big] = 1
$
for all $ t, \varepsilon \in ( 0, \infty ) $
in the case $ d = 1 $.
The solution process thus takes 
$ \mathbb{P} $-almost surely 
values in 
$ 
    D( ( - A )^{
      1 / 4 - \varepsilon
    } )
  \cup \{ \infty \}
$
in the case $ d = 1 $
where $ \varepsilon \in ( 0, \infty) $
is arbitrarily small.
Here and below the solution process
takes the value $ \infty $ after
its possible blow up
(e.g., if $ \kappa_n > 0 $).

\emph{In the case $ d = 2 $} 
the renormalization
is necessary and 
can not be avoided
(see Walsh~\cite{Walsh1986}
and, e.g., Section~1 in 
Hairer et al.~\cite{HairerRyserWeber2012}).
In the case $ d = 2 $
local existence, uniqueness
and regularity of solutions 
of \eqref{eq:SPDE_intro}
have been established in 
Proposition~4.4 in
Da Prato \& 
Debussche~\cite{DaPratoDebussche2003}
if the condition
\begin{equation}
\label{eq:DaPrato_initial_ass}
  \eta > 
  \inf_{ p \in ( n , \infty ) }
  \left(
  \max\!\left\{
    \frac{
      - 2
    }{
      p \left( 2 n + 1 \right)
    }
    , 
    \frac{
      - 1
    }{
      \left( n - 1 \right)
    }
    \left(
      1 - \frac{ n }{ p }
    \right)
  \right\}
  \right)
=
  -
  \sup_{ p \in ( n , \infty ) }
  \left(
  \min\!\left\{
    \frac{
      2
    }{
      p \left( 2 n + 1 \right)
    }
    , 
    \frac{
      1
    }{
      \left( n - 1 \right)
    }
    \left(
      1 - \frac{ n }{ p }
    \right)
  \right\}
  \right)
\end{equation}
is fulfilled beside
other assumptions
(see also
Theorem~4.2 in
\cite{DaPratoDebussche2003} for
the corresponding global existence result).
The first main result of this article,
Theorem~\ref{thm:2D}
in Subsection~\ref{sec:two_dim},
extends Da Prato \& Debussche's
result 
by establishing local
existence of strong solutions
in the case $ d = 2 $
for a larger class of initial 
values, that is, 
if the condition
\begin{equation}
\label{eq:Weinan_initial_ass}
  \eta >
  - \frac{ 2 }{ n }
\end{equation}
is fulfilled
instead of \eqref{eq:DaPrato_initial_ass}.
Clearly, 
assumption 
\eqref{eq:Weinan_initial_ass}
is less restrictive 
than assumption 
\eqref{eq:DaPrato_initial_ass}.
In addition,
under assumption
\eqref{eq:Weinan_initial_ass},
Theorem~\ref{thm:2D}
establishes more regularity
of the solution process
of the SPDE~\eqref{eq:SPDE_intro}.
The reader is referred to
\eqref{eq:regularity_DPD}
in Subsection~\ref{sec:two_dim}
for a detailed comparison
of the regularity statement
in Proposition~4.4 in
Da Prato \& 
Debussche~\cite{DaPratoDebussche2003}
and of the regularity statement in
Theorem~\ref{thm:2D} 
below.
Under assumption
\eqref{eq:Weinan_initial_ass},
Theorem~\ref{thm:2D} 
also shows 
that
the solution process
$ ( X_t )_{ t \in [0, \infty ) } $
of the SPDE~\eqref{eq:SPDE_intro}
satisfies
$
  \mathbb{P}\!\left[
    X_t \in 
    D( ( - A )^{ - \varepsilon } )
    \cup \{ \infty \}
  \right] = 1
$
for all $ t, \varepsilon \in ( 0, \infty ) $
and all $ r \in ( - \infty, 0 ) $
in the case $ d = 2 $.
The solution process thus takes 
$ \mathbb{P} $-almost surely 
values in 
$ 
    D( ( - A )^{
      - \varepsilon
    } )
  \cup \{ \infty \}
$
in the case $ d = 2 $
where $ \varepsilon \in ( 0, \infty) $
is arbitrarily small.

The next main result of this article
is devoted to \emph{the case $ d = 3 $
and $ n = 2 $}.
More precisely,
Theorem~\ref{thm:quadratic}
in Subsection~\ref{sec:three_dim},
proves local existence, uniqueness
and regularity of strong solutions
of \eqref{eq:SPDE_intro} in the
case $ d = 3 $ and $ n = 2 $
if the condition
$ 
  \eta > - 1
$
is fulfilled.
Under these assumptions,
Theorem~\ref{thm:quadratic}
proves that
the solution process of the 
SPDE~\eqref{eq:SPDE_intro}
satisfies
$
  \mathbb{P}\!\left[
    X_t \in 
    D( ( - A )^{
      - 1 / 4 - \varepsilon
    } )
    \cup \{ \infty \}
  \right] = 1
$
for all $ t, \varepsilon \in ( 0, \infty ) $.
The solution process thus takes 
$ \mathbb{P} $-almost surely 
values in 
$ 
    D( ( - A )^{
      - 1 / 4 - \varepsilon
    } )
  \cup \{ \infty \}
$
in the case $ d = 3 $ and
$ n = 2 $ and $ \eta > - 1 $
where $ \varepsilon \in ( 0, \infty) $
is arbitrarily small.
To the best of our knowledge,
Theorem~\ref{thm:quadratic}
is the first result in the literature
that establish local existence of
solutions of the 
SPDE~\eqref{eq:SPDE_intro}
in the three
dimensional case $ d = 3 $.
The proof of 
Theorem~\ref{thm:quadratic}
is based on a detailed analysis
of mild solutions of 
determinisitic nonautonomous
partial differential equations
in 
Subsection~\ref{sec:det_existence}
and on the analysis of
$ : \! ( V_t )^2 \! : $, $ t \in \R $, in three
dimensions $ d = 3 $
(see Section~\ref{sec:RenormalizedPowers}).

\subsubsection*{Acknowledgements}

Jan van Neerven, Alesandra Lunardi
and Philipp Doersiek are gratefully
acknowledged for a number of 
quite useful comments and 
references concerning
analytic semigroups and their
infinitesmal generators.

The work of E and Shen is supported in part by 
grants from ARO (W911NF-11-1-0101) and ONR (N00014-01-1-0674).
The work of Jentzen is supported in part by the research project "Numerical solutions of stochastic differential
equations with non-globally Lipschitz continuous coefficients" funded by
the German Research Foundation.

\subsection{Notation}
\label{sec:notation}

Throughout this article the
following conventions are used.
If $ \Omega $ is a set and 
$ 
  \mathcal{F} \subset 
  \mathcal{P}( \Omega ) 
$
is a subsets of the 
power set of $ \Omega $, then we denote
by
$ \sigma_{ \Omega }( \mathcal{F} ) $ the sigma-algebra
on $ \Omega $ which is generated by $ \mathcal{F} $.
If $ ( E, \mathcal{E} ) $ is a topological
space, then we denote by
$
  \mathcal{B}( E ) := \sigma_E( \mathcal{E} )
$
the Borel sigma-algebra of $ (E, \mathcal{E}) $.
Furthermore,
if $ d \in \N := \{ 1, 2, \dots \} $, then
we denote by
$
  C_{ \mathcal{P} }( [0, 2\pi]^d, \R ) 
$
the $ \R $-Banach space of 
periodic continuous functions
from $ [0, 2 \pi]^d $ to $ \R $
and by
$ 
  \mathcal{A}_d
  \colon 
  D(
    \mathcal{A}_d 
  ) 
  \subset 
  C_{ \mathcal{P} }( [0, 2\pi]^d, \R ) 
  \to
  C_{ \mathcal{P} }( [0, 2\pi]^d, \R ) 
$
the generator of a strongly continuous
analytic semigroup which satisfies
\begin{multline}
\label{eq:def_Laplacian}
  D(
    \mathcal{A}_d
  ) 
  \supset
  \Bigg\{
    v \in 
    C_P( [0, 2 \pi]^d, \R ) 
    \colon
    \bigg(
      \exists
      \,
      w \in 
      C^2( \R^d, \R ) 
      \colon
\\
      \Big[
        \forall \, x \in \R^d \colon
        \forall \, j \in \{ 1, \dots, d \} 
        \colon
        w(x) 
        = 
        w\big(
          x + 2 \pi e_j^{ (d) } 
        \big)
      \Big]
      \wedge
      \Big[ 
        w|_{ [0,2\pi]^d } = v 
      \Big]
    \bigg)
  \Bigg\}
\end{multline}
and
$
  \mathcal{A}_d v 
  = 
  \triangle v - v
$
for all $ v \in D( \mathcal{A}_d ) $.
The fact that such an operator exists and is unique
can, e.g., be proved by considering the Laplacian
on the whole $ \R^d $.
In addition, if $ d \in \N $ 
and $ r \in \R $, then we denote by
\begin{equation}
\label{eq:def_Cr}
\begin{split}
  \big(
    \mathcal{C}^r_{ \mathcal{P} 
    }( [0, 2 \pi ]^d, \R ) , 
    \left\| \cdot \right\|_{
      \mathcal{C}^r_{ \mathcal{P} 
      }( [0, 2 \pi ]^d, \R ) 
    }
  \big)
  :=
  \left(
    D\big( 
      \left( 
        - \mathcal{A}_d 
      \right)^{ r / 2 } 
    \big),
    \big\| 
      \left( 
        - \mathcal{A}_d 
      \right)^{ 
        r / 2
      }\!( \cdot ) 
    \big\|_{
      C( [0, 2 \pi ]^d, \R )
    }
  \right)
\end{split}
\end{equation}
the $ \R $-Banach space
of the domain of 
the 
$ \frac{ r }{ 2 } $-fractional power
of 
$ 
  \mathcal{A}_d
$.
Finally, we observe that
there exist real numbers
$
  c^{ (d) }_{ \alpha, \beta, \gamma } \in [0,\infty)
$,
$
  \alpha, \beta, \gamma \in \R
$,
$ 
  d \in \N 
$,
such that
for every
$ d \in \N $,
every
$ \alpha, \beta, \gamma \in \R $
with $ \alpha + \beta > 0 $
and $ \gamma < \min( \alpha, \beta ) $,
every
$ v \in \mathcal{C}^{ \alpha }_{ \mathcal{P} }( [0, 2\pi]^d, \R ) $
and every
$ w \in \mathcal{C}^{ \beta }_{ \mathcal{P} }( [0, 2\pi]^d, \R ) $
it holds that
$
  v \cdot w \in \mathcal{C}^{ \gamma }_{ \mathcal{P} }( [0, 2\pi]^d, \R )
$
and that
\begin{equation}
\label{eq:multiplication}
  \left\| v \cdot w \right\|_{ 
    \mathcal{C}^{ \gamma }_{ \mathcal{P} }( [0, 2\pi]^d, \R )
  }
  \leq
  c^{ (d) }_{ \alpha, \beta, \gamma }
  \left\| v \right\|_{ 
    \mathcal{C}^{ \alpha }_{ \mathcal{P} }( [0, 2\pi]^d, \R )
  }
  \left\| w \right\|_{ 
    \mathcal{C}^{ \beta }_{ \mathcal{P} }( [0, 2\pi]^d, \R )
  } .
\end{equation}
More details on interpolation spaces and 
analytic semigroups can, e.g, 
be found in the 
excellent 
books
Lunardi~\cite{l09}, 
Van Neerven~\cite{neerven_adjoint_1992} 
and Sell \& You~\cite{sy02}.
Finally, throughout this article, 
if $ ( V, \left\| \cdot \right\|_V ) $ is an $ \R $-Banach space,
then we equip the set 
$ 
  V \cup \{ \infty \}  
$
with the topology
\begin{multline}
\label{eq:inf_topology}
  \Bigg\{
    A \subset 
    \Big( 
      V \cup \{ \infty \}
    \Big)
    \colon
    \bigg(
      \forall \, a \in A \backslash \{ \infty \}
      \colon
      \Big[
        \exists \, \varepsilon \in ( 0, \infty )
        \colon
        \left\{
          y \in V 
          \colon
          \left\| y - v 
          \right\|_V 
          <
          \varepsilon
        \right\}
        \subset A
      \Big]
    \bigg)
\\
    \text{and }
    \bigg(
      \infty \in A 
      \Rightarrow
      \Big[
        \exists \, R \in ( 0, \infty )
        \colon
          \left\{
            y \in V 
            \colon
            \left\| y  
            \right\|_V 
            >
            R
          \right\}
          \subset
          A
      \Big]
    \bigg)
  \Bigg\}
\end{multline}
and we observe that 
the
pairing
consisting of
$
    V \cup \{ \infty \}
$
and \eqref{eq:inf_topology}
is a complete metrizable topological space.

\section{Renormalized powers of Ornstein-Uhlenbeck processes}
\label{sec:RenormalizedPowers}

\subsection{Setting and assumptions}
\label{sec:setting}

Throughout Section~\ref{sec:RenormalizedPowers} we will frequently
assume that the following setting is fulfilled. 
Let $ d \in \N $, let
$ \delta \colon \Z^d \times \Z^d \to \R $
be a function defined through
\begin{equation}
  \delta_{ v, w }
  :=
  \begin{cases}
    1 
  & 
    \colon v = w
  \\
    0
  & 
    \colon v \neq w
  \end{cases}
\end{equation}
for all $ v, w \in \R^d $ and let
$ 
  g_v \colon [0, 2 \pi]^d \to \C
$,
$ v \in \Z^d $,
be a family of functions defined through
\begin{equation}
  g_v( x ) := 
  e^{ 
    {\bf i} 
    \left< v, x \right>_{ \R^d }
  }  
  =
  e^{ 
    {\bf i} 
    \left(
      v_1 x_1
      +
      \ldots 
      +
      v_d x_d
    \right)
  }  
\end{equation}
for all 
$ v = ( v_1, \dots, v_d ) \in \Z^d $
and all
$ x = ( x_1, \dots, x_d) \in [ 0, 2 \pi ]^d $. 
Next let
$ 
  \big( 
    H := L^2( (0, 2 \pi)^d ; \C ) ,
    \left< \cdot, \cdot \right>_H
    \left\| \cdot \right\|_H
  \big)
$
be the $ \C $-Hilbert space
of equivalence classes of
Lebesgue square integrable
functions from $ ( 0, 2 \pi )^d $
to $ \C $
with
$
  \left< v, w \right>_H
  =
  \int_{ ( 0, 2 \pi )^d }
  \overline{ v(x) }
  \cdot
  w(x) \,
  dx  
$
for all $ v, w \in H $.
Observe that
$ 
  \left( 2 \pi \right)^{ - \frac{ d }{ 2 } }
  g_v
$,
$ v \in \Z^d $,
is an orthonormal basis of $ H $
and that
$
  y =
  \sum_{
    v \in \Z^d
  }
  \frac{ 1 }{
    \left( 2 \pi \right)^d
  }
  \left<
    g_v, y
  \right>_H
  g_v
$
for all
$ 
  y \in H
$.
Moreover,
let
$
  \N_0 
  := \{ 0, 1, 2, \dots \}
$,
let
$ 
  \mathcal{P}_m
  :=
  \{
    (i, j) \in \{ 1, 2, \dots, m \}^2
    \colon
    i < j
  \}
$,
$ m \in \N $,
be sets
and let
$
  \Theta
  \colon
  \cup_{ 
    m = 1
  }^{ \infty }
 \left( \N_0
 \right)^{
  \mathcal{P}_m
 } 
 \to
 \cup_{ m = 1 }^{ \infty }
 \left( \N_0 \right)^m
$
be a function defined 
through
\begin{equation}
  \Theta( \alpha )
:=
  \left(
    \sum_{ 
      \substack{
        (i, j) \in 
        \mathcal{P}_{ m }
      \\
        i = 1 
        \text{ or }
        j = 1
      } 
    }
    \alpha_{ ( i, j ) }
    \,
    ,
    \quad
    \dots
    \quad
    ,
    \sum_{ 
      \substack{
        (i, j) \in 
        \mathcal{P}_{ m }
      \\
        i = m
        \text{ or }
        j = m
      } 
    }
    \alpha_{ ( i, j ) }
  \right)
  \in 
  \left( \N_0 \right)^m 
\end{equation}
for all 
$
  \alpha 
  \in
  \left( \N_0 \right)^{
    \mathcal{P}_m
  }
$
and all
$ m \in \N $.
Furthermore, we denote 
by
\begin{equation}
  \Phi
  :=
  \left\{
    \varphi
    \colon 
    \mathbb{Z}^d \to 
    [0,\infty)
    \colon
    \left(
      \forall \,
      v \in \Z^d
      \colon
      \varphi_v =
      \varphi_{ - v }
    \right)
  \right\}
\end{equation}
the set of all
functions from $ \Z^d $ to $ [0,\infty) $
that are symmetric with respect to
the origin
and equipp it with the Fr\'{e}chet metric
\begin{equation}
  d_{ \Phi }(  
    \varphi, 
    \psi
  ) 
  :=
  \sum_{
    k \in
    \mathbb{Z}^d
  }
  \frac{
    \min\!
    \left(
      1, \varphi_k - \psi_k 
    \right)
  }{
    2^{ 
      \left(
        \left| k_1 \right|
        +
        \ldots
        +
        \left| k_d \right|
      \right)
    }
  }
\end{equation}
for all 
$ 
  \varphi,
  \psi \in 
  \Phi
$.
Next 
define 
$ 
  \Phi_0 :=
  \{ 
    \varphi \in \Phi 
    \colon
    \varphi_k = 0
    \text{ for almost all }
    k \in \Z^d
  \}
  \subset \Phi
$
and
$
  \Phi_{ 0, \leq 1 }
  :=
  \{
    \varphi \in \Phi_0
    \colon
    (
      \forall \, k \in \Z^d 
      \colon
      \varphi_k \in [0,1]
    )
  \}
  \subset \Phi
$.
In addition, let $ \left( \Omega, \mathcal{F}, \mathbb{P} \right) $
be a probability space
and
let
$
  \beta^v
  \colon 
  \R \times \Omega \to \C
$,
$ v \in \Z^d $,
be a family of jointly Gaussian 
complex valued stochastic processes
with continuous sample paths and with
\begin{equation}
  \overline{ \beta^v_t }
  =
  \beta^{ - v }_t
\qquad
  \text{and}
\qquad
  \E\bigg[
    \overline{
      \beta^{ v }_{ t_1 }
    }
    \beta^{ w }_{ t_2 }
   \bigg]
  =
  \begin{cases}
    \delta_{
      v, w 
    }
    \min\!\left(
      \left| t_1 \right|, \left| t_2 \right|
    \right)
  &
    \colon
    t_1 \cdot t_2 \geq 0
  \\
    0
  &
    \colon
    t_1 \cdot t_2 < 0
  \end{cases}
\end{equation}
for all $ t, t_1, t_2 \in \R $
and all $ v, w \in \Z^d $.
Observe that $ \beta^v $, $ v \in \Z^d $,
are two-sided complex valued standard Brownian motions.
Moreover,
let 
$
  V^{ \varphi }
  \colon
  \R \times
  \Omega
  \to 
  \mathcal{C}_{
    \mathcal{P}
  }(
    [0, 2 \pi]^d, \R
  )
$,
$
  \varphi \in \Phi_0
$,
be a family of stochastic processes
with continuous sample paths satisfying
\begin{equation}
  V^{ \varphi }_t
  =
  \sum_{
    v \in \Z^d
  }
  \sqrt{ 2 } \,
  \varphi_v
  \left[
    \int_{ - \infty }^t
    e^{ 
      - \lambda_v
      \left( t - s \right)
    }
    \, d\beta^v_s
  \right]
  g_v
\end{equation}
$ \mathbb{P} $-almost surely 
for all $ t \in \R $
and all 
$ 
  \varphi 
  = ( \varphi_v )_{
    v \in \Z^d
  }
  \in \Phi_0 
$.
Observe that
\begin{equation}
\label{eq:covariance}
  \frac{ 1 }{ 
    \left( 2 \pi \right)^{ 2 d }
  }
  \,
  \mathbb{E}\!\left[
    \overline{
      \left<
        g_{ v_1 } ,
        V_{ t_1 
        }^{ \varphi^{ (1) } }
      \right>_{ \! H }
    }
    \left<
      g_{ v_2 } ,
      V_{ t_2
      }^{ \varphi^{ (2) } }
    \right>_{ \! H }
  \right]
  =  
  \frac{ 
    \delta_{ v_1, v_2 }
    \,
    \varphi_{ v_1 }^{ (1) }
    \varphi_{ v_2 }^{ (2) }
    \,
    e^{ 
      - \lambda_{ v_1 } \left| t_2 - t_1 \right| 
    }
  }{
    \lambda_{ v_1 }
  }
\end{equation}
for all $ t_1, t_2 \in \R $,
$ 
  v_1, v_2 \in \Z^d 
$
and all 
$ 
  \varphi^{ (1) } = 
  ( \varphi_v^{ (1) } )_{ v \in \Z^d } 
  ,
  \varphi^{ (2) } = 
  ( \varphi_v^{ (2) } )_{ v \in \Z^d } 
  \in 
  \Phi_0
$
and that
\begin{equation}
\label{eq:space_correlation}
\begin{split}
&
  \E\!\left[
  \overline{
    V^{
      \varphi^{ (1) }
    }_{ t_1 }\!( x_1 )
  }
  \,
    V^{
      \varphi^{ (2) }
    }_{ t_2 }\!( x_2 )
  \right]
  =
  \sum_{
    v \in \Z^d
  }
  \frac{ 1 }{
    \left( 2 \pi \right)^{ 2 d }
  }
  \,
  \mathbb{E}\!\left[
    \overline{
      \left<
        g_{ v } ,
        V_{ t_1 
        }^{ \varphi^{ (1) } }
      \right>_{ \! H }
    }
    \left<
      g_{ v } ,
      V_{ t_2
      }^{ \varphi^{ (2) } }
    \right>_{ \! H }
  \right]
  g_v( x_2 - x_1 )
\\ & =
  \sum_{
    v \in \Z^d
  }
  \frac{ 
    \varphi^{ (1) }_v
    \varphi^{ (2) }_v
    \,
    e^{ 
      - \lambda_v
      \left|
        t_2 - t_1
      \right|
    }
    \,
    g_v( x_1 - x_2 )
  }{ 
    \lambda_v
  }
\end{split}
\end{equation}
for all
$ 
  \varphi^{ (1) },
  \varphi^{ (2) }
  \in \Phi_0
$,
$ t_1, t_2 \in \R $
and all
$ x_1, x_2 \in [0, 2 \pi]^d $.
Moreover, if $ n \in \N $,
then we denote by
$ 
  \mathcal{W}_n \subset 
  L^2( \Omega; \R )
$
the closure 
in 
$ 
  L^2( \Omega; \R ) 
$
of the set
\begin{equation}
    \bigcup_{ k \in \N }
    \bigcup_{ 
      \substack{
        p \colon \R^k \to \R
        \text{ is a}
      \\
        \text{polyn.\ of degree $ n $}
      }
    }
    \bigcup_{ 
      \substack{
        v_1, \dots, v_k 
      \\
        \in \Z^d
      }
    }
    \bigcup_{ 
      \substack{
        t_1, \dots, t_k 
      \\
        \in \R
      }
    }
    \Big\{
      p\!\left( 
        \beta^{ v_1 }_{ t_1 }
        ,
        \dots ,
        \beta^{ v_k }_{ t_k }
      \right)
    \Big\}
  .
\end{equation}
Note for every $ n \in \N $
that the $ \R $-Hilbert
space 
$ 
  \mathcal{W}_n 
$
is the direct sum of the first
$ n $ Wiener chaoses;
see, e.g., Section~4
in 
Da Prato \&
Tubaro~\cite{DaPratoTubaro2007}
and
Section~A.1
in Hairer~\cite{Hairer2012}.
Furthermore, let
$ 
  H_n \colon \R \to \R 
$,
$ n \in \{ 0, 1, 2, \dots \} $,
be the unique functions
satisfying
\begin{equation}
  e^{ - \frac{ t^2 }{ 2 } + t x
  }
  = 
  \sum_{ n = 0 }^{ \infty }
  \frac{ t^n }{ n! } \cdot 
  H_n( x )  
\end{equation}
for all $ t, x \in \R $.
The functions $ H_n $, $ n \in \{ 0, 1, 2, \dots \} $,
are typically referred as \emph{(probabilists') Hermite polynomials}
in the literature.
Note that
$ H_0( x ) = 1 $,
$ H_1( x ) = x $,
$ H_2( x ) = x^2 - 1 $,
$
  H_3( x ) = x^3 - 3 x
$,
$
  H_4( x ) = x^4 - 6 x^2 + 3
$,
$ \dots $
for all $ x \in \R $.
In addition,
if
$
  Z \colon \Omega
  \rightarrow \R
$ 
is a centered real valued Gaussian random variable
and if $ n \in \N_0 $,
then we denote by
$
  : \!
  Z^n \!\!
  : \;
  \colon
  \Omega \to \R
$
the \emph{$ n $-th Wick power} 
of $ Z $, that is, 
the random variable given by
\begin{equation}
\label{eq:wickGRV}
  : \!
  Z^n \!\!
  :
  \;
  = 
  \begin{cases}
    \left(
      \E\big[
         Z^2  
      \big]
    \right)^{ 
      \frac{ n }{ 2 }
    }
    H_n\!\left( 
      \frac{ 
        Z 
      }{
        \sqrt{ 
         \E[  Z^2 ] 
        }
      }
    \right) 
  &
    \colon \E[ Z^2 ] > 0
  \\
    Z^n
  &
    \colon \E[ Z^2 ] = 0
  \end{cases}
\end{equation}
(see, e.g., page 9 
in Simon~\cite{simon_functional_1979}).
Moreover, 
we denote by
$
  : \!
  (
    V^{ \varphi }
  )^n
  \!\! : \;
  \colon
  \R \times \Omega
  \rightarrow
  \mathcal{C}_{ \mathcal{P} }( [ 0, 2 \pi ]^d, \R )
$,
$ \varphi \in \Phi_0 $,
$ n \in \N_0 $,
the stochastic processes
with continuous sample paths
given by
\begin{equation}
  \big(
    : \!
      \left( 
        V_t^{ \varphi } 
      \right)^n \!
    \! :
  \big)( x )
  = \;
    : \!
      \left( 
        V_t^{ \varphi }(x) 
      \right)^n     
    \! :
\end{equation}
for all $ t \in \R $,
$ x \in [0, 2 \pi]^d $,
$ \varphi \in \Phi_0 $
and all
$ n \in \N_0 $.
Note that
$
  : \!
    \left( 
      V_t^{ \varphi } 
    \right)^0 \!
  \! :
  \; =
  1
$,
$
  : \!
    \left( 
      V_t^{ \varphi } 
    \right)^1 \!
  \! : \;
  =
  V_t^{ \varphi }
$,
$
  : \! 
    \left( 
      V_t^{ \varphi } 
    \right)^2 \!
  \! : \; 
  =
  ( 
    V_t^{ \varphi } 
  )^2 
  -
  \E\big[
  ( 
    V_t^{ \varphi } 
  )^2 
  \big]
  =
  ( 
    V_t^{ \varphi } 
  )^2 
  -
  \sum_{ v \in \Z^d }
  \frac{ 
    ( \varphi_v )^2
  }{
    \lambda_v
  }
$,
$
  : \! 
    \left( 
      V_t^{ \varphi } 
    \right)^3 \!
  \! : \;
  =
  ( 
    V_t^{ \varphi } 
  )^3
  -
  3
  V_t^{ \varphi } 
  \E\big[
  ( 
    V_t^{ \varphi } 
  )^2 
  \big]
  =
  ( 
    V_t^{ \varphi } 
  )^3
  -
  3 V_t^{ \varphi } 
  \big(
  \sum_{ v \in \Z^d }
  \frac{ 
    ( \varphi_v )^2
  }{
    \lambda_v
  }
  \big)
$,
$ \dots $
for all
$ t \in \R $
and all
$
  \varphi \in \Phi_0
$.
In addition, 
we denote by
$
  \circ
  \,
  \big(
    V^{ \varphi }_{ t_0, ( \cdot ) }
  \big)^n
  \circ
  \colon
  [ t_0, \infty ) \times \Omega
  \rightarrow
  C_{ \mathcal{P} }( [ 0, 2 \pi ]^d, \R )
$,
$ \varphi \in \Phi_0 $,
$ n \in \N_0 $,
$ t_0 \in \R $,
the stochastic processes
with continuous sample paths
defined by
\begin{equation}
\label{eq:def_AWP}
  \circ \,
    ( V^{ \varphi }_{ t_0, t } )^n
  \circ
:=
  \int_{ t_0 }^t
  : \!
  ( V^{ \varphi }_s )^n 
  \! :
  ds
\end{equation}
for all 
$ \varphi \in \Phi_0 $,
$ n \in \N_0 $
and all
$ t_0, t \in \R $
with $ t_0 \leq t $
and we denote by
$
  \bullet
  \,
  \big(
    V^{ \varphi }
  \big)^n
  \bullet
  \colon
  \R
  \times \Omega
  \rightarrow
  C_{ \mathcal{P} 
  }( [ 0, 2 \pi ]^d, \R )
$,
$ \varphi \in \Phi_0 $,
$ n \in \N_0 $,
the stochastic processes
with continuous sample 
paths defined by
\begin{equation}
  \bullet \,
    ( V^{ \varphi }_t )^n
  \bullet
:=
  \int_{ - \infty }^t
  e^{ 
    \mathcal{A}_{ 
      d 
    } 
    ( t - s ) 
  }
  \big[
    : \! 
      ( V^{ \varphi }_s )^n 
    \! :
  \big]
  \,
  ds
\end{equation}
for all 
$ \varphi \in \Phi_0 $,
$ n \in \N_0 $
and all
$ t \in \R $.
The readers who are familiar with quantum field theory 
should distinguish the concept of the "time-ordered product" in quantum field theory
(see, for instance, Peskin \& Schroeder~\cite{PeskinSchroeder1995}) 
from the averaged and the
convolutional Wick power defined above.
Finally, note that
$
  \left( 
    V^{ \varphi }_t( x ) 
  \right)^n	
$,
$
  : \!\!
  \left( 
    V^{ \varphi }_t( x ) 
  \right)^n
  \!\!\! : \,
$,
$
  \circ
  \left( 
    V^{ \varphi }_{ t_0, t }( x ) 
  \right)^n
  \circ
$,
$
  \bullet
  \left( 
    V^{ \varphi }_t( x ) 
  \right)^n
  \bullet
  \in 
  \mathcal{W}_n
$
for all 
$ n \in \N $,
$ x \in [0, 2 \pi]^d $
and all 
$ t_0, t \in \R $
with $ t_0 \leq t $.


\subsection{Hypercontractivity estimates}

The following lemma allows us to 
calculate regularities of 
suitable stochastic processes 
by computing their correlations
in Fourier space.
It is quite similar to
Proposition~A.2 in
Hairer~\cite{Hairer2012}.

\begin{lem}
\label{lem:hypercontractivity}
Assume the setting of 
Subsection~\ref{sec:setting},
let $ n \in \N $
and let 
$ a, b \in \R $
with $ a < b $.
Then there exist
real numbers
$ 
  \chi^{
    n, d, p, a, b
  }_{
    \alpha, \hat{ \alpha },
    \beta, \hat{ \beta }
  }
  \in [0, \infty)
$,
$ 
  p, \alpha, \hat{ \alpha },
  \beta, \hat{ \beta } \in \R 
$,
such that 
\begin{align}
\label{eq:hypercontractivity_main}
&
  \left\|
    X
  \right\|_{
    L^p(
      \Omega ;
      C^{ \alpha 
      }(
        [ a, b ], 
        \mathcal{C}^{ 
          2 \beta
        }_{ \mathcal{P} }(
          [ 0, 2 \pi ]^d, \mathbb{R}
        )
      )
    )
  }
\\ &
  \leq
  \chi^{
    n, d, p, a, b
  }_{
    \alpha, \hat{ \alpha },
    \beta, \hat{ \beta }
  }
  \left[
  \sup_{ 
    \substack{
      t_1, t_2 
    \\
      \in [ a, b ] ,
    \\
      t_1 \neq t_2
    }
  }
  \sum_{ 
    \substack{
      v_1, 
    \\
      v_2 
    \\
      \in \Z^d 
    }
  }
  \left[
  \tfrac{
  \left|
  \E\!\left[
    \overline{
      \left<
        g_{ v_1 },
        X_{ t_1 }
      \right>_{ \! H }
    }
    \left<
      g_{ v_2 },
      X_{ t_1 }
    \right>_{ \! H }
  \right]
  \right|
  }{
    \left(
      \lambda_{ v_1 }
      \lambda_{ v_2 }
    \right)^{ - \hat{ \beta } }
  }
  +
  \tfrac{
  \left|
  \E\!\left[
    \overline{
      \left< g_{ v_1 } ,
        X_{ t_1 } - X_{ t_2 }
      \right>_{ \! H }
    }
    \left< g_{ v_2 } ,
      X_{ t_1 } - X_{ t_2 } 
    \right>_{ \! H }
  \right]
  \right|
  }{
    \left(
      \lambda_{ v_1 }
      \lambda_{ v_2 }
    \right)^{ - \hat{ \beta } }
    \left| 
      t_1 - t_2 
    \right|^{
      2 \hat{ \alpha }
    }
  }
  \right]
  \right]^{ \! \frac{ 1 }{ 2 } }
\nonumber
\end{align}
for all
$
  p \in ( 0, \infty )
$,
$
  \hat{ \alpha } \in ( \alpha , 1 ) 
$,
$
  \alpha \in ( 0, 1 )
$,
$
  \hat{ \beta } 
  \in ( \beta, \infty )
$,
$
  \beta \in \R
$
and all stochastic processes
$ 
  X \colon 
  [a, b] \times \Omega \to
  \cap_{ r \in \R } 
  \mathcal{C}^r_{ \mathcal{ P } }( 
    [0, 2 \pi]^d, \R 
  )
$
with continuous sample paths
which satisfy 
for every $ t \in [a, b] $
and every $ x \in [0, 2 \pi]^d $
that
$
  X_t( x )
  \in
  \mathcal{W}_n
$.
\end{lem}

\begin{proof}[Proof
of Lemma~\ref{lem:hypercontractivity}]
Hypercontractivity
(see, e.g., Lemma~A.1 in
Hairer~\cite{Hairer2012})
ensures that
there exist real numbers
$ 
  \kappa_{ k, p }
  \in [0,\infty) 
$,
$ k \in \N $,
$ p \in [2, \infty) $,
such that
\begin{equation}
\label{eq:est_hypercontractivity}
  \E\big[
    \left| Y \right|^p
  \big]
\leq
  \kappa_{ k, p }
  \left(
    \E\!\left[
      \left| Y \right|^2
    \right]
  \right)^{
    \! \frac{ p }{ 2 }
  }
\end{equation}
for all 
$ p \in [2,\infty) $,
$ Y \in \mathcal{W}_k $
and all $ k \in \N $.
Note that
\begin{equation}
\begin{split}
&
  \big\|
    ( - A )^{ \hat{ \beta } } X
  \big\|_{
    C^{ \hat{ \alpha } }( 
      [a, b],
      L^p( \Omega;
        L^p( 
          ( 0, 2 \pi )^d; \R
         )
      )
    )
  }
  =
  \sup_{ t \in [a, b] }
  \big\|
    ( - A )^{ \hat{ \beta } } X_t
  \big\|_{
    L^p( 
      \Omega;
      L^p( 
        ( 0, 2 \pi )^d; \R 
       )
    )
  }
\\ &
  +
  \sup_{ 
    \substack{
      t_1, t_2 \in [ a, b ]
    \\
      t_1 \neq t_2
    }
  }
  \frac{
    \big\|
      ( - A )^{ \hat{ \beta } } 
      (
        X_{ t_1 } -
        X_{ t_2 } 
      )
    \big\|_{
      L^p\left( 
        \Omega;
        L^p( 
          ( 0, 2 \pi )^d; \R 
        )
      \right)
    }
  }{
    \left| t_1 - t_2 \right|^{
      \hat{ \alpha }
    }
  }
\\ & =
  \sup_{ t \in [a, b] }
  \left\{ 
    \int_{ 
      ( 0, 2 \pi )^d
    }
    \E\!\left[
      \big|
        (
          ( - A )^{ \hat{ \beta } } X_t 
        )(x)
      \big|^p
    \right] 
    dx
  \right\}^{ 
    \frac{ 1 }{ p } 
  }
\\ &
  +
  \sup_{ 
    \substack{
      t_1, t_2 \in [ a, b ]
    \\
      t_1 \neq t_2
    }
  }
  \frac{
    \left\{
      \int_{ 
        ( 0, 2 \pi )^d
      }
      \E\!\left[
        \big|
          (
            ( - A )^{ \hat{ \beta } } 
            (
              X_{ t_1 } -
              X_{ t_2 } 
            )
          )(x)
        \big|^p
      \right]
      dx
    \right\}^{ \frac{ 1 }{ p } }
  }{
    \left| t_1 - t_2 \right|^{
      \hat{ \alpha }
    }
  }
\\ & =
  \sup_{ t \in [a, b] }
  \left\{ 
    \int_{ 
      ( 0, 2 \pi )^d
    }
    \E\!\left[
      \bigg|
        \sum_{ v \in \Z^d }
        \left( \lambda_ v \right)^{
          \hat{ \beta }
        }
        \left< g_v, X_t \right>_{ \! H }
        g_v( x )
      \bigg|^p
    \right] 
    dx
  \right\}^{ 
    \!
    \frac{ 1 }{ p } 
  }
\\ &
  +
  \sup_{ 
    \substack{
      t_1, t_2 \in [ a, b ]
    \\
      t_1 \neq t_2
    }
  }
  \frac{
    \left\{
      \int_{ 
        ( 0, 2 \pi )^d
      }
      \E\!\left[
        \big|
          \sum_{ v \in \Z^d }
          \left( \lambda_ v \right)^{
            \hat{ \beta }
          }
          \left< 
            g_v, X_{ t_1 } - X_{ t_2 }
          \right>_{ \! H }
          g_v( x )
        \big|^p
      \right]
      dx
    \right\}^{ \frac{ 1 }{ p } }
  }{
    \left| t_1 - t_2 \right|^{
      \hat{ \alpha }
    }
  }
\end{split}
\end{equation}
for all
$
  p \in ( 0, \infty )
$,
$
  \hat{ \alpha } \in ( 0, 1 )
$,
$
  \hat{ \beta } \in \R
$
and all stochastic processes
$ 
  X \colon 
  [a, b] \times \Omega \to
  \cap_{ r \in \R } 
  \mathcal{C}^r_{ \mathcal{ P } }( 
    [0, 2 \pi]^d, \R 
  )
$.
Estimate~\eqref{eq:est_hypercontractivity}
hence implies that
\begin{equation}
\begin{split}
&
  \big\|
    ( - A )^{ \hat{ \beta } } X
  \big\|_{
    C^{ \hat{ \alpha } }( 
      [a, b],
      L^p( \Omega;
        L^p( 
          ( 0, 2 \pi )^d; \R
         )
      )
    )
  }
\\ & \leq 
  \frac{
    \kappa_{ n, p }
  }{
    \left( 2 \pi \right)^d
  }
  \Bigg[
    \sup_{ t \in [a, b] }
      \left\{ 
        \int_{  
          ( 0, 2 \pi )^d
        }
        \left(
          \E\!\left[
            \left|
              \sum\nolimits_{ v \in \Z^d }
              \left( \lambda_ v \right)^{
                \hat{ \beta }
              }
              \left< 
                g_v, X_t
              \right>_{ \! H }
              g_v( x )
            \right|^2
          \right] 
        \right)^{ 
          \! \frac{ p }{ 2 } 
        }
      dx
    \right\}^{ 
      \! \frac{ 1 }{ p } 
    }
\\ &  
    +
    \sup_{ 
      \substack{
        t_1, t_2 \in [ a, b ]
      \\
        t_1 \neq t_2
      }
    }
    \tfrac{
      \left\{
        \int_{ 
          ( 0, 2 \pi )^d
        }
        \left(
          \E\left[
            \left|
              \sum_{ v \in \Z^d }
              \left( \lambda_ v \right)^{
                \hat{ \beta }
              }
              \left< 
                g_v, X_{ t_1 } - X_{ t_2 }
              \right>_{ \! H }
              g_v( x )
            \right|^2
          \right]
        \right)^{ 
          \! \frac{ p }{ 2 }
        }
      dx
    \right\}^{ 
      \! \frac{ 1 }{ p } 
    }
  }{
    \left| t_1 - t_2 \right|^{
      \hat{ \alpha }
    }
  }
  \Bigg]
\\ & = 
  \frac{
    \kappa_{ n, p }
  }{
    \left( 2 \pi \right)^d
  }
  \Bigg[
    \sup_{ t \in [a, b] }
      \left\{ 
        \int_{  
          ( 0, 2 \pi )^d
        }
        \left(
          \sum_{ 
            v_1, v_2 \in \Z^d 
          }
          \tfrac{
            \E\!\left[
              \overline{
                \left< 
                  g_{ v_1 }, X_t
                \right>_{ \! H }
              }
                \left< 
                  g_{ v_2 }, X_t
                \right>_{ \! H }
            \right] 
            g_{ (v_2 - v_1) }( x )
          }{ 
              \left( 
                \lambda_{ v_1 } 
                \lambda_{ v_2 } 
              \right)^{
                - \hat{ \beta }
              }
          }
        \right)^{ 
          \!\! \frac{ p }{ 2 } 
        }
      dx
    \right\}^{ 
      \! \frac{ 1 }{ p } 
    }
\\ &  
    +
    \sup_{ 
      \substack{
        t_1, t_2 \in [ a, b ]
      \\
        t_1 \neq t_2
      }
    }
      \left\{
        \int_{ 
          ( 0, 2 \pi )^d
        }
        \left(
          \sum_{ v_1, v_2 \in \Z^d }
          \tfrac{
            \E\left[
              \overline{
                \left< 
                  g_{ v_1 }, X_{ t_1 } - X_{ t_2 }
                \right>_{ \! H }
              }
              \left< 
                g_{ v_2 }, X_{ t_1 } - X_{ t_2 }
              \right>_{ \! H }
            \right]
            g_{ ( v_2 - v_1 ) }( x )
          }{
            \left( 
              \lambda_{ v_1 } 
              \lambda_{ v_2 }
            \right)^{
              - \hat{ \beta }
            }
            \left| t_1 - t_2 \right|^{
              2 \hat{ \alpha }
            }
          }
        \right)^{ 
          \!\! \frac{ p }{ 2 }
        }
      dx
    \right\}^{ 
      \! \frac{ 1 }{ p } 
    }
  \Bigg]
\end{split}
\end{equation}
for all
$
  p \in ( 0, \infty )
$,
$
  \hat{ \alpha }
  \in ( 0, 1 )
$,
$
  \hat{ \beta } 
  \in \R
$
and all stochastic processes
$ 
  X \colon 
  [a, b] \times \Omega \to
  \cap_{ r \in \R } 
  \mathcal{C}^r_{ \mathcal{ P } }( 
    [0, 2 \pi]^d, \R 
  )
$
which satisfy 
for every $ t \in [a, b] $
and every $ x \in [0, 2 \pi]^d $
that
$
  X_t( x )
  \in
  \mathcal{W}_n
$.
This implies
\begin{equation}
\begin{split}
&
  \big\|
    ( - A )^{ \hat{ \beta } } X
  \big\|_{
    C^{ \hat{ \alpha } }( 
      [a, b],
      L^p( \Omega;
        L^p( 
          ( 0, 2 \pi )^d; \R
         )
      )
    )
  }
\\ & \leq 
  \frac{
    \kappa_{ n, p }
  }{
    \left( 2 \pi \right)^d
  }
  \Bigg[
    \sup_{ t \in [a, b] }
      \left\{ 
        \int_{  
          ( 0, 2 \pi )^d
        }
        \left(
          \sum\nolimits_{ 
            v_1, v_2 \in \Z^d 
          }
          \tfrac{
            \left|
            \E\!\left[
              \overline{
                \left< 
                  g_{ v_1 }, X_t
                \right>_{ \! H }
              }
                \left< 
                  g_{ v_2 }, X_t
                \right>_{ \! H }
            \right] 
            \right|
          }{ 
              \left( 
                \lambda_{ v_1 } 
                \lambda_{ v_2 } 
              \right)^{
                - \hat{ \beta }
              }
          }
        \right)^{ 
          \! \frac{ p }{ 2 } 
        }
      dx
    \right\}^{ 
      \! \frac{ 1 }{ p } 
    }
\\ &  
    +
    \sup_{ 
      \substack{
        t_1, t_2 \in [ a, b ]
      \\
        t_1 \neq t_2
      }
    }
      \left\{
        \int_{ 
          ( 0, 2 \pi )^d
        }
        \left(
          \sum\nolimits_{ v_1, v_2 \in \Z^d }
          \tfrac{
            \left|
            \E\left[
              \overline{
                \left< 
                  g_{ v_1 }, X_{ t_1 } - X_{ t_2 }
                \right>_{ \! H }
              }
              \left< 
                g_{ v_2 }, X_{ t_1 } - X_{ t_2 }
              \right>_{ \! H }
            \right]
            \right|
          }{
            \left( 
              \lambda_{ v_1 } 
              \lambda_{ v_2 }
            \right)^{
              - \hat{ \beta }
            }
            \left| t_1 - t_2 \right|^{
              2 \hat{ \alpha }
            }
          }
        \right)^{ 
          \! \frac{ p }{ 2 }
        }
      dx
    \right\}^{ 
      \! \frac{ 1 }{ p } 
    }
  \Bigg]
\\ & =
  \frac{
    \kappa_{ n, p }
  }{
    \left( 2 \pi \right)^d
  }
  \Bigg[
    \sup_{ t \in [a, b] }
      \left\{ 
          \sum\nolimits_{ 
            v_1, v_2 \in \Z^d 
          }
          \tfrac{
            \left|
            \E\!\left[
              \overline{
                \left< 
                  g_{ v_1 }, X_t
                \right>_{ \! H }
              }
                \left< 
                  g_{ v_2 }, X_t
                \right>_{ \! H }
            \right] 
            \right|
          }{ 
              \left( 
                \lambda_{ v_1 } 
                \lambda_{ v_2 } 
              \right)^{
                - \hat{ \beta }
              }
          }
    \right\}^{ 
      \! \frac{ 1 }{ 2 } 
    }
\\ &  
    +
    \sup_{ 
      \substack{
        t_1, t_2 \in [ a, b ]
      \\
        t_1 \neq t_2
      }
    }
      \left\{
          \sum\nolimits_{ v_1, v_2 \in \Z^d }
          \tfrac{
            \left|
            \E\left[
              \overline{
                \left< 
                  g_{ v_1 }, X_{ t_1 } - X_{ t_2 }
                \right>_{ \! H }
              }
              \left< 
                g_{ v_2 }, X_{ t_1 } - X_{ t_2 }
              \right>_{ \! H }
            \right]
            \right|
          }{
            \left( 
              \lambda_{ v_1 } 
              \lambda_{ v_2 }
            \right)^{
              - \hat{ \beta }
            }
            \left| t_1 - t_2 \right|^{
              2 \hat{ \alpha }
            }
          }
    \right\}^{ 
      \! \frac{ 1 }{ 2 } 
    }
  \Bigg]
\end{split}
\end{equation}
and hence
\begin{equation}
\begin{split}
&
  \big\|
    ( - A )^{ \hat{ \beta } } X
  \big\|_{
    C^{ \hat{ \alpha } }( 
      [a, b],
      L^p( \Omega;
        L^p( 
          ( 0, 2 \pi )^d; \R
         )
      )
    )
  }
\\ & \leq 
  \kappa_{ n, p }
  \left[
    \sup_{ 
      \substack{
        t_1, t_2 
      \\
        \in [ a, b ]
      \\
        t_1 \neq t_2
      }
    }
      \sum_{ 
        \substack{
          v_1, v_2 
        \\
          \in \Z^d 
        }
      }
    \Bigg\{
          \tfrac{
            \left|
            \E\left[
              \overline{
                \left< 
                  g_{ v_1 }, X_{ t_1 }
                \right>_{ \! H }
              }
                \left< 
                  g_{ v_2 }, X_{ t_1 }
                \right>_{ \! H }
            \right] 
            \right|
          }{ 
              \left( 
                \lambda_{ v_1 } 
                \lambda_{ v_2 } 
              \right)^{
                - \hat{ \beta }
              }
          }
          +
          \tfrac{
            \left|
            \E\left[
              \overline{
                \left< 
                  g_{ v_1 }, X_{ t_1 } - X_{ t_2 }
                \right>_{ \! H }
              }
              \left< 
                g_{ v_2 }, X_{ t_1 } - X_{ t_2 }
              \right>_{ \! H }
            \right]
            \right|
          }{
            \left( 
              \lambda_{ v_1 } 
              \lambda_{ v_2 }
            \right)^{
              - \hat{ \beta }
            }
            \left| t_1 - t_2 \right|^{
              2 \hat{ \alpha }
            }
          }
    \Bigg\}
  \right]^{
    \! \frac{ 1 }{ 2 }
  }
\end{split}
\label{eq:hypercontractivity}
\end{equation}
for all
$
  p \in ( 0, \infty )
$,
$
  \hat{ \alpha }
  \in ( 0, 1 )
$,
$
  \hat{ \beta } 
  \in \R
$
and all stochastic processes
$ 
  X \colon 
  [a, b] \times \Omega \to
  \cap_{ r \in \R } 
  \mathcal{C}^r_{ \mathcal{ P } }( 
    [0, 2 \pi]^d, \R 
  )
$
which satisfy 
for every $ t \in [a, b] $
and every $ x \in [0, 2 \pi]^d $
that
$
  X_t( x )
  \in
  \mathcal{W}_n
$.
Moreover, the Sobolev
embedding theorem 
ensures that there
exist real numbers
$ 
  \rho^{ 
    p, \alpha, \tilde{ \alpha } 
  }_{
    \beta, \hat{ \beta }
  } 
  \in [ 0, \infty )
$,
$ 
  p, \alpha, \tilde{ \alpha },
  \beta, \hat{ \beta } \in \R
$,
and
$
  \bar{ \rho }^{
    p , \tilde{ \alpha }, \hat{ \alpha }
  }
  \in [ 0, \infty )
$,
$
  p, \tilde{ \alpha },
  \hat{ \alpha } \in \R
$,
such that
\begin{equation}
\label{eq:sobolev}
\begin{split}
&
  \left\|
    X
  \right\|_{
    L^p(
      \Omega ;
      C^{ \alpha }(
        [ a, b ], 
        \mathcal{C}^{ 2 \beta  
        }_{ \mathcal{P} }(
          [ 0, 2 \pi ]^d
          , \R
        )
      )
    )
  }
=
  \left\|
    ( - A )^{ \beta }
    X
  \right\|_{
    L^p(
      \Omega ;
      C^{ \alpha }(
        [ a, b ], 
        C_{ \mathcal{P} }(
          [ 0, 2 \pi ]^d
          , \R
        )
      )
    )
  }
\\ & \leq
  \rho^{
    p,
    \alpha, \tilde{ \alpha }
  }_{
    \beta, \hat{ \beta }
  }
  \,
  \big\|
    ( - A )^{ \hat{ \beta } }
    X
  \big\|_{
    L^p(
      \Omega ;
      W^{ \tilde{ \alpha }, p }(
        [ a, b ], 
        L^p(
          ( 0, 2 \pi )^d
          ; \R
        )
      )
    )
  }
\\ & =
  \rho^{
    p,
    \alpha, \tilde{ \alpha }
  }_{
    \beta, \hat{ \beta }
  }
  \,
  \big\|
    ( - A )^{ \hat{ \beta } }
    X
  \big\|_{
    W^{ \tilde{ \alpha }, p }(
      [ a, b ], 
      L^p(
        \Omega ;
        L^p(
          ( 0, 2 \pi )^d
          ; \R
        )
      )
    )
  }
\\ & \leq
  \rho^{
    p,
    \alpha, \tilde{ \alpha }
  }_{
    \beta, \hat{ \beta }
  }
  \,
  \bar{ \rho }^{
    p, \tilde{ \alpha },
    \hat{ \alpha }
  }
  \,
  \big\|
    ( - A )^{ \hat{ \beta } }
    X
  \big\|_{
    C^{ \hat{ \alpha } }(
      [ a, b ], 
      L^p(
        \Omega ;
        L^p(
          ( 0, 2 \pi )^d
          ; \R
        )
      )
    )
  }
\end{split}
\end{equation}
for all stochastic processes
$ 
  X \colon 
  [a, b] \times \Omega \to
  \cap_{ r \in \R } 
  \mathcal{C}^r_{ \mathcal{ P } }( 
    [0, 2 \pi]^d, \R 
  )
$
with continuous 
sample paths
and all
$
  p \in ( 0, \infty )
$,
$
  \alpha, 
  \tilde{ \alpha },
  \hat{ \alpha }
  \in ( 0, 1 )
$,
$
  \beta ,
  \hat{ \beta } 
  \in \R
$
with
$
  \hat{ \alpha }
  >
  \tilde{ \alpha }
$,
$
  \tilde{ \alpha }
  - 
  \alpha
  > \frac{ 1 }{ p }
$
and
$
  \hat{ \beta } - \beta 
  >
  \frac{ d }{ p }
$.
Combining 
\eqref{eq:hypercontractivity}
and \eqref{eq:sobolev}
implies
\eqref{eq:hypercontractivity_main}
and this completes the proof
of Lemma~\ref{lem:hypercontractivity}.
\end{proof}

%
%
%
%
%
%

\subsection{Estimates for discrete
convolutions}

We first state three
well known lemmas that
we will use below.

\begin{lem}[Finiteness of 
infinite sums]
\label{lem:sum_finiteness}
Let $ d \in \N $,
$ 
  \alpha \in \R
$
and let
$ \lambda_x \in [1, \infty) $,
$ x \in \R^d $, 
be real numbers
with 
$
  \lambda_x = 1 + 
  \left( x_1 \right)^2 +
  \ldots + \left( x_d \right)^2
$
for all 
$ x = ( x_1, \dots, x_d ) \in \R^d $.
Then
$
  \sum_{
    \substack{ 
      k \in \Z^d
    }
  }
  \frac{  
    1
  }{
    \left( 
      \lambda_k 
    \right)^{ \alpha }
  }
  < \infty
$
if and only if 
$ 
  \alpha > \frac{ d }{ 2 }
$.
\end{lem}

\begin{lem}[Growth
rate of finite sums]
\label{lem:smallvalues}
Let $ d \in \N $,
$ 
  \alpha \in [ 0, \frac{ d }{ 2 } ) 
$,
$
  \beta \in
  \R
$,
$ 
  c \in ( 0, \infty ) 
$
and let
$ \lambda_x \in [1, \infty) $,
$ x \in \R^d $, 
be real numbers
with 
$
  \lambda_x = 1 + 
  \left( x_1 \right)^2 +
  \ldots + \left( x_d \right)^2
$
for all 
$ x = ( x_1, \dots, x_d ) \in \R^d $.
Then
$
  \sup_{ v \in \Z^d }
  \left[
  \sum_{
    \substack{ 
      k \in \Z^d ,
      \| k \|_{ \R^d } \leq 
      c \| v \|_{ \R^d }
    }
  }
  \frac{  
    \left(
      \lambda_v
    \right)^{ \beta }
  }{
    \left( 
      \lambda_k 
    \right)^{ \alpha }
  }
  \right]
  < \infty
$
if and only if 
$ 
  \beta \leq \alpha - \frac{ d }{ 2 } 
$.
\end{lem}

\begin{lem}[Growth rate
of infinite sums]
\label{lem:largevalues}
Let $ d \in \N $,
$ 
  \alpha \in ( \frac{ d }{ 2 } , \infty ) 
$,
$
  \beta \in
  \R
$,
$ 
  c \in ( 0, \infty ) 
$
and let
$ \lambda_x \in [1, \infty) $,
$ x \in \R^d $, 
be real numbers
with 
$
  \lambda_x = 1 + 
  \left( x_1 \right)^2 +
  \ldots + \left( x_d \right)^2
$
for all 
$ x = ( x_1, \dots, x_d ) \in \R^d $.
Then
$
  \sup_{ v \in \Z^d }
  \left[
  \sum_{
    \substack{ 
      k \in \Z^d ,
      \| k \|_{ \R^d } >
      c \| v \|_{ \R^d }
    }
  }
  \frac{  
    \left(
      \lambda_v
    \right)^{ \beta }
  }{
    \left( 
      \lambda_k 
    \right)^{ \alpha }
  }
  \right]
  < \infty
$
if and only if 
$ 
  \beta \leq \alpha - \frac{ d }{ 2 } 
$.
\end{lem}

Lemmas~\ref{lem:sum_finiteness}--\ref{lem:largevalues}
can all be proved
by estimating the sums through
suitable Lebesgue integrals 
and then by 
using polar coordinates.
The proofs of
Lemmas~\ref{lem:sum_finiteness}--\ref{lem:largevalues}
are straightforward and well known
and therefore
omitted.

\begin{lem}[Two-sided bounds
for discrete convolutions]
\label{lem:twosided_dc}
Let $ d \in \N $
and let
$ \lambda_x \in [1, \infty) $,
$ x \in \R^d $, 
be real numbers
with 
$
  \lambda_x = 1 + 
  \left( x_1 \right)^2 +
  \ldots + \left( x_d \right)^2
$
for all 
$ x = ( x_1, \dots, x_d ) \in \R^d $.
Then
\begin{equation}
\label{eq:twosided_first}
  \frac{
    4^{ - \beta }
  }{
    \left(
      \lambda_{ v } 
    \right)^{ \beta }
  }
  \left[
  \sum_{
    \substack{ 
      k \in \Z^d ,
    \\
      \| k \|_{ \R^d } \leq 
    \\
      \frac{ 1 }{ 2 } \| v \|_{ \R^d }
    }
  }
  \frac{  
    1
  }{
    \left( 
      \lambda_k 
    \right)^{ \alpha }
  }
  \right]
\leq
  \sum_{
    \substack{ 
      k \in \Z^d ,
    \\
      \| k \|_{ \R^d } \leq 
    \\
      \frac{ 1 }{ 2 } \| v \|_{ \R^d }
    }
  }
  \frac{ 1 }{
    \left( 
      \lambda_k 
    \right)^{ \alpha }
    \left(
      \lambda_{ v - k } 
    \right)^{ \beta }
  }
\leq
  \frac{
    4^{ \beta } 
  }{
    \left(
      \lambda_{ v } 
    \right)^{ \beta }
  }
  \left[
  \sum_{
    \substack{ 
      k \in \Z^d ,
    \\
      \| k \|_{ \R^d } \leq 
    \\
      \frac{ 1 }{ 2 } \| v \|_{ \R^d }
    }
  }
  \frac{ 
    1
  }{
    \left( 
      \lambda_k 
    \right)^{ \alpha }
  }
  \right] ,
\end{equation}
\begin{equation}
\label{eq:twosided_second}
  \frac{ 
    4^{ - \alpha }
  }{
    \left( 
      \lambda_v
    \right)^{ \alpha }
  }
  \left[
  \sum_{
    \substack{
      k \in \Z^d ,
    \\
      \| k \|_{ \R^d } 
      \leq 
    \\
      \frac{ 
        1 
      }{ 3 } 
      \| v \|_{ \R^d }
    }
  }
  \frac{ 
    1
  }{
    \left(
      \lambda_k 
    \right)^{ \beta }
  }
  \right]
\leq
  \sum_{
    \substack{
      k \in \Z^d ,
      \frac{ 
        1
      }{ 2 } 
        \| v \|_{ \R^d }
      <
    \\
      \| k \|_{ \R^d } 
      \leq 
      2 
        \| v \|_{ \R^d }
    }
  }
  \frac{ 1 }{
    \left( 
      \lambda_k 
    \right)^{ \alpha }
    \left(
      \lambda_{ v - k } 
    \right)^{ \beta }
  }
\leq
  \frac{
    4^{ \alpha }
  }{
    \left( 
      \lambda_v 
    \right)^{ \alpha }
  }
  \left[
  \sum_{
    \substack{
      k \in \Z^d ,
    \\
      \| k \|_{ \R^d } 
      \leq 
    \\
      3
      \| v \|_{ \R^d }
    }
  }
  \frac{ 
    1
  }{
    \left(
      \lambda_k 
    \right)^{ \beta }
  }
  \right]
  ,
\end{equation}
\begin{equation}
\label{eq:twosided_third}
  4^{ - \beta }
  \left[
  \sum_{ 
    \substack{
      k \in \Z^d ,
    \\
      \left\| k \right\|_{ \R^d }
      > 
    \\
      2
      \left\| v \right\|_{ \R^d }
    }
  }
  \frac{ 
    1
  }{
    \left( 
      \lambda_k 
    \right)^{ ( \alpha + \beta) }
  }
  \right]
\leq
  \sum_{ 
    \substack{
      k \in \Z^d ,
    \\
      \left\| k \right\|_{ \R^d }
      >
    \\
      2 
      \left\| v \right\|_{ \R^d }
    }
  }
  \frac{ 1 }{
    \left( 
      \lambda_k 
    \right)^{ \alpha }
    \left(
      \lambda_{ v - k } 
    \right)^{ \beta }
  }
\leq
  4^{ \beta }
  \left[
  \sum_{ 
    \substack{
      k \in \Z^d ,
    \\
      \left\| k \right\|_{ \R^d }
      >
    \\
      2 
      \left\| v \right\|_{ \R^d }
    }
  }
  \frac{ 
    1
  }{
    \left( 
      \lambda_k
    \right)^{ ( \alpha + \beta ) }
  }
  \right]
\end{equation}
for all $ v \in \Z^d $
and all
$ \alpha, \beta \in [0,\infty) $.
\end{lem}

\begin{proof}[Proof
of Lemma~\ref{lem:twosided_dc}]
First of all, observe that
\begin{equation}
\label{eq:small}
\begin{split}
&
  \sum_{
    \substack{ 
      k \in \Z^d ,
    \\
      \| k \|_{ \R^d } \leq 
    \\
      \frac{ 1 }{ 2 } \| v \|_{ \R^d }
    }
  }
  \frac{  
    3^{ - \beta }
  }{
    \left( 
      \lambda_k 
    \right)^{ \alpha }
    \left(
      \lambda_{ v } 
    \right)^{ \beta }
  }
\leq
  \sum_{
    \substack{ 
      k \in \Z^d ,
    \\
      \| k \|_{ \R^d } \leq 
    \\
      \frac{ 1 }{ 2 } \| v \|_{ \R^d }
    }
  }
  \frac{ 1 }{
    \left( 
      \lambda_k 
    \right)^{ \alpha }
    \left(
      \frac{ 9 }{ 4 }
      \lambda_{ v } 
    \right)^{ \beta }
  }
\\ & \leq
  \sum_{
    \substack{ 
      k \in \Z^d ,
    \\
      \| k \|_{ \R^d } \leq 
    \\
      \frac{ 1 }{ 2 } \| v \|_{ \R^d }
    }
  }
  \frac{ 1 }{
    \left( 
      \lambda_k 
    \right)^{ \alpha }
    \left(
      \lambda_{ v - k } 
    \right)^{ \beta }
  }
\leq
  \sum_{
    \substack{ 
      k \in \Z^d ,
    \\
      \| k \|_{ \R^d } \leq 
    \\
      \frac{ 1 }{ 2 } \| v \|_{ \R^d }
    }
  }
  \frac{ 1 }{
    \left( 
      \lambda_k 
    \right)^{ \alpha }
    \left(
      \frac{
        \lambda_{ v } 
      }{
        4
      }
    \right)^{ \beta }
  }
\leq
  \sum_{
    \substack{ 
      k \in \Z^d ,
    \\
      \| k \|_{ \R^d } \leq 
    \\
      \frac{ 1 }{ 2 } \| v \|_{ \R^d }
    }
  }
  \frac{ 
    4^{ \beta } 
  }{
    \left( 
      \lambda_k 
    \right)^{ \alpha }
    \left(
      \lambda_{ v } 
    \right)^{ \beta }
  }
\end{split}
\end{equation}
for all
$ v \in \R^d $.
This proves \eqref{eq:twosided_first}.
Furthermore, note that
\begin{equation}
\label{eq:middle}
\begin{split}
&
  \sum_{
    \substack{
      k \in \Z^d ,
    \\
      \| k \|_{ \R^d } 
      < 
    \\
      \frac{ 
        1 
      }{ 2 } 
      \| v \|_{ \R^d }
    }
  }
  \frac{ 
    4^{ - \alpha }
  }{
    \left(
      \lambda_k 
    \right)^{ \beta }
    \left( 
      \lambda_v
    \right)^{ \alpha }
  }
\leq
  \sum_{
    \substack{
      k \in \Z^d ,
      \frac{ 
        1
      }{ 2 } 
      \| v \|_{ \R^d }
      <
    \\
      \| v - k \|_{ \R^d } 
      \leq 
      \frac{ 
        3 
      }{ 2 } 
        \| v \|_{ \R^d }
    }
  }
  \frac{ 
    4^{ - \alpha }
  }{
    \left(
      \lambda_k 
    \right)^{ \beta }
    \left( 
      \lambda_v
    \right)^{ \alpha }
  }
  =
  \sum_{
    \substack{
      k \in \Z^d ,
      \frac{ 
        1
      }{ 2 } 
      \| v \|_{ \R^d }
      <
    \\
      \| k \|_{ \R^d } 
      \leq 
      2
      \| v \|_{ \R^d }
    }
  }
  \frac{ 1 }{
    \left( 
      4
      \lambda_v
    \right)^{ \alpha }
    \left(
      \lambda_{ v - k } 
    \right)^{ \beta }
  }
\\ & \leq
  \sum_{
    \substack{
      k \in \Z^d ,
      \frac{ 
        1
      }{ 2 } 
      \| v \|_{ \R^d }
      <
    \\
      \| k \|_{ \R^d } 
      \leq 
      2 
      \| v \|_{ \R^d }
    }
  }
  \frac{ 1 }{
    \left( 
      \lambda_k 
    \right)^{ \alpha }
    \left(
      \lambda_{ v - k } 
    \right)^{ \beta }
  }
\leq
  \sum_{
    \substack{
      k \in \Z^d ,
      \frac{ 
        1
      }{ 2 } 
      \| v \|_{ \R^d }
      <
    \\
      \| k \|_{ \R^d } 
      \leq 
      2 
      \| v \|_{ \R^d }
    }
  }
  \frac{ 1 }{
    \left( 
      \frac{ \lambda_v }{ 4 }
    \right)^{ \alpha }
    \left(
      \lambda_{ v - k } 
    \right)^{ \beta }
  }
\\ & =
  \sum_{
    \substack{
      k \in \Z^d ,
      \frac{ 
        1
      }{ 2 } 
      \| v \|_{ \R^d }
      <
    \\
      \| v - k \|_{ \R^d } 
      \leq 
      2 
      \| v \|_{ \R^d }
    }
  }
  \frac{ 
    4^{ \alpha } 
  }{
    \left(
      \lambda_k 
    \right)^{ \beta }
    \left( 
      \lambda_v 
    \right)^{ \alpha }
  }
\leq
  \sum_{
    \substack{
      k \in \Z^d ,
    \\
      \| k \|_{ \R^d } 
      \leq 
    \\
      3
      \| v \|_{ \R^d }
    }
  }
  \frac{ 
    4^{ \alpha } 
  }{
    \left(
      \lambda_k 
    \right)^{ \beta }
    \left( 
      \lambda_v 
    \right)^{ \alpha }
  }
\end{split}
\end{equation}
for all $ v \in \Z^d $.
This establishes
\eqref{eq:twosided_second}.
Finally, observe that
\begin{equation}
\label{eq:large}
\begin{split}
&
  \sum_{ 
    \substack{
      k \in \Z^d ,
    \\
      \left\| k \right\|_{ \R^d }
      > 
    \\
      2 
      \left\| v \right\|_{ \R^d }
    }
  }
  \frac{ 3^{ - \beta } }{
    \left( 
      \lambda_k 
    \right)^{ ( \alpha + \beta) }
  }
\leq
  \sum_{ 
    \substack{
      k \in \Z^d ,
    \\
      \left\| k \right\|_{ \R^d }
      > 
    \\
      2 
      \left\| v \right\|_{ \R^d }
    }
  }
  \frac{ 1 }{
    \left( 
      \lambda_k 
    \right)^{ \alpha }
    \left[
      1 +
      \big[
        \left\| k \right\|_{ \R^d }
        +
        \left\| v \right\|_{ \R^d }
      \big]^2
    \right]^{ \beta }
  }
\\ & \leq
  \sum_{ 
    \substack{
      k \in \Z^d ,
    \\
      \left\| k \right\|_{ \R^d }
      >
    \\
      2 
      \left\| v \right\|_{ \R^d }
    }
  }
  \frac{ 1 }{
    \left( 
      \lambda_k 
    \right)^{ \alpha }
    \left(
      \lambda_{ v - k } 
    \right)^{ \beta }
  }
\leq
  \sum_{ 
    \substack{
      k \in \Z^d ,
    \\
      \left\| k \right\|_{ \R^d }
      >
    \\
      2 
      \left\| v \right\|_{ \R^d }
    }
  }
  \tfrac{ 1 }{
    \left( 
      \lambda_k 
    \right)^{ \alpha }
    \left(
      1 +
      \left[
        \left\| k \right\|_{ \R^d }
        -
        \left\| v \right\|_{ \R^d }
      \right]^2
    \right)^{ \beta }
  }
\\ & \leq
  \sum_{ 
    \substack{
      k \in \Z^d ,
    \\
      \left\| k \right\|_{ \R^d }
      >
    \\
      2 
      \left\| v \right\|_{ \R^d }
    }
  }
  \frac{ 1 }{
    \left( 
      \lambda_k
    \right)^{ \alpha }
    \left(
      \frac{ \lambda_{ k } }{ 4 } 
    \right)^{ \beta }
  }
=
  \sum_{ 
    \substack{
      k \in \Z^d ,
    \\
      \left\| k \right\|_{ \R^d }
      >
    \\
      2 
      \left\| v \right\|_{ \R^d }
    }
  }
  \frac{ 
    4^{ \beta } 
  }{
    \left( 
      \lambda_k
    \right)^{ ( \alpha + \beta ) }
  }
\end{split}
\end{equation}
for all
$ v \in \R^d $.
This shows 
\eqref{eq:twosided_third}.
The proof
of Lemma~\ref{lem:twosided_dc}
is thus completed.
\end{proof}

The next elementary 
lemma,
Lemma~\ref{lem:discrete_conv_finite}, 
is a direct consequence
of  
Lemma~\ref{lem:sum_finiteness}
and of
\eqref{eq:twosided_third}
in Lemma~\ref{lem:twosided_dc}.
The proof of 
Lemma~\ref{lem:discrete_conv_finite}
is clear and therefore omitted.

\begin{lem}[Finiteness of 
discrete convolutions]
\label{lem:discrete_conv_finite}
Let $ d \in \N $,
$ \alpha, \beta \in [0,\infty) $,
$ v \in \Z^d $
and let
$ \lambda_x \in [1, \infty) $,
$ x \in \R^d $, 
be real numbers
with 
$
  \lambda_x = 1 + 
  \left( x_1 \right)^2 +
  \ldots + \left( x_d \right)^2
$
for all $ x = ( x_1, \dots, x_d ) \in \R^d $.
Then
$
  \sum_{ k \in \Z^d }
  \frac{ 1 }{
    \left( 
      \lambda_k 
    \right)^{ \alpha }
    \left(
      \lambda_{ v - k } 
    \right)^{ \beta }
  }
  < \infty
$
if and only if 
$
  \alpha + \beta > \frac{ d }{ 2 } 
$.
\end{lem}

The next lemma,
Lemma~\ref{lem:discrete_conv},
follows from
Lemmas~\ref{lem:smallvalues},
\ref{lem:largevalues}
and \ref{lem:twosided_dc}.

\begin{lem}[Regularity
of discrete convolutions]
\label{lem:discrete_conv}
Let $ d \in \N $, 
$ 
  \alpha, \beta, \gamma \in [ 0, \infty ) 
$
be real numbers with
$
  \alpha + \beta 
  > \frac{ d }{ 2 } 
  \neq \max( \alpha, \beta )
$
and let
$ \lambda_x \in [1, \infty) $,
$ x \in \R^d $, 
be real numbers
with
$
  \lambda_x = 1 + 
  \left( x_1 \right)^2 +
  \ldots + \left( x_d \right)^2
$
for all $ x = ( x_1, \dots, x_d ) \in \R^d $.
Then 
\begin{equation}
  \sup_{ v \in \Z^d }
  \left[
  \sum_{ k \in \Z^d }
  \frac{ 
    \left( \lambda_v 
    \right)^{ 
      \gamma
    }
  }{
    \left( 
      \lambda_k 
    \right)^{ \alpha }
    \left(
      \lambda_{ v - k } 
    \right)^{ \beta }
  }
  \right]
  < \infty
\end{equation}
if and only if
$
  \gamma \leq
  \min( 
    \alpha, \beta, 
    \alpha + \beta - \frac{ d }{ 2 }
  )
$.
\end{lem}

\begin{proof}[Proof
of 
Lemma~\ref{lem:discrete_conv}]
Note that
$
  \sum_{ k \in \Z^d }
  \frac{ 
    \left( \lambda_v 
    \right)^{ 
      \gamma
    }
  }{
    \left( 
      \lambda_k 
    \right)^{ \alpha }
    \left(
      \lambda_{ v - k } 
    \right)^{ \beta }
  }
  =
  \sum_{ k \in \Z^d }
  \frac{ 
    \left( \lambda_v 
    \right)^{ 
      \gamma
    }
  }{
    \left( 
      \lambda_k 
    \right)^{ \beta }
    \left(
      \lambda_{ v - k } 
    \right)^{ \alpha }
  }
$
for all $ v \in \Z^d $.
W.l.o.g.\ we assume that
$ \alpha \leq \beta $.
This ensures that
$ \beta \neq \frac{ d }{ 2 } $.
Moreover,
Lemma~\ref{lem:smallvalues}
and \eqref{eq:twosided_first}
in Lemma~\ref{lem:twosided_dc}
prove that
\begin{equation}
\label{eq:aquivalence_first}
\begin{split}
&
  \left(
  \sup_{ v \in \Z^d }
  \left[
  \sum_{
    \substack{ 
      k \in \Z^d ,
      \| k \|_{ \R^d } \leq 
      \frac{ 1 }{ 2 } \| v \|_{ \R^d }
    }
  }
  \frac{ 
    \left(
      \lambda_v
    \right)^{ \gamma }
  }{
    \left( 
      \lambda_k 
    \right)^{ \alpha }
    \left(
      \lambda_{ v - k } 
    \right)^{ \beta }
  }
  \right]
  < \infty
  \right)
\\ &
  \Leftrightarrow
  \Bigg(
  \bigg[
    \Big(
      \gamma \leq 
      \alpha + \beta - \tfrac{ d }{ 2 }
    \Big)
    \wedge
    \Big(
      \alpha < \tfrac{ d }{ 2 }
    \Big)
  \bigg]
  \vee
  \bigg[
    \Big(
      \gamma <
      \beta 
    \Big)
    \wedge
    \Big(
      \alpha = \tfrac{ d }{ 2 }
    \Big)
  \bigg]
  \vee
  \bigg[
    \Big(
      \gamma \leq 
      \beta 
    \Big)
    \wedge
    \Big(
      \alpha > \tfrac{ d }{ 2 }
    \Big)
  \bigg]
  \Bigg) 
\\ & 
  \Leftarrow
  \Big(
    \gamma \leq
    \min( 
      \alpha, 
      \beta, 
      \alpha + \beta - \tfrac{ d }{ 2 }
    )
  \Big)
  .
\end{split}
\end{equation}
In addition,
Lemma~\ref{lem:smallvalues}
and \eqref{eq:twosided_second}
in Lemma~\ref{lem:twosided_dc}
show that
\begin{equation}
\label{eq:aquivalence_second}
\begin{split}
&
  \left(
  \sup_{ v \in \Z^d }
  \left[
  \sum_{
    \substack{
      k \in \Z^d ,
      \frac{ 
        1
      }{ 2 } 
      \| v \|_{ \R^d }
      <
      \| k \|_{ \R^d } 
      \leq 
      2 
      \| v \|_{ \R^d }
    }
  }
  \frac{ 
    \left(
      \lambda_v
    \right)^{ \gamma }
  }{
    \left( 
      \lambda_k 
    \right)^{ \alpha }
    \left(
      \lambda_{ v - k } 
    \right)^{ \beta }
  }
  \right]
  < \infty
  \right)
\\ &
  \Leftrightarrow
  \!
  \Bigg(
  \bigg[
    \Big(
      \gamma \leq 
      \alpha + \beta - \tfrac{ d }{ 2 }
    \Big)
    \wedge
    \Big(
      \beta < \tfrac{ d }{ 2 }
    \Big)
  \bigg]
  \vee
  \bigg[
    \Big(
      \gamma \leq 
      \alpha 
    \Big)
    \wedge
    \Big(
      \beta > \tfrac{ d }{ 2 }
    \Big)
  \bigg]
  \Bigg) 
\\ &
  \Leftrightarrow
  \!
  \Bigg(
  \bigg[
    \Big(
      \gamma \leq 
      \min( 
        \alpha , 
        \alpha + \beta - \tfrac{ d }{ 2 } 
      )
    \Big)
    \wedge
    \Big(
      \beta < \tfrac{ d }{ 2 }
    \Big)
  \bigg]
  \vee
  \bigg[
    \Big(
      \gamma \leq 
      \min( 
        \alpha , 
        \alpha + \beta - \tfrac{ d }{ 2 } 
      )
    \Big)
    \wedge
    \Big(
      \beta > \tfrac{ d }{ 2 }
    \Big)
  \bigg]
  \Bigg) 
\\ &
  \Leftrightarrow
    \Big(
      \gamma \leq 
      \min( 
        \alpha , 
        \alpha + \beta - \tfrac{ d }{ 2 } 
      )
    \Big) 
  \Leftrightarrow
    \Big(
      \gamma \leq 
      \min( 
        \alpha , 
        \beta ,
        \alpha + \beta - \tfrac{ d }{ 2 } 
      )
    \Big) .
\end{split}
\end{equation}
Finally, 
Lemma~\ref{lem:largevalues}
and \eqref{eq:twosided_third}
in Lemma~\ref{lem:twosided_dc}
prove that
\begin{equation}
\label{eq:aquivalence_third}
\begin{split}
&
  \left(
  \sup_{ v \in \Z^d }
  \left[
  \sum_{
    \substack{
      k \in \Z^d ,
      \| k \|_{ \R^d } 
      > 
      2 
      \| v \|_{ \R^d }
    }
  }
  \frac{ 
    \left(
      \lambda_v
    \right)^{ \gamma }
  }{
    \left( 
      \lambda_k 
    \right)^{ \alpha }
    \left(
      \lambda_{ v - k } 
    \right)^{ \beta }
  }
  \right]
  < \infty
  \right)
  \Leftrightarrow
    \Big(
      \gamma \leq 
      \alpha + \beta - \tfrac{ d }{ 2 }
    \Big)
  .
\end{split}
\end{equation}
Combining 
\eqref{eq:aquivalence_first}--\eqref{eq:aquivalence_third}
completes the proof
of Lemma~\ref{lem:discrete_conv}.
\end{proof}

\begin{cor}[Regularity
of discrete convolutions]
\label{cor:discrete_conv}
Let $ d \in \N $, 
$ 
  \alpha, \beta \in [ 0, \infty ) 
$
and let
$ \lambda_x \in [1, \infty) $,
$ x \in \R^d $, 
be real numbers
with
$
  \lambda_x = 1 + 
  \left( x_1 \right)^2 +
  \ldots + \left( x_d \right)^2
$
for all $ x = ( x_1, \dots, x_d ) \in \R^d $.
Then 
$
  \sup_{ v \in \Z^d }
  \left[
  \sum_{ k \in \Z^d }
  \frac{ 
    \left( \lambda_v 
    \right)^{ 
      \gamma
    }
  }{
    \left( 
      \lambda_k 
    \right)^{ \alpha }
    \left(
      \lambda_{ v - k } 
    \right)^{ \beta }
  }
  \right]
  < \infty
$
for all
$
  \gamma 
  \in
  \big[ 0,
    \min( 
      \alpha, \beta, 
      \alpha + \beta - \frac{ d }{ 2 }
    )
  \big)
$.
\end{cor}

%
%
%

\subsection{Wick powers 
of Ornstein-Uhlenbeck processes}


The next elementary lemma is, e.g., 
similar to
Lemma~2.4 in 
Da Prato \&
Tubaro~\cite{DaPratoTubaro2007}
and
Corollary 8.3.2 in Glimm \& Jaffe~\cite{glimm_quantum_1987}.

\begin{lem}[Expectations of products
of Wick powers of Gaussian random variables]
\label{lem:Generalized_WickTheorem}
Assume the setting of 
Subsection~\ref{sec:setting},
let $ m \in \N $
and let
$
  Z = ( Z_1, \dots, Z_m )
  \colon \Omega
  \to \R^m
$
be a centered
jointly normally distributed random variable.
Then
\begin{equation}
  \E\Big[
    \!
    \left(
      : \!
      ( Z_1 )^{ n_1 }
      \! :
    \right)
    \cdot
    \left(
      : \!
      ( Z_2 )^{ n_2 }
      \! :
    \right)
    \cdot
    \ldots
    \cdot
    \left(
      : \!
      ( Z_m )^{ n_m }
      \! :
    \right)
    \!
  \Big]
  =
  \!\!\!\!
  \sum_{
    \substack{
      \alpha \in
      ( \N_0 )^{
        \mathcal{P}_m
      }
    \\
      \Theta( \alpha ) =n
    }
  }
  \frac{
    n !
  }{
    \alpha !
  }
  \left[
  \prod_{ 
    ( i, j ) \in \mathcal{P}_m 
  }  
    \big(
      \E\big[
        Z_i
        Z_j
      \big]
    \big)^{ 
      \alpha_{ ( i, j ) } 
    }
  \right]
\end{equation}
for all 
$ n = (n_1, \dots, n_m) \in ( \N_0 )^m $.
\end{lem}
\begin{proof}[Proof
of 
Lemma~\ref{lem:Generalized_WickTheorem}]
W.l.o.g.\ we assume that
$
  \E\big[
    ( Z_i )^2
  \big]
  > 0
$
for all
$
  i \in \{ 1, 2, \dots, m \} 
$.
Next throughout this proof 
let 
$ 
  \hat{ Z }_i \colon \Omega \to \R 
$,
$ 
  i \in \{ 1, 2, \dots, m \} 
$,
be random variables defined through
\begin{equation}
  \hat{ Z }_i
  :=
  \frac{
    Z_i
  }{
    \big(
      \E\big[
        ( Z_i )^2
      \big]
    \big)^{ 1 / 2 }
  }
\end{equation}
for all $ i \in \{ 1, 2, \dots, m \} $.
The definition of the 
Hermite polynomials $ H_n $, $ n \in \{ 0, 1, 2, \dots \} $,
then proves that
\begin{equation}
\begin{split}
&
  \sum_{ 
    n = (n_1, n_2, \dots, n_m) \in \N_0 
  }
  \frac{ 
    \left( s_1 
    \right)^{ n_1 }
    \cdot
    \left( s_2
    \right)^{ n_2 }
    \cdot
    \ldots
    \cdot
    \left( s_m
    \right)^{ n_m }
    \cdot
    \E\big[
       \prod_{ i = 1 }^m
       H_{ n_i }\!\big(
        \hat{ Z }_i
      \big)
    \big]
  }{
    n_1 ! \, n_2 ! \,
    \dots
    \, n_m!
  } 
\\ & =
  \E\!\left[
    \prod_{ i = 1 }^m
    \left(
      \sum_{
        n_i = 0
      }^{ \infty }
      \frac{
        \left( s_i \right)^{ n_i }
      }{
        n_i !
      }
       H_{ n_i }\!\big(
        \hat{ Z }_k
      \big)
    \right)
  \right]
  =
  \E\!\left[
    \prod_{ i = 1 }^m
    \exp\!\left(
      -
      \frac{ \left( s_i \right)^2 }{ 2 }
      +
      s_i 
      \hat{Z}_i
    \right)
  \right]
\\ & =
  \exp\!\left(
    \frac{ 
      - 
      \sum_{ i = 1 }^{ m } 
      \left( s_i \right)^2 
    }{ 2 }
  \right)
  \,
  \E\!\left[
    \exp\!\left(
      \sum_{ i = 1 }^m 
      s_i 
      \hat{Z}_i
    \right)
  \right]
\\ & =
  \exp\!\left(
    \frac{ 
      - 
      \sum_{ i = 1 }^{ m } 
      \left( s_i \right)^2 
    }{ 
      2 
    }
    +
    \frac{ 1 }{ 2 }
    \,
    \E\!\left[
      \left(
        \sum_{ i = 1 }^m 
        s_i
        \hat{Z}_i
      \right)^{ \! 2 }
    \right]
  \right)
\\ &  
  =
  \exp\!\left(
    \tfrac{ 
      - \sum_{ i = 1 }^{ m } \left( s_i \right)^2 
      +
      \sum_{ i, j = 1 }^m 
      s_i s_j
      \E\left[
        \hat{Z}_i
        \hat{Z}_j
      \right]
    }{ 2 }
  \right)
  =
  \prod_{
    \substack{
      i, j \in \{ 1, 2, \dots, m \}
    \\ 
      i < j
    }
  }
  \exp\!\left(
    s_i
    s_j
    \E\!\left[
      \hat{Z}_{ i }
      \hat{Z}_{ j }
    \right]
  \right)
\end{split}
\end{equation}
and the identity
$ 
  e^{ s_1 s_2 c }
  = 
  \sum_{ n = 0 }^{ \infty }
  \frac{ 
    ( s_1 s_2 )^n
    c^n
  }{
    n!
  }
$
for all $ s_1, s_2, c \in \R $
therefore shows that
\begin{equation}
\begin{split}
&
  \sum_{ 
    n = (n_1, n_2, \dots, n_m) \in \N_0 
  }
  \frac{ 
    \left( s_1 
    \right)^{ n_1 }
    \cdot
    \left( s_2
    \right)^{ n_2 }
    \cdot
    \ldots
    \cdot
    \left( s_m
    \right)^{ n_m }
    \cdot
    \E\big[
       \prod_{ i = 1 }^m
       H_{ n_i }\!\big(
        \hat{ Z }_i
      \big)
    \big]
  }{
    n !
  } 
\\ & =
  \prod_{
    (i, j) \in \mathcal{P}_m
  }
  \left(
  \sum_{ 
    \alpha_{ (i, j) } = 0
  }^{ \infty }
  \frac{
    \left(
      s_i
      s_j
    \right)^{ 
      \alpha_{ (i, j) } 
    }
    \left\{
      \E\!\left[
        \hat{Z}_{ i }
        \hat{Z}_{ j }
      \right]
    \right\}^{
      \alpha_{ (i, j) } 
    }
  }{
    \alpha_{ (i, j) }!
  }
  \right) 
\\ & =
  \sum_{ 
    \alpha \in ( \N_0 )^{ \mathcal{P}_m }
  }
  \frac{ 
    1
  }{
    \alpha !
  }
  \left(
    \prod_{
      (i, j) \in \mathcal{P}_m
    }
    \left(
      s_i
      s_j
    \right)^{ 
      \alpha_{ (i, j) } 
    }
    \left\{
      \E\!\left[
        \hat{Z}_{ i }
        \hat{Z}_{ j }
      \right]
    \right\}^{
      \alpha_{ (i, j) } 
    }
  \right)
  .
\end{split}
\end{equation}
This implies
\begin{equation}
\begin{split}
&
  \frac{ 1 }{ n! }
  \,
    \E\!\left[
     \prod_{ i = 1 }^{ m }
       H_{ n_i }\!\big(
        \hat{Z}_i
      \big)
    \right]
  =
  \sum_{
    \substack{
      \alpha \in ( \N_0 )^{ \mathcal{P}_m }
    \\
      \Theta( \alpha ) = n 
    }
  }
  \frac{
    1
  }{
    \alpha !
  }
  \left[
  \prod_{ 
    ( i, j ) \in \mathcal{P}_m 
  }  
    \left\{
      \E\!\left[
        \hat{Z}_i
        \hat{Z}_j
      \right]
    \right\}^{ 
      \alpha_{ (i, j ) }
    }
  \right]
\end{split}
\end{equation}
and hence
\begin{equation}
\begin{split}
&
  \E\!\left[
    \prod_{ i = 1 }^{ m }
    \big(
      \!
      : \!
      (
        Z_i
      )^n
      \! :
      \!
    \big)
  \right]
  =
  \left\{
    \E\!\left[
      ( Z_1 )^2
    \right]
  \right\}^{
    \frac{ n_1 }{ 2 }
  }
  \cdot
  \ldots
  \cdot
  \left\{
    \E\!\left[
      ( Z_m )^2
    \right]
  \right\}^{
    \frac{ n_m }{ 2 }
  }  
\\ & \quad \cdot
  \sum_{
    \substack{
      \alpha \in ( \N_0 )^{ \mathcal{P}_m }
    \\
      \Theta( \alpha ) = n 
    }
  }
  \frac{
    n !
  }{
    \alpha !
  }
  \left[
  \prod_{ 
    ( i, j ) \in \mathcal{P}_m 
  }  
  \left\{
    \frac{
      \E\!\left[
        Z_i
        Z_j
      \right]
    }{
      \sqrt{
        \E\!\left[
          ( Z_i )^2
        \right]
        \E\!\left[
          ( Z_j )^2
        \right]
      }
    }
  \right\}^{  
    \!
    \alpha_{ (i, j ) }
  }
  \right]
\end{split}
\end{equation}
for all $ n = (n_1, \dots, n_m) \in ( \N_0 )^m $.
The definition of the function
$ 
  \Theta 
  \colon
  \cup_{ m = 1 }^{ \infty }
  ( \N_0 )^{
    \mathcal{P}_m
  }
  \to
  \cup_{ m = 1 }^{ \infty }
  ( \N_0 )^m
$ 
therefore completes the 
proof of Lemma~\ref{lem:Generalized_WickTheorem}.
\end{proof}

\begin{remark}[Wick's theorem]
Assume the setting of 
Subsection~\ref{sec:setting},
let $ m \in \N $
and let
$
  Z = ( Z_1, \dots, Z_m )
  \colon \Omega
  \to \R^m
$
be a centered jointly normally distributed
random variable.
Then Lemma~\ref{lem:Generalized_WickTheorem}
implies that
\begin{equation}
\label{eq:wick_theorem}
  \E\big[
    Z_1 
    \cdot
    Z_2 
    \cdot
    \ldots
    \cdot
    Z_m 
  \big]
  =
  \!\!\!\!
  \sum_{
    \substack{
      \alpha \in 
      ( \N_0 )^{ 
        \mathcal{P}_m
      }
    \\
      \Theta( \alpha ) = (1,1,\dots,1)
    }
  }
  n !
  \left[
  \prod_{ 
    ( i, j ) \in \mathcal{P}_m 
  }  
    \big(
      \E\big[
        Z_i
        Z_j
      \big]
    \big)^{ 
      \alpha_{ ( i, j) }
    }
  \right] .
\end{equation}
Equation~\eqref{eq:wick_theorem}
is often referred as \emph{Wick's theorem}
in the literature
(see, e.g., Proposition~5.2
in Hairer~\cite{Hairer2012}).
\end{remark}

The next lemma is a direct
consequence
of Lemma~\ref{lem:Generalized_WickTheorem}.

\begin{cor}[Products of 
Wick powers of 
$ V^{ \varphi } $, 
$ \varphi \in \Phi_0 $,
in real space]
\label{cor:space_correlation}
Assume the setting of 
Subsection~\ref{sec:setting},
let 
$ m \in \N $
and let
$ 
  n = (n_1, n_2, \dots, n_m) \in 
  ( \N_0 )^m
  \backslash \{ 0 \} 
$.
Then
\begin{equation}
\begin{split}
&
  \E\bigg[
    \Big(
    \!
    : \! 
    \big( 
      V_{ t_1 }^{ \varphi^{ (1) } } 
    \big)^{ n_1 } 
    \!\! :
    \! \Big) 
    (x_1) 
    \cdot
    \Big(
    \!
    : \! 
    \big( 
      V_{ t_2 }^{ \varphi^{ (2) } } 
    \big)^{ n_2 } 
    \!\! : 
    \! \Big) 
    (x_2) 
    \cdot
    \ldots
    \cdot
    \Big(
    \!
    : \! 
    \big( 
      V_{ t_m }^{ \varphi^{ (m) } } 
    \big)^{ n_m } 
    \!\! :
    \! \Big) 
    (x_m) 
  \bigg]
\\ & =
  \sum_{
    \substack{
      \alpha 
      \in ( \N_0 )^{ 
        \mathcal{P}_m
      }
    \\
      \Theta( \alpha ) = n
    }
  }
  \frac{ 
    n!
  }{
    \alpha !
  }
  \left[
  \prod_{ ( i, j ) \in \mathcal{P}_m }  
  \left[
    \sum_{
      k \in \Z^d
    }
    \frac{
      \varphi^{ (i) }_k \,
      \varphi^{ (j) }_k \,
      g_k( x_i - x_j ) \,
      e^{ 
        - \lambda_k \left| t_i - t_j \right|
      }
    }{
      \lambda_k
    }
  \right]^{ 
    \alpha_{ 
      ( i, j ) 
    }
  }
  \right]
\\ & =
  \sum_{
    \substack{
      \alpha \in
      ( \N_0 )^{
        \mathcal{P}_m
      }
    \\
      \Theta( \alpha ) =n
    }
  }
  \tfrac{
    n !
  }{
    \alpha !
  }
  \!\!\!\!\!\!\!\!\!\!\!\!
    \sum_{
      k \in 
      ( \Z^d )^{
        \left\{ 
          \substack{
            ( A, l ) \in 
          \\
            \mathcal{P}_m \times \N \colon
          \\
            l \leq \alpha_{ A } 
          }
        \right\}
      }
    }
    \prod_{
      ( i, j, r )
      \in 
      \left\{ 
        \substack{
          ( A, l ) \in 
        \\
          \mathcal{P}_m \times \N \colon
        \\
          l \leq \alpha_{ A } 
        }
      \right\}
    }
      \tfrac{ 
        \varphi^{ (i) }_{ k_{ (i, j, r) } }
        \,
        \varphi^{ (j) }_{ k_{ (i, j, r) } }
        \,
        e^{ 
          - \lambda_{ k_{ (i, j, r) } }
          \left|
            t_i - t_j
          \right|
        }
        \,
        g_{ k_{ (i, j, r) } }( x_i - x_j )
      }{ 
        \lambda_{ k_{ ( i, j, r) } }
      }
\end{split}
\end{equation}
for all 
$ t_1, t_2, \dots, t_m \in \R $,
$ 
  x_1, x_2, \dots, x_m \in 
  [
    0, 2 \pi
  ]^d
$
and all
$ 
  \varphi^{ (1) } 
  = 
  ( \varphi^{ (1) }_k )_{ k \in \Z^d }
$,
$
  \varphi^{ (2) }
  ( \varphi^{ (2) }_k )_{ k \in \Z^d }
$,
$  
  \dots
$,
$
  \varphi^{ (m) } =
  ( \varphi^{ (m) }_k )_{ k \in \Z^d }
  \in \Phi_0 
$.
\end{cor}

\begin{proof}[Proof
of 
Corollary~\ref{cor:space_correlation}]
Combining
Lemma~\ref{lem:Generalized_WickTheorem}
and equation~\eqref{eq:space_correlation}
implies that
\begin{equation}
\begin{split}
&
  \E\bigg[
    \Big(
    \!
    : \! 
    \big( 
      V_{ t_1 }^{ \varphi^{ (1) } } 
    \big)^{ n_1 } 
    \!\! :
    \! \Big) 
    (x_1) 
    \cdot
    \Big(
    \!
    : \! 
    \big( 
      V_{ t_2 }^{ \varphi^{ (2) } } 
    \big)^{ n_2 } 
    \!\! :
    \! \Big) 
    (x_2) 
    \cdot
    \ldots
    \cdot
    \Big(
    \!
    : \! 
    \big( 
      V_{ t_m }^{ \varphi^{ (m) } } 
    \big)^{ n_m } 
    \!\! :
    \! \Big) 
    (x_m) 
  \bigg]
\\ & =
  \E\!\left[
    \left(
    : \! 
    \left( 
      V_{ t_1 }^{ \varphi^{ (1) } }( x_1 ) 
    \right)^{ \! n_1 } 
    \!\! :
    \right) 
    \cdot
    \left(
    : \! 
    \left( 
      V_{ t_2 }^{ \varphi^{ (2) } }( x_2 ) 
    \right)^{ \! n_2 } 
    \!\! :
    \right) 
    \cdot
    \ldots
    \cdot
    \left(
    : \! 
    \left( 
      V_{ t_{ n_m } 
      }^{ \varphi^{ ( n_m ) } }( x_{ n_m } ) 
    \right)^{ \! n_m } 
    \!\! :
    \right) 
  \right]
\\ & =
  \sum_{
    \substack{
      \alpha \in
      ( \N_0 )^{
        \mathcal{P}_m
      }
    \\
      \Theta( \alpha ) =n
    }
  }
  \frac{
    n !
  }{
    \alpha !
  }
  \left[
  \prod_{ 
    ( i, j ) \in \mathcal{P}_m 
  }  
    \left(
      \E\!\left[
        V^{ \varphi^{ (i) } }_{
          t_{ n_i }
        }\!( x_{ n_i } )
        \,
        V^{ \varphi^{ (j) } }_{
          t_{ n_j }
        }\!( x_{ n_j } )
      \right]
    \right)^{ 
      \alpha_{ ( i, j ) } 
    }
  \right]
\\ & =
  \sum_{
    \substack{
      \alpha \in
      ( \N_0 )^{
        \mathcal{P}_m
      }
    \\
      \Theta( \alpha ) =n
    }
  }
  \frac{
    n !
  }{
    \alpha !
  }
  \left(
    \prod_{ 
      ( i, j ) \in \mathcal{P}_m 
    }  
    \left[
      \sum_{
        k \in \Z^d
      }
      \frac{ 
        \varphi^{ (i) }_k
        \,
        \varphi^{ (j) }_k
        \,
        e^{ 
          - \lambda_k
          \left|
            t_i - t_j
          \right|
        }
        \,
        g_k( x_i - x_j )
      }{ 
        \lambda_k
      }
    \right]^{ 
      \! \alpha_{ ( i, j ) } 
    }
  \right)
\end{split}
\end{equation}
and therefore
\begin{equation}
\begin{split}
&
  \E\bigg[
    \Big(
    \!
    : \! 
    \big( 
      V_{ t_1 }^{ \varphi^{ (1) } } 
    \big)^{ n_1 } 
    \!\! :
    \! \Big) 
    (x_1) 
    \cdot
    \Big(
    \!
    : \! 
    \big( 
      V_{ t_2 }^{ \varphi^{ (2) } } 
    \big)^{ n_2 } 
    \!\! :
    \! \Big) 
    (x_2) 
    \cdot
    \ldots
    \cdot
    \Big(
    \!
    : \! 
    \big( 
      V_{ t_m }^{ \varphi^{ (m) } } 
    \big)^{ n_m } 
    \!\! :
    \! \Big) 
    (x_m) 
  \bigg]
\\ & =
  \sum_{
    \substack{
      \alpha \in
      ( \N_0 )^{
        \mathcal{P}_m
      }
    \\
      \Theta( \alpha ) =n
    }
  }
  \frac{
    n !
  }{
    \alpha !
  }
  \left(
    \prod_{ 
      ( i, j ) \in \mathcal{P}_m 
    }  
    \left[
  \sum_{
    \substack{
      k_1, k_2, 
      \dots, 
    \\
      k_{
        \alpha_{ (i, j) }
      }
      \in \Z^d
    }
  }
    \prod_{
      l = 1
    }^{
      \alpha_{ (i, j) }
    }
    \left\{
      \frac{ 
        \varphi^{ (i) }_{ k_l }
        \,
        \varphi^{ (j) }_{ k_l }
        \,
        e^{ 
          - \lambda_{ k_l }
          \left|
            t_i - t_j
          \right|
        }
        \,
        g_{ k_l }( x_i - x_j )
      }{ 
        \lambda_{ k_l }
      }
      \right\}
    \right]
  \right)
\\ & =
  \sum_{
    \substack{
      \alpha \in
      ( \N_0 )^{
        \mathcal{P}_m
      }
    \\
      \Theta( \alpha ) =n
    }
  }
  \tfrac{
    n !
  }{
    \alpha !
  }
  \!\!\!\!\!\!\!\!\!\!\!\!
    \sum_{
      k \in 
      ( \Z^d )^{
        \left\{ 
          \substack{
            ( A, l ) \in 
          \\
            \mathcal{P}_m \times \N \colon
          \\
            l \leq \alpha_{ A } 
          }
        \right\}
      }
    }
    \prod_{
      ( i, j, r )
      \in 
      \left\{ 
        \substack{
          ( A, l ) \in 
        \\
          \mathcal{P}_m \times \N \colon
        \\
          l \leq \alpha_{ A } 
        }
      \right\}
    }
      \tfrac{ 
        \varphi^{ (i) }_{ k_{ (i, j, r) } }
        \,
        \varphi^{ (j) }_{ k_{ (i, j, r) } }
        \,
        e^{ 
          - \lambda_{ k_{ (i, j, r) } }
          \left|
            t_i - t_j
          \right|
        }
        \,
        g_{ k_{ (i, j, r) } }( x_i - x_j )
      }{ 
        \lambda_{ k_{ ( i, j, r) } }
      }
\end{split}
\end{equation}
for all 
$ t_1, t_2, \dots, t_m \in \R $,
$ 
  x_1, x_2, \dots, x_m \in 
  [
    0, 2 \pi
  ]^d
$
and all
$ 
  \varphi^{ (1) },  
  \varphi^{ (2) },
  \dots,
  \varphi^{ (m) }
  \in \Phi_0 
$.
The proof of
Corollary~\ref{cor:space_correlation}
is thus completed.
\end{proof}

In the special
case $ m = 2 $,
Corollary~\ref{cor:space_correlation} 
reduces to
the following result.

\begin{cor}[Correlation of Wick powers
of 
$ V^{ \varphi } $, $ \varphi \in \Phi_0 $, 
in real space]
\label{cor:space_correlation2}
Assume the setting of 
Subsection~\ref{sec:setting}.
Then
\begin{equation}
\begin{split}
&
  \E\bigg[
    \Big(
    \!
    : \! 
    \big( 
      V_{ t_1 }^{ \varphi^{ (1) } } 
    \big)^{ n_1 } 
    \!\! :
    \! \Big) 
    (x_1) 
    \cdot
    \Big(
    \!
    : \! 
    \big( 
      V_{ t_2 }^{ \varphi^{ (2) } } 
    \big)^{ n_2 } 
    \!\! :
    \! \Big) 
    (x_2) 
  \bigg]
\\ & =
    n_1 !
    \,
    \delta_{ n_1, n_2 }
    \left[
      \sum_{
        k \in \Z^d
      }
      \frac{
        \varphi^{ (1) }_k \,
        \varphi^{ (2) }_k \,
        g_k( x_1 - x_2 ) \,
        e^{ 
          - \lambda_k 
          \left| t_1 - t_2 \right|
        }
      }{
        \lambda_k
      }
    \right]^{ 
      n_1
    }
\\ & =
\begin{cases}
    n_1 !
    \,
    \delta_{ n_1, n_2 }
    \!
    \left[
    \sum_{
      k_1, \dots, k_{ n_1 }
      \in 
      \Z^d 
    }
    \prod_{ r = 1 }^{
      n_1
    }
      \frac{ 
        \varphi^{ (1) }_{ k_r }
        \,
        \varphi^{ (2) }_{ k_r }
        \,
        e^{ 
          - \lambda_{ k_r }
          \left|
            t_1 - t_2
          \right|
        }
        \,
        g_{ k_r }( x_1 - x_2 )
      }{ 
        \lambda_{ k_r }
      }
    \right]
&
  \colon
  n_1 \cdot n_2
  \neq 0
\\
  \delta_{ n_1, n_2 }
&
  \colon
  n_1 \cdot n_2 = 0
\end{cases}
\end{split}
\end{equation}
for all 
$ 
  t_1, t_2
  \in \R $,
$ 
  x_1, x_2 
  \in
  [
    0, 2 \pi
  ]^d
$,
$ 
  \varphi^{ (1) },  
  \varphi^{ (2) }
  \in \Phi_0 
$
and all
$ 
  n_1, n_2  
  \in \N_0 
$.
\end{cor}

Corollary~\ref{cor:space_correlation2}
investigates 
correlations of Wick powers
of 
$ V^{ \varphi } $, 
$ \varphi \in \Phi_0 $,
in real space.
The next lemma studies 
correlations of Wick powers
of 
$ V^{ \varphi } $, 
$ \varphi \in \Phi_0 $,
in Fourier space.
Its proof makes use of
Corollary~\ref{cor:space_correlation2}.

\begin{lem}[Correlation of
Wick powers of 
$ V^{ \varphi } $, $ \varphi \in \Phi_0 $, 
in
Fourier space]
\label{lem:correlation}
Assume the setting of 
Subsection~\ref{sec:setting}.
Then
\begin{equation}
\label{eq:lem_correlation_main}
\begin{split} 
& 
  \frac{ 1 }{
    \left( 2 \pi \right)^{ 2 d }
  }
  \,
  \E\Bigg[
    \overline{
    \left<
      g_{ k_1 }
      ,
      : \! 
      \big( 
        V_{ t_1 }^{ \varphi^{ (1) } }
      \big)^{ n_1 }
      \!\! :
    \right>_{ \! H }
    }
    \left<
      g_{ k_2 } 
      ,
      : \!
      \big( 
        V_{ t_2 
        }^{
          \varphi^{ (2) }
        }
      \big)^{ n_2 }
      \!\! :
    \right>_{ \! H }
  \Bigg]
\\ & =
\begin{cases} 
  n_1 !
  \,
  \delta_{ n_1, n_2 }
  \,
  \delta_{ k_1, k_2 }
  \!
  \left[
  \sum_{
    \substack{
      l_{ 1 },
      \dots,
      l_{ n_1 }
      \in
      \Z^d
    \\
      l_1 
      +
      \ldots
      +
      l_{ n_1 } =
      k_1
    }
  }
  \left\{
  \prod_{ i = 1 }^{ n_1 }
  \frac{
    \varphi^{ (1) }_{
      l_i
    }
    \,
    \varphi^{ (2) }_{
      l_i
    }
    \,
    e^{
      - \lambda_{ l_i }
      | t_1 - t_2 |
    }
  }{
    \lambda_{ l_i }
  }
  \right\}
  \right]
&
  \colon 
  n_1 \cdot n_2
  \neq 0
\\
  \delta_{ n_1, n_2 }
  \,
  \delta_{ k_1, k_2 }
  \,
  \delta_{ k_1, 0 }
&
  \colon 
  n_1 \cdot n_2 = 0
\end{cases}
\end{split}
\end{equation}
for all 
$
  t_1, t_2 
  \in \R
$,
$ 
  k_1, k_2 \in \mathbb{Z}^d 
$,
$ 
  \varphi^{ (1) }, \varphi^{(2)} \in \Phi_0 
$ 
and all
$ 
  n_1, n_2 
  \in \N_0  
$.
\end{lem}

\begin{proof}[Proof
of Lemma~\ref{lem:correlation}]
First of all, observe that
\begin{equation}
\label{eq:lem_correlation}
\begin{split} 
& 
  \mathbb{E}\!
  \left[
    \overline{
    \left<
      g_{ k_1 } ,
      :  \!
      \big(
        V_{ t_1
        }^{
          \varphi^{ (1) }
        }
      \big)^{ n_1 }
      \!\! :
    \right>_{ \! H }
    }
    \left<
      g_{ k_2 }  
      ,
      : \!
      \big(
        V_{ t_2
        }^{
          \varphi^{ (2) }
        }
      \big)^{ n_2 }
      \!\! :
    \right>_{ \! H }
  \right]
\\ & = 
  \int_{
    ( 0, 2 \pi )^d
  }
  \int_{
    ( 0, 2 \pi )^d
  }
  \mathbb{E}\!
  \left[
    \overline{
    \left(
      : \!
      \big(
        V_{ t_1
        }^{
          \varphi^{ (1) }
        }
      \big)^{ n_1 
      }
    \!\! :
  \right)
  \!
  ( x_1 )
  }
  \cdot
  \left(
    : \!
    \big(
      V_{ t_2
      }^{
        \varphi^{ (2) 
        }
      }
    \big)^{
      n_2
    }
    \!\! :
  \right) \!
  (x_{2})
  \right]
  g_{ - k_1 }( x_1 )
  \,
  g_{ k_2 }( x_2 )
  \, dx_1 \, dx_2
\end{split}
\end{equation}
for all 
$
  t_1, t_2 
  \in \R
$,
$ 
  k_1, k_2 \in \mathbb{Z}^d 
$,
$ 
  \varphi^{ (1) }, \varphi^{(2)} \in \Phi_0 
$ 
and all
$ 
  n_1, n_2  
  \in \N_0 
$.
Equation~\eqref{eq:lem_correlation}
and
Corollary~\ref{cor:space_correlation2}
prove
\eqref{eq:lem_correlation_main}
in the case
$ 
  ( n_1, n_2 ) 
  \in ( \N_0 )^2 
  \backslash 
  \{ 
    (k, k) \in \N^2
    \colon
    k \in \N
  \}
$.
Furthermore,
equation~\eqref{eq:lem_correlation},
Corollary~\ref{cor:space_correlation2}
and the integral transformation
theorem imply that
\begin{equation}
\begin{split} 
& 
  \mathbb{E}\!
  \left[
    \overline{
    \left<
      g_{ k_1 } ,
      : \!  
      \big(
        V_{ t_1
        }^{
          \varphi^{ (1) }
        }
      \big)^{ n }
      \!\! :
    \right>_{ \! H }
    }
    \left<
      g_{ k_2 }  
      ,
      : \!
      \big(
        V_{ t_2
        }^{
          \varphi^{ (2) }
        }
      \big)^{ n }
      \!\! :
    \right>_{ \! H }
  \right]
\\ & = 
    n !
  \int_{
    ( 0, 2 \pi )^d
  }
  \int_{
    ( 0, 2 \pi )^d
  }
  \!\!
    \left[
    \sum_{  
      \substack{
        v_1, 
      \\
        \dots, 
      \\
        v_{ n }
      \\
        \in 
        \Z^d 
      }
    }
      \tfrac{ 
    \prod_{ r = 1 }^{
      n
    }
        \left[
        \varphi^{ (1) }_{ v_r }
        \,
        \varphi^{ (2) }_{ v_r }
        \,
        e^{ 
          - \lambda_{ v_r }
          \left|
            t_1 - t_2
          \right|
        }
        \,
        g_{ v_r }( x_1 - x_2 )
        \right]
      }{ 
        \lambda_{ v_1 }
        \cdot \ldots \cdot
        \lambda_{ v_n }
      }
    \right]
  g_{ - k_1 }( x_1 )
  \,
  g_{ k_2 }( x_2 )
  \, dx_1 \, dx_2
\\ & = 
      n !
  \int_{
    ( 0, 2 \pi )^d
  }
  \int_{
    ( 0, 2 \pi )^d - x_2
  }
   \left[
    \sum_{  
      \substack{
        v_1, \dots, 
      \\
        v_n
        \in 
        \Z^d 
      }
    }
      \tfrac{ 
    \prod_{ r = 1 }^{
      n
    }
        \left[
        \varphi^{ (1) }_{ v_r }
        \,
        \varphi^{ (2) }_{ v_r }
        \,
        e^{ 
          - \lambda_{ v_r }
          \left|
            t_1 - t_2
          \right|
        }
        \,
        g_{ v_r }( y )
        \right]
      }{ 
        \lambda_{ v_1 }
        \cdot \ldots \cdot
        \lambda_{ v_n }
      }
    \right]
  g_{ - k_1 }( y + x_2 )
  \,
  g_{ k_2 }( x_2 )
  \, dy \, dx_2
\\ & = 
  n !
  \int_{
    ( 0, 2 \pi )^d
  }
  \!\!
  \left[
    \sum_{  
      \substack{
        v_1, \dots, 
      \\
        v_n
        \in 
        \Z^d 
      }
    }
  \int_{
    ( 0, 2 \pi )^d - x_2
  }
      \tfrac{ 
    \prod_{ r = 1 }^{
      n
    }
        \left[
        \varphi^{ (1) }_{ v_r }
        \,
        \varphi^{ (2) }_{ v_r }
        \,
        e^{ 
          - \lambda_{ v_r }
          \left|
            t_1 - t_2
          \right|
        }
        \,
        g_{ v_r }( y )
        \right]
      }{ 
        \lambda_{ v_1 }
        \cdot \ldots \cdot
        \lambda_{ v_n }
      }
  \,
  g_{ - k_1 }( y )
  \,
  dy
  \right]
  \!
  g_{ ( k_2 - k_1 ) }( x_2 )
  \, dx_2
\\ & =
    \delta_{ k_1, k_2 }
    n !
    \left( 2 \pi \right)^d
  \left[
    \sum_{  
      \substack{
        v_1, \dots, 
        v_n
        \in 
        \Z^d 
      }
    }
  \int_{
    ( 0, 2 \pi )^d
  }
      \tfrac{ 
    \prod_{ r = 1 }^{
      n
    }
        \left[
          \varphi^{ (1) }_{ v_r }
          \,
          \varphi^{ (2) }_{ v_r }
          \,
        e^{ 
          - \lambda_{ v_r }
          \left|
            t_1 - t_2
          \right|
        }
        \,
        g_{ v_r }( y )
        \right]
      }{ 
        \lambda_{ v_1 }
        \cdot \ldots \cdot
        \lambda_{ v_n }
      }
  \,
  g_{ - k_1 }( y )
  \,
  dy
  \right]
\end{split}
\end{equation}
for all 
$
  t_1, t_2 
  \in \R
$,
$ 
  k_1, k_2 \in \mathbb{Z}^d 
$,
$ 
  \varphi^{ (1) }, \varphi^{(2)} \in \Phi_0 
$ 
and all
$ n \in \N $.
This shows that 
\begin{equation}
\begin{split} 
& 
  \mathbb{E}\!
  \left[
    \overline{
    \left<
      g_{ k_1 } ,
      : \! 
      \big(
        V_{ t_1
        }^{
          \varphi^{ (1) }
        }
      \big)^{ n }
      \!\! :
    \right>_{ \! H }
    }
    \left<
      g_{ k_2 }  
      ,
      : \!
      \big(
        V_{ t_2
        }^{
          \varphi^{ (2) }
        }
      \big)^{ n }
      \!\! :
    \right>_{ \! H }
  \right]
\\ & = 
    \delta_{ k_1, k_2 }
    n !
    \left( 2 \pi \right)^{
      d
    }
  \left[
    \sum_{  
      \substack{
        v_1, \dots, 
        v_n
        \in 
        \Z^d 
      }
    }
  \int_{
    ( 0, 2 \pi )^d
  }
      \tfrac{ 
    \prod_{ r = 1 }^{
      n
    }
        \left[
          \varphi^{ (1) }_{ v_r }
          \,
          \varphi^{ (2) }_{ v_r }
          \,
        e^{ 
          - \lambda_{ v_r }
          \left|
            t_1 - t_2
          \right|
        }
        \,
        g_{ v_r }( y )
        \right]
      }{ 
        \lambda_{ v_1 }
        \cdot \ldots \cdot
        \lambda_{ v_n }
      }
  \,
  g_{ - k_1 }( y )
  \,
  dy
  \right]
\\ & = 
    \delta_{ k_1, k_2 }
    n!
    \left( 2 \pi \right)^{
      d
    }
  \!\!\!
  \sum_{
    v_1, \dots, v_n \in \Z^d
  }
  \int_{
    ( 0, 2 \pi )^d 
  }
      \tfrac{
        \left[
          \prod_{ r = 1 }^n
          \varphi^{ (1) }_{ v_r } \,
          \varphi^{ (2) }_{ v_r } 
        \right]
        \,
        g_{ 
          ( 
            - k_1
            +
            \sum_{ r = 1 }^n 
            v_r
          ) 
        }( y ) \,
        e^{ 
          - 
          \left[
            \sum_{ r = 1 }^n
            \lambda_{ v_r }
          \right]
          \left| t_1 - t_2 \right|
        }
      }{
        \lambda_{ v_1 }
        \cdot
        \ldots
        \cdot
        \lambda_{ v_n }
      }
  \,
  dy
\\ & = 
    \delta_{ k_1, k_2 }
    n!
    \left( 2 \pi \right)^{
      2 d
    }
  \!\!\!
  \sum_{
    \substack{
      v_1, \dots, v_n \in \Z^d
    \\
      v_1 + \ldots + v_n = k_1
    }
  }
  \prod_{ i = 1 }^n
  \left[
      \frac{
        \varphi^{ (1) }_{ v_i } \,
        \varphi^{ (2) }_{ v_i } \,
        e^{ 
          - 
          \lambda_{ v_i }
          \left| t_1 - t_2 \right|
        }
      }{
        \lambda_{ v_i }
      }
  \right]
\end{split}
\end{equation}
for all 
$
  t_1, t_2 
  \in \R
$,
$ 
  k_1, k_2 \in \mathbb{Z}^d 
$,
$ 
  \varphi^{ (1) }, \varphi^{(2)} \in \Phi_0 
$ 
and all
$ n \in \N $.
The proof of Lemma~\ref{lem:correlation}
is thus completed.
\end{proof}

The next result,
Proposition~\ref{prop:regularity_Wick},
proves convergence of
Wick powers in the case
$
  (n, d)
  \in
  (
    \{ 2, 3, \dots \}
    \times \{ 2 \}
  )
  \cup
  \{
    ( 2, 3 )
  \}
$.
The proof of
Proposition~\ref{prop:regularity_Wick}
makes use of
Lemma~\ref{lem:correlation}.

\begin{prop}[Convergence of Wick
powers]
\label{prop:regularity_Wick}
Assume the setting of 
Subsection~\ref{sec:setting}
and let 
$
  (n, d)
  \in
  \left(
    \{ 2, 3, \dots \}
    \times \{ 2 \}
  \right)
  \cup
  \left\{
    ( 2, 3 )
  \right\}
$.
Then there exists
an up to 
indistinguishability unique 
stochastic process
$ 
  : \!\! ( V )^n \!\! : \;
  \colon
  \R \times \Omega
  \to 
  \cap_{
    \beta \in ( - \infty , 
      2 - d 
    )
  }
  \mathcal{C}^{
    \beta
  }_{
    \mathcal{P}
  }(
    [0, 2 \pi]^d,
    \R
  )
$
with continuous sample paths which
satisfies
for every $ T, p \in (0,\infty) $,
$ \alpha \in ( 0, \frac{ 4 - d }{ 4 } ) $,
$ \beta \in \R $
with 
$ 
  2 \alpha + \beta < 
  2 - d
$
that
\begin{equation}
  \left\|
    : \! 
    ( 
      V^{ \varphi } 
    )^n
    \! :
    -
    : \!
    ( 
      V
    )^n
    \! :
  \right\|_{
    L^p(
      \Omega;
      C^{ \alpha }
      (
        [ - T , T ],
        \mathcal{C}_{ \mathcal{P} 
        }^{ \beta  
        }(
          [0, 2 \pi]^d,
          \mathbb{R})
      )
    )
  }
  \rightarrow 0
\qquad
  \text{as}
\qquad
    \Phi_{ 0, \leq 1 } 
    \ni \varphi \rightarrow 1 
  .
\end{equation}
\end{prop}

\begin{proof}[Proof
of 
Proposition~\ref{prop:regularity_Wick}]
We apply 
Lemma~\ref{lem:correlation}
four times to obtain that
\begin{equation}
\label{eq:L_inf_part_000}
\begin{split}
&
      \E\!\left[
        \overline{
          \left< 
            g_{ k_1 } ,
            : \!
            \big(
              V_t^{ 
                \varphi 
              }
            \big)^n
            \!\! :
            \,
            -
            \,
            : \! \big(
              V_t^{ \psi }
            \big)^n
            \!\! :
          \right>_{ \! H } 
        }
          \left< 
            g_{ k_2 } ,
            : \!
            \big(
              V_t^{ 
                \varphi
              }
            \big)^n
            \!\! :
            \,
            -
            \,
            : \! \big(
              V_t^{ 
                \psi
              }
            \big)^n
            \!\! :
          \right>_{ \! H } 
      \right]
\\ & = 
  n!
  \left( 2 \pi \right)^{ 2 d }
  \delta_{ k_1, k_2 }
  \left[
    \sum_{
      \substack{
        l_1, \ldots, l_n
        \in \Z^d
      \\
        l_1 + \ldots + l_n 
        = k_1
      }
    }
    \left\{
      \prod_{
        i = 1 
      }^n
      \frac{ 
        \left[
          \varphi_{ l_i }
        \right]^2
      }{
        \lambda_{ l_i }
      }
      - 
      2
      \prod_{
        i = 1 
      }^n
      \frac{ 
        \varphi_{ l_i }
        \psi_{ l_i }
      }{
        \lambda_{ l_i }
      }
      +
      \prod_{
        i = 1 
      }^n
      \frac{ 
        \left[
          \psi_{ l_i }
        \right]^2
      }{
        \lambda_{ l_i }
      }
    \right\}
  \right]
\\ & = 
  n!
  \left( 2 \pi \right)^{ 2 d }
  \delta_{ k_1, k_2 }
  \left[
    \sum_{
      \substack{
        l_1, \ldots, l_n
        \in \Z^d
      \\
        l_1 + \ldots + l_n 
        = k_1
      }
    }
    \frac{
      \left(
        \prod_{ i = 1 }^n
        \varphi_{ l_i } 
        -
        \prod_{ i = 1 }^n
        \psi_{ l_i } 
      \right)^2
    }{
      \left(
        \prod_{ i = 1 }^n
        \lambda_{ l_i } 
      \right)
    }
  \right]
\end{split}
\end{equation}
for all $ t \in \R $,
$ \varphi, \psi \in \Phi_0 $
and all $ k_1, k_2 \in \Z^d $.
Next observe that
\begin{equation}
\label{eq:phipsi_est}
\begin{split} 
& 
  \left|
  \prod_{ i = 1 }^n
  \varphi_{ l_i }
  -
  \prod_{ i = 1 }^n
  \psi_{ l_i }
  \right|
  =
  \left|
  \sum_{
    i = 1
  }^n
    \left[
      \prod_{ j = 1 }^{ i }
      \varphi_{ l_i }
    \right]
    \left[
      \prod_{ j = i + 1 }^n
      \psi_{ l_i }
    \right]
    -
    \left[
      \prod_{ j = 1 }^{ i - 1 }
      \varphi_{ l_i }
    \right]
    \left[
      \prod_{ j = i }^n
      \psi_{ l_i }
    \right]
  \right|
\\ & =
  \sum_{
    i = 1
  }^n
  \underbrace{
    \left[
      \prod_{ j = 1 }^{ i - 1 }
      \varphi_{ l_i }
    \right]
  }_{ \leq 1 }
  \underbrace{
    \left[
      \prod_{ j = i + 1 }^n
      \psi_{ l_i }
    \right]
  }_{ \leq 1 }
  \left|
    \varphi_{ l_i }
    -
    \psi_{ l_i }
  \right|
\leq
  \underbrace{
  \sum_{
    i = 1
  }^n
  \left|
    \varphi_{ l_i }
    -
    \psi_{ l_i }
  \right|
  }_{
    \leq n
  }
\end{split}
\end{equation}
for all $ t \in \R $,
$ l_1, \dots, l_n \in \Z^d $
and all
$ \varphi, \psi \in \Phi_{ 0, \leq 1 } $.
Combining \eqref{eq:L_inf_part_000}
and \eqref{eq:phipsi_est}
implies that
\begin{equation}
\label{eq:L_inf_part_00}
\begin{split} 
&
    \sup_{
      t
      \in
      \R
    }
    \sum_{
      k_1, k_2
      \in
      \Z^d
    }
  \tfrac{
    \left|
      \E\left[
        \overline{
          \left< 
            g_{ k_1 } ,
            : 
            \left(
              V_t^{ 
                \varphi 
              }
            \right)^n
            :
            \,
            -
            \,
            : \left(
              V_t^{ \psi }
            \right)^n
            :
          \right>_{ \! H } 
        }
          \left< 
            g_{ k_2 } ,
            :
            \left(
              V_t^{ 
                \varphi
              }
            \right)^n
            :
            \,
            -
            \,
            : \left(
              V_t^{ 
                \psi
              }
            \right)^n
            :
          \right>_{ \! H } 
      \right]
    \right|
  }{
    \left(
      \lambda_{ k_1 }
      \lambda_{ k_2 }
    \right)^{ - \beta }
  }
\\ & \leq
      n!
      \left( 2 \pi \right)^{ 2 d }
  \left[
    \sum_{
      k
      \in
      \Z^d
    }
    \sum_{
      \substack{
        l_1, \ldots, l_n
        \in \Z^d
      \\
        l_1 + \ldots + l_n 
        = k
      }
    }
    \frac{
      \left(
        \lambda_k
      \right)^{ 2 \beta }
      \left(
        \sum_{ i = 1 }^n
        \left| 
          \varphi_{ l_i } 
          -
          \psi_{ l_i }
        \right|
      \right)^2
    }{
      \left(
        \prod_{ i = 1 }^n
        \lambda_{ l_i } 
      \right)
    }
  \right]
\end{split}
\end{equation}
for all $ \beta \in \R $
and all
$ \varphi, \psi \in \Phi_{ 0, \leq 1 } $.
Morever,
combining the identity
\begin{equation}
\begin{split}
&
  \sum_{
    \substack{
      l_1 , \ldots, l_n \in \Z^d
    \\
      l_1 + \ldots + l_n = k
    }
  }
  \frac{
    1
  }{
    \left(
      \prod_{ i = 1 }^n
      \lambda_{ l_i }
    \right)
  }
=
  \sum_{
    l_1 \in \Z^d
  }
  \frac{ 1 }{ 
    \lambda_{ l_1 } 
  }
  \left[
  \sum_{
    l_2 \in \Z^d
  }
  \frac{ 1 }{ 
    \lambda_{ l_2 } 
  }
  \left[
  \dots
  \left[
  \sum_{
    l_{ n - 1 } \in \Z^d
  }
  \frac{
    1
  }{
    \lambda_{ l_{ n - 1 } }
    \cdot
    \lambda_{
      ( k - l_1 - \ldots - l_{ n - 1 } ) 
    }
  }
  \right]
  \right]
  \right]
\end{split}
\end{equation}
for all $ k \in \Z^d $
with 
Corollary~\ref{cor:discrete_conv}
and 
with the assumption
$
  (n, d)
  \in
  \left(
    \{ 2, 3, \dots \}
    \times \{ 2 \}
  \right)
  \cup
  \left\{
    (2, 3)
  \right\}
$
proves that
$
  \sup_{ k \in \Z^d }
  \left[
  \sum_{
    \substack{
      l_1 , \ldots, l_n \in \Z^d
    \\
      l_1 + \ldots + l_n = k
    }
  }
  \frac{
    \left(
      \lambda_k
    \right)^{ \beta }
  }{
    \left(
      \prod_{ i = 1 }^n
      \lambda_{ l_i }
    \right)
  }
  \right]
  < \infty
$
for all 
$ 
  \beta \in 
  ( - \infty , 2 - \frac{ d }{ 2 } ) 
$.
This implies that
\begin{equation}
\label{eq:summability_WP_2}
  \sum_{ k \in \Z^d }
  \sum_{
    \substack{
      l_1 , \ldots, l_n \in \Z^d
    \\
      l_1 + \ldots + l_n = k
    }
  }
  \frac{
    \left(
      \lambda_k
    \right)^{ 2 \beta }
  }{
    \left(
      \prod_{ i = 1 }^n
      \lambda_{ l_i }
    \right)
  }
  < \infty
\end{equation}
for all 
$ 
  \beta \in 
  ( - \infty , \frac{ 2 - d }{ 2 } ) 
$.
Dominated convergence and
\eqref{eq:L_inf_part_00}
therefore show 
for every
$ 
  \beta \in 
  ( - \infty , \frac{ 2 - d }{ 2 } ) 
$
that
\begin{equation}
\label{eq:WP_first_limit}
\begin{split} 
&
    \sup_{
      t
      \in
      \R
    }
    \sum_{
      k_1, k_2
      \in
      \Z^d
    }
  \tfrac{
    \left|
      \E\left[
        \overline{
          \left< 
            g_{ k_1 } ,
            : 
            \left(
              V_t^{ 
                \varphi 
              }
            \right)^n
            :
            \,
            -
            \,
            : \left(
              V_t^{ \psi }
            \right)^n
            :
          \right>_{ \! H } 
        }
          \left< 
            g_{ k_2 } ,
            :
            \left(
              V_t^{ 
                \varphi
              }
            \right)^n
            :
            \,
            -
            \,
            : \left(
              V_t^{ 
                \psi
              }
            \right)^n
            :
          \right>_{ \! H } 
      \right]
    \right|
  }{
    \left(
      \lambda_{ k_1 }
      \lambda_{ k_2 }
    \right)^{ - \beta }
  }
  \to 0
\quad  
  \text{as}
\quad
      ( \Phi_{ 0, \leq 1 } )^2 
      \ni
      (
        \varphi,
        \psi
      )
      \rightarrow (1,1)
      .
\end{split}
\end{equation}
Next observe that
Lemma~\ref{lem:correlation}
shows that
\begin{equation}
\label{eq:proof_lem_time_diff}
\begin{split} 
& 
          \E\!\left[
          \begin{split}
          &
            \left<
              g_{ - k_1 } 
            ,
              \left[
                : \! (
                  V_{ t_1 }^{
                    \varphi
                  }
                )^n
                \! :
                -
                : \! (
                  V_{ t_1 }^{ 
                    \psi
                  }
                )^n
                \! :
              \right]
              -
              \left[
                : \! (
                  V_{ t_2 }^{
                    \varphi
                  }
                )^n
                \! :
                -
                : \! (
                  V_{ t_2 }^{ 
                    \psi
                  }
                )^n
                \! :
              \right]
            \right>_{ \! H }
          \\ & \;
            \cdot
            \left<
              g_{ k_2 } 
              ,
              \left[
                : \! (
                  V_{ t_1 }^{
                    \varphi
                  }
                )^n
                \! :
                -
                : \! (
                  V_{ t_1 }^{ 
                    \psi
                  }
                )^n
                \! :
              \right]
              -
              \left[
                : \! (
                  V_{ t_2 }^{
                    \varphi
                  }
                )^n
                \! :
                -
                : \! (
                  V_{ t_2 }^{ 
                    \psi
                  }
                )^n
                \! :
              \right]
            \right>_{ \! H }
          \end{split}
          \right]
\\ & =
  n!
  \left( 2 \pi \right)^{ 2 d }
  \delta_{
    k_1, k_2
  }
  \!\!\!\!\!\!\!\!\!
    \sum_{
      \substack{
        l_1, \ldots, l_n
        \in \Z^d
      \\
        l_1 + \ldots + l_n 
        = k_1
      }
    }
  \!\!\!\!\!\!\!\!
    \tfrac{
      \left(
        2
        \prod_{ i = 1 }^n
        \left( \varphi_{ l_i } \right)^2
        -
        4
        \prod_{ i = 1 }^n
        \varphi_{ l_i } 
        \psi_{ l_i } 
        +
        2
        \prod_{ i = 1 }^n
        \left( \psi_{ l_i } \right)^2
      \right)
      \left(
        1 - 
        e^{ 
          - \sum_{ i = 1 }^n \lambda_{ l_i }
          | t_1 - t_2 | 
        }
      \right)
    }{
      \left(
        \prod_{ i = 1 }^n
        \lambda_{ l_i } 
      \right)
    }
\\ & =
  2 \,
  n!
  \left( 2 \pi \right)^{ 2 d }
  \delta_{
    k_1, k_2
  }
  \left[
    \sum_{
      \substack{
        l_1, \ldots, l_n
        \in \Z^d
      \\
        l_1 + \ldots + l_n 
        = k_1
      }
    }
    \frac{
      \left(
        \prod_{ i = 1 }^n
        \varphi_{ l_i } 
        -
        \prod_{ i = 1 }^n
        \psi_{ l_i } 
      \right)^2
      \left(
        1 - 
        e^{ 
          - \sum_{ i = 1 }^n \lambda_{ l_i }
          | t_1 - t_2 | 
        }
      \right)
    }{
      \left(
        \prod_{ i = 1 }^n
        \lambda_{ l_i } 
      \right)
    }
  \right]
\\ & \leq 
  n!
  \left( 2 \pi \right)^{ 4 d }
  \delta_{
    k_1, k_2
  }
  \left[
    \sum_{
      \substack{
        l_1, \ldots, l_n
        \in \Z^d
      \\
        l_1 + \ldots + l_n 
        = k_1
      }
    }
    \frac{
      \left(
        \sum_{ i = 1 }^n
        \left|
        \varphi_{ l_i } 
        -
        \psi_{ l_i } 
        \right|
      \right)^2
      \left(
        \sum_{ i = 1 }^n
        \lambda_{ l_i }
      \right)^{ 2 \alpha }
    }{
      \left(
        \prod_{ i = 1 }^n
        \lambda_{ l_i } 
      \right)
    }
  \right]
  \left|
    t_1 - t_2 
  \right|^{ 2 \alpha }
\\ & \leq
  n!
  \left( 2 \pi \right)^{ 4 d }
  \delta_{
    k_1, k_2
  }
  \left[
    \sum_{
      \substack{
        l_1, \ldots, l_n
        \in \Z^d
      \\
        l_1 + \ldots + l_n 
        = k_1
      }
    }
    \frac{
      \left(
        \sum_{ i = 1 }^n
        \left|
        \varphi_{ l_i } 
        -
        \psi_{ l_i } 
        \right|
      \right)^2
      \left(
        \sum_{ i = 1 }^n
        \left(
          \lambda_{ l_i }
        \right)^{ 2 \alpha }
      \right)
    }{
      \left(
        \prod_{ i = 1 }^n
        \lambda_{ l_i } 
      \right)
    }
  \right]
  \left|
    t_1 - t_2 
  \right|^{ 2 \alpha }
\end{split}
\end{equation}
for all $ t_1, t_2 \in \R $,
$ k_1, k_2 \in \Z^d $,
$ \alpha \in ( 0, \frac{ 1 }{ 2 } ] $
and all
$ \varphi, \psi \in \Phi_{ 0, \leq 1 } $
where we used 
$
  1 - e^{ - x }
  \leq 
  x^{ 2 \alpha }
$ 
for all 
$ 
  \alpha \in [0, \frac{ 1 }{ 2 } ]
$ 
and all
$ 
  x \in [0,\infty)
$
and 
$
  \prod_{ i = 1 }^n
  \varphi_{ l_i }
  -
  \prod_{ i = 1 }^n
  \psi_{ l_i }
  \leq
      \sum_{ i = 1 }^n
      \left|
        \varphi_{ l_i } 
        -
        \psi_{ l_i } 
      \right|
$
for all
$ l_1, \dots, l_n \in \Z^d $
and all
$ \varphi, \psi \in \Phi_{ 0, \leq 1 } $
(cf.\ \eqref{eq:phipsi_est})
in the last but one line of 
\eqref{eq:proof_lem_time_diff}
and where we used
$
  \left(
    \sum_{ i = 1 }^n
    \lambda_{ l_i }
  \right)^{
    \!
    2 \alpha
  }
  \leq
  \sum_{ i = 1 }^n
  \left(
    \lambda_{ l_i }
  \right)^{ 2 \alpha }
$
for all $ \alpha \in [0, \frac{1}{2} ] $
in the last line of 
\eqref{eq:proof_lem_time_diff}.
Moreover, 
Corollary~\ref{cor:discrete_conv}
proves that 
\begin{equation}
\label{eq:55}
  \sum_{
    k \in \Z^d
  }
  \left[
    \sum_{
      \substack{
        l_1, \ldots, l_n
        \in \Z^d
      \\
        l_1 + \ldots + l_n 
        = k_1
      }
    }
    \frac{
      \left(
        \lambda_{ l_n }
      \right)^{ 2 \alpha }
      \left(
        \lambda_k
      \right)^{ 2 \beta }
    }{
      \left(
        \prod_{ i = 1 }^n
        \lambda_{ l_i } 
      \right)
    }
  \right]
  < \infty
\end{equation}
for all 
$ \alpha \in ( 0, \frac{ 4 - d }{ 4 } ) $,
$ \beta \in \R $
with 
$ 
  \alpha + \beta < 
  \frac{ 
    2 - d
  }{
    2
  }
$.
Dominated convergence
and \eqref{eq:proof_lem_time_diff}
therefore show for every 
$ \alpha \in ( 0, \frac{ 4 - d }{ 4 } ) $,
$ \beta \in \R $
with 
$ 
  \alpha + \beta < 
  \frac{ 
    2 - d
  }{
    2
  }
$
that
\begin{equation}
\label{eq:lem_time_diff}
  \sup_{ 
    \substack{
      t_1, t_2 \in \R 
    \\
      t_1 \neq t_2
    }
  }
  \left[
    \!\!\!
    \begin{array}{c}
      \sum\limits_{ 
        k_1, k_2 
        \in \Z^d
      } 
      \frac{
        \tiny{
        \left|
          \E\!\left[
          \begin{split}
          &
            \left<
              g_{ - k_1 } 
            ,
              \left[
                : (
                  V_{ t_1 }^{
                    \varphi
                  }
                )^n
                \! :
                -
                : (
                  V_{ t_1 }^{ 
                    \psi
                  }
                )^n
                \! :
              \right]
              -
              \left[
                : (
                  V_{ t_2 }^{
                    \varphi
                  }
                )^n
                \! :
                -
                : (
                  V_{ t_2 }^{ 
                    \psi
                  }
                )^n
                \! :
              \right]
            \right>_{ \! H }
          \\ & \;
            \cdot
            \left<
              g_{ k_2 } 
              ,
              \left[
                : (
                  V_{ t_1 }^{
                    \varphi
                  }
                )^n
                \! :
                -
                : (
                  V_{ t_1 }^{ 
                    \psi
                  }
                )^n
                \! :
              \right]
              -
              \left[
                : (
                  V_{ t_2 }^{
                    \varphi
                  }
                )^n
                \! :
                -
                : (
                  V_{ t_2 }^{ 
                    \psi
                  }
                )^n
                \! :
              \right]
            \right>_{ \! H }
          \end{split}
          \right]
        \right|
        }
      }{ 
        \left( 
          \lambda_{ k_1 }
          \lambda_{ k_2 }
        \right)^{ - \beta }
        \left|
          t_1 - t_2
        \right|^{ 2 \alpha }
      }
    \end{array}
    \!\!\!
  \right]
  \to 0
\;\;
  \text{as}
\;\;
      ( \Phi_{ 0, \leq 1 } )^2 
      \ni
      (
        \varphi,
        \psi
      )
      \rightarrow (1,1) .
\end{equation}
Combining \eqref{eq:WP_first_limit}
and \eqref{eq:lem_time_diff}
with Lemma~\ref{lem:hypercontractivity}
completes the proof
of 
Proposition~\ref{prop:regularity_Wick}.
\end{proof}

The next proposition is 
well known
in the literature 
(see, for instance,
Da Prato \& 
Zabczyk~\cite{dz92}
for related results and 
references)
and its proof is
therefore omitted.

\begin{prop}[Ornstein-Uhlenbeck processes]
\label{prop:regularity_OE}
Assume the setting of 
Subsection~\ref{sec:setting}
and let 
$ d \in \N $.
Then there exists
an up to 
indistinguishability unique 
stochastic process
$ 
  V
  \colon
  \R \times \Omega
  \to 
  \cap_{
    \beta \in ( - \infty , 
      \frac{ 2 - d }{ 2 }
    )
  }
$
$
  \mathcal{C}^{
    \beta
  }_{
    \mathcal{P}
  }(
    [0, 2 \pi]^d,
    \R
  )
$
with continuous sample paths which
satisfies
for every $ T, p \in (0,\infty) $,
$ \alpha \in ( 0, \frac{ 1 }{ 2 } ) $,
$ \beta \in \R $
with 
$ 
  2 \alpha + \beta < 
  \frac{ 
    2 - d
  }{
    2
  }
$
that
\begin{equation}
  \left\|
      V^{ \varphi } 
    -
      V
  \right\|_{
    L^p(
      \Omega;
      C^{ \alpha }
      (
        [ - T , T ],
        \mathcal{C}_{ \mathcal{P} 
        }^{ \beta  
        }(
          [0, 2 \pi]^d,
          \mathbb{R})
      )
    )
  }
  \rightarrow 0
\qquad
  \text{as}
\qquad
    \Phi_{ 0, \leq 1 } 
    \ni \varphi \to 1 .
\end{equation}
\end{prop}

Proposition~\ref{prop:regularity_Wick}
shows convergence of Wick
powers in the case
$
  (n, d)
  \in
  (
    \{ 2, 3, \dots \}
    \times \{ 2 \}
  )
  \cup
  \{
    ( 2, 3 )
  \}
$.
In the case
$
  (n, d)
  \in
  (
    \{
      3, 4, \dots
    \}
    \times \{ 3 \}
  )
  \cup
  (
    \{ 2, 3, \dots \}
    \times \{ 4, 5, \dots \}
  )
$,
Wick powers do not
converge anymore.
This is the subject of 
the next lemma.
In the case $ d = n = 3 $,
a statement similar to the next lemma
has been formulated in Section~7
in Da Prato \& Tubaro~\cite{DaPratoTubaro2007}.

\begin{lem}[Divergence of
Wick powers]
\label{lem:limitation}
Assume the setting of 
Subsection~\ref{sec:setting},
let 
$ 
  d \in \{ 3, 4, \dots \} 
$, 
$ n \in \{ 2, 3, \dots \}
$
be natural numbers with
$ d + n \geq 6 $
and let
$ 
  C_0, C_1, \dots, C_{ n - 1 } \colon
  \Phi_0
  \to \R
$
be arbitrary functions.
Then it holds for every
$ v \in \Z^d $ and every
$ t \in \R $ that
\begin{equation}
  \E\!
  \left[
    \left|
      \left<
        g_v ,
        \left( 
          V^{ \varphi }_t 
        \right)^n
        -
        \sum_{ k = 0 }^{ n - 1 }
          C_k( \varphi ) 
          \cdot
          \left( 
            V^{ \varphi }_t 
          \right)^k
      \right>_{ \! \! H }
    \right|^2
  \right] \to \infty
  \qquad
  \text{as}
  \qquad
      \Phi_0 \ni \varphi
      \rightarrow 1 .
\end{equation}
\end{lem}

\begin{proof}[Proof of Lemma~\ref{lem:limitation}]
Throughout this proof
let 
$ 
  \hat{C}_0, \hat{C}_1, \dots
  \hat{C}_n \colon \Phi_0
  \to \R
$
be the unique functions satisfying
$
  \hat{C}_0( 0 )
  =
  - C_0( 0 )
$,
$
  \hat{C}_1( 0 )
  =
  - C_1( 0 )
$,
$ \dots $,
$
  \hat{C}_{ n - 1 }( 0 )
  =
  - C_{ n - 1 }( 0 )
$,
$
  \hat{C}_n( 0 ) 
  = 1
$
and
\begin{equation}
\label{eq:lem_limitation_firstdispl}
\begin{split}
&
  x^n
  -
  \sum_{ k = 0 }^{ n - 1 }
  C_k( \varphi )
  \cdot
  x^k
=
  \sum_{ k = 0 }^n
  \hat{C}_k( \varphi )
  \cdot
  \left[
      \sum_{ v \in \Z^d }
      \tfrac{ ( \varphi_v )^2 }{
        \lambda_v
      }
  \right]^{
    \! \frac{ k }{ 2 }
  }
  \cdot
  H_k\!\left(
  \frac{
    x
  }{
    \sqrt{
      \sum_{ v \in \Z^d }
      \frac{ ( \varphi_v )^2 }{
        \lambda_v
      }
    }
  }
  \right)
\end{split}
\end{equation}
for all 
$ x \in \R $,
$ 
  \varphi =
  ( \varphi_v )_{ v \in \Z^d }
  \in \Phi_0 \backslash \{ 0 \} 
$
and all $ t \in \R $.
This ensures that
$ \hat{C}_n( \varphi ) = 1 $
and
\begin{equation}
  \left( 
    V^{ \varphi }_t 
  \right)^n
  -
  \sum_{ k = 0 }^{ n - 1 }
  C_k( \varphi )
  \cdot
  \left(
    V^{ \varphi }_t
  \right)^k
  =
  \sum_{ k = 0 }^n
  \hat{C}_k( \varphi )
  \cdot
  \left(
  : \!
    \left(
      V_t^{ \varphi }
    \right)^k
  \!\! :
  \right)
\end{equation}
for all $ \varphi \in \Phi_0 $
and all $ t \in \R $.
Lemma~\ref{lem:correlation}
hence implies that
\begin{equation}
\label{eq:limitation_main_estimate}
\begin{split}
&
  \E\!\left[
    \left|
      \left<
        g_v ,
        \left( V_t^{ \varphi }
        \right)^n
        -
        \sum_{ k = 0 }^{ n - 1 }
        C_k( \varphi ) \cdot
        ( V_t^{ \varphi })^k
      \right>_{ \!\! H }
    \right|^2
  \right]
  =
  \E\!\left[
    \left|
      \sum_{ k = 0 }^n
      \left<
        g_v ,
        \hat{C}_k( \varphi ) 
        \left(
          : \!
          ( V_t^{ \varphi }
          )^k
          \! :
        \right)
      \right>_{ \! H }
    \right|^2
  \right]
\\ & =
  \sum_{ k, l = 0 }^n
  \hat{C}_k( \varphi ) 
  \cdot
  \hat{C}_l( \varphi ) 
  \cdot
  \E\!\left[
    \overline{
      \left<
        g_v ,
        : \!
          ( V_t^{ \varphi }
          )^k
        \! :
      \right>_{ \! H }
    }
    \left<
      g_v ,
      : \!
        ( V_t^{ \varphi }
        )^l
      \! :
    \right>_{ \! H }
  \right]
\\ & =
  \sum_{ k = 0 }^n
  \left|
    \hat{C}_k( \varphi ) 
  \right|^2
  \E\!\left[
    \left|
      \left<
        g_v ,
          : \!
          ( V_t^{ \varphi }
          )^k
          \! :
      \right>_{ H }
    \right|^2
  \right]
\geq
  \left|
    \hat{C}_n( \varphi ) 
  \right|^2
  \E\!\left[
    \left|
      \left<
        g_v ,
          : \!
          ( V_t^{ \varphi }
          )^n
          \! :
      \right>_{ H }
    \right|^2
  \right]
\\ & =
  n!
  \left( 2 \pi \right)^{ 2 d }
  \left[
  \sum_{
    \substack{
      l_{ 1 },
      \dots,
      l_{ n }
      \in
      \Z^d
    \\
      l_1 
      +
      \ldots
      +
      l_n =
      v
    }
  }
  \left\{
  \prod_{ i = 1 }^n
  \frac{
    \left(
      \varphi_{
        l_i
      }
    \right)^2
  }{
    \lambda_{ l_i }
  }
  \right\}
  \right]
\geq
  \sum_{
    \substack{
      l_{ 1 },
      \dots,
      l_{ n }
      \in
      \Z^d
    \\
      l_1 
      +
      \ldots
      +
      l_n =
      v
    }
  }
  \frac{
    \left(
      \varphi_{
        l_1
      }
    \right)^2
    \cdot
    \ldots
    \cdot
    \left(
      \varphi_{
        l_n
      }
    \right)^2
  }{
    \lambda_{ l_1 }
    \cdot
    \ldots
    \cdot
    \lambda_{ l_n }
  }
\end{split}
\end{equation}
for all $ v \in \Z^d $
and all $ \varphi \in \Phi_0 $.
Next note that the estimate
\begin{equation}
  \sum_{ 
    l_1, l_2 \in \Z^3 
  }
  \frac{ 1 }{
    \lambda_{ l_1 }
    \lambda_{ l_2 }
    \lambda_{ ( v - l_1 - l_2 ) }
  }
  \geq
  \sum_{ 
    l_1, l_2 \in \Z^3 
  }
  \tfrac{ 1 }{
    3
    \left(
      1 + \left\| l_1 \right\|_{ \R^3 }^2
    \right)
    \left(
      1 + \left\| l_2 \right\|_{ \R^3 }^2
    \right)
    \left(
      1 
      + 
      \left\| l_1 \right\|_{ \R^3 }^2
      + 
      \left\| l_2 \right\|_{ \R^3 }^2
      +
      \left\| v \right\|_{ \R^3 }^2
    \right)
  }
  = \infty
\end{equation}
for all $ v \in \Z^d $
together with
the assumptions
$ d \geq 3 $,
$ n \geq 2 $
and 
$
  d + n \geq 6
$
and 
Lemma~\ref{lem:discrete_conv_finite}
implies that
\begin{equation}
  \sum_{
    \substack{
      l_{ 1 },
      \dots,
      l_{ n }
      \in
      \Z^d
    \\
      l_1 
      +
      \ldots
      +
      l_n =
      v
    }
  }
  \frac{
    1
  }{
    \lambda_{ l_1 }
    \cdot
    \ldots
    \cdot
    \lambda_{ l_n }
  }
  =
  \sum_{
    l_{ 1 },
    \dots,
    l_{ n - 1 }
    \in
    \Z^d
  }
  \frac{
    1
  }{
    \lambda_{ l_1 }
    \cdot
    \ldots
    \cdot
    \lambda_{ l_{ n - 1 } }
    \cdot
    \lambda_{ 
      ( v - l_1 - \ldots - l_{ n - 1 } ) 
    }
  }
  = \infty
\end{equation}
for all $ v \in \Z^d $.
Combining this with
\eqref{eq:limitation_main_estimate}
completes the proof of
Lemma~\ref{lem:limitation}.
\end{proof}

\subsection{Averaged Wick powers
of Ornstein-Uhlenbeck processes}

In the previous subsection 
it has been proved
in the case $ d = 3 $ 
that
for every $ t \in \R $
the family
$ 
  : \! 
    ( V^{ \varphi }_t )^3
  \! :
$,
$ 
  \varphi \in \Phi_{ 0, \leq 1 }
$,
\emph{does not converge} as
$
  \Phi_{ 0, \leq 1 }
  \ni
  \varphi 
  \to 1
  \in \Phi_{ 0, \leq 1 }
$
(see Lemma~\ref{lem:limitation}).
In this subsection
we prove 
in the case $ d = 3 $
that for every 
$ 
  (t_0, t) \in 
  \{ (s_0, s) \in \R^2 \colon
  s_0 \leq s \} 
$
the family
$ 
  \circ
    ( V^{ \varphi }_{ t_0, t } )^3
  \circ
  =
  \int_{ t_0 }^t
  : \!
    ( V^{ \varphi }_s )^3
  \! :
  ds
$,
$ 
  \varphi \in \Phi_{ 0, \leq 1 }
$,
\emph{does converge} as
$
  \Phi_{ 0, \leq 1 }
  \ni
  \varphi 
  \to 1
  \in \Phi_{ 0, \leq 1 }
$
(see 
Proposition~\ref{prop:regularity_WickAver}).

\begin{lem}[Correlation of
averaged Wick powers of 
$ V^{ \varphi } $, $ \varphi \in \Phi_0 $, 
in
Fourier space]
\label{lem:correlationAver}
Assume the setting of 
Subsection~\ref{sec:setting}.
Then
\begin{equation}
\begin{split} 
& 
  \frac{
    1
  }{
    \left( 2 \pi \right)^{ 2 d }
  }
  \,
  \E\Bigg[
    \overline{
    \left<
      g_{ k_1 }
      ,
      \circ
      \big( 
        V_{ t_0, t_1 }^{ \varphi^{ (1) } }
      \big)^{n_1}
      \circ 
    \right>_{ \! H }
    }
    \left<
      g_{ k_2 } 
      ,
      \circ
      \big( 
        V_{ t_0, t_2 
        }^{
          \varphi^{ (2) }
        }
      \big)^{n_2}
      \circ
    \right>_{ \! H }
  \Bigg]
\\ & = 
\begin{cases}
  n_1 !
  \, 
  \delta_{ n_1, n_2 }
  \,
  \delta_{ k_1, k_2 }
  \!\!\!\!\!\!\!
  \sum\limits_{
    \substack{
      l_{ 1 },
      \dots,
      l_{ n_1 }
      \in
      \Z^d
    \\
      l_1 
      +
      \ldots
      +
      l_{ n_1 } =
      k_1
    }
  }
  \!\!\!\!
  \tfrac{
      \prod_{ i = 1 }^{ n_1 }
      \varphi^{ (1) }_{
        l_i
      }
      \,
      \varphi^{ (2) }_{
        l_i
      }
  }{
    \prod_{ i = 1 }^{ n_1 }
    \lambda_{ l_i }
  }
  \int_{ t_0 }^{ t_1 }
  \int_{ t_0 }^{ t_2 }
    e^{ 
      - 
      \left(
        \sum_{ i = 1 }^n
        \lambda_{ l_i } 
      \right)
      \left| s_1 - s_2 \right|
    }
  \, 
  ds_2 \, ds_1
&
\colon
n_1 n_2 \neq 0
\\
  \delta_{ n_1, n_2 }
  \,
  \delta_{ k_1, k_2 }
&
\colon
n_1 n_2 = 0
\end{cases}
\end{split}
\end{equation}
for all 
$ 
  k_1, k_2 \in \mathbb{Z}^d 
$,
$ 
  \varphi^{ (1) }, \varphi^{(2)} \in \Phi_0 
$,
$ n_1, n_2 \in \N_0 $
and all
$
  t_0, t_1, t_2 
  \in \R
$
with
$
  t_0 \leq \min( t_1 , t_2 )
$.
\end{lem}

Lemma~\ref{lem:correlationAver}
is an immediate consequence
of
Lemma \ref{lem:correlation}
and the proof
of 
Lemma~\ref{lem:correlationAver}
is therefore omitted.

\begin{lem}[Time integrals
for averaged Wick powers]
\label{lem:time_integral}
Assume the setting of 
Subsection~\ref{sec:setting}.
Then
\begin{equation}
\label{eq:identity_time_integral}
\begin{split}
&
  \int_{ t_0 }^{ t_1 }
  \int_{ t_0 }^{ t_2 }
    e^{ 
      - 
      c
      \left| s_1 - s_2 \right|
    }
  \, ds_2 \, ds_1
  =
  \tfrac{
    2
    \left(
      \min( t_1, t_2 ) - t_0
    \right)
  }{
    c
  }
  +
  \tfrac{
    \left(
      \left[
        \sum_{ j = 1 }^2
        e^{
          -
          c
          \left( t_j - t_0 \right)
        }
      \right]
      - 1
      - 
      e^{
        - 
        c
        \left( t_2 - t_1 \right)
      }
    \right)
  }{
    c^2
  }
\end{split}
\end{equation}
and
\begin{equation}
\label{eq:identity_time_integral2}
\begin{split}
&
  \tfrac{ 
    \left( t_1 - t_0 \right) 
    \left(
      1 - 
      e^{
        - \frac{ c ( t_1 - t_0 ) }{ 2 }
      }
    \right)
  }{ c }
  \leq
  \int_{ t_0 }^{ t_1 }
  \int_{ t_0 }^{ t_1 }
    e^{ 
      - 
      c
      \left| s_1 - s_2 \right|
    }
  \, 
  ds_2 \, ds_1
=
  \frac{
    2
  }{
    c^2
  }
    \left(
      e^{
        -
        c
        \left( t_1 - t_0 \right)
      }
      - 
      \left[
        1
        -
        \left(
          t_1 - t_0
        \right) 
        c
      \right]
    \right)
\leq
  \frac{
    2
    \left( t_1 - t_0 \right)^{
      \theta
    }
  }{
    c^{
      \left( 2 - \theta \right)
    }
  }
\end{split}
\end{equation}
for all 
$ c \in ( 0, \infty ) $,
$ \theta \in [1, 2] $
and all
$ 
  t_0, t_1, t_2
  \in \R 
$
with
$ t_0 \leq \min( t_1, t_2 ) $.
\end{lem}

\begin{proof}[Proof
of Lemma~\ref{lem:time_integral}]
Note that
\begin{equation}
\label{eq:first_time_integral}
\begin{split}
&
  \int_{ t_0 }^{ t_1 }
  \int_{ t_0 }^{ t_2 }
  1_{
    \left\{
      ( u_1 , u_2 )
      \in \R^2
      \colon
      u_2 \leq u_1 
    \right\} 
  }
  ( s_1 , s_2 )
  \cdot 
  e^{
    - 
    c
    \left| s_1 - s_2 \right|
  }
  \, ds_2 \, ds_1
\\& =
  \int_{ t_0 }^{ t_1 }
  \int_{ t_0 }^{ s_1 }
  e^{
    - 
    c
    \left( s_1 - s_2 \right)
  }
  \, ds_2 \, ds_1
  =
  \int_{ t_0 }^{ t_1 }
  \frac{
    \left(
      1
      -
      e^{
        -
        c
        \left( s_1 - t_0 \right)
    }
  \right)
  }{
    c
  }
  \, ds_1
  =
  \frac{
    \left( t_1 - t_0 \right)
  }{
    c
  }
  +
  \frac{
    \left(
      e^{
        -
        c
        \left( t_1 - t_0 \right)
      }
      - 1
    \right)
  }{
    c^2
  }
\end{split}
\end{equation}
and hence
\begin{equation}
\begin{split}
&
  \int_{ t_0 }^{ t_1 }
  \int_{ t_0 }^{ t_2 }
  1_{
    \left\{ 
      ( u_1 , u_2 ) \in \R^2
      \colon u_1 \leq u_2 \leq t_1
    \right\} 
  }
  ( s_1 , s_2 )
  \cdot 
  e^{
    - c
    \left| s_1 - s_2 \right|
  }
  \, ds_1 \, ds_2
\\ & =
  \int_{ t_0 }^{ t_1 }
  \int_{ t_0 }^{ s_2 }
  e^{
    - c
    \left( s_2 - s_1 \right)
  }
  \, ds_1 \, ds_2
  =
  \frac{
    \left( t_1 - t_0 \right)
  }{
    c
  }
  +
  \frac{
    \left(
      e^{
        -
        c
        \left( t_1 - t_0 \right)
      }
      - 1
    \right)
  }{
    c^2
  }
\end{split}
\end{equation}
for all 
$ c \in ( 0, \infty ) $
and all
$
  t_0, t_1, t_2 
  \in \R
$
with
$
  t_0 \leq t_1 \leq t_2
$.
Furthermore, 
observe that
\begin{equation}
\label{eq:last_time_integral}
\begin{split}
&
  \int_{ t_0 }^{ t_1 }
  \int_{ t_0 }^{ t_2 }
  1_{
    \left\{ 
      ( u_1 , u_2 ) \in \R^2
      \colon 
      u_1 \leq t_1 \leq u_2 
    \right\} 
  }
  ( s_1 , s_2 )
  \cdot 
  e^{
    - c
    \left| s_1 - s_2 \right|
  }
  \, ds_2 \, ds_1
\\ & =
  \int_{ t_1 }^{ t_2 }
  \int_{ t_0 }^{ t_1 }
  e^{
    - c
    \left( s_2 - s_1 \right)
  }
  \, ds_1 \, ds_2
  =
  \int_{ t_1 }^{ t_2 }
  \frac{
    \left(
      e^{
        - c
        \left( s_2 - t_1 \right)
      }
      -
      e^{
        - c
        \left( s_2 - t_0 \right)
      }
    \right)
  }{
    c
  }
  \,
  ds_2
\\ & =
  \frac{ 
    -
    \left(
      e^{
        - c
        \left( t_2 - t_1 \right)
      }
      - 1
    \right)
    +
    \left(
      e^{
        - c
        \left( t_2 - t_0 \right)
      }
      -
      e^{
        - c
        \left( t_1 - t_0 \right)
      }
    \right)
  }{
    c^2
  }
  =
  \frac{ 
    1 
    +
      e^{
        - c
        \left( t_2 - t_0 \right)
      }
    -
      e^{
        - c
        \left( t_2 - t_1 \right)
      }
      -
      e^{
        - c
        \left( t_1 - t_0 \right)
      }
  }{
    c^2
  }
\end{split}
\end{equation}
for all 
$ 
  c \in ( 0, \infty )
$
and all
$
  t_0, t_1, t_2 
  \in \R
$
with
$
  t_0 \leq t_1 \leq t_2
$.
Combining
\eqref{eq:first_time_integral}--\eqref{eq:last_time_integral}
results in
\begin{equation}
\begin{split}
&
  \int_{ t_0 }^{ t_1 }
  \int_{ t_0 }^{ t_2 }
  e^{
    - c
    \left| s_1 - s_2 \right|
  }
  \, ds_2 \, ds_1
\\ & =
  \frac{
    2 
    \left( t_1 - t_0 \right)
  }{
    c
  }
  +
  \frac{
    2
    \left(
      e^{
        -
        c
        \left( t_1 - t_0 \right)
      }
      - 1
    \right)
  }{
    c^2
  }
  +
  \frac{ 
    \left(
    1 
    +
      e^{
        - c
        \left( t_2 - t_0 \right)
      }
    -
      e^{
        - c
        \left( t_2 - t_1 \right)
      }
      -
      e^{
        - c
        \left( t_1 - t_0 \right)
      }  
    \right)
  }{
    c^2
  }
\\ & =
  \frac{
    2 
    \left( t_1 - t_0 \right)
  }{
    c
  }
  +
  \frac{ 
    \left(
      e^{
        -
        c
        \left( t_1 - t_0 \right)
      }
      +
      e^{
        - c
        \left( t_2 - t_0 \right)
      }
      - 1
      -
      e^{
        - c
        \left( t_2 - t_1 \right)
      }
    \right)
  }{
    c^2
  }
\end{split}
\label{eq:exact_time_integral}
\end{equation}
for all 
$ 
  c \in ( 0, \infty )
$
and all
$
  t_0, t_1, t_2 
  \in \R
$
with
$
  t_0 \leq t_1 \leq t_2
$.
In addition, observe that
\begin{equation}
\begin{split}
&
  \frac{ 
    \left| y \right| 
    \left(
      1 - e^{
        \frac{ y }{ 2 }
      }
    \right)
  }{ 2 }
=
  \int_y^{
    \frac{ y }{ 2 }
  }
  1 - e^{
    \frac{ y }{ 2 }
  }
  \, ds
\leq
  \int_y^{
    \frac{ y }{ 2 }
  }
  1 - e^s
  \, ds
\leq
  \int_y^0
  1 - e^s
  \, ds
\\ & =
  e^y
  -
  \left( 
    1 + y 
  \right)
=
  \int_y^0
  1 - e^s
  \, ds
\leq
  \int_y^0
  \left[
    1 - e^s
  \right]^{ \theta }
  ds
\leq
  \int_y^0
  \left[
    \int_s^0
    e^u \, du
  \right]^{ \theta }
  ds
\\ & \leq
  \int_y^0
  \left| 
    s
  \right|^{ \theta }
  ds
=
  \int_0^{ - y }
    s^{ \theta }
  \, ds
=
  \frac{
    \left| y 
    \right|^{ 
      \left( 1 + \theta \right) 
    }
  }{
    \left( 1 + \theta \right)
  }
\leq
  \left| y 
  \right|^{ 
    \left( 1 + \theta \right) 
  }
\end{split}
\end{equation}
for all $ y \in ( - \infty, 0 ] $
and all $ \theta \in [0, 1] $.
Combining this with
\eqref{eq:exact_time_integral}
completes the proof of 
Lemma~\ref{lem:time_integral}.
\end{proof}

The next result,
Proposition~\ref{prop:regularity_WickAver},
establishes convergence of 
averaged Wick powers 
under the assumption
that
$ n, d \in \{ 2, 3, \dots \} $
with
$
  \frac{ n + 1 }{ n - 1 }
  > 
  \frac{ d }{ 2 }
$.
The proof of
Proposition~\ref{prop:regularity_WickAver}
exploits
Lemma~\ref{lem:hypercontractivity},
Lemma~\ref{lem:correlationAver}
and 
Lemma~\ref{lem:time_integral}.

\begin{prop}[Convergence of 
averaged Wick powers]
\label{prop:regularity_WickAver}
Assume the setting of 
Subsection~\ref{sec:setting},
let $ t_0 \in \R $
and let 
$ n, d \in \{ 2, 3, \dots \} $
with
$
  \frac{ n + 1 }{ n - 1 }
  > 
  \frac{ d }{ 2 }
$.
Then there exists
an up to 
indistinguishability unique 
stochastic process
\begin{equation} 
  \circ 
  \left( 
    V_{ t_0, ( \cdot ) } 
  \right)^n \circ
  \colon
  [ t_0, \infty ) \times \Omega
  \to 
  \cap_{
    \beta \in 
    (
      - \infty,
      1 
      + \frac{ 1 }{ n }
      - \frac{ d }{ 2 }
      + 
      ( n - 1 ) 
      \min(
        1 
        + 
        \frac{ 1 }{ n }
        - \frac{ d }{ 2 } , 0
      )
    )
  }
  \mathcal{C}^{
    \beta
  }_{
    \mathcal{P}
  }(
    [0, 2 \pi]^d,
    \R
  )
\end{equation}
with continuous sample paths which
satisfies for every
$ T \in (t_0,\infty) $, 
$ p \in (0,\infty) $,
$ \alpha \in ( 0, 1 ) $
and every
$   
    \beta \in 
    (
      - \infty,
      1 
      + \frac{ 1 }{ n }
      - \frac{ d }{ 2 }
      + 
      ( n - 1 ) 
      \min(
        1 
        + 
        \frac{ 1 }{ n }
        - \frac{ d }{ 2 } , 0
      )
    )
$
that
\begin{equation}
  \|
    \circ \!
    ( 
      V^{ \varphi }_{ t_0, ( \cdot ) }
    )^n
    \!
    \circ
    -
    \circ
    ( 
      V_{ t_0, ( \cdot ) }
    )^n \!
    \circ
  \|_{
    L^p(
      \Omega;
      C^{ \alpha }
      (
        [ t_0 , T ],
        \mathcal{C}_{ \mathcal{P} 
        }^{ \beta  
        }(
          [0, 2 \pi]^d,
          \mathbb{R})
      )
    )
  }
  \rightarrow 0
\qquad
  \text{as}
\qquad
  \Phi_{ 0, \leq 1 } 
  \ni \varphi 
  \rightarrow 1 .
\end{equation}
\end{prop}

\begin{proof}[Proof
of 
Proposition~\ref{prop:regularity_WickAver}]
Lemma~\ref{lem:correlationAver}
and 
Lemma~\ref{lem:time_integral}
imply
\begin{equation}
\label{eq:AWP_L_inf_part_00}
\begin{split}
&
  \frac{ 1 }{
    \left( 2 \pi \right)^{ 2 d }
  }
  \,
      \E\!\left[
        \overline{
          \left< 
            g_{ k_1 } ,
            \circ
            \big(
              V_{ \hat{t}, t}^{ 
                \varphi 
              }
            \big)^n
            \circ
            \,
            -
            \,
            \circ\big(
              V_{ \hat{t}, t}^{ \psi }
            \big)^n
            \circ
          \right>_{ \! H } 
        }
          \left< 
            g_{ k_2 } ,
            \circ
            \big(
              V_{ \hat{t}, t}^{ 
                \varphi
              }
            \big)^n
            \circ
            \,
            -
            \,
            \circ\big(
              V_{ \hat{t}, t}^{ 
                \psi
              }
            \big)^n
            \circ
          \right>_{ \! H } 
      \right]
\\ & = 
  n!
  \,
  \delta_{ k_1, k_2 }
    \sum_{
      \substack{
        l_1, \ldots, l_n
        \in \Z^d
      \\
        l_1 + \ldots + l_n 
        = k_1
      }
    }
  \left[
  \begin{array}{c}
    \left\{
      \prod_{
        i = 1 
      }^n
      \frac{ 
        \left[
          \varphi_{ l_i }
        \right]^2
      }{
        \lambda_{ l_i }
      }
      - 
      2
      \prod_{
        i = 1 
      }^n
      \frac{ 
        \varphi_{ l_i }
        \psi_{ l_i }
      }{
        \lambda_{ l_i }
      }
      +
      \prod_{
        i = 1 
      }^n
      \frac{ 
        \left[
          \psi_{ l_i }
        \right]^2
      }{
        \lambda_{ l_i }
      }
    \right\}
  \\ \cdot
  \int_{ \hat{t} }^{ t }
  \int_{ \hat{t} }^{ t }
    e^{ 
      - 
      \left[ 
        \sum_{ i = 1 }^n 
        \lambda_{ l_i }
      \right]
      \left| s_2 - s_1 \right|
    }
  \, ds_2 \, ds_1
  \end{array}
  \right]
\\ & = 
  n!
  \,
  \delta_{ k_1, k_2 }
    \sum_{
      \substack{
        l_1, \ldots, l_n
        \in \Z^d
      \\
        l_1 + \ldots + l_n 
        = k_1
      }
    }
      \frac{ 
        \left(
          \prod_{ i = 1 }^n
          \varphi_{ l_i }
          -
          \prod_{ i = 1 }^n
          \psi_{ l_i }
        \right)^2
      }{
        \left(
          \prod_{ i = 1 }^n
          \lambda_{ l_i }
        \right)
      }
  \int_{ \hat{t} }^{ t }
  \int_{ \hat{t} }^{ t }
    e^{ 
      - 
      \left[ 
        \sum_{ i = 1 }^n 
        \lambda_{ l_i }
      \right]
      \left| s_2 - s_1 \right|
    }
  \, ds_2 \, ds_1
\\ & \leq
  n!
  \,
  2
  \,
  \delta_{ k_1, k_2 }
  \sum_{
    \substack{
      l_1, \ldots, l_n
      \in \Z^d
    \\
      l_1 + \ldots + l_n 
      = k_1
    }
  }
  \frac{ 
    \left(
      \prod_{ i = 1 }^n
      \varphi_{ l_i }
      -
      \prod_{ i = 1 }^n
      \psi_{ l_i }
    \right)^2
    \left( t - \hat{t} \right)
  }{
    \left(
      \prod_{ i = 1 }^n
      \lambda_{ l_i }
    \right)
    \left(
      \sum_{ i = 1 }^n
      \lambda_{ l_i }
    \right)
  }
\\ & \leq
  n!
  \,
  2
  \,
  \delta_{ k_1, k_2 }
  \sum_{
    \substack{
      l_1, \ldots, l_n
      \in \Z^d
    \\
      l_1 + \ldots + l_n 
      = k_1
    }
  }
  \frac{ 
    \left(
      \prod_{ i = 1 }^n
      \varphi_{ l_i }
      -
      \prod_{ i = 1 }^n
      \psi_{ l_i }
    \right)^2
    \left( t - \hat{t} \right)
  }{
    \left(
      \prod_{ i = 1 }^n
      \left(
        \lambda_{ l_i }
      \right)^{
        \left( 1 + 1 / n \right)
      }
    \right)
  }
\end{split}
\end{equation}
for all 
$ k_1, k_2 \in \Z^d $
and all $ \hat{t}, t \in \R $
with $ \hat{t} \leq t $.
Next note that
Corollary~\ref{cor:discrete_conv}
ensures that
\begin{equation}
  \sup_{
    k_1 \in \Z^d
  }
  \left[
    \sum_{
      \substack{
        l_1, \ldots, l_n
        \in \Z^d
      \\
        l_1 + \ldots + l_n 
        = k_1
      }
    }
    \frac{
      \left(
        \lambda_{ k_1 }
      \right)^{
        \gamma
      }
    }{
      \left(
        \prod_{ i = 1 }^n
        \left(
          \lambda_{ l_i }
        \right)^{
          \left( 1 + 1 / n \right)
        } 
      \right)
    }
  \right]
  < \infty
\end{equation}
for all
$
  \gamma \in
  \big(
    0, 
    1 + \frac{ 1 }{ n }
    + \left( n - 1 \right)
    \min\!\left(
      1 + \frac{ 1 }{ n }
      - \frac{ d }{ 2 } , 0
    \right)
  \big)
$
and therefore, we obtain that
\begin{equation}
\label{eq:summability_AWP_2}
  \sum_{
    k_1 \in \Z^d
  }
    \sum_{
      \substack{
        l_1, \ldots, l_n
        \in \Z^d
      \\
        l_1 + \ldots + l_n 
        = k_1
      }
    }
    \frac{
      \left(
        \lambda_{k_1}
      \right)^{ 2 \beta }
    }{
      \left(
        \prod_{ i = 1 }^n
        \left(
          \lambda_{ l_i }
        \right)^{
          \left( 1 + 1 / n \right)
        } 
      \right)
    }
  <
  \infty
\end{equation}
for all
$
  \beta \in 
  \big(
    - \infty,
    \frac{ 1 }{ 2 } 
    + \frac{ 1 }{ 2 n }
    + 
    \left( n - 1 \right) 
    \min\!\left(
      \frac{ 1 }{ 2 } 
      + 
      \frac{ 1 }{ 2 n }
      - \frac{ d }{ 4 } , 0
    \right)
    - \frac{ d }{ 4 }
  \big)
$.
Combining this,
\eqref{eq:AWP_L_inf_part_00}
and dominated convergence 
implies for every
$
  \beta \in 
  \big(
    - \infty,
    \frac{ 1 }{ 2 } 
    + \frac{ 1 }{ 2 n }
    - \frac{ d }{ 4 }
    + 
    \left( n - 1 \right) 
    \min\!\left(
      \frac{ 1 }{ 2 } 
      + 
      \frac{ 1 }{ 2 n }
      - \frac{ d }{ 4 } , 0
    \right)
  \big)
$
that
\begin{equation}
\label{eq:L_inf_partAver}
    \sup_{
      \substack{
        \hat{t}, t
        \in
        \R ,
      \\
        \hat{t} < t
      }
    }
    \sum_{
      k_1, k_2
      \in
      \Z^d
    }
  \tfrac{
    \left|
      \E\left[
        \overline{
          \left< 
            g_{ k_1 } ,
            \circ
            \left(
              V_{\hat{t},t}^{ 
                \varphi 
              }
            \right)^n
            \circ
            \,
            -
            \,
            \circ\left(
              V_{\hat{t},t}^{ \psi }
            \right)^n
            \circ
          \right>_{ \! H } 
        }
          \left< 
            g_{ k_2 } ,
            \circ
            \left(
              V_{\hat{t},t}^{ 
                \varphi
              }
            \right)^n
            \circ
            \,
            -
            \,
            \circ\left(
              V_{\hat{t},t}^{ 
                \psi
              }
            \right)^n
            \circ
          \right>_{ \! H } 
      \right]
    \right|
  }{
    \left(
      t - \hat{t}
    \right)
    \cdot
    \left(
      \lambda_{ k_1 }
      \lambda_{ k_2 }
    \right)^{ - \beta }
  }
 \to 0
\end{equation}
as
$
      ( \Phi_{ 0, \leq 1 } )^2 
      \ni
      (
        \varphi,
        \psi
      )
      \rightarrow (1,1) 
$.
In the next step
observe that
Definition~\eqref{eq:def_AWP}
implies that
\begin{equation}
\begin{split} 
& 
          \E\!\left[
          \begin{split}
          &
            \left<
              g_{ - k_1 } 
            ,
              \left[
                \circ(
                  V_{ \hat{t}, t_1 }^{
                    \varphi
                  }
                )^n
                \circ
                -
                \circ(
                  V_{ \hat{t}, t_1 }^{ 
                    \psi
                  }
                )^n
                \circ
              \right]
              -
              \left[
                \circ(
                  V_{ \hat{t}, t_2 }^{
                    \varphi
                  }
                )^n
                \circ
                -
                \circ(
                  V_{ \hat{t}, t_2 }^{ 
                    \psi
                  }
                )^n
                \circ
              \right]
            \right>_{ \! H }
          \\ & \;
            \cdot
            \left<
              g_{ k_2 } 
              ,
              \left[
                \circ(
                  V_{ \hat{t}, t_1 }^{
                    \varphi
                  }
                )^n
                \circ
                -
                \circ(
                  V_{ \hat{t}, t_1 }^{ 
                    \psi
                  }
                )^n
                \circ
              \right]
              -
              \left[
                \circ(
                  V_{ \hat{t}, t_2 }^{
                    \varphi
                  }
                )^n
                \circ
                -
                \circ(
                  V_{ \hat{t}, t_2 }^{ 
                    \psi
                  }
                )^n
                \circ
              \right]
            \right>_{ \! H }
          \end{split}
          \right]
\\ & =
          \E\!\left[
          \begin{split}
          &
            \left<
              g_{ - k_1 } 
            ,
                \circ(
                  V_{ t_1, t_2 }^{
                    \varphi
                  }
                )^n
                \circ
                -
                \circ(
                  V_{ t_1, t_2 }^{ 
                    \psi
                  }
                )^n
                \circ
            \right>_{ \! H }
            \left<
              g_{ k_2 } 
              ,
                \circ(
                  V_{ t_1, t_2 }^{
                    \varphi
                  }
                )^n
                \circ
                -
                \circ(
                  V_{ t_1, t_2 }^{ 
                    \psi
                  }
                )^n
                \circ
            \right>_{ \! H }
          \end{split}
          \right]
\end{split}
\end{equation}
for all 
$ k_1, k_2 \in \Z^d $,
$ \varphi, \psi \in \Phi_{ 0, \leq 1 } $
and all
$ t_1, t_2 \in \R $
with $ \hat{t} \leq t_1 \leq t_2 $.
Combining this with
\eqref{eq:L_inf_partAver}
shows
for every
$
  \beta \in 
  \big(
    - \infty,
    \frac{ 1 }{ 2 } 
    + \frac{ 1 }{ 2 n }
    - \frac{ d }{ 4 }
    + 
    \left( n - 1 \right) 
    \min\!\left(
      \frac{ 1 }{ 2 } 
      + 
      \frac{ 1 }{ 2 n }
      - \frac{ d }{ 4 } , 0
    \right)
  \big)
$
that
\begin{equation}
\label{eq:lem_time_diffAver}
  \sup_{ 
    \substack{
      \hat{t} \in \R ,
    \\
       t_1, t_2 \in [ \hat{t}, \infty ) ,
    \\
      t_1 \neq t_2
    }
  }
  \left[
    \!\!\!
    \begin{array}{c}
      \sum\limits_{ 
        k_1, k_2 
        \in \Z^d
      } 
      \frac{
        \tiny{
        \left|
          \E\!\left[
          \begin{split}
          &
            \left<
              g_{ - k_1 } 
            ,
              \left[
                \circ(
                  V_{ \hat{t}, t_1 }^{
                    \varphi
                  }
                )^n
                \circ
                -
                \circ(
                  V_{ \hat{t}, t_1 }^{ 
                    \psi
                  }
                )^n
                \circ
              \right]
              -
              \left[
                \circ(
                  V_{ \hat{t}, t_2 }^{
                    \varphi
                  }
                )^n
                \circ
                -
                \circ(
                  V_{ \hat{t}, t_2 }^{ 
                    \psi
                  }
                )^n
                \circ
              \right]
            \right>_{ \! H }
          \\ & \;
            \cdot
            \left<
              g_{ k_2 } 
              ,
              \left[
                \circ(
                  V_{ \hat{t}, t_1 }^{
                    \varphi
                  }
                )^n
                \circ
                -
                \circ(
                  V_{ \hat{t}, t_1 }^{ 
                    \psi
                  }
                )^n
                \circ
              \right]
              -
              \left[
                \circ(
                  V_{ \hat{t}, t_2 }^{
                    \varphi
                  }
                )^n
                \circ
                -
                \circ(
                  V_{ \hat{t}, t_2 }^{ 
                    \psi
                  }
                )^n
                \circ
              \right]
            \right>_{ \! H }
          \end{split}
          \right]
        \right|
        }
      }{ 
        \left( 
          \lambda_{ k_1 }
          \lambda_{ k_2 }
        \right)^{ - \beta }
        \left|
          t_1 - t_2
        \right|
      }
    \end{array}
    \!\!\!
  \right]
  \to 0
\end{equation}
as
$
  ( \Phi_{ 0, \leq 1 } )^2 
  \ni
  (
    \varphi,
    \psi
  )
  \rightarrow (1,1) 
$.
Combining \eqref{eq:L_inf_partAver} and
\eqref{eq:lem_time_diffAver} with
Lemma~\ref{lem:hypercontractivity}
completes the proof
of 
Proposition~\ref{prop:regularity_WickAver}.
\end{proof}

Proposition~\ref{prop:regularity_WickAver}
shows convergence of 
averaged Wick powers 
under the assumption
that
$ n, d \in \{ 2, 3, \dots \} $
with
$
  \frac{ n + 1 }{ n - 1 }
  > 
  \frac{ d }{ 2 }
$.
Lemma~\ref{lem:limitationAWP}
below, in particular, 
proves that
averaged Wick powers 
fail to converge
if
$ n, d \in \{ 2, 3, \dots \} $
with
$
  \frac{ n + 1 }{ n - 1 }
  \leq
  \frac{ d }{ 2 }
$.
In the proof of
Lemma~\ref{lem:limitationAWP}
the following lemma is used.

\begin{lem}
\label{lem:sum_tool}
Assume the setting of 
Subsection~\ref{sec:setting}
and let 
$ n, d \in \{ 2, 3, \dots \} $
with
$
  \frac{ 
    ( n + 1 ) 
  }{
    ( n - 1 )
  }
  \leq
  \frac{ d }{ 2 }
$.
Then
$
  \sum_{
    \substack{
      l_{ 1 },
      \dots,
      l_{ n }
      \in
      \Z^d
    \\
      l_1 
      +
      \ldots
      +
      l_n =
      v
    }
  }
  \frac{
    1
  }{
    \left(
      \prod_{ i = 1 }^n
      \lambda_{ l_i }
    \right)
    \left(
      \lambda_v +
      \sum_{ i = 1 }^n  
      \lambda_{ l_i } 
    \right)
  }
  = \infty
$
for all $ v \in \Z^d $.
\end{lem}

\begin{proof}[Proof
of Lemma~\ref{lem:sum_tool}]
Note that
\begin{equation}
\begin{split}
&
  \sum_{
    \substack{
      l_{ 1 },
      \dots,
      l_{ n }
      \in
      \Z^d
    \\
      l_1 
      +
      \ldots
      +
      l_n =
      v
    }
  }
  \frac{
    1
  }{
    \left(
      \prod_{ i = 1 }^n
      \lambda_{ l_i }
    \right)
    \left(
      \lambda_v +
      \sum_{ i = 1 }^n  
      \lambda_{ l_i } 
    \right)
  }
\\ & \geq
  \sum_{
    \substack{
      l_{ 1 },
      \dots,
      l_{ n }
      \in
      \Z^d
    \\
      l_1 
      +
      \ldots
      +
      l_n =
      v
    }
  }
  \frac{
    1
  }{
    \left(
      \lambda_v +
      \sum_{ i = 1 }^n  
      \lambda_{ l_i } 
    \right)^{ ( n + 1 ) }
  }
\geq
  \sum_{
    \substack{
      l_{ 1 },
      \dots,
      l_{ n }
      \in
      \Z^d
    \\
      l_1 
      +
      \ldots
      +
      l_n =
      v
    }
  }
  \frac{
    1
  }{
    \left(
      \lambda_v 
      +
      \lambda_{ ( v - l_1 - \ldots - l_n ) }
      +
      \sum_{ i = 1 }^{ n - 1 }  
      \lambda_{ l_i } 
    \right)^{ ( n + 1 ) }
  }
\\ & \geq
  \frac{ 1 }{
    \left(
      n + 1
    \right)^{
      \left( n + 1 \right)
    }
  }
  \sum_{
    \substack{
      l_{ 1 },
      \dots,
      l_{ n }
      \in
      \Z^d
    \\
      l_1 
      +
      \ldots
      +
      l_n =
      v
    }
  }
  \frac{
    1
  }{
    \left(
      \lambda_v 
      +
      \sum_{ i = 1 }^{ n - 1 }  
      \lambda_{ l_i } 
    \right)^{ ( n + 1 ) }
  }
\\ & \geq
  \frac{ 1 }{
    \left(
      2 \lambda_v 
      \left( n + 1 \right)
    \right)^{
      \left( n + 1 \right)
    }
  }
  \left[
  \sum_{
    \substack{
      l_{ 1 },
      \dots,
      l_{ n }
      \in
      \Z^d
    \\
      l_1 
      +
      \ldots
      +
      l_n =
      v
    }
  }
  \frac{
    1
  }{
    \left(
      \sum_{ i = 1 }^{ n - 1 }  
      \lambda_{ l_i } 
    \right)^{ ( n + 1 ) }
  }
  \right]
\\ & =
  \frac{ 1 }{
    \left(
      2 \lambda_v 
      \left( n + 1 \right)
    \right)^{
      \left( n + 1 \right)
    }
  }
  \left[
  \sum_{
    k \in \Z^{ d ( n - 1 ) }
  }
  \frac{
    1
  }{
    \left(
      1 + \left\| k \right\|^2_{ 
        \R^{ d ( n - 1 ) }
      }
    \right)^{ ( n + 1 ) }
  }
  \right]
  = \infty
\end{split}
\end{equation}
for all $ v \in \Z^d $.
The proof of Lemma~\ref{lem:sum_tool}
is thus completed.
\end{proof}

\begin{lem}[Divergence of
averaged Wick powers]
\label{lem:limitationAWP}
Assume the setting of 
Subsection~\ref{sec:setting},
let 
$ n, d \in \{ 2, 3, \dots \} $
with
$
  \frac{ 
    ( n + 1 ) 
  }{
    ( n - 1 )
  }
  \leq
  \frac{ d }{ 2 }
$
and let
$ 
  C_0, C_1, \dots, C_{ n - 1 } \colon
  \Phi_0
  \to \R
$
be arbitrary functions.
Then it holds for every
$ v \in \Z^d $
and every $ t_0, t \in \R $
with $ t_0 < t $ that
\begin{equation}
  \E\!
  \left[
    \left|
      \left<
        g_v ,
        \int_{ t_0 }^t
        \left(
          \left( 
            V^{ \varphi }_s
          \right)^n
          -
          \sum_{ k = 0 }^{ n - 1 }
          C_k( \varphi ) 
          \cdot
          \left( 
            V^{ \varphi }_s 
          \right)^k
        \right)
        ds
      \right>_{ \! \! H }
    \right|^2
  \right] 
  \to \infty
  \qquad
  \text{as}
  \qquad
  \Phi_0 \ni \varphi
  \rightarrow 1
  .
\end{equation}
\end{lem}

\begin{proof}[Proof of Lemma~\ref{lem:limitationAWP}]
Throughout this proof
let 
$ 
  \hat{C}_0, \hat{C}_1, \dots
  \hat{C}_n \colon \Phi_0
  \to \R
$
be the unique functions satisfying
$
  \hat{C}_0( 0 )
  =
  - C_0( 0 )
$,
$
  \hat{C}_1( 0 )
  =
  - C_1( 0 )
$,
$ \dots $,
$
  \hat{C}_{ n - 1 }( 0 )
  =
  - C_{ n - 1 }( 0 )
$,
$
  \hat{C}_n( 0 ) 
  = 1
$
and
\begin{equation}
\begin{split}
&
  x^n
  -
  \sum_{ k = 0 }^{ n - 1 }
  C_k( \varphi )
  \cdot
  x^k
=
  \sum_{ k = 0 }^n
  \hat{C}_k( \varphi )
  \cdot
  \left[
      \sum_{ v \in \Z^d }
      \tfrac{ ( \varphi_v )^2 }{
        \lambda_v
      }
  \right]^{
    \! \frac{ k }{ 2 }
  }
  \cdot
  H_k\!\left(
  \frac{
    x
  }{
    \sqrt{
      \sum_{ v \in \Z^d }
      \frac{ ( \varphi_v )^2 }{
        \lambda_v
      }
    }
  }
  \right)
\end{split}
\end{equation}
for all 
$ x \in \R $,
$ 
  \varphi \in \Phi_0 \backslash \{ 0 \} 
$
and all $ t \in \R $
(cf.\ \eqref{eq:lem_limitation_firstdispl}).
Then Lemma~\ref{lem:correlationAver}
and Lemma~\ref{lem:time_integral}
imply that
\begin{equation}
\label{eq:limitation_main_estimateAver}
\begin{split}
&
  \E\!\left[
    \left|
      \left<
        g_v ,
        \int_{ t_0 }^t
        \left(
          \left( V_s^{ \varphi }
          \right)^n
          -
          \sum_{ k = 0 }^{ n - 1 }
          C_k( \varphi ) \cdot
          ( V_s^{ \varphi } )^k
        \right)
        ds
      \right>_{ \!\!\! H }
    \right|^2
  \right]
=
  \E\!\left[
    \left|
      \sum_{ k = 0 }^n
      \left<
        g_v ,
        \int_{ t_0 }^t
        \hat{C}_k( \varphi ) 
        \left(
          : \!
          ( V_s^{ \varphi }
          )^k
          \! :
        \right)
        ds
      \right>_{ \!\! H }
    \right|^2
  \right]
\\ & =
  \sum_{ k, l = 0 }^n
  \hat{C}_k( \varphi ) 
  \cdot
  \hat{C}_l( \varphi ) 
  \cdot
  \E\!\left[
    \overline{
      \left<
        g_v ,
        \circ
          ( V_{ t_0, t }^{ \varphi }
          )^k
        \circ
      \right>_{ \! H }
    }
    \left<
      g_v ,
      \circ
        ( V_{ t_0, t }^{ \varphi }
        )^l
      \circ
    \right>_{ \! H }
  \right]
\\ & =
  \sum_{ k = 0 }^n
  \left|
    \hat{C}_k( \varphi ) 
  \right|^2
  \E\!\left[
    \left|
      \left<
        g_v ,
          \circ
          ( V_{ t_0, t }^{ \varphi }
          )^k
          \circ
      \right>_{ H }
    \right|^2
  \right]
\geq
  \left|
    \hat{C}_n( \varphi ) 
  \right|^2
  \E\!\left[
    \left|
      \left<
        g_v ,
          \circ
          ( V_{ t_0, t }^{ \varphi }
          )^n
          \circ
      \right>_{ H }
    \right|^2
  \right]
\\ & =
  n!
  \left( 2 \pi \right)^{ 2 d }
  \left[
  \sum_{
    \substack{
      l_{ 1 },
      \dots,
      l_{ n }
      \in
      \Z^d
    \\
      l_1 
      +
      \ldots
      +
      l_n =
      v
    }
  }
  \left\{
  \prod_{ i = 1 }^n
  \frac{
    \left(
      \varphi_{
        l_i
      }
    \right)^2
  }{
    \lambda_{ l_i }
  }
  \right\}
  \left\{
    \int_{ t_0 }^{ t }
    \int_{ t_0 }^{ t }
      e^{ 
        - 
        \left(
          \sum_{ i = 1 }^n
          \lambda_{ l_i } 
        \right)
        \left| s_1 - s_2 \right|
      }
    \, 
    ds_2 \, ds_1
  \right\}
  \right]
\\ & \geq
  n!
  \left( 2 \pi \right)^{ 2 d }
  \left[
  \sum_{
    \substack{
      l_1 ,
      \dots,
      l_n
      \in
      \Z^d
    \\
      l_1 
      +
      \ldots
      +
      l_n =
      v
    }
  }
  \left(
  \prod_{ i = 1 }^n
  \frac{
    \left( \varphi_{ l_i } \right)^2
  }{
    \lambda_{ l_i }
  }
  \right)
  \frac{
    \left(
      t - t_0
    \right)
    \left(
      1 -
      e^{
        - \frac{ \left( t - t_0 \right) }{2}
      }
    \right)
  }{
    \left(
      \sum_{ i = 1 }^n
      \lambda_{ l_i }
    \right)
  }
  \right]
\\ & \geq
  n!
  \left( 2 \pi \right)^{ 2 d }
  \left[
  \sum_{
    \substack{
      l_1 ,
      \dots,
      l_n
      \in
      \Z^d
    \\
      l_1 
      +
      \ldots
      +
      l_n =
      v
    }
  }
  \frac{
    \left(
      \prod_{ i = 1 }^n
      \left( \varphi_{ l_i } \right)^2
    \right)
    \left(
      t - t_0
    \right)
    \left(
      1 -
      e^{
        - \frac{ \left( t - t_0 \right) }{2}
      }
    \right)
  }{
    \left(
      \prod_{ i = 1 }^n
      \lambda_{ l_i }
    \right)
    \left(
      \lambda_v +
      \sum_{ i = 1 }^n
      \lambda_{ l_i }
    \right)
  }
  \right]
\end{split}
\end{equation}
for all 
$ v \in \Z^d $,
$ t_0, t \in \R $
with $ t_0 \leq t $
and all $ \varphi \in \Phi_0 $.
Combining this with
Lemma~\ref{lem:sum_tool}
completes the proof of
Lemma~\ref{lem:limitationAWP}.
\end{proof}

\subsection{Convolutional 
Wick powers of Ornstein-Uhlenbeck
processes}
\label{sec:CWP}

\begin{lem}[Correlation of
convolutional Wick powers of 
$ V^{ \varphi } $, $ \varphi \in \Phi_0 $, 
in
Fourier space]
\label{lem:correlationConv}
Assume the setting of 
Subsection~\ref{sec:setting}.
Then
\begin{equation}
\begin{split}
&
  \frac{ 1 }{
    \left( 2 \pi \right)^{ 2 d }
  }
  \,
  \E\!\left[
  \overline{
    \left< g_{k_1}, 
      \bullet
        ( V_{ t_1 }^{ \varphi_1 } )^{n_1}
      \bullet
    \right>_H
  }
    \left< g_{ k_2 }, 
      \bullet
        ( V_{ t_2 }^{ \varphi_2 } )^{n_2}
      \bullet
    \right>_H
  \right]
\\ & =
\begin{cases}
  \begin{array}{c}
  n_1 ! 
  \, 
  \delta_{ n_1 , n_2 }
  \,
  \delta_{ k_1 , k_2 }
  \sum_{ 
    \substack{
      l_1, \dots, l_{ n_1 } \in \Z^d
    \\
      l_1 + \ldots + l_{ n_1 } = k_1
    }
  }
  \bigg[
  \tfrac{
    \left(
      \prod_{ i = 1 }^{ n_1 }
        \varphi^{ (1) }_{
          l_i
        }
        \,
        \varphi^{ (2) }_{
          l_i
        }
    \right)
  }{
    \left(
      \prod_{ i = 1 }^{ n_1 }
      \lambda_{ l_i }
    \right)
  }
\\ 
  \cdot
  \int_{ - \infty }^{ t_1 }
  \int_{ - \infty }^{ t_2 }
    e^{ 
      - \lambda_{k_1}
      \left( t_1 - s_1 + t_2 - s_2 \right)
      - 
      \left(
        \sum_{ i = 1 }^{ n_1 } 
        \lambda_{ l_i }
      \right)
      \left| s_1 - s_2 \right|
    }
  \, ds_2 \, ds_1
  \bigg]
\end{array}
&
  \colon
  n_1 n_2 \neq 0
\\
  \delta_{ n_1 , n_2 }
  \,
  \delta_{ k_1 , k_2 }
  \,
  \delta_{ k_1 , 0 }
& 
  \colon
  n_1 n_2 = 0
\end{cases}
\end{split}
\end{equation}
for all 
$
  t_1, t_2 
  \in \R
$,
$ 
  k_1, k_2 \in \mathbb{Z}^d 
$,
$ 
  \varphi^{ (1) }, \varphi^{(2)} \in \Phi_0 
$ 
and all
$ n_1, n_2 \in \N_0 $.
\end{lem}

\begin{proof}[Proof
of Lemma~\ref{lem:correlationConv}]
Combining the identity
\begin{equation}
\begin{split}
&
  \frac{ 1 }{
    \left( 2 \pi \right)^{ 2 d }
  }
  \,
  \E\!\left[
  \overline{
    \left< g_{ k_1 }, 
      \bullet    
        \big( 
          V_{ t_1 }^{ \varphi^{ (1) } } 
        \big)^{ n_1 }
      \bullet
    \right>_H
  }
    \left< g_{ k_2 }, 
      \bullet
        \big( 
          V_{ t_2 }^{ \varphi^{ (2) } } 
        \big)^{ n_2 }
      \bullet
    \right>_H
  \right]
\\& =
  \int_{ - \infty }^{ t_1 }
  \int_{ - \infty }^{ t_2 }
  e^{
    - \lambda_{ k_1 }
    \left( t_1 - s_1 \right)
  }
  \,
  e^{
    - \lambda_{ k_2 }
    \left( t_2 - s_2 \right)
  }
  \,
  \frac{
  \E\!\left[
  \overline{
    \left< g_{k_1}, 
      : \!    
        \big(
          V_{ s_1 }^{ \varphi^{ (1) } } 
        \big)^{ n_1 }
      \!\! :
    \right>_H
  }
    \left< g_{ k_2 }, 
      : \!
        \big( 
          V_{ s_2 }^{ \varphi^{ (2) } } 
        \big)^{ n_2 }
      \!\! :
    \right>_H
  \right]
  }{
    \left( 2 \pi \right)^{ 2 d }
  }
  \, ds_1 \, ds_2
\end{split}
\end{equation}
for all 
$   
  \varphi^{ (1) }, \varphi^{ (2) } 
  \in \Phi_0
$,
$ k_1, k_2 \in \Z^d $,
$ n_1, n_2 \in \N_0 $
and all
$ t_1, t_2 \in \R $
with $ t_1 \leq t_2 $ 
with
Lemma~\ref{lem:correlation}
completes the proof
of Lemma~\ref{lem:correlationConv}.
\end{proof}

\begin{lem}[Time integrals 
for convolutional 
Wick powers]
\label{lem:time_integralConv}
Assume the setting of 
Subsection~\ref{sec:setting}.
Then
\begin{equation}
\begin{split}
&
  \int_{ - \infty }^{ t_1 }
  \int_{ - \infty }^{ t_2 }
    e^{ 
      - a
      \left( t_1 - s_1 + t_2 - s_2 \right)
      - 
      b
      \left| s_1 - s_2 \right|
    }
  \, ds_2 \, ds_1
  =
  \frac{
    e^{
        - a ( t_2 - t_1 )
      }
  }{
    a
    \left(
      a + b
    \right)
  }
  +
\begin{cases}
  \frac{
    \left(
      e^{
        - b 
        (t_2 - t_1)
      }
      -
      e^{
        - a
        (t_2 - t_1) 
      }
    \right)
  }{
    \left(
      a - b 
    \right)
    \left(
      a + b
    \right)
  }
&
  \colon
  a \neq b
\\[2ex]
  \frac{
    \left( t_2 - t_1 \right)
     e^{
      - a
      ( t_2 -  t_1) 
    }
  }{
    \left(
      a + b
    \right)
  }
&
  \colon
  a = b
\end{cases}
\end{split}
\end{equation}
and
\begin{equation}
  \int_{ -\infty }^{ t }
  \int_{ -\infty }^{ t }
    e^{ 
      - a
      \left( 2 t - s_1 - s_2 \right)
      - 
      b
      \left| s_1 - s_2 \right|
    }
  \, ds_2 \, ds_1   =
  \frac{
    1
  }{
    a
    \left(
      a + b
    \right)
  }
\end{equation}
for all 
$  
  a, b \in (0, \infty) 
$
and all
$
  t, t_1, t_2 
  \in \R
$
with $ t_1 \leq t_2 $.
\end{lem}

\begin{proof}[Proof
of 
Lemma~\ref{lem:time_integralConv}]
First of all, note that
\begin{equation}
\label{eq:integral_first}
\begin{split}
&
  \int_{ - \infty }^{ t_1 }
  \int_{ - \infty }^{ t_2 }
    1_{
      \left\{
        ( u_1, u_2 ) \in \R^2
        \colon
        u_2 \leq u_1 
      \right\} 
    }(s_1, s_2)
    \cdot
    e^{ 
      - a
      \left( t_1 - s_1 + t_2 - s_2 \right)
      - 
      b
      \left| s_1 - s_2 \right|
    }
  \, ds_2 \, ds_1
\\ & =
  \int_{ - \infty }^{ t_1 }
  \int_{ - \infty }^{ t_1 }
    1_{
      \left\{
        ( u_1, u_2 ) \in \R^2
        \colon
        u_2 \leq u_1 
      \right\} 
    }(s_1, s_2)
    \cdot
    e^{ 
      - a
      \left( t_1 - s_1 + t_2 - s_2 \right)
      - 
      b
      \left| s_1 - s_2 \right|
    }
  \, ds_2 \, ds_1
\\ & =
  \int_{ - \infty }^{ t_1 }
  \int_{ - \infty }^{ s_1 }
    e^{ 
      - a
      \left( t_1 - s_1 + t_2 - s_2 \right)
      - 
      b
      \left( s_1 - s_2 \right)
    }
  \, ds_2 \, ds_1
\\ & =
  \frac{
    \int_{ - \infty }^{ t_1 }
      e^{ - a ( t_1 + t_2 - 2 s_1 ) }
    \, ds_1
  }{
    \left( a + b \right)
  }
=
  \frac{
    e^{
      - a ( t_2 -  t_1)
    }
    -
    e^{
      - a ( t_1 + t_2 - 2 t_0 )
    }
  }{
    2 a
    \left(
      a + b
    \right)
  } 
  ,
\end{split}
\end{equation}
that
\begin{equation}
\label{eq:integral_second}
\begin{split}
& 
  \int_{ - \infty }^{ t_1 }
  \int_{ - \infty }^{ t_2 }
    1_{
      \left\{
        ( u_1 , u_2 ) \in \R^2
        \colon
        u_1 \leq u_2 \leq t_1
      \right\} 
    }(s_1, s_2)
    \cdot
    e^{ 
      - a
      \left( t_1 - s_1 + t_2 - s_2 \right)
      - b
      \left| s_1 - s_2 \right|
    }
  \, ds_2 \, ds_1
\\ & =
  \int_{ - \infty }^{ t_1 }
  \int_{ - \infty }^{ t_1 }
    1_{
      \left\{
        ( u_1 , u_2 ) \in \R^2
        \colon
        u_1 \leq u_2 
      \right\} 
    }(s_1, s_2)
    \cdot
    e^{ 
      - a
      \left( t_1 - s_1 + t_2 - s_2 \right)
      - b
      \left| s_1 - s_2 \right|
    }
  \, ds_2 \, ds_1
\\ & =
  \int_{ - \infty }^{ t_1 }
  \int_{ - \infty }^{ t_1 }
    1_{
      \left\{
        ( u_1 , u_2 ) \in \R^2
        \colon
        u_2 \leq u_1
      \right\} 
    }(s_1, s_2)
    \cdot
    e^{ 
      - a
      \left( t_1 - s_1 + t_2 - s_2 \right)
      - b
      \left| s_1 - s_2 \right|
    }
  \, ds_2 \, ds_1
\\ & =
  \frac{
    e^{
      - a ( t_2 -  t_1)
    }
    -
    e^{
      - a ( t_1 + t_2 - 2 t_0 )
    }
  }{
    2 a
    \left(
      a + b
    \right)
  } 
\end{split}
\end{equation}
and that
\begin{equation}
\label{eq:integral_third}
\begin{split}
& 
  \int_{ - \infty }^{ t_1 }
  \int_{ - \infty }^{ t_2 }
    1_{
      \left\{
        ( u_1 , u_2 ) \in \R^2
        \colon
        u_1 \leq t_1 \leq u_2
      \right\} 
    }(s_1, s_2)
    \cdot
    e^{ 
      - a
      \left( t_1 - s_1 + t_2 - s_2 \right)
      - 
      b
      \left| s_1 - s_2 \right|
    }
  \, ds_2 \, ds_1
\\ & =
  \int_{ t_1 }^{ t_2 }
  \int_{ - \infty }^{ t_1 }
    e^{ 
      - a
      \left( t_1 - s_1 + t_2 - s_2 \right)
      - b
      \left( s_2 - s_1 \right)
    }
  \, ds_1 \, ds_2
\\ & =
  \frac{
    1
  }{
    \left( a + b \right)
  }
  \int_{ t_1 }^{ t_2 }
  \left(
    e^{
      - a
      ( t_2 -  s_2 )
      - b
      (s_2 - t_1)
    }
  \right)
  ds_2
=
\begin{cases}
  \frac{
    e^{ 
      - b  
      (t_2 - t_1)
    }
    -
    e^{
      - a 
      (t_2 - t_1) 
    }
  }{
    \left(
      a - b
    \right)
    \left(
      a + b
    \right)
  }
&
  \colon
    a \neq b
\\[2ex]
  \frac{
    \left( t_2 - t_1 \right)
    e^{
      - a (t_2 -  t_1)
    }
  }{
    \left(
      a + b
    \right)
  }
& \colon
  a = b
\end{cases}
\end{split}
\end{equation}
for all $ t_1, t_2 \in \R $
with $ t_1 \leq t_2 $.
Combining 
\eqref{eq:integral_first}--\eqref{eq:integral_third}
proves that
\begin{equation}
\begin{split}
&
  \int_{ - \infty }^{ t_1 }
  \int_{ - \infty }^{ t_2 }
    e^{ 
      - a
      \left( t_1 - s_1 + t_2 - s_2 \right)
      - 
      b
      \left| s_2 - s_1 \right|
    }
  \, ds_2 \, ds_1
  =
  \frac{
    e^{
      - a (t_2 - t_1)
    }
  }{
    a
    \left(
      a + b
    \right)
  }
  +
\begin{cases}
  \frac{
    e^{
      - a  
      ( t_2 - t_1 )
    }
    -
    e^{
      - b 
      ( t_2 - t_1 ) 
    }
  }{
    \left(
      a - b
    \right)
    \left(
      a + b
    \right)
  }
& 
\colon
  a \neq b
\\[2ex]
  \frac{
    \left( t_2 - t_1 \right)
    e^{
      - a
      ( t_2 -  t_1) 
    }
  }{
    \left(
      a + b
    \right)
  }
&
  \colon
  a = b
\end{cases}
\end{split}
\end{equation}  
for all $ t_1, t_2 \in \R $
with $ t_1 \leq t_2 $.
The proof of 
Lemma~\ref{lem:time_integralConv}
is thus completed.
\end{proof}

The next proposition
proves convergence of
convolutional Wick powers 
under the assumption
that
$ n, d \in \{ 2, 3, \dots \} $
with
$
  \frac{ n + 1 }{ n - 1 }
  > 
  \frac{ d }{ 2 }
$.
Its proof uses
Lemma~\ref{lem:correlationConv},
Lemma~\ref{lem:time_integralConv}
and
Lemma~\ref{lem:hypercontractivity}.

\begin{prop}[Convergence of 
convolutional Wick powers]
\label{prop:regularity_WickConv}
Assume the setting of 
Subsection~\ref{sec:setting}
and let 
$
  n, d \in \{ 2, 3, \dots \}
$
with
$
  \frac{ n + 1 }{ n - 1 }
  > 
  \frac{ d }{ 2 }
$.
Then there exists
an up to 
indistinguishability unique 
stochastic process
\begin{equation} 
  \bullet ( V )^n \bullet
  \colon
  \R \times \Omega
  \to 
  \cap_{
    \beta \in ( - \infty , 
      2 + \frac{ n ( 2 - d ) }{ 2 }
    )
  }
  \mathcal{C}^{
    \beta
  }_{
    \mathcal{P}
  }(
    [0, 2 \pi]^d,
    \R
  )
\end{equation}
with continuous sample 
paths which satisfies
for every $ T, p \in (0,\infty) $
and every
$ \alpha \in ( 0, \frac{ 1 }{ 2 } ) $,
$ \beta \in \R $
with 
$ 
  2 \alpha + \beta 
  < 
  2 +
  \frac{ n ( 2 - d ) }{ 2 }
$
that
\begin{equation}
  \left\|
    \bullet
    ( 
      V^{ \varphi } 
    )^n
    \bullet
    -
    \bullet
    ( 
      V
    )^n
    \bullet
  \right\|_{
    L^p(
      \Omega;
      \mathcal{C}^{ \alpha }
      (
        [ - T , T ],
        \mathcal{C}_{ \mathcal{P} 
        }^{ \beta  
        }(
          [0, 2 \pi]^d,
          \mathbb{R})
      )
    )
  }
  \to 0
  \qquad
  \text{as}
  \qquad
  \Phi_{ 0, \leq 1 } 
  \ni \varphi \to 1 .
\end{equation}
\end{prop}

\begin{proof}[Proof
of Proposition~\ref{prop:regularity_WickConv}]

Lemma~\ref{lem:correlationConv}
and 
Lemma~\ref{lem:time_integralConv}
imply
\begin{equation}
\label{eq:L_inf_part_CWP}
\begin{split}
&
  \frac{ 1 }{
    \left( 2 \pi \right)^{ 2 d }
  }
  \,
      \E\!\left[
        \overline{
          \left< 
            g_{ k_1 } ,
            \bullet
            \big(
              V_t^{ 
                \varphi 
              }
            \big)^n
            \bullet
            \,
            -
            \,
            \bullet\big(
              V_t^{ \psi }
            \big)^n
            \bullet
          \right>_{ \! H } 
        }
          \left< 
            g_{ k_2 } ,
            \bullet
            \big(
              V_t^{ 
                \varphi
              }
            \big)^n
            \bullet
            \,
            -
            \,
            \bullet\big(
              V_t^{ 
                \psi
              }
            \big)^n
            \bullet
          \right>_{ \! H } 
      \right]
\\ & = 
  n!
  \,
  \delta_{ k_1, k_2 }
    \sum_{
      \substack{
        l_1, \ldots, l_n
        \in \Z^d
      \\
        l_1 + \ldots + l_n 
        = k_1
      }
    }
  \left[
  \begin{array}{c}
    \left\{
      \prod_{
        i = 1 
      }^n
      \frac{ 
        \left[
          \varphi_{ l_i }
        \right]^2
      }{
        \lambda_{ l_i }
      }
      - 
      2
      \prod_{
        i = 1 
      }^n
      \frac{ 
        \varphi_{ l_i }
        \psi_{ l_i }
      }{
        \lambda_{ l_i }
      }
      +
      \prod_{
        i = 1 
      }^n
      \frac{ 
        \left[
          \psi_{ l_i }
        \right]^2
      }{
        \lambda_{ l_i }
      }
    \right\}
  \\ \cdot
  \int_{ - \infty }^{ t }
  \int_{ - \infty }^{ t }
    e^{ 
      - \lambda_{k_1}
      \left( 2 t - s_1 - s_2 \right)
      - 
      \left( 
       \sum_{i=1}^n \lambda_{l_i}
      \right)
      \left| s_2 - s_1 \right|
    }
  \, ds_2 \, ds_1
  \end{array}
  \right]
\\ & =
  n!
  \,
  \delta_{ k_1, k_2 }
  \sum_{
    \substack{
      l_1, \ldots, l_n
      \in \Z^d
    \\
      l_1 + \ldots + l_n 
      = k_1
    }
  }
  \frac{ 
    \left(
      \prod_{ i = 1 }^n
      \varphi_{ l_i }
      -
      \prod_{ i = 1 }^n
      \psi_{ l_i }
    \right)^2
  }{
    \left(
      \prod_{ i = 1 }^n
      \lambda_{ l_i }
    \right)
    \lambda_{ k_1 }
    \left(
      \lambda_{ k_1 } 
      +
      \sum_{ i = 1 }^n  
      \lambda_{ l_i } 
    \right)
  }
\\ & \leq
  \frac{
    n! \, \delta_{ k_1, k_2 }
  }{
    \left( 
      \lambda_{ k_1 }
    \right)^{
      \max( 4 - d, 1 )
    }
  }
  \left[
  \sum_{
    \substack{
      l_1, \ldots, l_n
      \in \Z^d
    \\
      l_1 + \ldots + l_n 
      = k_1
    }
  }
  \frac{ 
    \left(
      \prod_{ i = 1 }^n
      \varphi_{ l_i }
      -
      \prod_{ i = 1 }^n
      \psi_{ l_i }
    \right)^2
  }{
    \left(
      \prod_{ i = 1 }^n
      \left(
        \lambda_{ l_i }
      \right)^{
        \left(
          1 +
          \frac{ \min( d - 2 , 1 ) }{ n }
        \right)
      }
    \right)
  }
  \right]
\end{split}
\end{equation}
for all 
$ k_1, k_2 \in \Z^d $,
$ t \in \R $
and all
$
  \varphi, \psi \in \Phi_0
$.
Next note that
Corollary~\ref{cor:discrete_conv}
ensures that
\begin{equation}
  \sup_{
    k_1 \in \Z^d
  }
  \left[
    \sum_{
      \substack{
        l_1, \ldots, l_n
        \in \Z^d
      \\
        l_1 + \ldots + l_n 
        = k_1
      }
    }
    \frac{
      \left(
        \lambda_{ k_1 }
      \right)^{
        \gamma
      }
    }{
      \left(
        \prod_{ i = 1 }^n
        \left(
          \lambda_{ l_i }
        \right)^{
          \left(
            1 +
            \frac{ \min( d - 2 , 1 ) }{ n }
          \right)
        }
      \right)
    }
  \right]
  < \infty
\end{equation}
for all 
$ 
  \gamma \in 
  \big(
    - \infty ,
    \frac{ d }{ 2 } 
    +
    \min( d - 2, 1 )
    +
    n 
    \left( 1 - \frac{ d }{ 2 } \right)
  \big)
$.
Therefore,
we obtain that
\begin{equation}
\label{eq:summability_CWP_2}
  \sum_{ 
    k_1 \in \Z^d 
  }
  \sum_{
    \substack{
      l_1, \ldots, l_n
      \in \Z^d
    \\
      l_1 + \ldots + l_n 
      = k_1
    }
  }
    \frac{
      \left(
        \lambda_{ k_1 }
      \right)^{
        \left( 
          2 \beta 
          -
          \max( 4 - d, 1 )
        \right)
      }
    }{
      \left(
        \prod_{ i = 1 }^n
        \left(
          \lambda_{ l_i }
        \right)^{
          \left(
            1 +
            \frac{ \min( d - 2 , 1 ) }{ n }
          \right)
        }
      \right)
    }
  <
  \infty
\end{equation}
for all 
$ 
  \beta \in 
  \big(
    - \infty ,
    \frac{ 
      4 -
      \left( d - 2 \right) n
    }{ 4 } 
  \big)
$.
Combining this
with \eqref{eq:L_inf_part_CWP}
and dominated convergence 
shows for every
$ 
  \beta \in 
  \big(
    - \infty ,
    1 -
    \frac{ 
      \left( d - 2 \right) n
    }{ 4 } 
  \big)
$ 
that
\begin{equation}
\label{eq:CWP_firstlimit}
  \sup_{
    t
    \in
    \R
  }
    \sum_{
      k_1, k_2
      \in
      \Z^d
    }
  \tfrac{
    \left|
      \E\left[
        \overline{
          \left< 
            g_{ k_1 } ,
            \bullet
            \left(
              V_t^{ 
                \varphi 
              }
            \right)^n
            \bullet
            \,
            -
            \,
            \bullet\left(
              V_t^{ \psi }
            \right)^n
            \bullet
          \right>_{ \! H } 
        }
          \left< 
            g_{ k_2 } ,
            \bullet
            \left(
              V_t^{ 
                \varphi
              }
            \right)^n
            \bullet
            \,
            -
            \,
            \bullet\left(
              V_t^{ 
                \psi
              }
            \right)^n
            \bullet
          \right>_{ \! H } 
      \right]
    \right|
  }{
    \left(
      \lambda_{ k_1 }
      \lambda_{ k_2 }
    \right)^{ - \beta }
  }
  \to 0
  \quad
  \text{as}
  \quad
      ( \Phi_{ 0, \leq 1 } )^2 
      \ni
      (
        \varphi,
        \psi
      )
      \rightarrow (1,1) .
\end{equation}
In the next step
let 
$ 
  h_{ k_1, l_1, \dots, l_n } 
  \colon \R^2 \to \R 
$,
$
  k_1, l_1, \dots, l_n \in \Z^d
$,
be functions defined
through
\begin{equation}
  h_{
    k_1, l_1, \dots, l_n
  }( t_1, t_2 )
:=
  \int_{ - \infty }^{ t_1 }
  \int_{ - \infty }^{ t_2 }
    e^{ 
      - \lambda_{k_1}
      \left( t_1 - s_1 + t_2 - s_2 \right)
      - 
      \left(
        \sum_{ i = 1 }^{ n } 
        \lambda_{ l_i }
      \right)
      \left| s_1 - s_2 \right|
    }
  \, ds_2 \, ds_1
\end{equation}
for all
$ t_1, t_2 \in \R $
and all
$ k_1, l_1, \dots, l_n \in \Z^d $.
Then observe that 
Lemma~\ref{lem:time_integralConv}
implies that
\begin{equation}
\begin{split}
&
  h_{ k_1, l_1, \dots, l_n }(t_1, t_1)
  - 2 h_{ k_1, l_1, \dots, l_n }(t_1, t_2)
  + h_{ k_1, l_1, \dots, l_n }(t_2, t_2)
\\[2ex] & =
  \frac{
    2
    \left(
      1 -
      e^{
        - \lambda_{ k_1 } ( t_2 - t_1 )
      }
    \right)
  }{
    \lambda_{ k_1 }
    \left(
      \lambda_{k_1} 
      +
      \sum_{ i = 1 }^n 
      \lambda_{ l_i } 
    \right)
  }
  - 2 \cdot
\begin{cases}
  \frac{
    e^{
      - 
      \left(
        \sum_{ i = 1 }^n  \lambda_{ l_i } 
      \right)  
      ( t_2 - t_1 )
    }
    -
    e^{
      - \lambda_{ k_1 } 
      ( t_2 - t_1 ) 
    }
  }{
    \left(
      \lambda_{ k_1 } - 
      \sum_{ i = 1}^n 
      \lambda_{ l_i }  
    \right)
    \left(
      \lambda_{ k_1 } + 
      \sum_{ i = 1 }^n
      \lambda_{ l_i } 
    \right)
  }
&
\colon
  \lambda_{ k_1 } 
  \neq
  \sum_{ i = 1 }^n \lambda_{l_i}  
\\[2ex]
  \frac{
    ( t_2 - t_1 )
    e^{
      - \lambda_{ k_1 } 
      ( t_2 -  t_1) 
    }
  }{
    \left(
      \lambda_{k_1} + 
      \left[ 
        \sum_{ i = 1 }^n \lambda_{ l_i } 
      \right] 
    \right)
  }
&
\colon
  \lambda_{ k_1 } 
  =
  \sum_{ i = 1 }^n  \lambda_{ l_i } 
\end{cases}
\\[2ex] & \leq
  \frac{
    2
    \left(
      1 -
      e^{
        - \lambda_{ k_1 } ( t_2 - t_1 )
      }
    \right)
  }{
    \lambda_{ k_1 }
    \left(
      \lambda_{ k_1 } 
      +
      \sum_{ i = 1 }^n  \lambda_{ l_i } 
    \right)
  }
\leq
  \frac{
    2 
    \left(
      \lambda_{k_1}
    \right)^{ 2 \alpha }
    \left( 
      t_2 - t_1 
    \right)^{ 2 \alpha }
  }{
    \lambda_{k_1}
    \left(
      \lambda_{k_1} +
      \sum_{ i = 1 }^n  
      \lambda_{ l_i } 
    \right)
  }
=
  \frac{
    2 
    \left(
      \lambda_{ k_1 }
    \right)^{ \left( 2 \alpha - 1 \right) }
    \left( 
      t_2 - t_1 
    \right)^{ 2 \alpha }
  }{
    \left(
      \lambda_{ k_1 } +
      \sum_{ i = 1 }^n  
      \lambda_{ l_i } 
    \right)
  }
\end{split}
\end{equation}  
for all 
$ k_1, l_1, \dots, l_n \in \Z^d $,
$ 
  \alpha \in [ 0, \frac{ 1 }{ 2 } ]
$
and all
$ t_1, t_2 \in \R $
with $ t_1 \leq t_2 $.
Lemma~\ref{lem:correlationConv}
hence shows that
\begin{equation}
\label{eq:proof_lem_time_diffAver}
\begin{split} 
& 
  \frac{ 1 }{
    \left( 2 \pi \right)^{ 2 d }
  }
  \,
          \E\!\left[
          \begin{split}
          &
            \left<
              g_{ - k_1 } 
            ,
              \left[
                \bullet(
                  V_{ t_1 }^{
                    \varphi
                  }
                )^n
                \bullet
                -
                \bullet(
                  V_{ t_1 }^{ 
                    \psi
                  }
                )^n
                \bullet
              \right]
              -
              \left[
                \bullet(
                  V_{ t_2 }^{
                    \varphi
                  }
                )^n
                \bullet
                -
                \bullet(
                  V_{ t_2 }^{ 
                    \psi
                  }
                )^n
                \bullet
              \right]
            \right>_{ \! H }
          \\ & \;
            \cdot
            \left<
              g_{ k_2 } 
              ,
              \left[
                \bullet(
                  V_{ t_1 }^{
                    \varphi
                  }
                )^n
                \bullet
                -
                \bullet(
                  V_{ t_1 }^{ 
                    \psi
                  }
                )^n
                \bullet
              \right]
              -
              \left[
                \bullet(
                  V_{ t_2 }^{
                    \varphi
                  }
                )^n
                \bullet
                -
                \bullet(
                  V_{ t_2 }^{ 
                    \psi
                  }
                )^n
                \bullet
              \right]
            \right>_{ \! H }
          \end{split}
          \right]
\\ & =
  n! \,
  \delta_{
    k_1, k_2
  }
  \!\!\!\!\!\!\!\!\!
    \sum_{
      \substack{
        l_1, \ldots, l_n
        \in \Z^d
      \\
        l_1 + \ldots + l_n 
        = k_1
      }
    }
  \!\!\!\!\!\!\!\!
    \tfrac{
      \left(
        \prod_{ i = 1 }^n
        \left( \varphi_{ l_i } \right)^2
        - 2
        \prod_{ i = 1 }^n
        \varphi_{ l_i } 
        \psi_{ l_i } 
        +
        \prod_{ i = 1 }^n
        \left( \psi_{ l_i } \right)^2
      \right)
      \left(
        h_{ k_1, l_1, \dots, l_n }(t_1, t_1)
        - 2 h_{ k_1, l_1, \dots, l_n }(t_1, t_2)
        + h_{ k_1, l_1, \dots, l_n }(t_2, t_2)
      \right)
    }{
      \left(
        \prod_{ i = 1 }^n
        \lambda_{ l_i } 
      \right)
    }
\\ & =
  n! \,
  \delta_{
    k_1, k_2
  }
  \!\!\!\!\!\!\!\!\!
    \sum_{
      \substack{
        l_1, \ldots, l_n
        \in \Z^d
      \\
        l_1 + \ldots + l_n 
        = k_1
      }
    }
  \!\!\!\!\!\!\!\!
    \tfrac{
      \left(
        \prod_{ i = 1 }^n
        \varphi_{ l_i } 
        - 
        \prod_{ i = 1 }^n
        \psi_{ l_i } 
      \right)^2
      \left(
        h_{ k_1, l_1, \dots, l_n }(t_1, t_1)
        - 2 h_{ k_1, l_1, \dots, l_n }(t_1, t_2)
        + h_{ k_1, l_1, \dots, l_n }(t_2, t_2)
      \right)
    }{
      \left(
        \prod_{ i = 1 }^n
        \lambda_{ l_i } 
      \right)
    }
\\ & \leq
  \left( n + 1 \right) ! 
  \,
  \delta_{
    k_1, k_2
  }
  \left[
    \sum_{
      \substack{
        l_1, \ldots, l_n
        \in \Z^d
      \\
        l_1 + \ldots + l_n 
        = k_1
      }
    }
    \frac{
      \left(
        \prod_{ i = 1 }^n
        \varphi_{ l_i } 
        - 
        \prod_{ i = 1 }^n
        \psi_{ l_i } 
      \right)^2
      \left(
        \lambda_{ k_1 }
      \right)^{ \left( 2 \alpha - 1 \right) }
    }{
      \left(
        \prod_{ i = 1 }^n
        \lambda_{ l_i } 
      \right)
      \left(
        \lambda_{ k_1 } +
        \sum_{ i = 1 }^n  
        \lambda_{ l_i } 
      \right)
    }
  \right]
  \left( 
    t_2 - t_1 
  \right)^{ 2 \alpha }
\\ & \leq
  \frac{
    \left( n + 1 \right) ! 
    \,
    \delta_{
      k_1, k_2
    }
  }{
    \left( 
      \lambda_{ k_1 }
    \right)^{
      \left(
        \max( 4 - d, 1 ) - 2 \alpha 
      \right)
    }
  }
  \left[
    \sum_{
      \substack{
        l_1, \ldots, l_n
        \in \Z^d
      \\
        l_1 + \ldots + l_n 
        = k_1
      }
    }
    \frac{
      \left(
        \prod_{ i = 1 }^n
        \varphi_{ l_i } 
        - 
        \prod_{ i = 1 }^n
        \psi_{ l_i } 
      \right)^2
    }{
      \left(
        \prod_{ i = 1 }^n
        \left(
          \lambda_{ l_i } 
        \right)^{
          \left(
            1 +
            \frac{ \min( d - 2 , 1 ) }{ n }
          \right)
        }
      \right)
    }
  \right]
  \left( 
    t_2 - t_1 
  \right)^{ 2 \alpha }
\end{split}
\end{equation}
for all 
$ k_1, k_2 \in \Z^d $,
$ \varphi, \psi \in \Phi_{ 0, \leq 1 } $,
$
  \alpha \in [ 0, \frac{ 1 }{ 2 } ]
$
and all
$ t_1, t_2 \in \R $
with $ t_1 \leq t_2 $.
In addition,
Corollary~\ref{cor:discrete_conv}
ensures that
\begin{equation}
  \sup_{
    k_1 \in \Z^d
  }
  \left[
    \sum_{
      \substack{
        l_1, \ldots, l_n
        \in \Z^d
      \\
        l_1 + \ldots + l_n 
        = k_1
      }
    }
    \frac{
      \left(
        \lambda_{ k_1 }
      \right)^{
        \gamma
      }
    }{
      \left(
        \prod_{ i = 1 }^n
        \left(
          \lambda_{ l_i }
        \right)^{
          \left(
            1 +
            \frac{ \min( d - 2 , 1 ) }{ n }
          \right)
        }
      \right)
    }
  \right]
  < \infty
\end{equation}
for all 
$ 
  \gamma \in 
  \big(
    - \infty ,
    \frac{ d }{ 2 } 
    +
    \min( d - 2, 1 )
    +
    n 
    \left( 1 - \frac{ d }{ 2 } \right)
  \big)
$.
Therefore,
we obtain that
\begin{equation}
  \sum_{ 
    k_1 \in \Z^d 
  }
  \sum_{
    \substack{
      l_1, \ldots, l_n
      \in \Z^d
    \\
      l_1 + \ldots + l_n 
      = k_1
    }
  }
    \frac{
      \left(
        \lambda_{ k_1 }
      \right)^{
        \left( 
          2 \alpha + 2 \beta 
          -
          \max( 4 - d, 1 )
        \right)
      }
    }{
      \left(
        \prod_{ i = 1 }^n
        \left(
          \lambda_{ l_i }
        \right)^{
          \left(
            1 +
            \frac{ \min( d - 2 , 1 ) }{ n }
          \right)
        }
      \right)
    }
  <
  \infty
\end{equation}
for all 
$ 
  \alpha,
  \beta \in \R 
$
with
$
  \alpha + \beta
  <
  1 -
  \frac{ 
    \left( d - 2 \right) n
  }{ 4 } 
$.
Combining this
with \eqref{eq:proof_lem_time_diffAver}
and dominated convergence 
implies for every
$ \alpha \in [ 0, \frac{ 1 }{ 2 } ] $,
$ \beta \in \R $
with 
$ 
  \alpha + \beta 
  < 
  1 - 
  \frac{ ( d - 2 ) n }{ 4 }
$
that
\begin{equation}
\label{eq:lem_time_diffConv}
  \sup_{ 
    \substack{
      t_1, t_2 \in \R 
    \\
      t_1 \neq t_2
    }
  }
  \left[
    \!\!\!
    \begin{array}{c}
      \sum\limits_{ 
        k_1, k_2 
        \in \Z^d
      } 
      \frac{
        \tiny{
        \left|
          \E\!\left[
          \begin{split}
          &
            \left<
              g_{ - k_1 } 
              ,
              \left[
                \bullet(
                  V_{ t_1 }^{
                    \varphi
                  }
                )^n
                \bullet
                -
                \bullet(
                  V_{ t_1 }^{ 
                    \psi
                  }
                )^n
                \bullet
              \right]
              -
              \left[
                \bullet(
                  V_{ t_2 }^{
                    \varphi
                  }
                )^n
                \bullet
                -
                \bullet(
                  V_{ t_2 }^{ 
                    \psi
                  }
                )^n
                \bullet
              \right]
            \right>_{ \! H }
          \\ & \;
            \cdot
            \left<
              g_{ k_2 } 
              ,
              \left[
                \bullet(
                  V_{ t_1 }^{
                    \varphi
                  }
                )^n
                \bullet
                -
                \bullet(
                  V_{ t_1 }^{ 
                    \psi
                  }
                )^n
                \bullet
              \right]
              -
              \left[
                \bullet(
                  V_{ t_2 }^{
                    \varphi
                  }
                )^n
                \bullet
                -
                \bullet(
                  V_{ t_2 }^{ 
                    \psi
                  }
                )^n
                \bullet
              \right]
            \right>_{ \! H }
          \end{split}
          \right]
        \right|
        }
      }{ 
        \left( 
          \lambda_{ k_1 }
          \lambda_{ k_2 }
        \right)^{ - \beta }
        \left|
          t_1 - t_2
        \right|^{ 2 \alpha }
      }
    \end{array}
    \!\!\!
  \right]
  \to 0
  \;\;
  \text{as}
  \;\;
      ( \Phi_{ 0, \leq 1 } )^2 
      \ni
      (
        \varphi,
        \psi
      )
      \rightarrow (1,1) .
\end{equation}
Combining
\eqref{eq:CWP_firstlimit}
and
\eqref{eq:lem_time_diffConv}
with Lemma~\ref{lem:hypercontractivity}
completes the proof of 
Proposition~\ref{prop:regularity_WickConv}.
\end{proof}

Proposition~\ref{prop:regularity_WickConv}
shows convergence of 
convolutional Wick powers 
under the assumption
that
$ n, d \in \{ 2, 3, \dots \} $
with
$
  \frac{ n + 1 }{ n - 1 }
  > 
  \frac{ d }{ 2 }
$.
In the case
$ n, d \in \{ 2, 3, \dots \} $
with
$
  \frac{ n + 1 }{ n - 1 }
  \leq
  \frac{ d }{ 2 }
$,
convolutional Wick
powers fail to converge.
This is the subject of the
next lemma.

\begin{lem}[Divergence of
convolutional Wick powers]
\label{lem:limitationCWP}
Assume the setting of 
Subsection~\ref{sec:setting},
let 
$ n, d \in \{ 2, 3, \dots \} $
with
$
  \frac{ n + 1 }{ n - 1 }
  \leq
  \frac{ d }{ 2 }
$
and let
$ 
  C_0, C_1, \dots, C_{ n - 1 } \colon
  \Phi_0
  \to \R
$
be arbitrary functions.
Then it holds for every
$ v \in \Z^d $
and every $ t \in \R $ that
\begin{equation}
  \E\!
  \left[
    \left|
      \left<
        g_v ,
        \int_{ - \infty }^t
        e^{ A ( t - s ) }
        \left(
          \left( 
            V^{ \varphi }_s
          \right)^n
          -
          \sum_{ k = 0 }^{ n - 1 }
          C_k( \varphi ) 
          \cdot
          \left( 
            V^{ \varphi }_s 
          \right)^k
        \right)
        ds
      \right>_{ \! \! H }
    \right|^2
  \right] 
  \to \infty
  \qquad
  \text{as}
  \qquad
  \Phi_0 \ni \varphi
  \rightarrow 1 .
\end{equation}
\end{lem}

\begin{proof}[Proof of Lemma~\ref{lem:limitationCWP}]
Throughout this proof
let 
$ 
  \hat{C}_0, \hat{C}_1, \dots
  \hat{C}_n \colon \Phi_0
  \to \R
$
be the unique functions satisfying
$
  \hat{C}_0( 0 )
  =
  - C_0( 0 )
$,
$
  \hat{C}_1( 0 )
  =
  - C_1( 0 )
$,
$ \dots $,
$
  \hat{C}_{ n - 1 }( 0 )
  =
  - C_{ n - 1 }( 0 )
$,
$
  \hat{C}_n( 0 ) 
  = 1
$
and
\begin{equation}
\begin{split}
&
  x^n
  -
  \sum_{ k = 0 }^{ n - 1 }
  C_k( \varphi )
  \cdot
  x^k
=
  \sum_{ k = 0 }^n
  \hat{C}_k( \varphi )
  \cdot
  \left[
      \sum_{ v \in \Z^d }
      \tfrac{ ( \varphi_v )^2 }{
        \lambda_v
      }
  \right]^{
    \! \frac{ k }{ 2 }
  }
  \cdot
  H_k\!\left(
  \frac{
    x
  }{
    \sqrt{
      \sum_{ v \in \Z^d }
      \frac{ ( \varphi_v )^2 }{
        \lambda_v
      }
    }
  }
  \right)
\end{split}
\end{equation}
for all 
$ x \in \R $,
$ 
  \varphi \in \Phi_0 \backslash \{ 0 \} 
$
and all $ t \in \R $
(cf.\ \eqref{eq:lem_limitation_firstdispl}).
Then Lemma~\ref{lem:correlationConv}
implies that
\begin{equation}
\label{eq:limitation_main_estimateConv}
\begin{split}
&
  \E\!\left[
    \left|
      \left<
        g_v ,
        \int_{ - \infty }^t
        e^{ A ( t - s ) }
        \left(
          \left( V_s^{ \varphi }
          \right)^n
          -
          \sum_{ k = 0 }^{ n - 1 }
          C_k( \varphi ) \cdot
          ( V_s^{ \varphi } )^k
        \right)
        ds
      \right>_{ \!\! H }
    \right|^2
  \right]
\\ &  =
  \E\!\left[
    \left|
      \sum_{ k = 0 }^n
      \left<
        g_v ,
        \int_{ - \infty }^t
        e^{ A (t - s ) }
        \left[
          \hat{C}_k( \varphi ) 
          \left(
            : \!
            ( V_t^{ \varphi }
            )^k
            \! :
          \right)
        \right]
        ds
      \right>_{ \! H }
    \right|^2
  \right]
\\ & =
  \sum_{ k, l = 0 }^n
  \hat{C}_k( \varphi ) 
  \cdot
  \hat{C}_l( \varphi ) 
  \cdot
  \E\!\left[
    \overline{
      \left<
        g_v ,
        \bullet
          ( V_t^{ \varphi }
          )^k
        \bullet
      \right>_{ \! H }
    }
    \left<
      g_v ,
      \bullet
        ( V_t^{ \varphi }
        )^l
      \bullet
    \right>_{ \! H }
  \right]
\\ & =
  \sum_{ k = 0 }^n
  \left|
    \hat{C}_k( \varphi ) 
  \right|^2
  \E\!\left[
    \left|
      \left<
        g_v ,
          \bullet
          ( V_t^{ \varphi }
          )^k
          \bullet
      \right>_{ H }
    \right|^2
  \right]
\geq
  \left|
    \hat{C}_n( \varphi ) 
  \right|^2
  \E\!\left[
    \left|
      \left<
        g_v ,
          \bullet
          ( V_t^{ \varphi }
          )^n
          \bullet
      \right>_{ H }
    \right|^2
  \right]
\\ & =
  \frac{ 
    n!
    \left( 2 \pi \right)^{ 2 d }
  }{
    \lambda_v
  }
  \left[
  \sum_{
    \substack{
      l_{ 1 },
      \dots,
      l_{ n }
      \in
      \Z^d
    \\
      l_1 
      +
      \ldots
      +
      l_n =
      v
    }
  }
  \frac{
    \left(
      \prod_{ i = 1 }^n
      \left(
        \varphi_{
          l_i
        }
      \right)^2
    \right)
  }{
    \left(
      \prod_{ i = 1 }^n
      \lambda_{ l_i }
    \right)
    \left(
      \lambda_v 
      +
      \sum_{ i = 1 }^n  
      \lambda_{ l_i } 
    \right)
  }
  \right]
\end{split}
\end{equation}
for all 
$ t \in \R $,
$ v \in \Z^d $
and all $ \varphi \in \Phi_0 $.
Combining this with
Lemma~\ref{lem:sum_tool}
completes the proof of
Lemma~\ref{lem:limitationCWP}.
\end{proof}

\subsection{Summary}
\label{sec:summary}

The following table 
briefly summarizes
the results of
Proposition~\ref{prop:regularity_Wick},
Proposition~\ref{prop:regularity_WickAver}
and Proposition~\ref{prop:regularity_WickConv}
and of 
Lemma~\ref{lem:limitation},
Lemma~\ref{lem:limitationAWP}
and 
Lemma~\ref{lem:limitationCWP}.
Recall that the main arguments for the results 
from Propositions~\ref{prop:regularity_Wick},
\ref{prop:regularity_WickAver}
and \ref{prop:regularity_WickConv}
presented in the table are certain 
summability properities;
see \eqref{eq:L_inf_part_000} and 
\eqref{eq:summability_WP_2} in the case of 
Wick powers,
\eqref{eq:AWP_L_inf_part_00} and 
\eqref{eq:summability_AWP_2} 
in the 
case of averaged Wick powers
and \eqref{eq:L_inf_part_CWP} and \eqref{eq:summability_CWP_2}
in the case of convolutional
Wick powers.
In the table $ \varepsilon
\in ( 0, \infty) $
is an arbitrarily small
positive real number,
$
  \mathcal{C}^{ \alpha }_{
    \mathcal{P}
  } 
$
is an abbreviation for
$
  \mathcal{C}^{ \alpha }_{
    \mathcal{P}
  }( [0, 2 \pi]^d, \R )
$
where $ \alpha \in \R $
and $ d \in \N $
and the expressions \emph{WP}, 
\emph{AWP} and 
\emph{CWP} 
are abbreviations 
for \emph{Wick powers}, 
\emph{averaged Wick
powers} and 
\emph{convolutional Wick powers}
respectively.

\begin{center}
%
%
%
%
\begin{tabular}[c]{|| l || c | c | c | c | c | c ||}
\hline
\hline
\,\,\,\,\,\,
  \vdots
&
  \vdots
&
  \vdots
&
  \vdots
&
  \vdots
&
  \vdots
&
\\  
\hline
  \parbox[t]{0.9cm}{
      \ 
    \\
      \ 
    \\
      \ 
    \\
      n = 5
  }
&
  \, 
  \parbox[t]{1cm}{
    {\bf WP:}
  \\
    $
      \mathcal{C}^{ - \varepsilon 
      }_{ 
        \mathcal{P} 
      }
    $
  \\
    \ 
  \\
    {\bf AWP:}
  \\
    $
      \mathcal{C}^{ 
        1 / 5 - \varepsilon 
      }_{ \mathcal{P} 
      }
    $
  \\
  \ 
  \\
    {\bf CWP:}
    $
      \mathcal{C}^{ 2 - \varepsilon }_{ 
        \mathcal{P} 
      }
    $
  }  
  \,
&
  \parbox[t]{1.5cm}{
  \ 
  \\
    No WP
  \\
  \ 
  \\
  \ 
  \\
    No AWP
  \\
  \ 
  \\
  \ 
  \\
    No CWP
  }  
&
  \parbox[t]{1.5cm}{
  \ 
  \\
    No WP
  \\
  \ 
  \\
  \ 
  \\
    No AWP
  \\
  \ 
  \\
  \ 
  \\
    No CWP
  }  
&
  \parbox[t]{1.5cm}{
  \ 
  \\
    No WP
  \\
  \ 
  \\
  \ 
  \\
    No AWP
  \\
  \ 
  \\
  \ 
  \\
    No CWP
  }  
&
  \parbox[t]{1.5cm}{
  \ 
  \\
    No WP
  \\
  \ 
  \\
  \ 
  \\
    No AWP
  \\
  \ 
  \\
  \ 
  \\
    No CWP
  }  
& 
  \parbox[t]{0.6cm}{
      \ 
    \\
      \ 
    \\
      \ 
    \\
      \dots
  }
\\  
\hline
  \parbox[t]{1cm}{
      \ 
    \\
      \ 
    \\
      \ 
    \\
      n = 4
  }
&
  \parbox[t]{1cm}{
    {\bf WP:}
    $
      \mathcal{C}^{ - \varepsilon }_{ 
        \mathcal{P} 
      }
    $
  \\
    \ 
  \\
    {\bf AWP:}
    $
      \mathcal{C}^{ 
        1 / 4 - \varepsilon 
      }_{ \mathcal{P} 
      } 
    $
  \\
    \ 
  \\
    {\bf CWP:}
    $
      \mathcal{C}^{ 2 - \varepsilon }_{ 
        \mathcal{P} 
      } 
    $
  }  
& 
  \parbox[t]{1.2cm}{
   \ 
  \\
    No WP
  \\
    \ 
  \\
    {\bf AWP:}
    $
      \mathcal{C}^{ 
        - 1 - \varepsilon 
      }_{ \mathcal{P} 
      } 
    $
  \\
    \ 
  \\
    {\bf CWP:}
    $
      \mathcal{C}^{ - \varepsilon }_{ 
        \mathcal{P} 
      } 
    $
  }  
& 
  \parbox[t]{1.5cm}{
  \ 
  \\
    No WP
  \\
  \ 
  \\
  \ 
  \\
    No AWP
  \\
  \ 
  \\
  \ 
  \\
    No CWP
  }  
& 
  \parbox[t]{1.5cm}{
  \ 
  \\
    No WP
  \\
  \ 
  \\
  \ 
  \\
    No AWP
  \\
  \ 
  \\
  \ 
  \\
    No CWP
  }  
& 
  \parbox[t]{1.5cm}{
  \ 
  \\
    No WP
  \\
  \ 
  \\
  \ 
  \\
    No AWP
  \\
  \ 
  \\
  \ 
  \\
    No CWP
  }  
&
  \parbox[t]{0.6cm}{
      \ 
    \\
      \ 
    \\
      \ 
    \\
      \dots
  }

\\  
\hline
  \parbox[t]{1cm}{
      \ 
    \\
      \ 
    \\
      \ 
    \\
      n = 3
  }
&
  \parbox[t]{1cm}{
    {\bf WP:}
    $
      \mathcal{C}^{ - \varepsilon }_{ 
        \mathcal{P} 
      }
    $
  \\
    \ 
  \\
    {\bf AWP:}
    $
      \mathcal{C}^{ 
        1 / 3 - \varepsilon 
      }_{ 
        \mathcal{P} 
      } 
    $
  \\
  \ 
  \\
    {\bf CWP:}
    $
      \mathcal{C}^{ 2 - \varepsilon }_{ 
        \mathcal{P}
      }
    $
  }  
&
  \parbox[t]{1.2cm}{
  \ 
  \\
    No WP
  \\
    \ 
  \\
    {\bf AWP:}
    $
      \mathcal{C}^{ 
        - 1 / 2 - \varepsilon 
      }_{ 
        \mathcal{P} 
      } 
    $
  \\
    \ 
  \\
    {\bf CWP:}
    $
      \mathcal{C}^{ 1/2 - \varepsilon }_{ 
        \mathcal{P} 
      }
    $
  }  
&
  \parbox[t]{1.5cm}{
  \ 
  \\
    No WP
  \\
  \ 
  \\
  \ 
  \\
    No AWP
  \\
  \ 
  \\
  \ 
  \\
    No CWP
  }  
&
  \parbox[t]{1.5cm}{
  \ 
  \\
    No WP
  \\
  \ 
  \\
  \ 
  \\
    No AWP
  \\
  \ 
  \\
  \ 
  \\
    No CWP
  }  
&
  \parbox[t]{1.5cm}{
  \ 
  \\
    No WP
  \\
  \ 
  \\
  \ 
  \\
    No AWP
  \\
  \ 
  \\
  \ 
  \\
    No CWP
  }  
&
  \parbox[t]{0.6cm}{
      \ 
    \\
      \ 
    \\
      \ 
    \\
      \dots
  }

\\  
\hline
  \parbox[t]{1cm}{
      \ 
    \\
      \ 
    \\
      \ 
    \\
      n = 2
  }
&
  \parbox[t]{1cm}{
    {\bf WP:}
  \\
    $
      \mathcal{C}^{ - \varepsilon }_{ 
        \mathcal{P} 
      }
    $
  \\
    \ 
  \\
    {\bf AWP:}
    $
      \mathcal{C}^{ 
        1 / 2 - \varepsilon 
      }_{ 
        \mathcal{P} 
      } 
    $
  \\
    \ 
  \\
    {\bf CWP:}
    $
      \mathcal{C}^{ 2 - \varepsilon }_{ 
        \mathcal{P} 
      }
    $
  }  
&
  \parbox[t]{1.5cm}{
    {\bf WP:}
    $
      \mathcal{C}^{ - 1 - \varepsilon }_{ 
        \mathcal{P} 
      }
    $
  \\
    \ 
  \\
    {\bf AWP:}
    $
      \mathcal{C}^{ 
        - \varepsilon 
      }_{ 
        \mathcal{P} 
      } 
    $
  \\
    \ 
  \\
    {\bf CWP:}
    $
      \mathcal{C}^{ 1 - \varepsilon }_{ 
        \mathcal{P} 
      }
    $
  }  
&
  \parbox[t]{1.3cm}{
  \ 
  \\
    No WP
  \\
    \ 
  \\
    {\bf AWP:}
    $
      \mathcal{C}^{ 
        - 1 - \varepsilon 
      }_{ 
        \mathcal{P} 
      } 
    $
  \\
    \ 
  \\
    {\bf CWP:}
    $
      \mathcal{C}^{ - \varepsilon }_{ 
        \mathcal{P} 
      }
    $
  }  
&
  \parbox[t]{1.5cm}{
    \ 
  \\
    No WP
  \\
    \ 
  \\
    {\bf AWP:}
    $
      \mathcal{C}^{ 
        - 2 - \varepsilon 
      }_{ 
        \mathcal{P} 
      } 
    $
  \\
    \ 
  \\
    {\bf CWP:}
    $
      \mathcal{C}^{ - 1 - \varepsilon }_{ 
        \mathcal{P} 
      }
    $
  }  
&
  \parbox[t]{1.5cm}{
    \ 
  \\
    No WP 
  \\
    \ 
  \\
    \ 
  \\
    No AWP
  \\
    \ 
  \\
    \ 
  \\
    No CWP
  }  
&
  \parbox[t]{0.6cm}{
      \ 
    \\
      \ 
    \\
      \ 
    \\
      \dots
  }
\\ \hline \hline
&
  d = 2
&
  d = 3
& 
  d = 4
&
  d = 5
&
  d = 6
&
  \dots
\\ 
  \hline
  \hline
\end{tabular}
\end{center}

\bigskip

\section{Stochastic
partial differential equations
(SPDEs)}
\label{sec:SPDEs}

\subsection{Local existence
and uniqueness of mild 
solutions of deterministic nonautonomous
partial differential equations}
\label{sec:det_existence}

This subsection investigates
local existence and uniqueness
questions
for mild solutions of
deterministic nonautonomous
evolution equations
of the form
\begin{equation}
\label{eq:PDEs}
  \frac{ \partial }{ \partial t }
  x(t)
=
  A x(t)
  +
  \sum_{ i = 1 }^n
  F_i( t, x(t) ) 
\end{equation}
on a 
real Banach space
$
  \left( 
    U, \left\| \cdot \right\|_U
  \right)
$
for $ t \in [t_0,T] $
where
$ t_0, T \in \R $
are real numbers
with $ t_0 < T $,
where
$  
  A \colon D(A) \subset U 
  \to U
$
is a negative generator of
a strongly continuous analytic 
semigroup,
where $ n \in \N $
is a natural number
and where
$ 
  F_1, \dots, F_n
$
are suitable functions that are 
locally Lipschitz continuous 
on appropriate spaces.

To investigate these questions, 
we impose the following setting.
Throughout this subsection,
let 
$ 
  \left( 
    U, \left\| \cdot \right\|_{ U } 
  \right) 
$
be a real Banach space,
let 
$ A \colon D(A) \subset U \to U $
be a negative generator of a strongly continuous 
analytic semigroup on $ U $
and let
$ 
  \left( 
    U_r, 
    \left\| \cdot \right\|_{ U_r } 
  \right) 
  :=
  \left( 
    D( (- A)^r ) ,
    \left\|
      ( - A )^r
      ( \cdot )
    \right\|_U
  \right) 
$
for all $ r \in \R $.
Next define
\begin{equation}
\label{eq:Fnorm}
\begin{split}
&
  \left\|
    F
  \right\|_{
    \mathcal{C}^n_{ \alpha, \beta, \gamma, \delta }( [t_0, T] )
  }
  :=
\\ & 
  \sup_{ t \in [t_0,T] }
  \sum_{ i = 1 }^n
  \Bigg[
    \left\| F_i( t, 0 ) \right\|_{
      U_{ \alpha_i }
    }
    +
    \sup_{ 
      \substack{
        x, y \in 
        U_{ 
          \max( \beta_i, \gamma_i ) 
        } 
      \\
        x \neq y
      }
    }
    \frac{
      \left\|
        F_i( t, x ) - F_i( t, y )
      \right\|_{
        U_{ \alpha_i }
      }
    }{
      \big(
        1 
        + 
        \left\| x \right\|_{ U_{ \beta_i } }^{ \delta_i }
        + 
        \left\| y \right\|_{ U_{ \beta_i } }^{ \delta_i }
        \!
      \big)
      \left\| x - y \right\|_{ U_{ \gamma_i } }
    }
  \Bigg]
  \in [0,\infty]
\end{split}
\end{equation}
for all
$
  F = 
  ( F_1, \dots, F_n )
  \in 
  C( [t_0,T] \times U_{ \max( \beta_1, \gamma_1 ) }, U_{ \alpha_1 } )
    \times
  \ldots
    \times
  C( [t_0,T] \times U_{ \max( \beta_n, \gamma_n ) }, U_{ \alpha_n } )
$,
$
  \alpha = (\alpha_1, \dots, \alpha_n), 
  \beta = (\beta_1, \dots, \beta_n), 
  \gamma = (\gamma_1, \dots, \gamma_n) 
  \in \mathbb{R}^n
$,
$
  \delta = (\delta_1, \dots, \delta_n) 
  \in [0,\infty)^n
$,
$ t_0 \in ( - \infty, T ) $,
$ T \in \R $
and all
$ n \in \mathbb{N} $.
Furthermore, define
\begin{multline}
\label{eq:Fspace}
  \mathcal{C}^n_{ \alpha, \beta, \gamma, \delta }( [t_0, T] )
  :=
  \bigg\{
    F \in 
    \Big(
    C( [t_0,T] \times U_{ \max( \beta_1, \gamma_1 ) }, U_{ \alpha_1 } )
    \times
    \ldots
\\
    \times
    C( [t_0,T] \times U_{ \max( \beta_n, \gamma_n ) }, U_{ \alpha_n } )
    \Big)
    \colon
    \left\| F \right\|_{
      \mathcal{C}^n_{ \alpha, \beta, \gamma, \delta }( [t_0, T] )
    }
    < \infty
  \bigg\}
\end{multline}
for all
$
  \alpha = (\alpha_1, \dots, \alpha_n), 
  \beta = (\beta_1, \dots, \beta_n), 
  \gamma = (\gamma_1, \dots, \gamma_n) 
  \in \mathbb{R}^n
$,
$
  \delta = (\delta_1, \dots, \delta_n) 
$
$
  \in [0,\infty)^n
$,
$ t_0 \in (-\infty,T) $,
$ T \in \R $
and all
$ n \in \mathbb{N} $.
Observe that the pairs
$
  \big(
    \mathcal{C}^n_{ \alpha, \beta, \gamma, \delta }( [t_0, T] ) ,
$
$
    \left\|
      \cdot
    \right\|_{
      \mathcal{C}^n_{ \alpha, \beta, \gamma, \delta }( [t_0, T] )
    }
    \!
  \big)
$
for
$
  \alpha, \beta, \gamma \in \mathbb{R}^n
$,
$
  \delta \in [0,\infty)^n
$,
$ t_0 \in (-\infty,T) $,
$ T \in \R $
and
$ n \in \mathbb{N} $
are normed real vector spaces.
In the next step define 
\begin{multline}
\label{eq:Fspace_inf}
  \mathcal{C}^n_{ \alpha, \beta, \gamma, \delta }( [t_0, \infty) )
  :=
  \bigg\{
    ( F_1, \dots, F_n ) 
    \in 
    \Big(
    C( [t_0,\infty) \times U_{ \max( \beta_1, \gamma_1 ) }, U_{ \alpha_1 } )
    \times
    \ldots
\\
    \times
    C( [t_0,\infty) \times U_{ \max( \beta_n, \gamma_n ) }, U_{ \alpha_n } )
    \Big)
    \colon
   \Big(
    \forall \, T \in (t_0,\infty) 
    \colon
\\
    \| 
      ( 
        F_1|_{ 
          [t_0,T] \times U_{ \max( \beta_1, \gamma_1 ) } 
        }, 
        \dots,
        F_n|_{ 
          [t_0,T] \times U_{ \max( \beta_n, \gamma_n ) } 
        }
      )
    \|_{
      \mathcal{C}^n_{ \alpha, \beta, \gamma, \delta }( [t_0, T] )
    }
    < \infty
    \Big)
  \bigg\}
\end{multline}
for all
$
  \alpha = (\alpha_1, \dots, \alpha_n), 
  \beta = (\beta_1, \dots, \beta_n), 
  \gamma = (\gamma_1, \dots, \gamma_n) 
  \in \mathbb{R}^n
$,
$
  \delta = (\delta_1, \dots, \delta_n) 
$
$
  \in [0,\infty)^n
$,
$ t_0 \in \R $
and all
$ n \in \mathbb{N} $.
Moreover, we equip
$
  \mathcal{C}^n_{ \alpha, \beta, \gamma, \delta }( [t_0, \infty) )
$
with the metric
$
  d_{
    \mathcal{C}^n_{ \alpha, \beta, \gamma, \delta }( [t_0, \infty) )
  } 
  \colon
$
$
  \mathcal{C}^n_{ \alpha, \beta, \gamma, \delta }( [t_0, \infty) )
  \times
  \mathcal{C}^n_{ \alpha, \beta, \gamma, \delta }( [t_0, \infty) )
  \to [0,\infty)
$
defined through
\begin{multline}
  d_{
    \mathcal{C}^n_{ \alpha, \beta, \gamma, \delta }( [t_0, \infty) )
  }( F, G )
  :=
  \sum_{ k = 1 }^{ \infty }
  \frac{ 1 }{ 2^k }
  \min\!\Big( 1 ,
  \big\|
    \big( 
      ( F_1 - G_1 )|_{
        [t_0, t_0 + k] 
        \times 
        U_{
          \max( \beta_1, \gamma_1 )
        }
      }
      ,
      \dots 
      ,
\\ 
      ( F_n - G_n )|_{
        [t_0, t_0 + k]
        \times 
        U_{
          \max( \beta_n, \gamma_n )
        }
      }
    \big)
  \big\|_{
    \mathcal{C}^n_{
      \alpha, \beta, \gamma, \delta
    }(
      [t_0, t_0 + k]
    )
  }
  \Big)
\end{multline}
for all
$
  F = (F_1, \dots, F_n )
$,
$
  G = ( G_1, \dots, G_n )
  \in
  \mathcal{C}^n_{ \alpha, \beta, \gamma, \delta }( [t_0, \infty) )
$,
$
  \alpha = (\alpha_1, \dots, \alpha_n)
$,
$
  \beta = (\beta_1, \dots, \beta_n)
$, 
$
  \gamma = (\gamma_1, \dots, \gamma_n)
$
$
  \in \mathbb{R}^n
$,
$
  \delta = (\delta_1, \dots, \delta_n) 
$
$
  \in [0,\infty)^n
$,
$ t_0 \in \R $
and all
$ n \in \mathbb{N} $.
Finally, note that
the triangle inequality and
the definition of
$
  \left\| F \right\|_{
    \mathcal{C}^n_{ \alpha, \beta, \gamma, \delta }( [t_0, T] )
  }
$
imply that
\begin{equation}
\label{eq:linear_growth_bound}
\begin{split}
&
  \left\|
    F_i( t, x )
  \right\|_{
    U_{ \alpha_i }
  }
\leq
  \left\|
    F_i( t, x ) -
    F_i( t, 0 )
  \right\|_{
    U_{ \alpha_i }
  }
  +
  \left\|
    F_i( t, 0 )
  \right\|_{
    U_{ \alpha_i }
  }
\\ & \leq
  \left[
  \sup_{ 
    y \in U_{ \max( \beta_i, \gamma_i ) \backslash \{ 0 \} } 
  }
  \tfrac{
    \left\|
      F_i( t, y ) -
      F_i( t, 0 )
    \right\|_{
      U_{ \alpha_i } 
    }
  }{
    \big(
      1 + \| y \|_{ U_{ \beta_i } }^{ \delta_i }
    \big)
    \|
      y
    \|_{
      U_{ \gamma_i }
    }
  }
  +
  \left\|
    F_i( t, 0 )
  \right\|_{
    U_{ \alpha_i }
  }
  \right]
  \!
  \big(
    1 + \| x \|_{ U_{ \beta_i } }^{ \delta_i }
  \big)
  \big(
    1 +
    \|
      x
    \|_{
      U_{ \gamma_i }
    }
  \big)
\\ & \leq
  \left\| F \right\|_{
    \mathcal{C}^n_{ \alpha, \beta, \gamma, \delta }( [t_0, T] )
  }
  \big(
    1 + \| x \|_{ U_{ \beta_i } }^{ \delta_i }
  \big)
  \,
  \big(
    1 +
    \|
      x
    \|_{
      U_{ \gamma_i }
    }
  \big)
\end{split}
\end{equation}
for all 
$ t \in [t_0,T] 
$,
$
  x \in U_{
    \max( \beta_i, \gamma_i )
  }
$,
$ 
  i \in \{ 1, 2, \dots, n \}
$,
$
  F = (F_1, \dots, F_n) 
  \in 
$
$
  \mathcal{C}^n_{ \alpha, \beta, \gamma, \delta }( [t_0, T] )
$,
$
  \alpha = (\alpha_1, \dots, \alpha_n), 
  \beta = (\beta_1, \dots, \beta_n), 
  \gamma = (\gamma_1, \dots, \gamma_n) 
  \in \mathbb{R}^n
$,
$
  \delta = (\delta_1, \dots, \delta_n) 
  \in [0,\infty)^n
$,
$ t_0 \in (-\infty,T) $,
$ T \in \R $
and all
$ n \in \mathbb{N} $.

\begin{lem}[Local existence and uniqueness
of mild solutions]
\label{lem:existence}
Assume the setting in the beginning 
of Subsection~\ref{sec:det_existence},
let 
$ r_0, t_0 \in \R $,
$ T \in (t_0, \infty) $,
$ v \in U_{ r_0 } $,
$ n \in \N $,
$ 
  \alpha = ( \alpha_1, \dots, \alpha_n ) \in \R^n
$,
$
  \beta = ( \beta_1, \dots, \beta_n ) ,
  \gamma = ( \gamma_1, \dots, \gamma_n ) 
  \in [r_0, \infty)^n
$,
$
  \delta = ( \delta_1, \dots, \delta_n ) 
  \in [0,\infty)^n
$,
$
  r_1
  \in
  \left[  
    \max( \beta_1, \dots, \beta_n, \gamma_1, \dots, \gamma_n ) ,
    1 +
    \min( \alpha_1, \dots, \alpha_n )
  \right)
$
with
$
  \max_{ i \in \{1,\dots,n\} }
  [
    \gamma_i
    - \min( \alpha_i , r_0 ) 
    + 
    \delta_i
    ( \beta_i - r_0 ) 
  ]
  < 
  1 
$
and let
$
  F = ( F_1, \dots, F_n ) \in 
  \mathcal{C}^n_{ \alpha, \beta, \gamma, \delta }( [t_0, T] ) 
$.
Then there exist a real number
$ \tau \in (t_0,T] $
such that there exists
a unique continuous function
$
  x \colon [t_0,\tau] \to
  U_{ r_0 }
$
satisfying
$
  x|_{ (t_0,\tau ] }
  \in
  C(
    (t_0,\tau ], U_{ r_1 }
  )
$,
$
  \sup_{
    s \in (t_0, \tau ]
  }
  \left( s - t_0 \right)^{
    \left( r_1 - r_0 \right)
  }
  \| 
    x(s)
  \|_{ U_{ r_1 } }
  < \infty
$
and
$
  x(t)
=
  e^{ A ( t - t_0 ) } \, v
+
  \sum_{ i = 1 }^n
  \int_{ t_0 }^t
  e^{ A ( t - s ) }
  \,
  F_i( s, x(s) ) \, ds
$
for all 
$ t \in [t_0,\tau] $.
\end{lem}
Observe that all integrals appearing in Lemma~\ref{lem:existence}
are well-defined. Indeed, under the assumptions of
Lemma~\ref{lem:existence} it holds that if
$ \tau \in (t_0,T] $
and if
$
  x \colon [t_0,\tau] \to
  U_{ r_0 }
$
is a continuous function which satisfies
$
  x|_{ (t_0,\tau ] }
  \in
  C(
    (t_0,\tau ], U_{ r_1 }
  )
$
and
$
  \sup_{
    s \in (t_0, \tau]
  }
  \left( s - t_0 \right)^{
    \left( r_1 - r_0 \right)
  }
  \left\| 
    x(s)
  \right\|_{ U_{ r_1 } }
$
$
  < \infty
$,
then 
\eqref{eq:linear_growth_bound}
and interpolation 
(see, e.g., 
Theorem~37.6 in Sell \& You~\cite{sy02})
imply that
\begin{equation}
\label{eq:integral_welldefined}
\begin{split}
&
  \int_{ t_0 }^t
  \big\|
  e^{ A ( t - s ) }
  \,
  F_i( s, x(s) ) 
  \big\|_{
    U_{ r_1 }
  }
  ds
\leq
  \int_{ t_0 }^t
  \|
    e^{ A ( t - s ) }
  \|_{
    L( U_{ \alpha_i }, U_{ r_1 } )
  }
  \,
  \|
    F_i( s, x(s) ) 
  \|_{
    U_{ \alpha_i }
  }
  ds
\\ & \leq
  \left\| F \right\|_{
    \mathcal{C}^n_{ \alpha, \beta, \gamma, \delta }( [t_0, T] )
  }
  \left[
    \sup_{ s \in (0,T-t_0] }
    \!\!\!
    \tfrac{
      \|
        e^{ A s }
      \|_{
        L( U_{ \alpha_i }, U_{ r_1 } )
      }
    }{
      s^{
        \min( \alpha_i - r_1 , 0 )
      }
    }
  \right]
  \int_{ t_0 }^t
  \frac{
  \big(
    1 + \| x(s) \|_{ U_{ \beta_i } }^{ \delta_i }
  \big)
  \,
  \big(
    1 +
    \|
      x(s)
    \|_{
      U_{ \gamma_i }
    }
  \big)
  }{
    \left(
      t - s
    \right)^{
      \max( r_1 - \alpha_i , 0 )
    }
  }
  \, ds
\\ & \leq
  \left\| F \right\|_{
    \mathcal{C}^n_{ \alpha, \beta, \gamma, \delta }( [t_0, T] )
  }
  \left[
    \sup_{ s \in (0,T-t_0] }
    \!\!\!
    \tfrac{
      \|
        e^{ A s }
      \|_{
        L( U_{ \alpha_i }, U_{ r_1 } )
      }
    }{
      s^{
        \min( \alpha_i - r_1 , 0 )
      }
    }
  \right]
  \left[
    \sup_{
      s \in (t_0,\tau]
    }
    \frac{
      (
      1 +
      \|
        x(s)
      \|_{
        U_{ \gamma_i }
      }
      )
    }{
      \left( s - t_0 \right)^{
        ( r_0 - \gamma_i )
      }
    }
  \right]
\\ & 
  \cdot
  \left[
    \sup_{
      s \in (t_0,\tau]
    }
    \frac{
      (
      1 +
      \|
        x(s)
      \|_{
        U_{ \beta_i }
      }^{ \delta_i }
      )
    }{
      \left( s - t_0 \right)^{
        \delta_i ( r_0 - \beta_i )
      }
    }
  \right]
  \int_{ t_0 }^t
  \frac{
    1
  }{
    \left(
      t - s
    \right)^{
      \max( r_1 - \alpha_i , 0 )
    }
    \left( s - t_0 \right)^{
      \left(
        \gamma_i - r_0 + \delta_i ( \beta_i - r_0 )
      \right)
    }
  }
  \, ds
  < \infty
\end{split}
\end{equation}
for all 
$ t \in [t_0,\tau] $
and all
$ i \in \{ 1, 2, \dots, n \} $
where we used 
$
  r_1 
  < 
  1 + \min_{ j \in \{ 1, \dots, n \} } \alpha_j
  \leq 1 + \alpha_i
$
and 
$
  \gamma_i - r_0 + \delta_i ( \beta_i - r_0 )
  < 1
$
for all $ i \in \{ 1, 2, \dots, n \} $
in the last line of \eqref{eq:integral_welldefined}.
We now present the proof
of Lemma~\ref{lem:existence}.
\begin{proof}[Proof
of Lemma~\ref{lem:existence}]
Lemma~\ref{lem:existence}
follows from an application
of the Banach fixed point theorem.
For this several preparations are
needed.
First, let $ \kappa \in [0,\infty) $
be a real number 
defined through
\begin{equation}
\begin{split}
  \kappa
  :=
&
  \left[
    2 + r_1 - r_0 + T + 
    \left\| F \right\|_{
      \mathcal{C}^n_{ \alpha, \beta, \gamma, \delta }( [t_0, T] )
    }
    + 
    \sum_{ i = 1 }^n \delta_i
  \right]^{
    \left( 
      4  
      + \left| r_0 \right|
      + \left| r_1 \right|
      +
      \max_{ i \in \{ 1, \dots, n \} }
      \left| \alpha_i \right|
    \right)
  }
\\ &
  +
  \sum_{ j = 0 }^1
  \sum_{ i = 1 }^n
  \left[
  \frac{ 
    1
  }{
    \min\!\left( 
      1 
      + 
      \alpha_i 
      - r_j, 1
    \right)
  }
  +
  B_{
    \left(
      1 + \min( \alpha_i - r_j , 0 ) ,
      1 + r_0 - \gamma_i +
      \delta_i
      \left( r_0 - \beta_i \right) 
    \right)
  }
  \right]
\\ &
  +
  \max_{ j \in \{ 0, 1 \} }
  \max_{ 
    \theta 
    \in 
    \{
      r_0, r_1, \alpha_1, \dots, \alpha_n
    \}
  }
  \sup_{ t \in (t_0,T] }
  \left[
    \left( t - t_0 
    \right)^{
      \max\left( r_j - \theta , 0 \right)
    } 
  \big\|
    e^{ 
      A \left( t - t_0 \right) 
    }
  \big\|_{ 
    L( U_{ \theta }, U_{ r_j } ) 
  }
  \right]
\\ & 
  +
  \max_{
  \substack{
    \theta \in 
    \{ 
      \beta_1, \dots, \beta_n 
    \}
  \\
    \cup
    \{
      \gamma_1, \dots, \gamma_n 
    \} 
  }
  }
  \sup_{ 
    \substack{
      v \in U_{ r_1 } 
    \\
      v \neq 0
    }
  }
  \left[
  1 +
  \frac{
    \left\|  
      v
    \right\|_{ U_{ \theta } }
  }{
    \left\|
      v
    \right\|_{ U_{ r_1 } 
    }^{
      \frac{
        ( \theta - r_0 )
      }{
        ( r_1 - r_0 )
      }
    }
    \left\|
      v
    \right\|_{ U_{ r_0 } }^{
      \frac{
        ( r_1 - \theta )
      }{
        ( r_1 - r_0 )
      }
    } 
  }
  \right]^{
    \!
    ( 1 + \sum_{ i = 1 }^n \delta_i )
  }
  < \infty 
\end{split}
\end{equation}
where
$
  B \colon (0,\infty)^2 \to (0,\infty)
$
is the Beta function 
defined through
$
  B_{ (x, y) }
:=
  \int_0^1
  \left( 1 - s \right)^{
    \left( x - 1 \right)
  }
$
$
  s^{ \left( y - 1 \right) }
  \, ds
$
for all $ x, y \in (0,\infty) $.
Observe that the quantity $ \kappa $ is indeed finite;
see, e.g., Theorems~37.5 and 37.6 
in Sell \& You~\cite{sy02}.
Next define real vector spaces
$
  \mathcal{E}_{ [t_0, \tau] }
  \subset
  C( [t_0,\tau], U_{ r_0 } )
$,
$
  \tau \in (t_0,T]
$,
through
\begin{equation}
  \mathcal{E}_{ [t_0, \tau] }
  :=
  \left\{
    x \in 
    C( [t_0,\tau], U_{ r_0 } )
    \colon
    \left(
    \begin{array}{c}
    x|_{
      (t_0,\tau]
    }
    \in
    C( (t_0,\tau], U_{ r_1 } )
    \text{    and}
    \\[1ex]
    \sup_{ t \in (t_0,\tau] }
      \left( t - t_0 \right)^{ \left( r_1 - r_0 \right) }
      \left\| 
        x(t) 
      \right\|_{
        U_{ r_1 }
      }
    < \infty
    \end{array}
    \right)
  \right\} 
\end{equation}
for all $ \tau \in (t_0,T] $,
define norms
$
  \left\|
    \cdot
  \right\|_{
    \mathcal{E}_{ [t_0, \tau] }
  }
  \colon
  \mathcal{E}_{ [t_0, \tau] }
  \to 
  [0,\infty)
$,
$ \tau \in (t_0,T] $,
through
\begin{equation}
  \left\|
    x
  \right\|_{
    \mathcal{E}_{ [t_0, \tau] }
  }
  :=
  \sum_{ j = 0 }^1 
  \left[
  \sup_{ t \in (t_0,\tau] }
  \left[
    \left( t - t_0 \right)^{
      \left( r_j - r_0 \right)
    }
    \left\| 
      x(t) 
    \right\|_{
      U_{ r_j }
    }
  \right] 
  \right]
\end{equation}
for all $ \tau \in (t_0,T] $,
define
sets
$
  \mathcal{E}_{ [t_0, \tau], v }
  \subset
  \mathcal{E}_{ [t_0, \tau] }
$,
$
  \tau \in (t_0,T]
$,
$
  v \in U_{ r_0 }
$,
through
\begin{equation}
\label{eq:defE}
  \mathcal{E}_{ [t_0, \tau], v }
  :=
  \left\{
    x \in 
    \mathcal{E}_{ [t_0, \tau] }
    \colon
    \| 
      x 
    \|_{
      \mathcal{E}_{ [t_0, \tau] }
    }
    \leq 
    \kappa^{ 7 }
    \left(
      1 +
      \|
        v
      \|_{ U_{ r_0 } }
    \right)
  \right\}
\end{equation}
for all
$
  \tau \in (t_0,T]
$
and all
$
  v \in U_{ r_0 }
$
and define mappings
$ 
  \Phi_{ 
    [t_0, \tau], v
  } 
  \colon 
  \mathcal{E}_{ 
    [ t_0, \tau ]
  }
  \to 
  \mathcal{E}_{ 
    [ t_0, \tau ]
  }
$,
$
  \tau \in (t_0,T]
$,
$
  v \in U_{ r_0 }
$,
through
\begin{equation}
\label{eq:defEnorm}
  (
    \Phi_{ [t_0, \tau], v } 
    x
  )( t )
:=
  e^{ A ( t - t_0 ) } \, v
  +
  \sum_{ i = 1 }^n
  \int_{ t_0 }^t 
  e^{ A ( t - s ) }
  \,
  F_i( s, x(s) )
  \,
  ds
\end{equation}
for all 
$ t \in [t_0,\tau] $,
$ 
  x \in 
  \mathcal{E}_{ [t_0, \tau] }
$,
$
  \tau \in (t_0, T]  
$
and all
$ 
  v \in U_{ r_0 }
$.
Note that 
\eqref{eq:integral_welldefined}
ensures that
the mappings 
$ 
  \Phi_{ [t_0, \tau], v } 
$,
$
  \tau \in (t_0,T]
$,
$
  v \in U_{ r_0 }
$,
are well-defined.
We now establish
a few estimates for 
the mappings
$ 
  \Phi_{ [t_0, \tau], v } 
$,
$ \tau \in (t_0,T] $,
$ v \in U_{ r_0 } $.
First, observe that
\begin{equation}
\begin{split}
&
  \left\|
    ( \Phi_{ [t_0, \tau], v } 0 )(t)
  \right\|_{ U_{ r_j } }
\leq
  \big\|
    e^{ A ( t - t_0 ) } v
  \big\|_{ U_{ r_j } }
  +
  \sum_{ i = 1 }^n
  \int_{ t_0 }^t 
  \big\|
    e^{ A ( t - s ) } 
  \big\|_{ 
    L( U_{ \alpha_i }, U_{ r_j } )
  } 
  \,
  \big\|
    F_i( s, 0 )
  \big\|_{ U_{ \alpha_i } }
  ds
\\ & \leq
  \big\|
    e^{ A ( t - t_0 ) }
  \big\|_{ 
    L( U_{ r_0 } , U_{ r_j } )
  }
  \,
  \|
    v
  \|_{ U_{ r_0 } }
  +
  \kappa^2
  \sum_{ i = 1 }^n
  \int_{ t_0 }^t 
  \left( t - s \right)^{
    \min( 
      \alpha_i - r_j, 0
    )
  }
  ds
\\ & \leq
  \kappa 
  \left( t - t_0 \right)^{
    \left( r_0 - r_j \right)
  }
  \|
    v
  \|_{ U_{ r_0 } }
  +
  \kappa^2
  \sum_{ i = 1 }^n
  \frac{
    \left( t - t_0 \right)^{
      \min\left( 
        1 + \alpha_i - r_j, 1
      \right)
    }
  }{
    \min\!\left( 1 + \alpha_i - r_j , 1 \right)
  }
\\ & \leq
  \kappa^5
  \left( t - t_0 \right)^{
    \left( r_0 - r_j \right)
  }
  \left(
    1 +
    \|
      v
    \|_{ U_{ r_0 } }
  \right)
\end{split}
\end{equation}
for all 
$ j \in \{ 0, 1 \} $,
$ t \in (t_0,\tau] $,
$ \tau \in (t_0,T] $,
$ v \in U_{ r_0 } $
and hence
\begin{equation}
  \left\|
    \Phi_{ [t_0, \tau], v }( 0 )
  \right\|_{
    \mathcal{E}_{ [t_0, \tau] }
  }
\leq
  \kappa^{ 6 }
  \left(
    1 +
    \|
      v
    \|_{ U_{ r_0 } }
  \right)
\label{eq:local_zero}
\end{equation}
for all 
$ \tau \in (t_0,T] $,
$ v \in U_{ r_0 } $.
In the next step observe that
\begin{equation}
\begin{split}
&
  \left\|
    ( \Phi_{ [t_0, \tau], v } x )(t)
    -
    ( \Phi_{ [t_0, \tau], v } y )(t)
  \right\|_{ U_{ r_j } }
\leq
  \sum_{ i = 1 }^n
  \int_{ t_0 }^t 
  \big\|
    e^{ A ( t - s ) }
    \big[ 
      F_i( s, x(s) ) - F_i( s, y(s) )
    \big] 
  \big\|_{ U_{ r_j } }
  ds
\\ & \leq
  \sum_{ i = 1 }^n
  \int_{ t_0 }^t 
  \big\|
    e^{ A ( t - s ) }
  \big\|_{
    L( U_{ \alpha_i }, U_{ r_j } )
  }
  \,
  \big\|
    F_i( s, x(s) ) - F_i( s, y(s) )
  \big\|_{ U_{ \alpha_i } }
  ds
\\ & \leq
  \kappa
  \sum_{ i = 1 }^n
  \int_{ t_0 }^t 
  \left( t - s \right)^{
    \min( \alpha_i - r_j , 0 )
  }
  \big\|
    F_i( s, x(s) ) - F_i( s, y(s) )
  \big\|_{ U_{ \alpha_i } }
  ds
\\ & \leq
  \kappa^2
  \sum_{ i = 1 }^n
  \int_{ t_0 }^t 
  \left( t - s \right)^{
    \min( \alpha_i - r_j , 0 )
  }
  \left[
    1 
    +
    \left\|
      x(s) 
    \right\|_{ U_{ \beta_i } 
    }^{ \delta_i }
    +
    \left\|
      y(s) 
    \right\|_{ U_{ \beta_i } 
    }^{ \delta_i }
  \right]
  \left\|
    x(s) - y(s)
  \right\|_{ U_{ \gamma_i } }
  ds
\\ & \leq
  \kappa^{ 
    3
  }
  \sum_{ i = 1 }^n
  \int_{ t_0 }^t 
  \left( t - s \right)^{
    \min( \alpha_i - r_j , 0 )
  }
  \left\|
    x(s) - y(s)
  \right\|_{ U_{ r_1 } 
  }^{
    \frac{
      ( \gamma_i - r_0 )
    }{
      ( r_1 - r_0 )
    }
  }
  \left\|
    x(s) - y(s)
  \right\|_{ U_{ r_0 } }^{
    \frac{
      ( r_1 - \gamma_i )
    }{
      ( r_1 - r_0 )
    }
  }
\\ & \quad
  \cdot
  \left[
    1 
    +
    \left\|
      x(s) 
    \right\|_{ U_{ r_1 } }^{
      \frac{  
        ( \beta_i - r_0 ) \delta_i
      }{
        ( r_1 - r_0 )
      }
    }
    \left\|
      x(s) 
    \right\|_{ U_{ r_0 } }^{
      \frac{  
        ( r_1 - \beta_i ) \delta_i
      }{
        ( r_1 - r_0 )
      }
    }
    +
    \left\|
      y(s) 
    \right\|_{ U_{ r_1 } }^{
      \frac{  
        ( \beta_i - r_0 ) \delta_i
      }{
        ( r_1 - r_0 )
      }
    }
    \left\|
      y(s) 
    \right\|_{ U_{ r_0 } }^{
      \frac{  
        ( r_1 - \beta_i ) \delta_i
      }{
        ( r_1 - r_0 )
      }
    }
  \right]
  ds
\end{split}
\end{equation}
and therefore 
\begin{equation}
\begin{split}
&
  \left(
    t - t_0
  \right)^{
    \left( r_j - r_0 \right)
  }
  \left\|
    (\Phi_{ [t_0, \tau], v } x)(t)
    -
    (\Phi_{ [t_0, \tau], v } y)(t)
  \right\|_{ U_{ r_j } }
\\ & 
  \leq
  \kappa^{ 
    3
  }
  \left(
    t - t_0
  \right)^{
    \left( r_j - r_0 \right)
  }
  \left\|
    x - y
  \right\|_{ 
    \mathcal{E}_{ [t_0, t] }
  }
  \sum_{ i = 1 }^n
  \int_{ t_0 }^t 
  \left( t - s \right)^{
    \min( \alpha_i - r_j , 0 )
  }
\\ & 
  \cdot
  \left(
    \left( s - t_0 \right)^{ 
      \left( 
        r_0 - \gamma_i 
      \right)
    }
    +
    \left( s - t_0 \right)^{ 
      \left( 
        r_0 - \gamma_i 
        +
        \delta_i
        \left( r_0 - \beta_i \right) 
      \right)
    }
    \left[
    \left\|
      x 
    \right\|_{ 
      \mathcal{E}_{ [t_0, t] } 
    }^{
      \delta_i
    }
    +
    \left\|
      y
    \right\|_{ \mathcal{E}_{ [t_0, t] } 
    }^{
      \delta_i
    }
    \right]
  \right)
  ds
\\ & \leq
  \kappa^{ 
    5
  }
  \left(
    t - t_0
  \right)^{
    \left( r_j - r_0 \right)
  }
  \left[ 
    1 +
    \left\|
      x 
    \right\|_{ \mathcal{E}_{ [t_0, t] } }
    +
    \left\|
      y
    \right\|_{ \mathcal{E}_{ [t_0, t] } }
  \right]^{
    \left( 
      1 + \sum_{ i = 1 }^n \delta_i
    \right)
  }
  \left\|
    x - y
  \right\|_{ 
    \mathcal{E}_{ [t_0, t] }
  }
\\ & 
  \cdot
  \sum_{ i = 1 }^n
  \int_{ 0 }^{ ( t - t_0 ) } 
  \left( t - t_0 - s \right)^{
    \min( \alpha_i - r_j , 0 )
  }
  s^{ 
    \left( 
      r_0 - \gamma_i +
      \delta_i
      \left( r_0 - \beta_i \right) 
    \right)
  }
  \, ds
\\ & =
  \kappa^{ 
    5
  }
  \left[ 
    1 +
    \left\|
      x 
    \right\|_{ \mathcal{E}_{ [t_0, t] } }
    +
    \left\|
      y
    \right\|_{ \mathcal{E}_{ [t_0, t] } }
  \right]^{
    \left( 
      1 + \sum_{ i = 1 }^n \delta_i
    \right)
  }
  \left\|
    x - y
  \right\|_{ 
    \mathcal{E}_{ [t_0, t] }
  }
\\ & \cdot
  \sum_{ i = 1 }^n
  \left( t - t_0 \right)^{ 
    \left(
      1 + \min( \alpha_i , r_j ) 
      - \gamma_i 
      +
      \delta_i \left( r_0 - \beta_i \right) 
    \right)
  }
  B_{
    \left(
      1 + \min( \alpha_i - r_j , 0 ) ,
      1 + r_0 - \gamma_i +
      \delta_i \left( r_0 - \beta_i \right) 
    \right)
  }
\end{split}
\end{equation}
and hence
\begin{equation}
\begin{split}
&
  \left(
    t - t_0
  \right)^{
    \left( r_j - r_0 \right)
  }
  \left\|
    (\Phi_{ [t_0, \tau], v } x)(t)
    -
    (\Phi_{ [t_0, \tau], v } y)(t)
  \right\|_{ U_{ r_j } }
\\ & \leq
  \kappa^{ 
    6
  }
  \left( t - t_0 \right)^{ 
    \min_{
      i \in \{ 1, \dots, n \}
    }
    \left[
      1 
      -
      \left( 
        \gamma_i
        - \min( \alpha_i , r_j ) 
        + 
        \delta_i 
        \left( \beta_i - r_0 \right) 
      \right)
    \right]
  }
  \left[ 
    1 +
    \left\|
      x 
    \right\|_{ \mathcal{E}_{ [t_0, t] } }
    +
    \left\|
      y
    \right\|_{ \mathcal{E}_{ [t_0, t] } }
  \right]^{
    \left( 
      1 + \sum_{ i = 1 }^n \delta_i
    \right)
  }
\\ & \quad
  \cdot
  \left\|
    x - y
  \right\|_{ 
    \mathcal{E}_{ [t_0, t] }
  }
  \left[
  \sum_{ i = 1 }^n
  B_{
    \left(
      1 + \min( \alpha_i - r_j , 0 ) ,
      1 + r_0 - \gamma_i 
      +
      \delta_i 
      \left( 
        r_0 - \beta_i 
      \right) 
    \right)
  }
  \right]
\\ & \leq
  \kappa^{ 
    7
  }
  \left( t - t_0 \right)^{ 
    \left[
      1 
      -
      \max_{
        i \in \{ 1, \dots, n \}
      }
      \left( 
        \gamma_i
        - \min( \alpha_i , r_j ) 
        +
        \delta_i 
        \left( \beta_i - r_0 
        \right) 
      \right)
    \right]
  }
  \left[ 
    1 +
    \left\|
      x 
    \right\|_{ \mathcal{E}_{ [t_0, t] } }
    +
    \left\|
      y
    \right\|_{ \mathcal{E}_{ [t_0, t] } }
  \right]^{
    \left( 
      1 + \sum_{ i = 1 }^n \delta_i
    \right)
  }
\\ & \quad \cdot
  \left\|
    x - y
  \right\|_{ 
    \mathcal{E}_{ [t_0, t] }
  }
\end{split}
\end{equation}
for all
$ j \in \{ 0, 1 \} $,
$ t \in ( t_0, \tau ] $,
$ x, y \in \mathcal{E}_{ [t_0, \tau] } $,
$ \tau \in (t_0,T] $,
$ v \in U_{ r_0 } $.
Hence, we get
\begin{equation}
\label{eq:local_Lipschitz}
\begin{split}
&
  \left\|
    \Phi_{ [t_0, \tau], v }( x ) -
    \Phi_{ [t_0, \tau], v }( y )
  \right\|_{ 
    \mathcal{E}_{ [t_0, \tau] }
  }
\\ & \leq
  \kappa^{ 
    8
  }
  \left( \tau - t_0 \right)^{ 
    \left[
      1 
      -
      \max_{
        i \in \{ 1, \dots, n \}
      }
      \left( 
        \gamma_i
        - \min( \alpha_i , r_0 ) 
        +
        \delta_i 
        ( \beta_i - r_0 ) 
      \right)
    \right]
  }
\\ & \quad 
  \cdot
  \left[ 
    1 +
    \left\|
      x 
    \right\|_{ \mathcal{E}_{ [t_0, \tau] } }
    +
    \left\|
      y
    \right\|_{ \mathcal{E}_{ [t_0, \tau] } }
  \right]^{
    \kappa
  }
  \left\|
    x - y
  \right\|_{ 
    \mathcal{E}_{ [t_0, \tau] }
  }
\end{split}
\end{equation}
for all
$ x, y \in \mathcal{E}_{ [t_0, \tau] } $,
$ v \in U_{ r_0 } $,
$ \tau \in (t_0,T] $.
Combining \eqref{eq:local_zero}
and \eqref{eq:local_Lipschitz}
results in
\begin{equation}
\label{eq:local_lineargrowth}
\begin{split}
&
  \left\|
    \Phi_{ [t_0, \tau], v }( x )
  \right\|_{ 
    \mathcal{E}_{ [t_0, \tau] }
  }
\leq
  \left\|
    \Phi_{ [t_0, \tau], v }( x ) -
    \Phi_{ [t_0, \tau], v }( 0 )
  \right\|_{ 
    \mathcal{E}_{ [t_0, \tau] }
  }
  +
  \left\|
    \Phi_{ [t_0, \tau], v }( 0 )
  \right\|_{ 
    \mathcal{E}_{ [t_0, \tau] }
  }
\\ & \leq
  \kappa^{ 
    8
  }
  \left( \tau - t_0 \right)^{ 
    \left[
      1 
      -
      \max_{
        i \in \{ 1, \dots, n \}
      }
      \left( 
        \gamma_i
        - \min( \alpha_i , r_0 ) 
        + 
        \delta_i ( \beta_i - r_0 ) 
      \right)
    \right]
  }
  \left[ 
    1 +
    \left\|
      x 
    \right\|_{ \mathcal{E}_{ [t_0, \tau] } }
  \right]^{
    \kappa
  }
  \left\|
    x 
  \right\|_{ 
    \mathcal{E}_{ [t_0, \tau] }
  }
\\ & \quad
  +
  \kappa^{ 
    6
  }
  \left(
    1 +
    \|
      v
    \|_{ U_{ r_0 } }
  \right)
\end{split}
\end{equation}
for all
$ x \in \mathcal{E}_{ [t_0, \tau] } $,
$ v \in U_{ r_0 } $,
$ \tau \in (t_0,T] $.
The assumption
\begin{equation}
  \max_{ i \in \{1,\dots,n\} }
  \left[
    \gamma_i
    - \min( \alpha_i , r_0 ) 
    + 
    \delta_i
    \left( \beta_i - r_0 \right) 
  \right]
  < 
  1 
\end{equation}
together with inequalities
\eqref{eq:local_Lipschitz}
and 
\eqref{eq:local_lineargrowth}
implies that
there exists a mapping
$
  \rho \colon U_{ r_0 } \to
  (t_0,T]
$
such that
\begin{equation}
\begin{split}
  \left\|
    \Phi_{ [t_0, \rho(v)], v }( x )
  \right\|_{ 
    \mathcal{E}_{ [t_0, \rho(v)] }
  }
& \leq
  1
  +
  \kappa^{ 
    6
  }
  \left(
    1 +
    \|
      v
    \|_{ U_{ r_0 } }
  \right) 
  ,
\\
  \left\|
    \Phi_{ [t_0, \rho(v)], v }( x ) -
    \Phi_{ [t_0, \rho(v)], v }( y )
  \right\|_{ 
    \mathcal{E}_{ [t_0, \rho(v)] }
  }
& \leq
  \frac{ 1 }{ 2 }
  \left\|
    x - y
  \right\|_{ 
    \mathcal{E}_{ [t_0, \rho(v)] }
  }
\end{split}
\end{equation}
for all
$ x, y \in \mathcal{E}_{ [t_0, \rho(v)], v } $,
$ v \in U_{ r_0 } $.
This ensures that
$
  \Phi_{ [t_0, \rho(v)], v }\!\left( 
    \mathcal{E}_{ [t_0, \rho(v)], v }
  \right)
  \subset
  \mathcal{E}_{ [t_0, \rho(v)], v }
$
for all
$ v \in U_{ r_0 } $.
The Banach fixed
point theorem hence
proves 
that there exist 
unique functions
$ 
  x_v \in 
  \mathcal{E}_{ [t_0, \rho(v)], v }
$, 
$ v \in U_{ r_0 } $,
such that
$
  \Phi_{ [t_0, \rho(v)], v }( x_v )
  =
  x_v
$
for all
$ 
  v \in U_{ r_0 }
$.
This completes the
proof of Lemma~\ref{lem:existence}.
\end{proof}

Lemma~\ref{lem:existence} shows, under
suitable assumptions, that
there exists a unique local mild solution
of \eqref{eq:PDEs}.
This solution can
be extended to a maximal interval of
definition. This is the subject of the next
corollary. It follows directly from
Lemma~\ref{lem:existence} and a standard
argument from the ordinary 
differential equations literature
and its proof is therefore omitted.

\begin{cor}[Maximal mild solutions]
\label{cor:maximal}
Assume the setting 
in the beginning 
of Subsection~\ref{sec:det_existence},
let 
$ r_0, t_0 \in \R $,
$ T \in (t_0, \infty) $,
$ v \in U_{ r_0 } $,
$ n \in \N $,
$ 
  \alpha = ( \alpha_1, \dots, \alpha_n ) \in \R^n
$,
$
  \beta = ( \beta_1, \dots, \beta_n ) ,
  \gamma = ( \gamma_1, \dots, \gamma_n ) 
  \in [r_0, \infty)^n
$,
$
  \delta = ( \delta_1, \dots, \delta_n ) 
  \in [0,\infty)^n
$,
$
  r_1
  \in
  \left[  
    \max( \beta_1, \dots, \beta_n, \gamma_1, \dots, \gamma_n ) ,
    1 +
    \min( \alpha_1, \dots, \alpha_n )
  \right)
$
with
$
  \max_{ i \in \{1,\dots,n\} }
  [
  \gamma_i
  - \min( \alpha_i , r_0 ) 
  + 
    \delta_i
    ( \beta_i - r_0 ) 
  ]
  < 
  1 
$
and 
let
$
  F = ( F_1, \dots, F_n ) \in 
  \mathcal{C}^n_{ \alpha, \beta, \gamma, \delta }( [t_0, T] ) 
$.
Then there exist a unique real number
$ \tau \in (t_0,T] $
and a unique continuous function
$
  x \colon [t_0,\tau) \to
  U_{ r_0 }
$
satisfying
$
  x|_{ (t_0,\tau ) }
  \in
  C(
    (t_0,\tau ), U_{ r_1 }
  )
$,
$
  \sup_{
    s \in (t_0, t ]
  }
  \left( s - t_0 \right)^{
    \left( r_1 - r_0 \right)
  }
  \left\| 
    x(s)
  \right\|_{ U_{ r_1 } }
  < \infty
$,
$
  \lim_{ s \nearrow \tau }
  \big[
    \frac{ 1 }{ 
      ( T - s ) 
    }
    +
    \| 
      x(s)
    \|_{
      U_{ r_1 }
    }
  \big]
  = \infty
$
and
$
  x(t)
=
  e^{ A ( t - t_0 ) } \, v
+
  \sum_{ i = 1 }^n
  \int_{ t_0 }^t
  e^{ A ( t - s ) }
  \,
  F_i( s, x(s) ) \, ds
$
for all
$ t \in ( t_0, \tau ) $.
\end{cor}

The next result shows, under suitable
assumptions, that the unique 
maximal mild solution of \eqref{eq:PDEs} 
enjoys a bit more regularity 
than the regularity asserted 
in Corollary~\ref{cor:maximal}.

\begin{cor}[More regularity for maximal mild solutions]
\label{cor:maximal2}
Assume the setting 
in the beginning 
of Subsection~\ref{sec:det_existence},
let 
$ r_0, t_0 \in \R $,
$ T \in (t_0, \infty) $,
$ v \in U_{ r_0 } $,
$ n \in \N $,
$ 
  \alpha = ( \alpha_1, \dots, \alpha_n ) \in \R^n
$,
$
  \beta = ( \beta_1, \dots, \beta_n ) ,
  \gamma = ( \gamma_1, \dots, \gamma_n ) 
  \in [r_0, \infty)^n
$,
$
  \delta = ( \delta_1, \dots, \delta_n ) 
  \in [0,\infty)^n
$
with
$
  \max( \beta_1, \dots, \beta_n, \gamma_1, \dots, \gamma_n ) 
  <
  1 +
  \min( \alpha_1, \dots, \alpha_n )
$
and
$
  \max_{ i \in \{1,\dots,n\} }
  [
  \gamma_i
  - \min( \alpha_i , r_0 ) 
  + 
    \delta_i
    ( \beta_i - r_0 ) 
  ]
  < 
  1 
$
and 
let
$
  F = ( F_1, \dots, F_n ) \in 
  \mathcal{C}^n_{ \alpha, \beta, \gamma, \delta }( [t_0, T] ) 
$.
Then there exist a unique real number
$ \tau \in (t_0,T] $
and a unique continuous function
$
  x \colon [t_0,\tau) \to
  U_{ r_0 }
$
satisfying
$
  x|_{ (t_0,\tau ) }
  \in
  C(
    (t_0,\tau ), U_{ r_1 }
  )
$,
$
  \sup_{
    s \in (t_0, t ]
  }
  \left( s - t_0 \right)^{
    \left( r_1 - r_0 \right)
  }
  \left\| 
    x(s)
  \right\|_{ U_{ r_1 } }
  < \infty
$,
$
  \lim_{ s \nearrow \tau }
  \big[
    \| 
      x(s)
    \|_{
      U_{ 
        \max( \beta_1, \dots, \beta_n, \gamma_1, \dots, \gamma_n )
      }
    }
$
$
    +
    \frac{ 1 }{ 
      ( T - s ) 
    }
  \big]
  = \infty
$
and
$
  x(t)
=
  e^{ A ( t - t_0 ) } \, v
+
  \sum_{ i = 1 }^n
  \int_{ t_0 }^t
  e^{ A ( t - s ) }
  \,
  F_i( s, x(s) ) \, ds
$
for all
$ t \in ( t_0, \tau ) $
and all
$
  r_1
  \in
  \left[  
    r_0, 
    1 +
    \min( \alpha_1, \dots, \alpha_n )
  \right)
$.
\end{cor}

\begin{proof}[Proof
of Corollary~\ref{cor:maximal2}]
First of all, Corollary~\ref{cor:maximal}
implies that there exists a unique real number
$ \tau \in (t_0,T] $
and a unique continuous function
$
  x \colon [t_0,\tau) \to
  U_{ r_0 }
$
satisfying
$
  x|_{ (t_0,\tau ) }
  \in
  C(
    (t_0,\tau ), 
    U_{ 
      \max( \beta_1, \dots, \beta_n, \gamma_1, \dots, \gamma_n ) 
    }
  )
$,
$
  \lim_{ s \nearrow \tau }
  \big[
    \frac{ 1 }{ 
      ( T - s ) 
    }
    +
    \| 
      x(s)
    \|_{
      U_{ 
        \max( \beta_1, \dots, \beta_n, \gamma_1, \dots, \gamma_n ) 
      }
    }
  \big]
  = 
  \infty
$
and
\begin{equation}
  \sup_{
    s \in (t_0, t ]
  }
  \left( s - t_0 \right)^{
    \left( 
      \max( \beta_1, \dots, \beta_n, \gamma_1, \dots, \gamma_n ) 
      - r_0 
    \right)
  }
  \left\| 
    x(s)
  \right\|_{ 
    U_{ 
      \max( \beta_1, \dots, \beta_n, \gamma_1, \dots, \gamma_n ) 
    } 
  }
  < \infty
\end{equation}
and
\begin{equation}
  x(t)
=
  e^{ A ( t - t_0 ) } \, v
+
  \sum_{ i = 1 }^n
  \int_{ t_0 }^t
  e^{ A ( t - s ) }
  \,
  F_i( s, x(s) ) \, ds
\end{equation}
for all
$ t \in ( t_0, \tau ) $.
Next we observe 
similar as in \eqref{eq:integral_welldefined} that
\eqref{eq:linear_growth_bound} 
and interpolation (see, e.g., 
Theorem~37.6 in Sell \& You~\cite{sy02})
imply that
\begin{equation}
\begin{split}
&
  \int_{ t_0 }^t
  \big\|
  e^{ A ( t - s ) }
  \,
  F_i( s, x(s) ) 
  \big\|_{
    U_{ r_1 }
  }
  ds
\leq
  \int_{ t_0 }^t
  \|
    e^{ A ( t - s ) }
  \|_{
    L( U_{ \alpha_i }, U_{ r_1 } )
  }
  \,
  \|
    F_i( s, x(s) ) 
  \|_{
    U_{ \alpha_i }
  }
  ds
\\ & \leq
  \left\| F \right\|_{
    \mathcal{C}^n_{ \alpha, \beta, \gamma, \delta }( [t_0, T] )
  }
  \left[
    \sup_{ s \in (0,T-t_0] }
    \!\!\!
    \tfrac{
      \|
        e^{ A s }
      \|_{
        L( U_{ \alpha_i }, U_{ r_1 } )
      }
    }{
      s^{
        \min( \alpha_i - r_1 , 0 )
      }
    }
  \right]
  \int_{ t_0 }^t
  \frac{
  \big(
    1 + \| x(s) \|_{ U_{ \beta_i } }^{ \delta_i }
  \big)
  \,
  \big(
    1 +
    \|
      x(s)
    \|_{
      U_{ \gamma_i }
    }
  \big)
  }{
    \left(
      t - s
    \right)^{
      \max( r_1 - \alpha_i , 0 )
    }
  }
  \, ds
\\ & \leq
  \left\| F \right\|_{
    \mathcal{C}^n_{ \alpha, \beta, \gamma, \delta }( [t_0, T] )
  }
  \left[
    \sup_{ s \in (0,T-t_0] }
    \!\!\!
    \tfrac{
      \|
        e^{ A s }
      \|_{
        L( U_{ \alpha_i }, U_{ r_1 } )
      }
    }{
      s^{
        \min( \alpha_i - r_1 , 0 )
      }
    }
  \right]
  \left[
    \sup_{
      s \in (t_0,t]
    }
    \frac{
      (
      1 +
      \|
        x(s)
      \|_{
        U_{ \gamma_i }
      }
      )
    }{
      \left( s - t_0 \right)^{
        ( r_0 - \gamma_i )
      }
    }
  \right]
\\ & 
  \cdot
  \left[
    \sup_{
      s \in (t_0,t]
    }
    \frac{
      (
      1 +
      \|
        x(s)
      \|_{
        U_{ \beta_i }
      }^{ \delta_i }
      )
    }{
      \left( s - t_0 \right)^{
        \delta_i ( r_0 - \beta_i )
      }
    }
  \right]
  \int_{ t_0 }^t
  \frac{
    1
  }{
    \left(
      t - s
    \right)^{
      \max( r_1 - \alpha_i , 0 )
    }
    \left( s - t_0 \right)^{
      \left(
        \gamma_i - r_0 + \delta_i ( \beta_i - r_0 )
      \right)
    }
  }
  \, ds
  < \infty
\end{split}
\label{eq:integral_morereg}
\end{equation}
for all 
$ t \in [t_0,\tau) $,
$ i \in \{ 1, 2, \dots, n \} $
and all
$
  r_1
  \in 
  ( 
    - \infty ,
    1 +
    \min( \alpha_1, \dots, \alpha_n )
  )
$
where we used 
$
  \gamma_i - r_0 + \delta_i ( \beta_i - r_0 )
  < 1
$
for all $ i \in \{ 1, 2, \dots, n \} $
in the last line of \eqref{eq:integral_morereg}.
This proves that
$ 
  x(t) \in U_{ r_1 }
$
for all
$ t \in ( t_0, \tau ) $
and all
$
  r_1 
  \in 
  ( 
    - \infty, 
    1 + \min( \alpha_1, \dots, \alpha_n ) 
  )
$
and that
\begin{equation}
  \sup_{
    s \in ( t_0, t ]
  }
  \left(
    s - t_0
  \right)^{
    \left( r_1 - r_0 \right)
  }
  \left\|
    x(s)
  \right\|_{
    U_{ r_1 }
  }
  < 
  \infty
\end{equation}
for all
$ t \in ( t_0, \tau ) $
and 
all
$
  r_1 
  \in 
  [
    r_0 ,
    1 + \min( \alpha_1, \dots, \alpha_n ) 
  )
$.
Applying Lemma~\ref{lem:existence}
then proves that
$
  x|_{
    (t_0, \tau)
  }
  \in
  C( (t_0, \tau), U_{ r_1 } )
$
for all
$
  r_1 
  \in 
  [
    r_0 ,
    1 + \min( \alpha_1, \dots, \alpha_n ) 
  )
$.
This completes the proof
of Lemma~\ref{cor:maximal2}.
\end{proof}

We now present and prove the main result of this subsection.
It shows, under suitable assumptions, that
the unique local mild solutions of \eqref{eq:PDEs} 
depend continuously in an appropriate sense on the possibly 
nonlinear vector fields in \eqref{eq:PDEs}.

\begin{thm}[Continuous dependence on the data
on bounded time intervals]
\label{thm:continuity}
Assume the setting in the beginning 
of 
Subsection~\ref{sec:det_existence}
and let 
$ r_0 \in \R $,
$ n \in \N $,
$ 
  \alpha = ( \alpha_1, \dots, \alpha_n ) \in \R^n
$,
$
  \beta = ( \beta_1, \dots, \beta_n ), 
  \gamma = ( \gamma_1, \dots, \gamma_n )
  \in [r_0, \infty)^n
$,
$
  \delta = ( \delta_1, \dots, \delta_n ) 
  \in [0,\infty)^n
$
with
$
  \max( \beta_1, \dots, \beta_n, \gamma_1, \dots, \gamma_n ) 
  <
  1 +
  \min( \alpha_1, \dots, \alpha_n )
$
and
$
  \max_{ i \in \{ 1, \dots, n \} }
  \big[
  \gamma_i - \min( \alpha_i , r_0 ) 
  + ( \beta_i - r_0 ) \delta_i
  \big]
  < 1 
$.
Then 
there exist unique lower
semicontinuous functions
$ 
  \tau^{ t_0, T } 
  \colon 
  \mathcal{C}^n_{ \alpha, \beta, \gamma, \delta }( [t_0, T] ) 
  \times
  U_{ r_0 } 
  \to (t_0,T] 
$,
$ t_0, T \in \R $ with $ t_0 < T $,
and unique functions
$
  x^{ t_0, T } 
  \colon 
  \mathcal{C}^n_{ \alpha, \beta, \gamma, \delta }( [t_0, T] ) 
$
$
  \times
  U_{ r_0 } 
  \to
  \cup_{ s \in (t_0,T] } 
  C( [t_0,s), 
$
$
  U_{ r_0 } )
$,
$ t_0, T \in \R $ with $ t_0 < T $,
which satisfy
$
  x^{ t_0, T }_{ F, v }
  \in
  C( [t_0, \tau^{ t_0, T }_{ F, v } ), U_{ r_0 } )
$,
$
  x^{ t_0, T }_{ F, v }|_{ ( t_0, \tau^{ t_0, T }_{ F, v } ) }
  \in
  C(
    (t_0,
$
$
    \tau^{ t_0, T }_{ F, v } ), U_{ r_1 }
  )
$,
$
  \sup_{
    s \in (t_0,t]
  }
  ( s - t_0 )^{
    ( r_1 - r_0 )
  }
  \,
  \| 
    x^{ t_0, T }_{ F, v }(s)
  \|_{ U_{ r_1 } }
  < \infty 
$
and
\begin{equation}
  \lim_{ 
    s \nearrow \tau^{ t_0, T }_{ F, v }
  }
  \left[
    \tfrac{ 1
    }{
      ( T - s )
    }
    +
    \| 
      x^{ t_0, T }_{ F, v }(s) 
    \|_{ 
      U_{ 
        \max( \beta_1, \dots, \beta_n, \gamma_1, \dots, \gamma_n )
      } 
    }
  \right]
  = \infty 
\end{equation}
and
\begin{equation}
  x^{ t_0, T }_{ F, v }(t)
=
  e^{ A ( t - t_0 ) } \, v
+
  \sum_{ i = 1 }^n
  \int_{ t_0 }^t
  e^{ A ( t - s ) }
  \,
  F_i( s, x^{ t_0, T }_{ F, v }(s) ) \, ds
\end{equation}
for all 
$ t \in (t_0,\tau^{ t_0, T }_{ F, v } ) $,
$ v \in U_{ r_0 } $,
$
  r_1
  \in
  [  
    r_0 ,
    1 +
    \min( \alpha_1, \dots, 
$
$
    \alpha_n )
  )
$,
$
  F = ( F_1, \dots, F_n )
  \in 
  \mathcal{C}^n_{ \alpha, \beta, \gamma, \delta }( [t_0, T] )
$
and all
$ t_0, T \in \R $
with $ t_0 < T $.
In addition, it holds 
for every
$ t_0, T \in \R $ 
with $ t_0 < T $,
every
$ t \in (t_0, T] $
and every
$
  r_1
  \in
  \left[  
    r_0 ,
    1 +
    \min( \alpha_1, \dots, \alpha_n )
  \right)
$
that the function
\begin{equation}
\label{eq:x_mess}
  \mathcal{C}^n_{ \alpha, \beta, \gamma, \delta }( [t_0, T] )
  \times U_{ r_0 }
  \ni (F,v) \mapsto 
  \left\{
  \begin{array}{ll}
    x^{ t_0, T }_{ F, v }( t ) 
  &
    \colon
    t < \tau_{ F, v } 
  \\
    \infty 
  &
    \colon
    t \geq \tau_{ F, v }
  \end{array}
  \right\}
  \in U_{ r_1 } \cup \{ \infty \}
\end{equation}
is Borel measurable.
Moreover, it holds that
\begin{equation}
\label{eq:x_cont}
  \lim_{ N \to \infty }
  \sup_{ s \in (t_0, t] }
  \left[
  \begin{array}{c}
    \left( s - t_0 \right)^{
      \left( r_1 - r_0 \right)
    }
    \|
      x^{ t_0, T }_{ F_1, v_1 }(s)
      -
      x^{ t_0, T }_{ F_N, v_N }(s)
    \|_{
      U_{ r_1 }
    }
  \\[1ex]
    +
    \,
    \|
      x^{ t_0, T }_{ F_1, v_1 }(s)
      -
      x^{ t_0, T }_{ F_N, v_N }(s)
    \|_{
      U_{ r_0 }
    }
  \end{array}
  \right]
  = 0
\end{equation}
for all 
$ t \in ( t_0,\tau_{ F, v } ) $,
$
  r_1
  \in
  \left[  
    r_0 ,
    1 +
    \min( \alpha_1, \dots, \alpha_n )
  \right)
$,
$ 
  ( v_N )_{ N \in \N } 
$
$
  \subset U_{ r_0 } 
$,
$
  ( F_N )_{ N \in \N }
  \subset
  \mathcal{C}^n_{ \alpha, \beta, \gamma, \delta }( [t_0, T] )
$
with
$
  \lim_{ N \to \infty } 
  \left\|
    F_1 - F_N
  \right\|_{
    \mathcal{C}^n_{ \alpha, \beta, \gamma, \delta }( [ t_0, T ] )
  }
$
$
  =
$
$
  \lim_{ N \to \infty } 
  \|
    v_1 - v_N
  \|_{
    U_{ r_0 }
  }
$
$
  = 0
$
and all
$ t_0, T \in \R $
with $ t_0 < T $.
\end{thm}

\begin{proof}[Proof
of Theorem~\ref{thm:continuity}]
First of all, observe that
Corollary~\ref{cor:maximal2} ensures that
there exist unique functions
$ 
  \tau^{ t_0, T } 
  \colon 
  \mathcal{C}^n_{ \alpha, \beta, \gamma, \delta }( [t_0, T] ) 
  \times
  U_{ r_0 } 
  \to (t_0,T] 
$,
$ t_0, T \in \R $
with $ t_0 < T $,
and 
$
  x^{ t_0, T } \colon 
$
$
  \mathcal{C}^n_{ \alpha, \beta, \gamma, \delta }( [t_0, T] ) 
  \times
  U_{ r_0 } 
  \to
  \cup_{ s \in (t_0,T] }
  C( [t_0,s), U_{ r_0 } )
$,
$ t_0, T \in \R $
with $ t_0 < T $,
satisfying
$
  x^{ t_0, T }_{ F, v }
  \in
  C( [t_0, \tau^{ t_0, T }_{ F, v } ), U_{ r_0 } )
$,
$
  x^{ t_0, T }_{ F, v }|_{ (t_0,\tau^{ t_0, T }_{ F, v } ) }
  \in
  C(
    (t_0,\tau^{ t_0, T }_{ F, v } ), U_{ r_1 }
  )
$,
$
  \sup_{
    s \in (t_0,t]
  }
  \left( s - t_0 \right)^{
    \left( r_1 - r_0 \right)
  }
  \| 
    x^{ t_0, T }_{ F, v }(s)
  \|_{ U_{ r_1 } }
  < 
  \infty
$
and
\begin{equation}
  \lim_{ 
    s \nearrow \tau_{ F, v }
  }
  \left[
    \frac{ 1
    }{
      ( T - s )
    }
    +
    \| 
      x^{ t_0, T }_{ F, v }(s) 
    \|_{ 
      U_{ 
        \max( \beta_1, \dots, \beta_n, \gamma_1, \dots, \gamma_n ) 
      } 
    }
  \right]
  = \infty
\end{equation}
and
\begin{equation}
  x^{ t_0, T }_{ F, v }(t)
=
  e^{ A ( t - t_0 ) } \, v
+
  \sum_{ i = 1 }^n
  \int_{ t_0 }^t
  e^{ A ( t - s ) }
  \,
  F_i( s, x^{ t_0, T }_{ F, v }(s) ) \, ds
\end{equation}
for all 
$ t \in (t_0,\tau^{ t_0, T }_{ F, v }) $,
$ v \in U_{ r_0 } $,
$
  F = ( F_1, \dots, F_n )
  \in 
  \mathcal{C}^n_{ \alpha, \beta, \gamma, \delta }( [t_0, T] )
$,
$ t_0, T \in \R $
with $ t_0 < T $
and all
$
  r_1
  \in
  \left[  
    r_0 ,
    1 +
    \min( \alpha_1, \dots, \alpha_n )
  \right)
$.
It thus remains to prove 
that $ \tau^{ t_0, T } $,
$ t_0, T \in \R $ with $ t_0 < T $,
are lower semicontinuous and
that \eqref{eq:x_mess} and \eqref{eq:x_cont} are fulfilled.

For this let
$
  r_1
  \in
$
$
  [  
    \max( \beta_1, \dots, \beta_n, \delta_1, \dots, \delta_n ) ,
$
$
    1 +
    \min( \alpha_1, \dots, \alpha_n )
  )
$
be an arbitrary real number 
and 
let $ \kappa_{ [t_0, T] } \in [0,\infty) $,
$ t_0, T \in \R $ with $ t_0 < T $,
be real numbers defined through
\begin{equation}
\begin{split}
&
  \kappa_{
    [t_0, T]
  }
  :=
  \sum_{ j = 0 }^1
  \sum_{ i = 1 }^n
  \left[
  \frac{ 
    1
  }{
    \left( 
      1 
      + 
      \alpha_i 
      - r_j
    \right)
  }
  +
  B_{
    \left(
      1 + \min( \alpha_i - r_j , 0 ) ,
      1 + r_0 - \gamma_i +
      \delta_i
      \left( r_0 - \beta_i \right) 
    \right)
  }
  \right]
\\ &
  +
  \left[
    2 + n + r_1 - r_0 
    + | T - t_0 | 
    + 
    \sum_{ i = 1 }^n \delta_i
  \right]^{
    \left( 
      4  
      + \left| r_0 \right|
      + \left| r_1 \right|
      +
      \max_{ i \in \{ 1, \dots, n \} }
      \left| \alpha_i \right|
    \right)
  }
\\ &
  +
  \max_{ j \in \{ 0, 1 \} }
  \max_{ 
    \theta 
    \in 
    \{
      r_0, r_1, \alpha_1, \dots, \alpha_n
    \}
  }
  \sup_{ t \in (t_0,T] }
  \left[
    \left( t - t_0 
    \right)^{
      \max\left( r_j - \theta , 0 \right)
    } 
  \|
    e^{ A \left( t - t_0 \right) }
  \|_{ 
    L( U_{ \theta }, U_{ r_j } ) 
  }
  \right]
\\ & 
  +
  \max_{
  \substack{
    \theta \in 
    \{ 
      \beta_1, \dots, \beta_n 
    \}
  \\
    \cup
    \{
      \gamma_1, \dots, \gamma_n 
    \} 
  }
  }
  \sup_{ 
    \substack{
      v \in U_{ r_1 } 
    \\
      v \neq 0
    }
  }
  \left[
  1 +
  \frac{
    \left\|  
      v
    \right\|_{ U_{ \theta } }
  }{
    \left\|
      v
    \right\|_{ U_{ r_1 } 
    }^{
      \frac{
        ( \theta - r_0 )
      }{
        ( r_1 - r_0 )
      }
    }
    \left\|
      v
    \right\|_{ U_{ r_0 } }^{
      \frac{
        ( r_1 - \theta )
      }{
        ( r_1 - r_0 )
      }
    } 
  }
  +
  \frac{
    \left\|  
      v
    \right\|_{ U_{ r_0 } }
  }{
    \left\|
      v
    \right\|_{ U_{ r_1 } 
    }
  }
  \right]^{
    \!
    ( 1 + \sum_{ i = 1 }^n \delta_i )
  }
  < \infty 
\end{split}
\end{equation}
for all $ t_0, T \in \R $
with $ t_0 < T $
where
$
  B \colon (0,\infty)^2 \to (0,\infty)
$
is the Beta function defined through
$
  B_{ (x, y) }
:=
  \int_0^1
  \left( 1 - s \right)^{
    \left( x - 1 \right)
  }
$
$
  s^{ \left( y - 1 \right) }
  \, ds
$
for all $ x, y \in (0,\infty) $.
Then observe that
\begin{equation}
\begin{split}
&
  \big\|
    x_{ F, v }^{ t_0, T }(t) 
    -
    x_{ \tilde{F}, \tilde{v} }^{ t_0, T }(t)
  \big\|_{ 
    U_{ r_j } 
  }
\leq 
  \big\|
    e^{ A ( t - t_0 ) } 
  \big\|_{
    L( U_{ r_k }, U_{ r_j } )
  }
  \left\|
    v - \tilde{v} 
  \right\|_{
    U_{ r_k }
  }
\\ & 
  +
  \sum_{ i = 1 }^n
  \int_{ t_0 }^t
  \big\|
    e^{ A ( t - s ) }
  \big\|_{
    L( U_{ \alpha_i }, U_{ r_j } )
  }
  \,
  \big\|
    F_i( s, x_{ F, v }^{ t_0, T }(s) ) -
    \tilde{F}_i( s, x_{ \tilde{F}, \tilde{v} }^{ t_0, T }(s) ) 
  \big\|_{
    U_{ \alpha_i }
  }
  ds
\\ & \leq 
  \kappa_{ [t_0, T] }
  \left( t - t_0 \right)^{
    \min\left( r_k - r_j , 0 \right)
  }
  \left\|
    v - \tilde{v} 
  \right\|_{
    U_{ r_k }
  }
\\ & 
  +
  \sum_{ i = 1 }^n
  \int_{ t_0 }^t
  \kappa_{ [t_0, T] }
  \left( t - t_0 \right)^{
    \min\left( \alpha_i - r_j, 0 \right)
  }
  \|
    F_i( s, x_{ F, v }^{ t_0, T }(s) ) -
    \tilde{F}_i( s, x_{ F, v }^{ t_0, T }(s) ) 
  \|_{
    U_{ \alpha_i }
  }
  \, ds
\\ &  
  +
  \sum_{ i = 1 }^n
  \int_{ t_0 }^t
  \kappa_{ [t_0, T] }
  \left( t - t_0 \right)^{
    \min\left( \alpha_i - r_j, 0 \right)
  }
  \|
    \tilde{F}_i( s, x_{ F, v }^{ t_0, T }(s) ) -
    \tilde{F}_i( s, x_{ \tilde{F}, \tilde{v} }^{ t_0, T }(s) ) 
  \|_{
    U_{ \alpha_i }
  }
  \, ds
\end{split}
\end{equation}
and inequality~\eqref{eq:linear_growth_bound}
therefore implies that
\begin{equation}
\begin{split}
&
  \big\|
    x^{ t_0, T }_{ F, v }(t) 
    -
    x^{ t_0, T }_{ \tilde{F}, \tilde{v} }(t)
  \big\|_{ 
    U_{ r_j } 
  }
  \leq 
  \kappa_{ [t_0, T] }
  \left( t - t_0 \right)^{
    \min\left( r_k - r_j , 0 \right)
  }
  \left\|
    v - \tilde{v} 
  \right\|_{
    U_{ r_k }
  }
\\ & 
  +
  \kappa_{ [t_0, T] }
  \,
  \| F - \tilde{F} \|_{
    \mathcal{C}^n_{ \alpha, \beta, \gamma, \delta }( [t_0,T] )
  }
\\ & \cdot
  \sum_{ i = 1 }^n
  \int_{ t_0 }^t
  \left( t - s \right)^{
    \min\left( \alpha_i - r_j, 0 \right)
  }
  \big(
    1 + 
    \| 
      x^{ t_0, T }_{ F, v }(s) 
    \|_{ \beta_i }^{ \delta_i }
  \big)
  \,
  \big(
    1 +
    \|
      x^{ t_0, T }_{ F, v }(s) 
    \|_{
      U_{ \gamma_i }
    }
  \big)
  \, ds
\\ &  
  +
  \kappa_{ [t_0, T] }
  \,
  \|
    \tilde{F}
  \|_{
    \mathcal{C}^n_{ \alpha, \beta, \gamma, \delta }( [t_0,T] )
  }
\\ & \cdot
  \sum_{ i = 1 }^n
  \int_{ t_0 }^t
  \left( t - s \right)^{
    \min\left( \alpha_i - r_j, 0 \right)
  }
  \left(
    1 
    +
    \|
      x^{ t_0, T }_{ F, v }(s) 
    \|_{
      U_{ \beta_i }
    }^{ \delta_i }
    +
    \|
      x^{ t_0, T }_{ \tilde{F}, \tilde{v} }(s) 
    \|_{
      U_{ \beta_i }
    }^{ \delta_i }
  \right)
\\ & \cdot
  \|
    x^{ t_0, T }_{ F, v }(s) -
    x^{ t_0, T }_{ \tilde{F}, \tilde{v} }(s) 
  \|_{
    U_{ \gamma_i }
  }
  \, ds
\end{split}
\end{equation}
and the definition 
of $ \kappa_{ [t_0,T] } $
hence shows that
\begin{equation}
\label{eq:main_estimate}
\begin{split}
&
  \big\|
    x^{ t_0, T }_{ F, v }(t) 
    -
    x^{ t_0, T }_{ \tilde{F}, \tilde{v} }(t)
  \big\|_{ 
    U_{ r_j } 
  }
  \leq 
  \kappa_{ [t_0, T] }
  \left( t - t_0 \right)^{
    \min\left( r_k - r_j , 0 \right)
  }
  \left\|
    v - \tilde{v} 
  \right\|_{
    U_{ r_k }
  }
\\ & 
  +
  \left[
    \kappa_{ [t_0, T] } 
  \right]^{
    3
  } 
  \| F - \tilde{F} \|_{
    \mathcal{C}^n_{ \alpha, \beta, \gamma, \delta }( [t_0,T] )
  }
  \sum_{ i = 1 }^n
  \int_{ t_0 }^t
  \left( t - s \right)^{
    \min\left( \alpha_i - r_j, 0 \right)
  }
\\ & 
  \cdot
  \left[
    1 + 
    \| 
      x^{ t_0, T }_{ F, v }(s) 
    \|_{ 
      U_{ r_0 } 
    }^{ 
      \frac{ 
        ( r_1 - \beta_i ) \delta_i 
      }{
        ( r_1 - r_0 )
      }
    }
    \| 
      x^{ t_0, T }_{ F, v }(s) 
    \|_{ U_{ r_1 } 
    }^{ 
      \frac{ 
        ( \beta_i - r_0 ) \delta_i 
      }{
        ( r_1 - r_0 )
      }
    }
  \right]
  \left[
    1 
    +
    \|
      x^{ t_0, T }_{ F, v }(s) 
    \|_{
      U_{ r_0 }
    }^{
      \frac{
        ( r_1 - \gamma_i )
      }{
        ( r_1 - r_0 )
      }
    }
    \|
      x^{ t_0, T }_{ F, v }(s) 
    \|_{
      U_{ r_1 }
    }^{
      \frac{
        ( \gamma_i - r_0 )
      }{
        ( r_1 - r_0 )
      }
    }
  \right]
  ds
\\ &
  +
  \left[
    \kappa_{ [t_0, T] }
  \right]^3
  \|
    \tilde{F}
  \|_{
    \mathcal{C}^n_{ \alpha, \beta, \gamma, \delta }( [t_0,T] )
  }
  \sum_{ i = 1 }^n
  \int_{ t_0 }^t
  \left( t - s \right)^{
    \min\left( \alpha_i - r_j, 0 \right)
  }
\\ & \cdot
  \left[
    1 
    +
    \|
      x^{ t_0, T }_{ F, v }(s) 
    \|_{
      U_{ r_0 }
    }^{ 
      \frac{ 
        ( r_1 - \beta_i ) 
        \delta_i 
      }{
        ( r_1 - r_0 )
      }
    }
    \|
      x^{ t_0, T }_{ F, v }(s) 
    \|_{
      U_{ r_1 }
    }^{ 
      \frac{ 
        ( \beta_i - r_0 ) 
        \delta_i 
      }{
        ( r_1 - r_0 )
      }
    }
    +
    \|
      x^{ t_0, T }_{ \tilde{F}, \tilde{v} }(s) 
    \|_{
      U_{ r_0 }
    }^{ 
      \frac{ 
        ( r_1 - \beta_i ) 
        \delta_i 
      }{
        ( r_1 - r_0 )
      }
    }
    \|
      x^{ t_0, T }_{ \tilde{F}, \tilde{v} }(s) 
    \|_{
      U_{ r_1 }
    }^{ 
      \frac{ 
        ( \beta_i - r_0 ) 
        \delta_i 
      }{
        ( r_1 - r_0 )
      }
    }
  \right]
\\ & \cdot
  \|
    x^{ t_0, T }_{ F, v }(s) -
    x^{ t_0, T }_{ \tilde{F}, \tilde{v} }(s) 
  \|_{
    U_{ r_0 }
  }^{
    \frac{
      ( r_1 - \gamma_i )
    }{
      ( r_1 - r_0 )
    }
  }
  \|
    x^{ t_0, T }_{ F, v }(s) -
    x^{ t_0, T }_{ \tilde{F}, \tilde{v} }(s) 
  \|_{
    U_{ r_1 }
  }^{
    \frac{
      ( \gamma_i - r_0 )
    }{
      ( r_1 - r_0 )
    }
  }
  \, ds
\end{split}
\end{equation}
for all 
$
  j, k \in \{ 0, 1 \}
$,
$ 
  t \in 
  (
    t_0, 
    \tau^{ t_0, T }_{ F, v } 
  )
  \cap 
  (  
    t_0, 
    \tau^{ t_0, T }_{ \tilde{F}, \tilde{v} } 
  )
$,
$
  v, \tilde{v} \in U_{ r_0 }
$,
$
  F, \tilde{F} \in
  \mathcal{C}^n_{ \alpha, \beta, \gamma, \delta }( [t_0, T] )
$
and all
$ t_0, T \in \R $
with $ t_0 < T $.
This, in particular, implies that
\begin{equation}
\begin{split}
&
  \big\|
    x^{ t_0, T }_{ F, v }(t) 
    -
    x^{ t_0, T }_{ \tilde{F}, \tilde{v} }(t)
  \big\|_{ 
    U_{ r_1 } 
  }
\leq 
  \kappa_{ [t_0, T] }
  \left\|
    v - \tilde{v} 
  \right\|_{
    U_{ r_1 }
  }
\\ & 
  +
  \left[
    \kappa_{ [t_0, T] } 
  \right]^{
    5
  } 
  \| F - \tilde{F} \|_{
    \mathcal{C}^n_{ \alpha, \beta, \gamma, \delta }( [t_0,T] )
  }
  \sum_{ i = 1 }^n
  \int_{ t_0 }^t
  \left( t - s \right)^{
    \min\left( \alpha_i - r_1, 0 \right)
  }
\\ & 
  \cdot
  \left[
    1 + 
    \| 
      x^{ t_0, T }_{ F, v }(s) 
    \|_{ U_{ r_1 } 
    }^{ 
      \delta_i
    }
  \right]
  \left[
    1 
    +
    \|
      x^{ t_0, T }_{ F, v }(s) 
    \|_{
      U_{ r_1 }
    }
  \right]
  ds
\\ &
  +
  \left[
    \kappa_{ [t_0, T] }
  \right]^5
  \|
    \tilde{F}
  \|_{
    \mathcal{C}^n_{ \alpha, \beta, \gamma, \delta }( [t_0,T] )
  }
  \sum_{ i = 1 }^n
  \int_{ t_0 }^t
  \left( t - s \right)^{
    \min\left( \alpha_i - r_1, 0 \right)
  }
\\ & \cdot
  \left[
    1 
    +
    \|
      x^{ t_0, T }_{ F, v }(s) 
    \|_{
      U_{ r_1 }
    }^{ 
      \delta_i 
    }
    +
    \|
      x^{ t_0, T }_{ \tilde{F}, \tilde{v} }(s) 
    \|_{
      U_{ r_1 }
    }^{ 
      \delta_i 
    }
  \right]
  \|
    x^{ t_0, T }_{ F, v }(s) -
    x^{ t_0, T }_{ \tilde{F}, \tilde{v} }(s) 
  \|_{
    U_{ r_1 }
  }
  \, ds
\end{split}
\end{equation}
and 
the estimates
$
  \big(
    1 + | x |^{ \delta_i }
  \big)
  \big(
    1 + | x |
  \big)
  \leq 
  \kappa 
  \left(
    1 + \left| x \right| 
  \right)^{
    ( 2 + \sum_{ j = 1 }^n \delta_j )
  }
$
and
$
  \big(
    1 + 
    | x |^{ \delta_i }
    + 
    | y |^{ \delta_i }
  \big)
  \leq 
  \kappa 
  \left(
    1 
    + \left| x \right| 
    + \left| y \right| 
  \right)^{
    ( 1 + \sum_{ j = 1 }^n \delta_j )
  }
$
for all 
$ x, y \in \R $
and all
$ i \in \{ 1, \dots, n \} $
hence give
\begin{equation}
\begin{split}
&
  \big\|
    x^{ t_0, T }_{ F, v }(t) 
    -
    x^{ t_0, T }_{ \tilde{F}, \tilde{v} }(t)
  \big\|_{ 
    U_{ r_1 } 
  }
\leq 
  \kappa_{ [t_0, T] }
  \left\|
    v - \tilde{v} 
  \right\|_{
    U_{ r_1 }
  }
\\ & 
  +
  \left[
    \kappa_{ [t_0, T] } 
  \right]^{
    6
  } 
  \| F - \tilde{F} \|_{
    \mathcal{C}^n_{ \alpha, \beta, \gamma, \delta }( [t_0,T] )
  }
  \left[
    1 +
    \sup_{ s \in [t_0, t] }
    \| 
      x^{ t_0, T }_{ F, v }( s ) 
    \|_{ 
      U_{ r_1 } 
    }
  \right]^{
    ( 2 + \sum_{ i = 1 }^n \delta_i )
  }
\\ & 
  \cdot
  \sum_{ i = 1 }^n
  \int_{ t_0 }^t
  \left( t - s \right)^{
    \min\left( \alpha_i - r_1, 0 \right)
  }
  ds
\\ &
  +
  \left[
    \kappa_{ [t_0, T] }
  \right]^6
  \!
  \|
    \tilde{F}
  \|_{
    \mathcal{C}^n_{ \alpha, \beta, \gamma, \delta }( [t_0,T] )
  }
  \!
  \left[
    1 
    +
    \sup_{ s \in [t_0, t] }
    \|
      x^{ t_0, T }_{ F, v }(s)
    \|_{ 
      U_{ r_1 } 
    }
    +
    \sup_{ s \in [t_0, t] }
    \|
      x^{ t_0, T }_{ \tilde{F}, \tilde{v} }(s)
    \|_{
      U_{ r_1 } 
    }
  \right]^{
    ( 
      1 
      + 
      \sum_{ i = 1 }^n \delta_i 
    )
  }
\\ & \cdot
  \sum_{ i = 1 }^n
  \int_{ t_0 }^t
  \left( t - s \right)^{
    \min\left( \alpha_i - r_1, 0 \right)
  }
  \|
    x^{ t_0, T }_{ F, v }(s) -
    x^{ t_0, T }_{ \tilde{F}, \tilde{v} }(s) 
  \|_{
    U_{ r_1 }
  }
  \, ds
\end{split}
\end{equation}
for all 
$ 
  t \in 
  (
    t_0, 
    \tau^{ t_0, T }_{ F, v } 
  )
  \cap 
  (  
    t_0, 
    \tau^{ t_0, T }_{ \tilde{F}, \tilde{v} } 
  )
$,
$
  v, \tilde{v} \in U_{ r_1 }
$,
$
  F, \tilde{F} \in
  \mathcal{C}^n_{ \alpha, \beta, \gamma, \delta }( [t_0, T] )
$
and all
$ t_0, T \in \R $
with $ t_0 < T $.
Therefore, we obtain that
\begin{equation}
\begin{split}
&
  \big\|
    x^{ t_0, T }_{ F, v }(t) 
    -
    x^{ t_0, T }_{ \tilde{F}, \tilde{v} }(t)
  \big\|_{ 
    U_{ r_1 } 
  }
\leq 
  \kappa_{ [t_0, T] }
  \left\|
    v - \tilde{v} 
  \right\|_{
    U_{ r_1 }
  }
\\ & 
  +
  \left[
    \kappa_{ [t_0, T] } 
  \right]^{
    8
  } 
  \| F - \tilde{F} \|_{
    \mathcal{C}^n_{ \alpha, \beta, \gamma, \delta }( [t_0,T] )
  }
  \left[
    1 + 
    \sup_{ s \in [t_0, t] }
    \| 
      x^{ t_0, T }_{ F, v }( s ) 
    \|_{ 
      U_{ r_1 } 
    }
  \right]^{
    ( 2 + \sum_{ i = 1 }^n \delta_i )
  }
\\ & 
  \cdot
  \int_{ t_0 }^t
  \left( t - s \right)^{
    \min\left( \alpha_1 - r_1, \dots, \alpha_n - r_1, 0 \right)
  }
  ds
\\ &
  +
  \left[
    \kappa_{ [t_0, T] }
  \right]^8
  \|
    \tilde{F}
  \|_{
    \mathcal{C}^n_{ \alpha, \beta, \gamma, \delta }( [t_0,T] )
  }
  \left[
    1 
    +
    \sup_{ s \in [t_0, t] }
    \|
      x^{ t_0, T }_{ F, v }(s)
    \|_{ 
      U_{ r_1 } 
    }
    +
    \sup_{ s \in [t_0, t] }
    \|
      x^{ t_0, T }_{ \tilde{F}, \tilde{v} }(s)
    \|_{
      U_{ r_1 }
    }
  \right]^{
    ( 
      1 
      + 
      \sum_{ i = 1 }^n \delta_i 
    )
  }
\\ & \cdot
  \int_{ t_0 }^t
  \left( t - s \right)^{
    \min\left( \alpha_1 - r_1, \dots, \alpha_n - r_1, 0 \right)
  }
  \|
    x^{ t_0, T }_{ F, v }(s) -
    x^{ t_0, T }_{ \tilde{F}, \tilde{v} }(s) 
  \|_{
    U_{ r_1 }
  }
  \, ds
\end{split}
\end{equation}
and hence
\begin{equation}
\begin{split}
&
  \big\|
    x^{ t_0, T }_{ F, v }(t) 
    -
    x^{ t_0, T }_{ \tilde{F}, \tilde{v} }(t)
  \big\|_{ 
    U_{ r_1 } 
  }
\\ & \leq
  \left[
    \kappa_{ [t_0, T] } 
  \right]^{
    10
  } 
  \left[ 
    \left\|
      v - \tilde{v} 
    \right\|_{
      U_{ r_1 }
    }
    +
    \| F - \tilde{F} \|_{
      \mathcal{C}^n_{ \alpha, \beta, \gamma, \delta }( [t_0,T] )
    }
  \right]
  \left[
    1 + 
    \sup_{ s \in [t_0, t] }
    \| 
      x^{ t_0, T }_{ F, v }(s) 
    \|_{
      U_{ r_1 }
    }
  \right]^{
    ( 2 + \sum_{ i = 1 }^n \delta_i )
  }
\\ &
  +
  \left[
    \kappa_{ [t_0, T] }
  \right]^8
  \|
    \tilde{F}
  \|_{
    \mathcal{C}^n_{ \alpha, \beta, \gamma, \delta }( [t_0,T] )
  }
  \left[
    1 
    +
    \sup_{ s \in [t_0, t] }
    \|
      x^{ t_0, T }_{ F, v }(s)
    \|_{ 
      U_{ r_1 } 
    }
    +
    \sup_{ s \in [t_0, t] }
    \|
      x^{ t_0, T }_{ \tilde{F}, \tilde{v} }(s)
    \|_{
      U_{ r_1 } 
    }
  \right]^{
    ( 
      1 
      + 
      \sum_{ i = 1 }^n \delta_i 
    )
  }
\\ & \cdot
  \int_{ t_0 }^t
  \left( t - s \right)^{
    \min\left( \alpha_1 - r_1, \dots, \alpha_n - r_1, 0 \right)
  }
  \|
    x^{ t_0, T }_{ F, v }(s) -
    x^{ t_0, T }_{ \tilde{F}, \tilde{v} }(s) 
  \|_{
    U_{ r_1 }
  }
  \, ds
\end{split}
\end{equation}
for all 
$ 
  t \in 
  (
    t_0, 
    \tau^{ t_0, T }_{ F, v } 
  )
  \cap 
  (  
    t_0, 
    \tau^{ t_0, T }_{ \tilde{F}, \tilde{v} } 
  )
$,
$
  v, \tilde{v} \in U_{ r_1 }
$,
$
  F, \tilde{F} \in
  \mathcal{C}^n_{ \alpha, \beta, \gamma, \delta }( [t_0, T] )
$
and all
$ t_0, T \in \R $
with $ t_0 < T $.
A generalization of Gronwall's lemma 
(see Lemma~7.1.1 in 
Henry~\cite{h81})
therefore implies
\begin{equation}
\label{eq:continuity_second}
\begin{split}
&
  \sup_{ 
    s \in [t_0,t]
  }
  \big\|
    x^{ t_0, T }_{ F, v }(s)
    -
    x^{ t_0, T }_{ \tilde{F}, \tilde{v} }(s)
  \big\|_{ 
    U_{ r_1 } 
  }
\leq
  E_{
    \min\left( \alpha_1 - r_1, \dots, \alpha_n - r_1, 0 \right)
  }
  \Bigg[ 
  \left[
    \kappa_{ [t_0, T] }
  \right]^{ 9 }
  \|
    \tilde{F}
  \|_{
    \mathcal{C}^n_{ \alpha, \beta, \gamma, \delta }( [t_0,T] )
  }
\\ &  
  \cdot
  \Big[
    1 
    +
    \sup_{ 
      s \in [t_0,t]
    }
    \|
      x^{ t_0, T }_{ F, v }(s)
    \|_{ 
      U_{ r_1 }
    }
    +
    \sup_{ 
      s \in [t_0,t]
    }
    \|
      x^{ t_0, T }_{ \tilde{F}, \tilde{v} }(s)
    \|_{
      U_{ r_1 }
    }
  \Big]^{
    ( 
      1 
      + 
      \sum_{ i = 1 }^n \delta_i 
    )
  }
  \Bigg]
  \left[
    \kappa_{ [t_0, T] } 
  \right]^{
    10
  } 
\\ & \cdot
  \left[ 
    \left\|
      v - \tilde{v} 
    \right\|_{
      U_{ r_1 }
    }
    +
    \| F - \tilde{F} \|_{
      \mathcal{C}^n_{ \alpha, \beta, \gamma, \delta }( [t_0,T] )
    }
  \right]
  \Big[
    1 + 
    \sup_{ s \in [t_0, t] }    
    \| 
      x^{ t_0, T }_{ F, v }(s) 
    \|_{ 
      U_{ r_1 }
    }
  \Big]^{
    ( 2 + \sum_{ i = 1 }^n \delta_i )
  }
\end{split}
\end{equation}
for all 
$ 
  t \in 
  (
    t_0, 
    \tau^{ t_0, T }_{ F, v } 
  )
  \cap 
  (  
    t_0, 
    \tau^{ t_0, T }_{ \tilde{F}, \tilde{v} } 
  )
$,
$
  v, \tilde{v} \in U_{ r_1 }
$,
$
  F, \tilde{F} \in
  \mathcal{C}^n_{ \alpha, \beta, \gamma, \delta }( [t_0, T] )
$
and all
$ t_0, T \in \R $
with $ t_0 < T $
where 
$ E_r \colon [0,\infty) \to [0,\infty) $,
$ r \in ( - 1, 0 ] $,
is a family of functions defined
through
$
  E_r( x )
  :=
  \sum_{ n = 0 }^{ \infty }
  \frac{
    (
      x \cdot
      \Gamma( r + 1 )
    )^{ 
      n
    }
  }{
    \Gamma( n ( r + 1 ) + 1 )
  }
$
for all 
$ x \in [0,\infty) $
and all
$ r \in ( - 1, 0 ] $.
As in \eqref{eq:defE} and \eqref{eq:defEnorm},
we now define sets
$ \mathcal{E}_{ [t_0,T] } $,
$ t_0, T \in \R $
with $ t_0 < T $,
and functions
$
  \left\| \cdot \right\|_{
    \mathcal{E}_{ [t_0,T] }
  }
  \colon
  \mathcal{E}_{ [t_0,T] }
  \to
  [0,\infty)
$,
$ t_0, T \in \R $
with $ t_0 < T $,
by
\begin{equation}
  \mathcal{E}_{ [t_0,T] }
  :=
  \left\{
    y \in 
    C( [t_0,T], U_{ r_0 } )
    \colon
    \left(
    \begin{array}{c}
      y|_{
        (t_0,T]
      }
      \in
      C( (t_0,T], U_{ r_1 } )
      \text{    and}
      \\
      \sup_{ t \in (t_0,T] }
        \left( t - t_0 \right)^{ \left( r_1 - r_0 \right) }
        \left\| 
          y(t) 
        \right\|_{
          U_{ r_1 }
        }
      < \infty
    \end{array}
    \right)
  \right\} 
\end{equation}
for all 
$ t_0, T \in \R $
with $ t_0 < T $
and by
$
  \left\| 
    y
  \right\|_{
    \mathcal{E}_{ [t_0,T] }
  }
  :=
  \sum_{ j = 0 }^1 
  \sup_{ t \in (t_0,\tau] }
    \left( t - t_0 \right)^{
      \left( r_j - r_0 \right)
    }
    \| 
      y(t) 
    \|_{
      U_{ r_j }
    }
$
for all 
$ y \in \mathcal{E}_{ [t_0, T] } $,
$ t_0, T \in \R $ with $ t_0 < T $.
Then we get
from \eqref{eq:main_estimate} that
\begin{equation}
\begin{split}
&
  \big\|
    x^{ t_0, T }_{ F, v }(t) 
    -
    x^{ t_0, T }_{ \tilde{F}, \tilde{v} }(t)
  \big\|_{ 
    U_{ r_j } 
  }
  \leq 
  \kappa_{ [t_0, T] }
  \left( t - t_0 \right)^{
    \left( r_0 - r_j \right)
  }
  \left\|
    v - \tilde{v} 
  \right\|_{
    U_{ r_0 }
  }
\\ & 
  +
  \left[
    \kappa_{ [t_0, T] } 
  \right]^{
    4
  } 
  \| F - \tilde{F} \|_{
    \mathcal{C}^n_{ \alpha, \beta, \gamma, \delta }( [t_0,T] )
  }
  \sum_{ i = 1 }^n
  \int_{ t_0 }^t
  \left( t - s \right)^{
    \min\left( \alpha_i - r_j, 0 \right)
  }
  \left(
    s - t_0
  \right)^{
    ( 
      r_0 - \gamma_i
      +
      \delta_i 
      \left(
        r_0 - \beta_i 
      \right)
    )
  }
\\ & 
  \cdot
  \left[
    1 + 
    \| 
      x^{ t_0, T }_{ F, v }(s) 
    \|_{ 
      U_{ r_0 } 
    }^{ 
      \frac{ 
        ( r_1 - \beta_i ) \delta_i 
      }{
        ( r_1 - r_0 )
      }
    }
    \left(
      s - t_0
    \right)^{
      ( \beta_i - r_0 ) \delta_i
    }
    \| 
      x^{ t_0, T }_{ F, v }(s) 
    \|_{ U_{ r_1 } 
    }^{ 
      \frac{ 
        ( \beta_i - r_0 ) \delta_i 
      }{
        ( r_1 - r_0 )
      }
    }
  \right]
\\ & \cdot
  \left[
    1 
    +
    \|
      x^{ t_0, T }_{ F, v }(s) 
    \|_{
      U_{ r_0 }
    }^{
      \frac{
        ( r_1 - \gamma_i )
      }{
        ( r_1 - r_0 )
      }
    }
    \left(
      s - t_0
    \right)^{
      ( \gamma_i - r_0 ) 
    }
    \|
      x^{ t_0, T }_{ F, v }(s) 
    \|_{
      U_{ r_1 }
    }^{
      \frac{
        ( \gamma_i - r_0 )
      }{
        ( r_1 - r_0 )
      }
    }
  \right]
  ds
\\ &
  +
  \left[
    \kappa_{ [t_0, T] }
  \right]^4
  \|
    \tilde{F}
  \|_{
    \mathcal{C}^n_{ \alpha, \beta, \gamma, \delta }( [t_0,T] )
  }
  \sum_{ i = 1 }^n
  \int_{ t_0 }^t
  \left( t - s \right)^{
    \min\left( \alpha_i - r_j, 0 \right)
  }
  \left(
    s - t_0
  \right)^{
    (
      r_0 - \gamma_i 
      + \delta_i ( r_0 - \beta_i )
    )
  }
\\ & \cdot
  \bigg[
    1 
    +
    \|
      x^{ t_0, T }_{ F, v }(s) 
    \|_{
      U_{ r_0 }
    }^{ 
      \frac{ 
        ( r_1 - \beta_i ) 
        \delta_i 
      }{
        ( r_1 - r_0 )
      }
    }
    \left(
      s - t_0
    \right)^{
      ( \beta_i - r_0 ) \delta_i
    }
    \|
      x^{ t_0, T }_{ F, v }(s) 
    \|_{
      U_{ r_1 }
    }^{ 
      \frac{ 
        ( \beta_i - r_0 ) 
        \delta_i 
      }{
        ( r_1 - r_0 )
      }
    }
\\ & 
    +
    \|
      x^{ t_0, T }_{ \tilde{F}, \tilde{v} }(s) 
    \|_{
      U_{ r_0 }
    }^{ 
      \frac{ 
        ( r_1 - \beta_i ) 
        \delta_i 
      }{
        ( r_1 - r_0 )
      }
    }
    \left(
      s - t_0
    \right)^{
      ( \beta_i - r_0 ) \delta_i
    }
    \|
      x^{ t_0, T }_{ \tilde{F}, \tilde{v} }(s) 
    \|_{
      U_{ r_1 }
    }^{ 
      \frac{ 
        ( \beta_i - r_0 ) 
        \delta_i 
      }{
        ( r_1 - r_0 )
      }
    }
  \bigg]
\\ & \cdot
  \|
    x^{ t_0, T }_{ F, v }(s) -
    x^{ t_0, T }_{ \tilde{F}, \tilde{v} }(s) 
  \|_{
    U_{ r_0 }
  }^{
    \frac{
      ( r_1 - \gamma_i )
    }{
      ( r_1 - r_0 )
    }
  }
  \left(
    s - t_0
  \right)^{
    ( \gamma_i - r_0 )
  }
  \|
    x^{ t_0, T }_{ F, v }(s) -
    x^{ t_0, T }_{ \tilde{F}, \tilde{v} }(s) 
  \|_{
    U_{ r_1 }
  }^{
    \frac{
      ( \gamma_i - r_0 )
    }{
      ( r_1 - r_0 )
    }
  }
  \, ds
\end{split}
\end{equation}
and therefore
\begin{equation}
\begin{split}
&
  \big\|
    x^{ t_0, T }_{ F, v }(t) 
    -
    x^{ t_0, T }_{ \tilde{F}, \tilde{v} }(t)
  \big\|_{ 
    U_{ r_j } 
  }
  \leq 
  \kappa_{ [t_0, T] }
  \left( t - t_0 \right)^{
    \left( r_0 - r_j \right)
  }
  \left\|
    v - \tilde{v} 
  \right\|_{
    U_{ r_0 }
  }
\\ & 
  +
  \left[
    \kappa_{ [t_0, T] } 
  \right]^{
    4
  } 
  \| F - \tilde{F} \|_{
    \mathcal{C}^n_{ \alpha, \beta, \gamma, \delta }( [t_0,T] )
  }
  \sum_{ i = 1 }^n
  \int_{ t_0 }^t
  \left( t - s \right)^{
    \min\left( \alpha_i - r_j, 0 \right)
  }
  \left(
    s - t_0
  \right)^{
    ( 
      r_0 - \gamma_i
      +
      \delta_i 
      \left(
        r_0 - \beta_i 
      \right)
    )
  }
  ds
\\ & 
  \cdot
  \left[
    1 + 
    \| 
      x^{ t_0, T }_{ F, v 
      }|_{
        [t_0, t]
      }
    \|_{ 
      \mathcal{E}_{ [t_0, t] } 
    }^{ \delta_i }
  \right]
  \left[
    1 
    +
    \|
      x^{ t_0, T }_{ F, v 
      }|_{
        [t_0, t]
      }
    \|_{
      \mathcal{E}_{ [t_0, t] }
    }
  \right]
\\ &
  +
  \left[
    \kappa_{ [t_0, T] }
  \right]^4
  \|
    \tilde{F}
  \|_{
    \mathcal{C}^n_{ \alpha, \beta, \gamma, \delta }( [t_0,T] )
  }
  \sum_{ i = 1 }^n
  \int_{ t_0 }^t
  \left( t - s \right)^{
    \min\left( \alpha_i - r_j, 0 \right)
  }
  \left(
    s - t_0
  \right)^{
    (
      r_0 - \gamma_i 
      + \delta_i ( r_0 - \beta_i )
    )
  }
  ds
\\ & \cdot
  \left[
    1 
    +
    \|
      x^{ t_0, T }_{ F, v 
      }|_{
        [t_0, t]
      }
    \|_{
      \mathcal{E}_{ [t_0,t] }
    }^{ \delta_i }
    +
    \|
      x^{ t_0, T }_{ \tilde{F}, \tilde{v} 
      }|_{
        [t_0, t]
      }
    \|_{
      \mathcal{E}_{ [t_0,t] }
    }^{ \delta_i }
  \right]
  \|
    (
      x^{ t_0, T }_{ F, v } -
      x^{ t_0, T }_{ \tilde{F}, \tilde{v} }
    )|_{
      [t_0, t]
    }
  \|_{
    \mathcal{E}_{ [t_0, t] }
  }
\end{split}
\end{equation}
and the estimates
$
  \big(
    1 + | x |^{ \delta_i }
  \big)
  \big(
    1 + | x |
  \big)
  \leq 
  \kappa 
  \left(
    1 + \left| x \right| 
  \right)^{
    ( 2 + \sum_{ j = 1 }^n \delta_j )
  }
$
and
$
  \big(
    1 + 
    | x |^{ \delta_i }
    + 
    | y |^{ \delta_i }
  \big)
  \leq 
  \kappa 
  \left(
    1 
    + \left| x \right| 
    + \left| y \right| 
  \right)^{
    ( 1 + \sum_{ j = 1 }^n \delta_j )
  }
$
for all 
$ x, y \in \R $
and all
$ i \in \{ 1, \dots, n \} $
hence show that
\begin{equation}
\begin{split}
&
  \big\|
    x^{ t_0, T }_{ F, v }(t) 
    -
    x^{ t_0, T }_{ \tilde{F}, \tilde{v} }(t)
  \big\|_{ 
    U_{ r_j } 
  }
  \leq 
  \kappa_{ [t_0, T] }
  \left( t - t_0 \right)^{
    \left( r_0 - r_j \right)
  }
  \left\|
    v - \tilde{v} 
  \right\|_{
    U_{ r_0 }
  }
\\ & 
  +
  \left[
    \kappa_{ [t_0, T] } 
  \right]^{
    5
  } 
  \| F - \tilde{F} \|_{
    \mathcal{C}^n_{ \alpha, \beta, \gamma, \delta }( [t_0,T] )
  }
  \left[
    1 + 
    \| 
      x^{ t_0, T }_{ F, v 
      }|_{
        [t_0, t]
      }
    \|_{ 
      \mathcal{E}_{ [t_0, t] } 
    }
  \right]^{ 
    ( 2 + \sum_{ i = 1 }^n \delta_i )
  }
\\ & 
  \cdot
  \sum_{ i = 1 }^n
  \int_{ t_0 }^t
  \left( t - s \right)^{
    \min\left( \alpha_i - r_j, 0 \right)
  }
  \left(
    s - t_0
  \right)^{
    ( 
      r_0 - \gamma_i
      +
      \delta_i 
      \left(
        r_0 - \beta_i 
      \right)
    )
  }
  ds
\\ &
  +
  \left[
    \kappa_{ [t_0, T] }
  \right]^5
  \|
    \tilde{F}
  \|_{
    \mathcal{C}^n_{ \alpha, \beta, \gamma, \delta }( [t_0,T] )
  }
  \left[
    1 
    +
    \|
      x^{ t_0, T }_{ F, v 
      }|_{
        [t_0, t]
      }
    \|_{
      \mathcal{E}_{ [t_0,t] }
    }
    +
    \|
      x^{ t_0, T }_{ \tilde{F}, \tilde{v} 
      }|_{
        [t_0, t]
      }
    \|_{
      \mathcal{E}_{ [t_0,t] }
    }
  \right]^{
    ( 
      1 + \sum_{ i = 1 }^n
      \delta_i
    )
  }
\\ & \cdot
  \|
    (
      x^{ t_0, T }_{ F, v } -
      x^{ t_0, T }_{ \tilde{F}, \tilde{v} }
    )|_{
      [t_0, t]
    }
  \|_{
    \mathcal{E}_{ [t_0, t] }
  }
  \sum_{ i = 1 }^n
  \int_{ t_0 }^t
  \left( t - s \right)^{
    \min\left( \alpha_i - r_j, 0 \right)
  }
  \left(
    s - t_0
  \right)^{
    (
      r_0 - \gamma_i 
      + \delta_i ( r_0 - \beta_i )
    )
  }
  ds
\end{split}
\end{equation}
for all 
$
  j \in \{ 0, 1 \}
$,
$ 
  t \in 
  (
    t_0, 
    \tau^{ t_0, T }_{ F, v } 
  )
  \cap 
  (  
    t_0, 
    \tau^{ t_0, T }_{ \tilde{F}, \tilde{v} } 
  )
$,
$
  v, \tilde{v} \in U_{ r_0 }
$,
$
  F, \tilde{F} \in
  \mathcal{C}^n_{ \alpha, \beta, \gamma, \delta }( [t_0, T] )
$
and all
$ t_0, T \in \R $
with $ t_0 < T $.
The estimate
\begin{equation}
\begin{split}
&
  \sum_{ i = 1 }^n
  \int_{ t_0 }^t
  \left( t - s \right)^{
    \min\left( \alpha_i - r_j, 0 \right)
  }
  \left(
    s - t_0
  \right)^{
    (
      r_0 - \gamma_i 
      + \delta_i ( r_0 - \beta_i )
    )
  }
  ds
\\ & =
  \sum_{ i = 1 }^n
  \left( t - t_0 \right)^{
    \left(
      1 +
      \min\left( \alpha_i - r_j, 0 \right)
      +
      r_0 - \gamma_i 
      + \delta_i ( r_0 - \beta_i )
    \right)
  }
  B_{
    (
      1 +
      \min\left( \alpha_i - r_j, 0 \right)
    ,
      1 +
      r_0 - \gamma_i 
      + \delta_i ( r_0 - \beta_i )
    )
  }
\\ & \leq
  \kappa_{ [t_0, T] } 
  \left( t - t_0 \right)^{
    \left[
      r_0 - r_j
      +
      \min_{ i \in \{ 1, \dots, n \} }
      \left(
        1 +
        \min\left( \alpha_i , r_0 \right)
        - \gamma_i 
        + \delta_i ( r_0 - \beta_i )
      \right)
    \right]
  }
\\ & \quad \cdot
  \sum_{ i = 1 }^n
  B_{
    (
      1 +
      \min\left( \alpha_i - r_j, 0 \right)
    ,
      1 +
      r_0 - \gamma_i 
      + \delta_i ( r_0 - \beta_i )
    )
  }
\\ & \leq
  \left[ 
    \kappa_{ [t_0, T] } 
  \right]^2
  \left( t - t_0 \right)^{
    \left[
      r_0 - r_j
      +
      \min_{ i \in \{ 1, \dots, n \} }
      \left(
        1 +
        \min\left( \alpha_i , r_0 \right)
        - \gamma_i 
        + \delta_i ( r_0 - \beta_i )
      \right)
    \right]
  }
\end{split}
\end{equation}
for all 
$
  j \in \{ 0, 1 \}
$,
$ 
  t \in 
  (t_0,T]
$
and all
$ t_0, T \in \R $
with $ t_0 < T $
therefore proves that
\begin{equation}
\begin{split}
&
  \left( t - t_0 \right)^{
    ( r_j - r_0 )
  }
  \big\|
    x^{ t_0, T }_{ F, v }(t) 
    -
    x^{ t_0, T }_{ \tilde{F}, \tilde{v} }(t)
  \big\|_{ 
    U_{ r_j } 
  }
  \leq 
  \kappa_{ [t_0, T] }
  \left\|
    v - \tilde{v} 
  \right\|_{
    U_{ r_0 }
  }
\\ & 
  +
  \left[
    \kappa_{ [t_0, T] } 
  \right]^{
    7
  } 
  \| F - \tilde{F} \|_{
    \mathcal{C}^n_{ \alpha, \beta, \gamma, \delta }( [t_0,T] )
  }
  \left[
    1 + 
    \| 
      x^{ t_0, T }_{ F, v 
      }|_{
        [t_0, t]
      }
    \|_{ 
      \mathcal{E}_{ [t_0, t] } 
    }
  \right]^{ 
    ( 2 + \sum_{ i = 1 }^n \delta_i )
  }
\\ & 
  \cdot
  \left( t - t_0 \right)^{
    \min_{ i \in \{ 1, \dots, n \} }
    \left(
      1 +
      \min\left( \alpha_i , r_0 \right)
      - \gamma_i 
      + \delta_i ( r_0 - \beta_i )
    \right)
  }
\\ &
  +
  \left[
    \kappa_{ [t_0, T] }
  \right]^7
  \|
    \tilde{F}
  \|_{
    \mathcal{C}^n_{ \alpha, \beta, \gamma, \delta }( [t_0,T] )
  }
  \left[
    1 
    +
    \|
      x^{ t_0, T }_{ F, v 
      }|_{
        [t_0, t]
      }
    \|_{
      \mathcal{E}_{ [t_0,t] }
    }
    +
    \|
      x^{ t_0, T }_{ \tilde{F}, \tilde{v} 
      }|_{
        [t_0, t]
      }
    \|_{
      \mathcal{E}_{ [t_0,t] }
    }
  \right]^{
    ( 
      1 +
      \sum_{ i = 1 }^n
      \delta_i
    )
  }
\\ & \cdot
  \|
    (
      x^{ t_0, T }_{ F, v } -
      x^{ t_0, T }_{ \tilde{F}, \tilde{v} }
    )|_{
      [t_0, t]
    }
  \|_{
    \mathcal{E}_{ [t_0, t] }
  }
  \left( t - t_0 \right)^{
    \min_{ i \in \{ 1, \dots, n \} }
    \left(
      1 +
      \min\left( \alpha_i , r_0 \right)
      - \gamma_i 
      + \delta_i ( r_0 - \beta_i )
    \right)
  }
\end{split}
\end{equation}
for all 
$
  j \in \{ 0, 1 \}
$,
$ 
  t \in 
  (
    t_0, 
    \tau^{ t_0, T }_{ F, v } 
  )
  \cap 
  (  
    t_0, 
    \tau^{ t_0, T }_{ \tilde{F}, \tilde{v} } 
  )
$,
$
  v, \tilde{v} \in U_{ r_0 }
$,
$
  F, \tilde{F} \in
  \mathcal{C}^n_{ \alpha, \beta, \gamma, \delta }( [t_0, T] )
$
and all
$ t_0, T \in \R $
with $ t_0 < T $.
Hence, we obtain
\begin{equation}
\begin{split}
&
  \|
    (
      x^{ t_0, T }_{ F, v }
      -
      x^{ t_0, T }_{ \tilde{F}, \tilde{v} }
    )|_{
      [t_0, t]
    }
  \|_{ 
    \mathcal{E}_{ [t_0,t] }
  }
  \leq 
  \left[
    \kappa_{ [t_0, T] }
  \right]^2
  \left\|
    v - \tilde{v} 
  \right\|_{
    U_{ r_0 }
  }
\\ & 
  +
  \left[
    \kappa_{ [t_0, T] } 
  \right]^{
    8
  } 
  \| F - \tilde{F} \|_{
    \mathcal{C}^n_{ \alpha, \beta, \gamma, \delta }( [t_0,T] )
  }
  \left[
    1 
    + 
    \| 
      x^{ t_0, T }_{ F, v 
      }|_{
        [t_0, t]
      }
    \|_{ 
      \mathcal{E}_{ [t_0, t] } 
    }
  \right]^{ 
    ( 2 + \sum_{ i = 1 }^n \delta_i )
  }
\\ & \cdot
  \left( t - t_0 \right)^{
    \min_{ i \in \{ 1, \dots, n \} }
    \left(
      1 +
      \min\left( \alpha_i , r_0 \right)
      - \gamma_i 
      + \delta_i ( r_0 - \beta_i )
    \right)
  }
\\ & +
  \left[
    \kappa_{ [t_0, T] } 
  \right]^{
    8
  } 
  \|
    \tilde{F}
  \|_{
    \mathcal{C}^n_{ \alpha, \beta, \gamma, \delta }( [t_0,T] )
  }
  \left[
    1 
    + 
    \| 
      x^{ t_0, T }_{ F, v 
      }|_{
        [t_0, t]
      }
    \|_{ 
      \mathcal{E}_{ [t_0, t] } 
    }
    + 
    \| 
      x^{ t_0, T }_{ \tilde{F}, \tilde{v} 
      }|_{
        [t_0, t]
      }
    \|_{ 
      \mathcal{E}_{ [t_0, t] } 
    }
  \right]^{ 
    ( 1 + \sum_{ i = 1 }^n \delta_i )
  }
\\ & 
  \cdot
  \|
    (
      x^{ t_0, T }_{ F, v } -
      x^{ t_0, T }_{ \tilde{F}, \tilde{v} }
    )|_{
      [t_0, t]
    }
  \|_{
    \mathcal{E}_{ [t_0, t] }
  }
  \left( t - t_0 \right)^{
    \min_{ i \in \{ 1, \dots, n \} }
    \left(
      1 +
      \min\left( \alpha_i , r_0 \right)
      - \gamma_i 
      + \delta_i ( r_0 - \beta_i )
    \right)
  }
\end{split}
\end{equation}
for all 
$ 
  t \in 
  (
    t_0, 
    \tau^{ t_0, T }_{ F, v } 
  )
  \cap 
  (  
    t_0, 
    \tau^{ t_0, T }_{ \tilde{F}, \tilde{v} } 
  )
$,
$
  v, \tilde{v} \in U_{ r_0 }
$,
$
  F, \tilde{F} \in
  \mathcal{C}^n_{ \alpha, \beta, \gamma, \delta }( [t_0, T] )
$
and all
$ t_0, T \in \R $
with $ t_0 < T $.
Rearranging finally results in
\begin{equation}
\begin{split}
&
  \|
    (
      x^{ t_0, T }_{ F, v }
      -
      x^{ t_0, T }_{ \tilde{F}, \tilde{v} }
    )|_{
      [t_0, t]
    }
  \|_{ 
    \mathcal{E}_{ [t_0,t] }
  }
  \bigg[
    1 -
    \left( t - t_0 \right)^{
      \min_{ i \in \{ 1, \dots, n \} }
      \left(
        1 +
        \min\left( \alpha_i , r_0 \right)
        - \gamma_i 
        + \delta_i ( r_0 - \beta_i )
      \right)
    }   
\\ & 
    \cdot
    \left[
      \kappa_{ [t_0, T] } 
    \right]^{
      8
    } 
    \|
      \tilde{F}
    \|_{
      \mathcal{C}^n_{ \alpha, \beta, \gamma, \delta }( [t_0,T] )
    }
    \left[
      1 
      + 
      \| 
        x^{ t_0, T }_{ F, v 
        }|_{
          [t_0, t] 
        }
      \|_{ 
        \mathcal{E}_{ [t_0, t] } 
      }
      + 
      \| 
        x^{ t_0, T }_{ \tilde{F}, \tilde{v} 
        }|_{
          [t_0, t] 
        }
      \|_{ 
        \mathcal{E}_{ [t_0, t] } 
      }
    \right]^{ 
      ( 1 + \sum_{ i = 1 }^n \delta_i )
    }
  \bigg]
\\ & \leq 
  \left[
    \kappa_{ [t_0, T] } 
  \right]^{
    9
  } 
  \left[
    \| F - \tilde{F} \|_{
      \mathcal{C}^n_{ \alpha, \beta, \gamma, \delta }( [t_0,T] )
    }
    +
    \left\|
      v - \tilde{v} 
    \right\|_{
      U_{ r_0 }
    }
  \right]
  \left[
    1 
    + 
    \| 
      x^{ t_0, T }_{ F, v 
      }|_{
        [t_0, t] 
      }
    \|_{ 
      \mathcal{E}_{ [t_0, t] } 
    }
  \right]^{ 
    ( 2 + \sum_{ i = 1 }^n \delta_i )
  }
\end{split}
\label{eq:continuity_first}
\end{equation}
for all 
$ 
  t \in 
  (
    t_0, 
    \tau^{ t_0, T }_{ F, v } 
  )
  \cap 
  (  
    t_0, 
    \tau^{ t_0, T }_{ \tilde{F}, \tilde{v} } 
  )
$,
$
  v, \tilde{v} \in U_{ r_0 }
$,
$
  F, \tilde{F} \in
  \mathcal{C}^n_{ \alpha, \beta, \gamma, \delta }( [t_0, T] )
$
and all
$ t_0, T \in \R $
with $ t_0 < T $.

We now use 
\eqref{eq:continuity_second}
and \eqref{eq:continuity_first}
to prove \eqref{eq:x_cont}.
For this let 
$ t_0, T \in \R $ 
be real numbers with $ t_0 < T $,
let $ \varepsilon \in (0,1] $
be a real number defined through
$
  \varepsilon
  := 
  \min_{ i \in \{ 1, \dots, n \} }
  \left(
    1 +
    \min( \alpha_i , r_0 )
    - \gamma_i 
    + \delta_i ( r_0 - \beta_i )
  \right)
$
and let
$ ( v_N )_{ N \in \N } \subset U_{ r_0 } $
and
$
  F_N = ( F_{ N, 1 }, \dots, F_{ N, n } ) 
  \in
  \mathcal{C}^n_{ \alpha, \beta, \gamma, \delta }( [t_0, T] )
$
$, N \in \N 
$,
be sequences with 
$
  \lim_{ N \to \infty } 
  \|
    v_1 - v_N
  \|_{
    U_{ r_0 }
  }
  =
  \lim_{ N \to \infty } 
  \|
    F_1 - F_N
  \|_{
    \mathcal{C}^n_{ \alpha, \beta, \gamma, \delta }( [ t_0, T ] )
  }
  = 0
$
and 
$
  \| F_1 - F_N \|_{ 
    \mathcal{C}^n_{ \alpha, \beta, \gamma, \delta }( [t_0, T] )
  } 
  \leq 1
$
for all 
$ 
  N \in \N 
$.
Then observe that
\eqref{eq:continuity_first}
ensures that
\begin{equation}
\begin{split}
&
  \big\|
    (
      x^{ t_0, T }_{ F_1, v_1 }
      -
      x^{ t_0, T }_{ F_N, v_N }
    )|_{
      [t_0, t]
    }
  \big\|_{ 
    \mathcal{E}_{ [t_0,t] }
  }
  \bigg[
    1 -
    \left[
      \kappa_{ [t_0, T] } 
    \right]^{
      8
    } 
    \left( t - t_0 \right)^{ 
      \varepsilon
    }   
    \Big[
      2
      + 
      \| 
        x^{ t_0, T }_{ F_1, v_1 
        }|_{
          [t_0, t]
        }
      \|_{ 
        \mathcal{E}_{ [t_0, t] } 
      }
\\ & 
      + 
      \| 
        x^{ t_0, T }_{ F_N, v_N 
        }|_{
          [t_0, t]
        }
      \|_{ 
        \mathcal{E}_{ [t_0, t] } 
      }
      +
      \|
        F_1
      \|_{
        \mathcal{C}^n_{ \alpha, \beta, \gamma, \delta }( [t_0,T] )
      }
    \Big]^{ 
      ( 2 + \sum_{ i = 1 }^n \delta_i )
    }
  \bigg]
  \leq 
  \left[
    \kappa_{ [t_0, T] } 
  \right]^{
    9
  } 
\\ & 
  \cdot
  \left[
    \| F_1 - F_N \|_{
      \mathcal{C}^n_{ \alpha, \beta, \gamma, \delta }( [t_0,T] )
    }
    +
    \left\|
      v_1 - v_N 
    \right\|_{
      U_{ r_0 }
    }
  \right]
  \left[
    1 
    + 
    \| 
      x^{ t_0, T }_{ F_1, v_1 
      }|_{
        [t_0, t]
      }
    \|_{ 
      \mathcal{E}_{ [t_0, t] } 
    }
  \right]^{ 
    ( 2 + \sum_{ i = 1 }^n \delta_i )
  }
\end{split}
\end{equation}
for all 
$ 
  t \in 
  (
    t_0, 
    \tau^{ t_0, T }_{ F_1, v_1 } 
  )
  \cap 
  (  
    t_0, 
    \tau^{ t_0, T }_{ F_N, v_N } 
  )
$
and all
$
  N \in \N
$.
This implies that
\begin{equation}
\begin{split}
&
  \|
    (
      x^{ t_0, T }_{ F_1, v_1 }
      -
      x^{ t_0, T }_{ F_N, v_N }
    )|_{
      [t_0, t]
    }
  \|_{ 
    \mathcal{E}_{ [t_0,t] }
  }
\\ & \cdot
  \bigg[
    1 -
    \left[
      \kappa_{ [t_0, T] } 
    \right]^{
      8
    } 
    \left( t - t_0 \right)^{ 
      \varepsilon
    }   
    \Big[
      4
      + 
      2 
      \,
      \| 
        x^{ t_0, T }_{ F_1, v_1 
        }|_{
          [t_0, t]
        }
      \|_{ 
        \mathcal{E}_{ [t_0, t] } 
      }
      +
      \|
        F_1
      \|_{
        \mathcal{C}^n_{ \alpha, \beta, \gamma, \delta }( [t_0,T] )
      }
    \Big]^{ 
      ( 2 + \sum_{ i = 1 }^n \delta_i )
    }
  \bigg]
\\ & \leq 
  \left[
    \kappa_{ [t_0, T] } 
  \right]^{
    9
  } 
  \left[
    \| F_1 - F_N \|_{
      \mathcal{C}^n_{ \alpha, \beta, \gamma, \delta }( [t_0,T] )
    }
    +
    \left\|
      v_1 - v_N 
    \right\|_{
      U_{ r_0 }
    }
  \right]
  \left[
    1 
    + 
    \| 
      x^{ t_0, T }_{ F_1, v_1 
      }|_{
        [t_0, t]
      }
    \|_{ 
      \mathcal{E}_{ [t_0, t] } 
    }
  \right]^{ 
    ( 2 + \sum_{ i = 1 }^n \delta_i )
  }
\end{split}
\label{eq:continuity_firstB}
\end{equation}
for all 
$ 
  t \in 
  \big\{
    s \in 
    (
      t_0, 
      \tau^{ t_0, T }_{ F_1, v_1 } 
    )
    \cap
    (  
      t_0, 
      \tau^{ t_0, T }_{ F_N, v_N } 
    )
    \colon
    \| 
      x^{ t_0, T }_{ F_N, v_N }|_{
        [t_0, t]
      }
    \|_{  
      \mathcal{E}_{ 
        [t_0, s] 
      }
    }
    \leq
    2 +
    \| 
      x^{ t_0, T }_{ F_1, v_1 
      }|_{
        [t_0, t]
      }
    \|_{  
      \mathcal{E}_{ 
        [t_0, s] 
      }
    }
  \big\}
$
and all
$
  N \in \N
$.
In the next step let 
$ 
  \hat{t} \in 
  ( t_0, \tau^{ t_0, T }_{ F_1, v_1 } )
$
and
$
  \hat{N} \in \N
$
be real numbers with the property that
\begin{equation}
    \left[
      \kappa_{ [t_0, T] } 
    \right]^{
      8
    } 
    \left( t - t_0 \right)^{ 
      \varepsilon
    }   
    \Big[
      4
      + 
      2 
      \,
      \| 
        x^{ t_0, T }_{ F_1, v_1 }|_{
          [t_0, t]
        }
      \|_{ 
        \mathcal{E}_{ [t_0, t] } 
      }
      +
      \|
        F_1
      \|_{
        \mathcal{C}^n_{ \alpha, \beta, \gamma, \delta }( [t_0,T] )
      }
    \Big]^{ 
      ( 2 + \sum_{ i = 1 }^n \delta_i )
    }
  \leq 
  \tfrac{ 1 }{ 2 }
\end{equation}
for all $ t \in (t_0, \hat{t}] $
and with the property that
\begin{equation}
\begin{split}
&
  \left[
    \kappa_{ [t_0, T] } 
  \right]^{
    9
  } 
  \left[
    \| F_1 - F_N \|_{
      \mathcal{C}^n_{ \alpha, \beta, \gamma, \delta }( [t_0,T] )
    }
    +
    \left\|
      v_1 - v_N 
    \right\|_{
      U_{ r_0 }
    }
  \right]
\\ & 
  \cdot
  \left[
    1 
    + 
    \| 
      x^{ t_0, T }_{ F_1, v_1 }|_{
        [t_0, t]
      }
    \|_{ 
      \mathcal{E}_{ [t_0, t] } 
    }
  \right]^{ 
    ( 2 + \sum_{ i = 1 }^n \delta_i )
  }
  \leq 
  \tfrac{ 1 }{ 2 }
\end{split}
\end{equation}
for all $ N \in \{ \hat{N}, \hat{N} + 1, \dots \} =: \hat{ \N } $.
Then we obtain from \eqref{eq:continuity_firstB}
that
\begin{equation}
\begin{split}
&
  \|
    (
      x^{ t_0, T }_{ F_1, v_1 }
      -
      x^{ t_0, T }_{ F_N, v_N }
    )|_{
      [ t_0, t ]
    }
  \|_{ 
    \mathcal{E}_{ [t_0,t] }
  }
\leq 
  2
  \left[
    \kappa_{ [t_0, T] } 
  \right]^{
    9
  } 
\\ & 
  \cdot
  \left[
    \| F_1 - F_N \|_{
      \mathcal{C}^n_{ \alpha, \beta, \gamma, \delta }( [t_0,T] )
    }
    +
    \left\|
      v_1 - v_N 
    \right\|_{
      U_{ r_0 }
    }
  \right]
  \left[
    1 
    + 
    \| 
      x^{ t_0, T }_{ F_1, v_1 }|_{
        [t_0, t]
      }
    \|_{ 
      \mathcal{E}_{ [t_0, t] } 
    }
  \right]^{ 
    ( 2 + \sum_{ i = 1 }^n \delta_i )
  }
  \leq 
  1
\end{split}
\end{equation}
for all 
$ 
  t \in 
  \big\{
    s \in 
    (
      t_0, 
      \hat{t}
    ]
    \cap
    (  
      t_0, 
      \tau^{ t_0, T }_{ F_N, v_N } 
    )
    \colon
    \| 
      x^{ t_0, T }_{ F_N, v_N }|_{
        [t_0, t]
      }
    \|_{  
      \mathcal{E}_{ 
        [t_0, s] 
      }
    }
    \leq
    2 +
    \| 
      x^{ t_0, T }_{ F_1, v_1 
      }|_{
        [t_0, t]
      }
    \|_{  
      \mathcal{E}_{ 
        [t_0, s] 
      }
    }
  \big\}
$
and all
$
  N \in 
  \hat{ \N }
$.
This implies that
\begin{equation}
\label{eq:continuity_firstC}
\begin{split}
&
  \|
    (
      x^{ t_0, T }_{ F_1, v_1 }
      -
      x^{ t_0, T }_{ F_N, v_N }
    )|_{
      [t_0, \hat{t}]
    }
  \|_{ 
    \mathcal{E}_{ [t_0, \hat{t}] }
  }
  \leq 
  2
  \left[
    \kappa_{ [t_0, T] } 
  \right]^{
    9
  } 
\\ & 
  \cdot
  \left[
    \| F_1 - F_N \|_{
      \mathcal{C}^n_{ \alpha, \beta, \gamma, \delta }( [t_0,T] )
    }
    +
    \left\|
      v_1 - v_N 
    \right\|_{
      U_{ r_0 }
    }
  \right]
  \left[
    1 
    + 
    \| 
      x^{ t_0, T }_{ F_1, v_1 
      }|_{
        [t_0, \hat{t}]
      }
    \|_{ 
      \mathcal{E}_{ [t_0, \hat{t}] } 
    }
  \right]^{ 
    ( 2 + \sum_{ i = 1 }^n \delta_i )
  }
\end{split}
\end{equation}
for all 
$
  N \in 
  \hat{ \N }
$
and we hence get
\begin{equation}
\label{eq:limithatt}
  \lim_{ N \to \infty }
  \|
    (
      x^{ t_0, T }_{ F_1, v_1 }
      -
      x^{ t_0, T }_{ F_N, v_N }
    )|_{
      [t_0, \hat{t}]
    }
  \|_{ 
    \mathcal{E}_{ [t_0, \hat{t}] }
  }
  = 0 .
\end{equation}
In the next step we define
$ \hat{v}_N \in U_{ r_1 } $,
$ 
  N \in 
  \hat{ \N } \cup \{ 1 \}
$,
through
$ 
  \hat{v}_N := x_{ F_N, v_N }^{ t_0, T }( \hat{t} )
$
for all $ N \in \hat{ \N } \cup \{ 1 \} $
and we define
$
  \hat{F}_N 
  \in
  C^n_{ \alpha, \beta, \gamma, \delta }( [ \hat{t}, T ] )
$,
$ 
  N \in 
  \hat{ \N } \cup \{ 1 \}
$,
through
$
  \hat{F}_{ N }
  :=
  ( 
    F_{ N, 1 }|_{
      [ \hat{t}, T ]
    }
    ,
    \dots,
    F_{ N, n }|_{
      [ \hat{t}, T ]
    }
  )
$
for all 
$ N \in \hat{ \N } \cup \{ 1 \} $.
Note that 
$ ( \hat{v}_N )_{ N \in \hat{ \N } \cup \{ 1 \} } $
is well-defined 
since
$ \hat{t} < \tau^{ t_0, T }_{ F_N, v_N } $
for all $ N \in \hat{ \N } \cup \{ 1 \} $.
Furthermore, we obtain from \eqref{eq:continuity_second} that
\begin{equation}
\begin{split}
&
  \sup_{ 
    s \in 
    [
      \hat{t}, t
    ]
  }
  \big\|
    x^{ \hat{ t }, T }_{ 
      \hat{F}_1, \hat{v}_1 
    }( s )
    -
    x^{ \hat{ t }, T }_{ \hat{ F }_N, \hat{ v }_N }( s )
  \big\|_{ 
    U_{ r_1 }
  }
\leq
  E_{
    \min\left( \alpha_1 - r_1, \dots, \alpha_n - r_1, 0 \right)
  }
  \Bigg[ 
  [
    \kappa_{ [ \hat{t}, T] }
  ]^{ 9 }
  \,
  \|
    \hat{F}_N
  \|_{
    \mathcal{C}^n_{ \alpha, \beta, \gamma, \delta }( [ \hat{t}, T ] )
  }
\\ &  
  \cdot
  \Big[
    1 
    +
    \sup_{ s \in [ \hat{t}, t ] }
    \|
      x^{ \hat{t}, T }_{ \hat{F}_1, \hat{v}_1 }(s)
    \|_{ 
      U_{ r_1 }
    }
    +
    \sup_{ s \in [ \hat{t}, t ] }
    \|
      x^{ \hat{t}, T }_{ 
        \hat{F}_N, \hat{v}_N 
      }(s)
    \|_{
      U_{ r_1 }
    }
  \Big]^{
    ( 
      1 
      + 
      \sum_{ i = 1 }^n \delta_i 
    )
  }
  \Bigg]
  [
    \kappa_{ [ \hat{t}, T] } 
  ]^{
    10
  } 
\\ & \cdot
  \left[ 
    \left\|
      \hat{v}_1 - \hat{v}_N 
    \right\|_{
      U_{ r_1 }
    }
    +
    \| 
      \hat{F}_1 - \hat{F}_N 
    \|_{
      \mathcal{C}^n_{ \alpha, \beta, \gamma, \delta }( [\hat{t}, T] )
    }
  \right]
  \left[
    1 + 
    \sup_{ 
      s \in [ \hat{t}, t ] 
    }
    \| 
      x^{ \hat{t}, T }_{ \hat{F}_1, \hat{v}_1 }( s )
    \|_{ 
      U_{ r_1 }
    }
  \right]^{
    ( 2 + \sum_{ i = 1 }^n \delta_i )
  }
\end{split}
\end{equation}
for all 
$ 
  t \in 
  (
    \hat{ t }, 
    \tau^{ \hat{t}, T }_{ \hat{F}_1, \hat{v}_1 } 
  )
  \cap 
  (  
    \hat{t}, 
    \tau^{ \hat{t}, T }_{ \hat{F}_N, \hat{v}_N } 
  )
$
and all
$ N \in \hat{ \N } $.
This implies that
\begin{equation}
\begin{split}
&
  \sup_{ 
    s \in 
    [
      \hat{t}, t
    ]
  }
  \big\|
    x^{ \hat{ t }, T }_{ 
      \hat{F}_1, \hat{v}_1 
    }( s )
    -
    x^{ \hat{ t }, T }_{ \hat{ F }_N, \hat{ v }_N }( s )
  \big\|_{ 
    U_{ r_1 }
  }
\leq
  [
    \kappa_{ [ t_0, T] }
  ]^{ 10 }
  \cdot
  E_{
    \min\left( \alpha_1 - r_1, \dots, \alpha_n - r_1, 0 \right)
  }
  \Bigg[ 
  [
    \kappa_{ [ t_0, T] }
  ]^{ 9 }
\\ &  
  \cdot
  \Big[
    2
    +
    \sup_{ 
      s \in [ \hat{t}, t ]
    }
    \|
      x^{ \hat{t}, T }_{ \hat{F}_1, \hat{v}_1 }( s )
    \|_{ 
      U_{ r_1 } 
    }
    +
    \sup_{ 
      s \in [ \hat{t}, t ]
    }
    \|
      x^{ \hat{t}, T }_{ 
        \hat{F}_N, \hat{v}_N 
      }( s )
    \|_{
      U_{ r_1 }
    }
    +
    \|
      F_1
    \|_{
      \mathcal{C}^n_{ \alpha, \beta, \gamma, \delta }( [ t_0, T ] )
    }
  \Big]^{
    ( 
      2
      + 
      \sum_{ i = 1 }^n \delta_i 
    )
  }
  \Bigg]
\\ & \cdot
  \left[ 
    \left\|
      \hat{v}_1 - \hat{v}_N 
    \right\|_{
      U_{ r_1 }
    }
    +
    \| 
      F_1 - F_N 
    \|_{
      \mathcal{C}^n_{ \alpha, \beta, \gamma, \delta }( [t_0, T] )
    }
  \right]
  \left[
    1 +
    \frac{
      \| 
        x^{ t_0, T }_{ F_1, v_1 
        }|_{
          [ t_0, t ]
        }
      \|_{ 
        \mathcal{E}_{ [t_0, t] }
      }
    }{
      ( 
        \hat{t} - t_0 
      )^{
        ( r_1 - r_0 )
      }
    }    
  \right]^{
    ( 2 + \sum_{ i = 1 }^n \delta_i )
  }
\end{split}
\end{equation}
and therefore
\begin{equation}
\label{eq:continuity_secondB}
\begin{split}
&
  \sup_{ 
    s \in 
    [
      \hat{t}, t
    ]
  }
  \big\|
    x^{ t_0, T }_{ 
      F_1, v_1 
    }( s )
    -
    x^{ t_0, T }_{ F_N, v_N }( s )
  \big\|_{ 
    U_{ r_1 }
  }
\leq
  [
    \kappa_{ [ t_0, T] }
  ]^{ 10 }
  \cdot
  E_{
    \min\left( \alpha_1 - r_1, \dots, \alpha_n - r_1, 0 \right)
  }
  \Bigg[ 
  [
    \kappa_{ [ t_0, T] }
  ]^{ 9 }
\\ &  
  \cdot
  \Big[
    4
    +
    2
    \sup_{ 
      s \in 
      [
        \hat{t}, t
      ]
    }
    \|
      x^{ t_0, T }_{ 
        F_1, v_1 
      }( s )
    \|_{ 
      U_{ r_1 }
    }
    +
    \|
      F_1
    \|_{
      \mathcal{C}^n_{ \alpha, \beta, \gamma, \delta }( [ t_0, T ] )
    }
  \Big]^{
    ( 
      2
      + 
      \sum_{ i = 1 }^n \delta_i 
    )
  }
  \Bigg]
\\ & 
  \cdot
  \left[ 
    \| \hat{v}_1 - \hat{v}_N \|_{
      U_{ r_1 }
    }
    +
    \| 
      F_1 - F_N 
    \|_{
      \mathcal{C}^n_{ \alpha, \beta, \gamma, \delta }( [t_0, T] )
    }
  \right]
  \left[
    1 +
    \frac{
      \| 
        x^{ t_0, T }_{ F_1, v_1 
        }|_{
          [ t_0, t ]
        }
      \|_{ 
        \mathcal{E}_{ [t_0, t] }
      }
    }{
      ( 
        \hat{t} - t_0 
      )^{
        ( r_1 - r_0 )
      }
    }    
  \right]^{
    ( 2 + \sum_{ i = 1 }^n \delta_i )
  }
\end{split}
\end{equation}
for all 
$ 
  t \in 
  \big\{
    s \in
    (
      \hat{ t }, 
      \tau^{ t_0, T }_{ F_1, v_1 } 
    )
    \,
    \cap 
    \,
    (  
      \hat{t}, 
      \tau^{ t_0, T }_{ F_N, v_N } 
    )
    \colon
    \sup_{
      u \in [ \hat{t} , s ]
    }
    \|
      x^{ t_0, T 
      }_{ F_N, v_N 
      }( u )
    \|_{
      U_{ r_1 } 
    }
    \leq
    2 
    +
    \sup_{
      u \in [ \hat{t} , s ]
    }
$
$
    \|
      x^{ 
        t_0, T 
      }_{ 
        F_1, v_1 
      }( u )
    \|_{
      U_{ r_1 }
    }
  \big\}
$
and all
$ N \in \hat{ \N } $.
In the next step we observe that
\eqref{eq:limithatt}
proves that there exists
a non-decreasing family 
$ N_t \in \hat{ \N } $,
$ t \in ( \hat{t}, \tau^{ t_0, T }_{ F_1, v_1 } ) $,
of natural numbers such that
\begin{equation}
\begin{split}
&
  E_{
    \min\left( \alpha_1 - r_1, \dots, \alpha_n - r_1, 0 \right)
  }
  \Bigg[ 
  \Big[
    4
    +
    2
    \sup_{ 
      s \in 
      [
        \hat{t}, t
      ]
    }
    \|
      x^{ t_0, T }_{ 
        F_1, v_1 
      }( s )
    \|_{ 
      U_{ r_1 }
    }
    +
    \|
      F_1
    \|_{
      \mathcal{C}^n_{ \alpha, \beta, \gamma, \delta }( [ t_0, T ] )
    }
  \Big]^{
    ( 
      2
      + 
      \sum_{ i = 1 }^n \delta_i 
    )
  }
\\ & 
  \cdot
    [
      \kappa_{ [ t_0, T] }
    ]^{ 9 }
  \Bigg]
  [
    \kappa_{ [ t_0, T] }
  ]^{ 10 }
  \left[ 
    \| 
      \hat{v}_1 - \hat{v}_N 
    \|_{
      U_{ r_1 }
    }
    +
    \| 
      F_1 - F_N 
    \|_{
      \mathcal{C}^n_{ \alpha, \beta, \gamma, \delta }( [t_0, T] )
    }
  \right]
\\ & 
  \cdot
  \left[
    1 +
    \frac{
      \| 
        x^{ t_0, T }_{ F_1, v_1 
        }|_{
          [ t_0, t ]
        }
      \|_{ 
        \mathcal{E}_{ [t_0, t] }
      }
    }{
      ( 
        \hat{t} - t_0 
      )^{
        ( r_1 - r_0 )
      }
    }    
  \right]^{
    ( 2 + \sum_{ i = 1 }^n \delta_i )
  }
  \leq
  1
\end{split}
\end{equation}
for all 
$ N \in \{ N_t, N_t + 1, \dots \} $
and all
$ 
  t \in 
  (
    \hat{ t }, 
    \tau^{ t_0, T }_{ F_1, v_1 } 
  )
$.
Combining this
with \eqref{eq:continuity_secondB}
results in
\begin{equation}
\begin{split}
&
  \sup_{
    s \in [ \hat{t}, t ]
  }
  \|
    x^{ t_0, T }_{ F_1, v_1 }( s )
    -
    x^{ t_0, T }_{ F_N, v_N }( s )
  \|_{ 
    U_{ r_1 } 
  }
\leq
  E_{
    \min\left( \alpha_1 - r_1, \dots, \alpha_n - r_1, 0 \right)
  }
  \bigg[ 
  [
    \kappa_{ [ t_0, T] }
  ]^{ 9 }
\\ &  
  \cdot
  \Big[
    4
    +
    2
    \,
    \sup_{ s \in [ \hat{t}, t ] }
    \|
      x^{ t_0, T }_{ F_1, v_1 }(s)
    \|_{ 
      U_{ r_1 }
    }
    +
    \|
      F_1
    \|_{
      \mathcal{C}^n_{ \alpha, \beta, \gamma, \delta }( [ t_0, T ] )
    }
  \Big]^{
    ( 
      2
      + 
      \sum_{ i = 1 }^n \delta_i 
    )
  }
  \bigg]
  \,
  [
    \kappa_{ [ t_0, T] } 
  ]^{
    10
  } 
\\ & \cdot
  \left[ 
    \|
      \hat{v}_1 - \hat{v}_N
    \|_{ U_{ r_1 } }
    +
    \| 
      F_1 - F_N 
    \|_{
      \mathcal{C}^n_{ \alpha, \beta, \gamma, \delta }( [t_0, T] )
    }
  \right]
  \left[
    1 +
    \frac{
      \| 
        x^{ t_0, T }_{ F_1, v_1 }|_{
          [ t_0, t ]
        }
      \|_{ 
        \mathcal{E}_{ [t_0, t] }
      }
    }{
      ( 
        \hat{t} - t_0 
      )^{
        ( r_1 - r_0 )
      }
    }    
  \right]^{
    ( 2 + \sum_{ i = 1 }^n \delta_i )
  }
\end{split}
\label{eq:last_cont}
\end{equation}
for all 
$ N \in \{ N_t, N_t + 1, \dots \} $
and all
$ 
  t \in ( \hat{t}, \tau^{ t_0, T }_{ F_1, v_1 } ) 
$.
Inequality~\eqref{eq:last_cont} 
implies that $ \tau^{ t_0, T } $ is lower
semicontinuous
and combining \eqref{eq:last_cont}
with \eqref{eq:continuity_firstC}
proves that 
\begin{equation}
\label{eq:x_cont_ALMOST}
  \lim_{ N \to \infty }
  \sup_{ s \in (t_0, t] }
  \left[
  \begin{array}{c}
    \left( s - t_0 \right)^{
      \left( r_1 - r_0 \right)
    }
    \|
      x^{ t_0, T }_{ F_1, v_1 }(s)
      -
      x^{ t_0, T }_{ F_N, v_N }(s)
    \|_{
      U_{ r_1 }
    }
  \\[1ex]
    +
    \,
    \|
      x^{ t_0, T }_{ F_1, v_1 }(s)
      -
      x^{ t_0, T }_{ F_N, v_N }(s)
    \|_{
      U_{ r_0 }
    }
  \end{array}
  \right]
  = 0
\end{equation}
for all 
$ t \in ( t_0,\tau_{ F_1, v_1 } ) $.
Interpolation (see, e.g.,
Theorem~37.6 in Sell \& You~\cite{sy02}) 
hence implies that \eqref{eq:x_cont}
is fulfilled.
Since every lower semicontinuous function is Borel measurable,
we obtain that $ \tau^{ t_0, T } $ is Borel measurable.
Therefore, we get 
for every $ t \in [t_0, T] $ that the sets
$
  \{
    ( F, v ) \in
    \mathcal{C}^n_{ \alpha, \beta, \gamma, \delta }( [t_0, T] ) 
    \times U_{ r_0 } 
    \colon
    \tau_{ F, v }^{ t_0, T } > t
  \}
$
and
$
  \{
    ( F, v ) \in
    \mathcal{C}^n_{ \alpha, \beta, \gamma, \delta }( [t_0, T] ) 
    \times U_{ r_0 } 
    \colon
    \tau_{ F, v }^{ t_0, T } \leq t
  \}
$
are Borel measurable subsets of
$
  \mathcal{C}^n_{ \alpha, \beta, \gamma, \delta }( [t_0, T] ) 
  \times U_{ r_0 } 
$	
and \eqref{eq:x_cont} implies 
for every $ t \in (t_0, T] $
and every $ r \in [r_0, r_1] $ 
that the mapping
$
  \{
    ( F, v ) \in
    \mathcal{C}^n_{ \alpha, \beta, \gamma, \delta }( [t_0, T] ) 
    \times U_{ r_0 } 
    \colon
    \tau_{ F, v }^{ t_0, T } > t
  \}
  \ni
  ( F, v )
  \mapsto 
  x_{ F, v }(t)
  \in
  U_{ r }
$
is continuous and, in particular, Borel measurable.
These two facts imply \eqref{eq:x_mess}
and this completes the proof of Theorem~\ref{thm:continuity}.
\end{proof}

Theorem~\ref{thm:continuity}
investigates solutions of \eqref{eq:PDEs}
on a bounded time interval. 
The next corollary extends 
this result to unbounded
time intervals.

\begin{cor}[Continuous dependence on the data
on unbounded time intervals]
\label{cor:continuity}
Assume the setting in the beginning 
of Subsection~\ref{sec:det_existence} 
and let 
$ t_0, r_0 \in \R $,
$ n \in \N $,
$ 
  \alpha = ( \alpha_1, \dots, \alpha_n ) \in \R^n
$,
$
  \beta = ( \beta_1, \dots, \beta_n ), 
  \gamma = ( \gamma_1, \dots, \gamma_n )
  \in [r_0, \infty)^n
$,
$
  \delta = ( \delta_1, \dots, \delta_n ) 
$
$
  \in [0,\infty)^n
$
with
$
  \max( \beta_1, \dots, \beta_n, 
$
$
  \gamma_1, \dots, \gamma_n ) 
  <
  1 +
  \min( \alpha_1, \dots, \alpha_n )
$
and
\begin{equation}
  \max_{ i \in \{ 1, \dots, n \} }
  \big[
  \gamma_i - \min( \alpha_i , r_0 ) 
  + 
  \delta_i
  ( \beta_i - r_0 ) 
  \big]
  < 1 .
\end{equation}
Then 
there exist a unique lower
semicontinuous function
$ 
  \tau
  \colon 
  \mathcal{C}^n_{ \alpha, \beta, \gamma, \delta }( [t_0, \infty) ) 
  \times
  U_{ r_0 } 
  \to (t_0, \infty] 
$
and a unique function
$
  x
  \colon 
  \mathcal{C}^n_{ \alpha, \beta, \gamma, \delta }( 
    [t_0, \infty) 
  ) 
  \times
  U_{ r_0 } 
  \to
  \cup_{ s \in (t_0,\infty] } 
  C( [t_0, 
$
$
  s), 
  U_{ r_0 } )
$
satisfying
$
  x_{ F, v }
  \in
  C( [t_0, \tau_{ F, v } ), U_{ r_0 } )
$,
$
  x_{ F, v }|_{ 
    ( t_0, \tau_{ F, v } ) 
  }
  \in
  C(
    (t_0,
    \tau_{ F, v } ), U_{ r_1 }
  )
$,
$
  \sup_{
    s \in (t_0,t]
  }
  ( s - t_0 )^{
    ( r_1 - r_0 )
  }
  \,
  \| 
    x_{ F, v }(s)
  \|_{ U_{ r_1 } }
  <
  \infty
$
and
\begin{equation}
  \lim_{ 
    s \nearrow \tau_{ F, v }
  }
  \big[
  \tau_{ F, v }
  +
  \| 
    x_{ F, v }(s) 
  \|_{ 
    U_{ 
      \max( \beta_1, \dots, \beta_n, \gamma_1, \dots, \gamma_n )
    } 
  }
  \big]
  = \infty
\end{equation}
and
\begin{equation}
  x_{ F, v }(t)
=
  e^{ A ( t - t_0 ) } \, v
+
  \sum_{ i = 1 }^n
  \int_{ t_0 }^t
  e^{ A ( t - s ) }
  \,
  F_i( s, x_{ F, v }(s) ) \, ds
\end{equation}
for all 
$ t \in ( t_0, \tau_{ F, v } ) $,
$ v \in U_{ r_0 } $,
$
  r_1
  \in
  [  
    r_0 ,
    1 +
    \min( \alpha_1, \dots, 
    \alpha_n )
  )
$
and all
$
  F = ( F_1, \dots, F_n )
  \in 
  \mathcal{C}^n_{ \alpha, \beta, \gamma, \delta }( [t_0, \infty) )
$.
In addition, it holds 
for every
$ t \in (t_0, \infty) $
and every
$
  r_1
  \in
  [  
    r_0 ,
    1 +
    \min( \alpha_1, \dots, \alpha_n )
  )
$
that the function
\begin{equation}
\label{eq:x_mess2}
\tag{Measurability property}
  \mathcal{C}^n_{ \alpha, \beta, \gamma, \delta }( [t_0, \infty) )
  \times U_{ r_0 }
  \ni (F,v) \mapsto 
  \left\{
  \begin{array}{ll}
    x_{ F, v }( t ) 
  &
    \colon
    t < \tau_{ F, v } 
  \\
    \infty 
  &
    \colon
    t \geq \tau_{ F, v }
  \end{array}
  \right\}
  \in U_{ r_1 } \cup \{ \infty \}
\end{equation}
is Borel measurable.
Moreover, it holds that
\begin{equation}
\label{eq:x_cont2}
\tag{Continuity property}
  \lim_{ N \to \infty }
  \sup_{ s \in (t_0, t] }
  \left[
  \begin{array}{c}
    \left( s - t_0 \right)^{
      \left( r_1 - r_0 \right)
    }
    \|
      x_{ F_1, v_1 }(s)
      -
      x_{ F_N, v_N }(s)
    \|_{
      U_{ r_1 }
    }
  \\[1ex]
    +
    \,
    \|
      x_{ F_1, v_1 }(s)
      -
      x_{ F_N, v_N }(s)
    \|_{
      U_{ r_0 }
    }
  \end{array}
  \right]
  = 0
\end{equation}
for all 
$ t \in ( t_0,\tau_{ F_1, v_1 } ) $,
$
  r_1
  \in
  \left[  
    r_0 ,
    1 +
    \min( \alpha_1, \dots, \alpha_n )
  \right)
$,
$ 
  ( v_N )_{ N \in \N } 
$
$
  \subset U_{ r_0 } 
$,
$
  ( F_N )_{ N \in \N }
  \subset
  \mathcal{C}^n_{ \alpha, \beta, \gamma, \delta }( [t_0, \infty) )
$
with
$
  \lim_{ N \to \infty } 
  d_{
    \mathcal{C}^n_{ \alpha, \beta, \gamma, \delta }( [ t_0, \infty) )
  }(
    F_1, F_N
  )
$
$
  =
  \lim_{ N \to \infty } 
  \|
    v_1 - v_N
  \|_{
    U_{ r_0 }
  }
$
$
  = 0
$.
\end{cor}

Corollary~\ref{cor:continuity}
follows immediately from
Theorem~\ref{thm:continuity}
and its proof is therefore omitted.

\subsection{SPDEs  
with space-time white noise
and polynomial nonlinearities
in two space dimensions}
\label{sec:two_dim}

The aim of this subsection is
to prove local existence and
uniqueness of mild solutions
of SPDEs in two space
dimensions with polynomial nonlinearities
of the form
\begin{equation}
\label{eq:2D_original}
  d X_t =
  \left[
    \triangle X_t
    +
    \kappa_n(t) 
    : ( X_t )^n :
    +
    \ldots 
    +
    \kappa_2(t)
    : ( X_t )^2 :
    +
    \kappa_1(t)
    X_t
    +
    \kappa_0(t)
  \right]
  dt
  + dW_t
\end{equation}
for $ t \in [0,\infty) $
with periodic boundary
conditions on
$ (0, 2 \pi)^2 $
where $ n \in \N $ is an arbitrary natural number,
where
$ \kappa_0, \kappa_1, \dots, \kappa_n \in C( [0,\infty) , \R ) $
are arbitrary continuous functions,
where $ ( W_t )_{ t \geq 0 } $ is a cylindrical
$ I $-Wiener process and where 
$ : ( X_t )^2 : , \, \dots , \, : ( X_t )^n : $
are suitable renormalizations
of $ ( X_t )^2 , \, \dots , \, ( X_t )^n $
for $ t \in [0,\infty) $.
The precise result is formulated in the following theorem.

\begin{thm}[Polynomial nonlinearities 
in two space dimensions]
\label{thm:2D}
Let 
$ 
\left( 
  \Omega, \mathcal{F}, \mathbb{P} 
\right) 
$
be a probability space,
let 
$ n \in \N $,
$ 
  t_0 \in \R 
$,
$
  \kappa_0, \kappa_1, \dots, \kappa_n 
  \in C( [t_0,\infty), \R ) 
$,
$ 
  \eta \in ( - \frac{ 2 }{ n } , 0 )
$,
let
$
  V =
  \; : \! ( V )^1 \!\! :
  ,
  : \! ( V )^2 \!\! : ,
  \dots ,
  : \! ( V )^n \!\! : \;
  \colon
  [ t_0, \infty ) \times \Omega
  \to 
  \cap_{ r \in ( - \infty, 0 ) }
  \,
  \mathcal{C}_{ \mathcal{P} }^r
  (
    [0, 2 \pi]^2, \R
  )
$
be stochastic processes
with continuous sample paths
given by
Propositions~\ref{prop:regularity_Wick}
and \ref{prop:regularity_OE}
and let 
$ 
  \xi \colon \Omega \to 
  \mathcal{C}^{ \eta }_{ 
    \mathcal{P}
  }
  ( [0, 2 \pi]^2, \R )
$
be a random variable.
Then there exists a unique random variable
$ 
  \tau \colon \Omega \to (t_0,\infty]
$
and a unique stochastic process
$
  X \colon [t_0,\infty) \times \Omega
  \to 
  \mathcal{C}^{ \eta }_{ 
    \mathcal{P}
  }
  ( [0, 2 \pi]^2, \R )
  \cup \{ \infty \}
$
such that 
for every $ \omega \in \Omega $ 
it holds that
$
  X_t( \omega ) = \infty
$
for all $ t \in [ \tau(\omega), \infty ) $,
that
\begin{equation}
  ( 
    X_s( \omega )
  )_{ 
    s \in 
    [ t_0, \tau(\omega) ) 
  }
  \in
  C\big( 
    [ t_0, \tau(\omega) ) ,
    \mathcal{C}^{ \eta }_{ \mathcal{P} }( [0, 2 \pi]^2, \R )
  \big) ,
\end{equation} 
\begin{equation}
  ( 
    X_s(\omega) - V_s(\omega) 
  )_{ s \in (t_0,\infty) } 
  \in 
  C\big( 
    (t_0,\infty), 
    \cap_{ \nu \in ( 0, 2 ) }
    \,
    [
      \mathcal{C}^{ \nu }_{ \mathcal{P} }( [0, 2 \pi]^2, 
      \R )
      \cup \{ \infty \}
    ]
  \big) ,
\end{equation}
\begin{equation}
  \sup_{
    s \in (t_0, t]
  }
  ( s - t_0 )^{ \frac{ ( r - \eta ) }{ 2 } }
  \| 
    X_s( \omega ) - V_s( \omega )
  \|_{
    \mathcal{C}_{ \mathcal{P} }^r(
      [0, 2 \pi]^2, \R
    )
  }
  < \infty
\end{equation}
for all
$ r \in [ \eta , 2 ) $
and all
$
  t \in ( t_0, \tau( \omega ) )
$
and that
\begin{equation}
\label{eq:2D}
\begin{split}
  X_t( \omega ) 
& =
  e^{     
    \mathcal{A}_2 
    \left( t - t_0 \right) 
  } 
  \,
  \xi( \omega ) 
  +
  V_t( \omega )
  - 
  e^{ 
    \mathcal{A}_2 
    \left( t - t_0 \right) 
  } 
  \, V_{ t_0 }( \omega )
  +
  \int_{ t_0 }^t
  e^{ 
    \mathcal{A}_2 
    (t - s ) 
  }
  \Bigg(
    \kappa_0(t) +
    \left(
      \kappa_1(t) + 1
    \right) X_t( \omega )
\\ & 
  +
    \sum_{ w = 2 }^n
    \kappa_w(t)
    \Bigg[
    \left( 
      X_t( \omega ) - V_t( \omega ) 
    \right)^w
    +
    \sum_{
      k = 0
    }^{ w - 1 }
    \left(
      \begin{array}{c}
        w \\ k
      \end{array}
    \right)
    \left( 
      X_t( \omega ) - V_t( \omega ) 
    \right)^k
    \left( 
      : \! ( V_t )^{ (w - k) } \! :
    \right)\!( \omega )
    \Bigg]
  \Bigg)
  \, ds
\end{split}
\end{equation}
for all $ t \in [t_0, \tau(\omega) ) $.
%
In that sense, the stochastic process 
$ X $ is a local mild solution of the
SPDE~\eqref{eq:2D_original}.
\end{thm}

Let us briefly compare 
Proposition~4.4
in 
Da Prato \& 
Debussche~\cite{DaPratoDebussche2003}
with
Theorem~\ref{thm:2D}
above.
In the setting of Theorem~\ref{thm:2D}
we note that
\begin{equation}
\begin{split}
&
  \int_{ t_0 }^t
  \left\|
    X_s( \omega ) - V_s( \omega )
  \right\|^p_{
    \mathcal{C}^r_{ \mathcal{P} }(
      [0, 2 \pi]^2, \R
    )
  }
  ds
\\ & \leq
  \left[
    \sup_{ 
      s \in ( t_0, t ]
    }
    \left(
      s - t_0
    \right)^{
      \frac{ p ( r - \eta ) }{ 2 }
    }
    \left\|
      X_s( \omega ) - V_s( \omega )
    \right\|_{
      \mathcal{C}^r_{ \mathcal{P} }( [0, 2 \pi ]^2, \R )
    }^p
  \right]
  \int_{ t_0 }^t
  \left(
    s - t_0
  \right)^{
    \frac{ 
      p \left( \eta - r \right)
    }{
      2
    }
  }
  ds
\\ & =
  \left[
    \sup_{ 
      s \in ( t_0, t ]
    }
    \left(
      s - t_0
    \right)^{
      \frac{ p ( r - \eta ) }{ 2 }
    }
    \left\|
      X_s( \omega ) - V_s( \omega )
    \right\|_{
      \mathcal{C}^r_{ \mathcal{P} }( [0, 2 \pi ]^2, \R )
    }^p
  \right]
  \frac{
    \left(
      t - t_0
    \right)^{
      \left(
        1 +
        \frac{ 
          p \left( \eta - r \right)
        }{
          2
        }
      \right)
    }
  }{
    \left(
      1 +
      \frac{ 
        p \left( \eta - r \right)
      }{
        2
      }
    \right)
  }
  < \infty
\end{split}
\end{equation}
and hence
\begin{equation}
\label{eq:regularity_DPD}
  \left( 
    X_s( \omega ) - V_s( \omega ) 
  \right)_{
    s \in [ t_0, t]
  }
  \in
  C\big(
    [ t_0, t ];
    \mathcal{C}^{ \eta }_{
      \mathcal{P}
    }(
      [ 0, 2 \pi ]^2, \R
    )
  \big)
  \cap
  L^p\big(
    [ t_0, t] ;
    \mathcal{C}^r_{
      \mathcal{P}
    }(
      [0, 2 \pi]^2, \R
    )
  \big)
\end{equation}
for all 
$ t \in ( t_0, \tau(\omega) ) $,
$ \omega \in \Omega $,
$ r \in [ \eta , \frac{ 2 }{ p } + \eta ) $
and all
$ p \in (0, \infty) $.
Equation~\eqref{eq:regularity_DPD}
implies the regularity statement
in 
Proposition~4.4
in 
Da Prato \& Debussche~\cite{DaPratoDebussche2003}
and this demonstrates 
that Theorem~\ref{thm:2D} above
implies Proposition~4.4 in 
Da Prato \& Debussche~\cite{DaPratoDebussche2003}.


\begin{proof}[Proof
of Theorem~\ref{thm:2D}]
We show 
Theorem~\ref{thm:2D}
through an application of
Corollary~\ref{cor:continuity}.
For this application define
$ 
  \left( U, \left\| \cdot \right\|_U \right) 
  := 
  \big( 
    \mathcal{C}^0_{ \mathcal{P} 
    }( [0, 2 \pi ]^2, \R ) ,
    \| \cdot \|_{ 
      \mathcal{C}^0_{ \mathcal{P} 
      }( [0, 2 \pi ]^2, \R )    
    }
  \big)
$
and 
$ 
  \left( 
    U_r, 
    \left\| \cdot \right\|_{ U_r } 
  \right) 
  :=
  \left( 
    D( ( - \mathcal{A}_2 )^r ) ,
    \left\|
      ( - \mathcal{A}_2 )^r
      ( \cdot )
    \right\|_U
  \right) 
$
for all $ r \in \R $.
Moreover, define
$ r_0 := \frac{ \eta }{ 2 } \in ( - \frac{ 1 }{ n }, 0 ) $
and let
$
  \varepsilon
  \in (0, \min( \frac{ 1 }{ 2 } , \frac{ 1 }{ n } + r_0 ) )
$
be a real number.
Observe that this ensures that
$
  n \left( \varepsilon - r_0 \right)
  < 1
$.
Next define
$
  \alpha
  :=
  - \varepsilon 
$,
$
  \beta
  :=
  \varepsilon 
$,
$
  \gamma
  :=
  \varepsilon 
$,
$  
  \delta
  :=
  n - 1
$
and
let
$
  F^{ \omega } \colon 
  [t_0,\infty) \times
  U_{ \max( \beta, \gamma ) } 
  \to 
  U_{ \alpha }
$,
$ \omega \in \Omega 
$,
be functions defined through
\begin{equation}
\label{eq:2D_Fomega}
  F^{ \omega }( t, y ) 
  =
    \kappa_0(t)
    +
    \left(
      \kappa_1(t) + 1
    \right)
    \left(
      y + V_t( \omega )
    \right)
    +
    \sum_{ w = 2 }^n
    \kappa_w(t)
    \left[
    y^w
    +
    \sum_{
      k = 0
    }^{ w - 1 }
    \left(
      \begin{array}{c}
        w \\ k
      \end{array}
    \right)
    y^k
    \left( 
      : \! ( V_t )^{ (w - k) } \! :
    \right)\!( \omega )
    \right]
\end{equation}
for all 
$ y \in U_{ \max( \beta, \gamma ) } $,
$ t \in [t_0, \infty) $,
$ \omega \in \Omega $.
Then note 
for every $ \omega \in \Omega $
that
$ 
  F^{ \omega }
  \in
  \mathcal{C}^1_{ \alpha, \beta, \gamma, \delta }( [t_0, \infty) )
$;
see \eqref{eq:Fspace_inf} for 
the definition of
$
  \mathcal{C}^1_{ \alpha, \beta, \gamma, \delta }( [t_0, \infty) )
$.
Next observe that
\begin{equation}
  \big[
    \max( \beta,  \gamma ) ,
    1 +
    \alpha
  \big)
  =
  \big[
    \varepsilon ,
    1 - \varepsilon 
  \big)
  \neq 
  \emptyset
\end{equation}
and
\begin{equation}
\begin{split}
  \gamma - \min( \alpha, r_0 )
  + ( \beta - r_0 ) \delta
& =
  \varepsilon  
  - \min(  - \varepsilon, r_0 )
  + 
  ( \varepsilon - r_0 ) 
  \left( n - 1 \right)
\\ & 
  =
  \varepsilon - r_0
  + 
  ( \varepsilon - r_0 ) 
  \left( n - 1 \right)
  =
  n \left( \varepsilon - r_0 \right) 
  < 1 .
\end{split}
\end{equation}
We can thus apply Corollary~\ref{cor:continuity}
to obtain the existence of 
a unique lower
semicontinuous function
$ 
  \rho \colon 
  \mathcal{C}^1_{ 
    \alpha, \beta, \gamma, \delta 
  }( [t_0, T] )
  \times U_{ r_0 }
  \to
  (t_0,\infty]
$
and to obtain the
existence of a unique function
$
  y \colon 
  \mathcal{C}^1_{ \alpha, \beta, \gamma, \delta }( [t_0, \infty) )
  \times U_{ r_0 }
  \to
  \cup_{ s \in ( t_0, \infty ] 
  }
  \,
  C( [t_0, s), U_{ r_0 } )
$
which satisfy
$
  y_{ G, v }
  \in
  C( [t_0, \rho_{ G, v } ), U_{ r_0 } )
$,
$
  y_{ G, v }|_{ 
    ( t_0, \rho_{ G, v } )   
  }
  \in
  C(
    ( t_0, \rho_{ G, v } ), U_{ r_1 }
  )
$
and
\begin{equation}
  \sup_{
    s \in (t_0,t]
  }
  \left( s - t_0 \right)^{
    \left( r_1 - r_0 \right)
  }
  \| 
    y_{ G, v }(s)
  \|_{ U_{ r_1 } }
  <
  \infty
  =
  \lim_{ 
    s \nearrow \rho_{ G, v } 
  }
  \big[
    \rho_{ G, V }
    +
    \| y_{ G, v }(s) \|_{ 
      U_{ 
        \varepsilon
      } 
    }
  \big]
\end{equation}
and
\begin{equation}
\label{eq:use_y}
  y_{ G, v }(t)
=
  e^{ \mathcal{A}_2 ( t - t_0 ) } \, v
+
  \int_{ t_0 }^t
  e^{ \mathcal{A}_2 ( t - s ) }
  \,
  G( s, y_{ G, v }(s) ) \, ds
\end{equation}
for all 
$ 
  t \in ( t_0, \rho_{ G, v } ) 
$,
$ v \in U_{ r_0 } $,
$
  r_1
  \in
  \left[  
    \frac{ \eta }{ 2 } ,
    1 - \varepsilon
  \right)
$,
$
  G 
  \in 
  \mathcal{C}^1_{ 
    \alpha, \beta, \gamma, \delta 
  }( [t_0, T] )
$
and all
$ T \in (0,\infty) $.
Next we define functions
$ 
  \tau \colon \Omega \to (t_0, \infty] 
$
and 
$ 
  X \colon [t_0, \infty) \times \Omega 
  \to U_{ r_0 } \cup \{ \infty \} 
$
through
$
  \tau( \omega ) :=
  \rho_{
    F^{ \omega }
    , 
    \,
    \xi( \omega ) - 
    V_{ t_0 }( \omega ) 
  }
$
for all 
$ \omega \in \Omega $
and through
\begin{equation}
  X_t( \omega ) :=
  \begin{cases}
     y_{
      F^{ \omega }
      , 
      \,
      \xi( \omega ) - 
      V_{ t_0 }( \omega ) 
     }( t )
     +
     V_t( \omega ) 
   & 
     \colon 
     t < \tau( \omega )
   \\ 
     \infty 
   &
     \colon
     t \geq \tau( \omega )
  \end{cases}
\end{equation}
for all $ t \in [t_0,\infty) $
and all $ \omega \in \Omega $.
This definition together
with \eqref{eq:use_y}
ensures that
\begin{equation}
\begin{split}
  X_t( \omega )
  - 
  V_t( \omega )
=
  e^{ \mathcal{A}_2 ( t - t_0 ) } 
  \big(
    \xi( \omega ) 
    -
    V_{ t_0 }( \omega )
  \big)
  +
  \int_{ t_0 }^t
  e^{ \mathcal{A}_2 ( t - s ) }
  \,
  F^{ \omega }\big( 
    s, 
    X_t( \omega ) -  
    V_t( \omega )
  \big) \, 
  ds
\end{split}
\end{equation}
for all $ t \in ( t_0, \tau( \omega ) $
and all $ \omega \in \Omega $.
Combining this with \eqref{eq:2D_Fomega}
proves that $ X $ fulfills \eqref{eq:2D}.
In the next step we note that
\begin{equation}
\begin{split}
&
  \mathcal{B}\!\left(
    C\big( 
      [t_0, \infty) , 
      \big[
        \mathcal{C}^{ - \varepsilon / 2 }_{ \mathcal{P} }( [0, 2 \pi]^2, \R ) 
      \big]^{
        \! \times n
      }      
    \big)
  \right)
\\ & =
  \sigma_{ 
    C( 
      [t_0, \infty), 
      [
        \mathcal{C}^{ - \varepsilon / 2 }_{ \mathcal{P} }( [0, 2 \pi]^2, \R ) 
      ]^{
        \times n
      }
    )
  }\!\bigg(
    C\big( 
      [t_0, \infty), 
      \big[
        \mathcal{C}^{ - \varepsilon / 2 }_{ \mathcal{P} }( [0, 2 \pi]^2, \R ) 
      \big]^{
        \! \times n
      }      
    \big)
\\ & \quad
    \ni
    f \mapsto
    f(t)
    \in
    \big[
      \mathcal{C}^{ - \varepsilon / 2 }_{ \mathcal{P} }( [0, 2 \pi]^2, \R ) 
    \big]^{
      \! \times n
    }      
    \colon
    t \in [t_0, \infty)
  \bigg) .
\end{split}
\end{equation}
This implies that the mapping
\begin{multline}
  \Omega \ni
  \omega \mapsto
  \big( 
    V_t(\omega) ,
    ( : \!\! ( V_t )^2 \!\! : )(\omega) ,
    \dots ,
    ( : \!\! ( V_t )^n \!\! : )(\omega)
  \big)_{
    t \in [t_0, \infty)
  }
  \in
  C\big( 
    [t_0, \infty) , 
    \big[
      \mathcal{C}^{ - \varepsilon / 2 }_{ \mathcal{P} }( [0, 2 \pi]^2, \R ) 
    \big]^{
      \! \times n
    }      
  \big)
\end{multline}
is 
$ 
  \mathcal{F} 
$/$ 
  \mathcal{B}\big( 
    C( 
      [t_0, \infty), 
      \big[
        \mathcal{C}^{ - \varepsilon / 2 }_{ \mathcal{P} }( [0, 2 \pi]^2, \R ) 
      \big]^{
        \! \times n
      }      
    )
  \big)
$-measurable.
This ensures that
the mapping
\begin{equation}
  \Omega \ni \omega
  \mapsto 
  \big(
    F^{ \omega } ,
    \xi( \omega )
  \big)
  \in
  \mathcal{C}^1_{ \alpha, \beta, \gamma, \delta }( [t_0, \infty) )
  \times
  U_{ r_0 }
\end{equation}
is 
$ 
  \mathcal{F} 
$/$ 
  \mathcal{B}\big( 
    \mathcal{C}^1_{ \alpha, \beta, \gamma, \delta }( [t_0, \infty) )
    \times
    U_{ r_0 }
  \big) 
$-measurable.
Combining this with 
Corollary~\ref{cor:continuity}
proves that $ \tau $ is a random variable
and that $ X $ is a stochastic process
(see \eqref{eq:x_mess2} in
Corollary~\ref{cor:continuity}
for details).
Since 
$
  \varepsilon
  \in (0, \min( \frac{ 1 }{ 2 } , \frac{ 1 }{ n } + r_0 ) )
$
was arbitrary,
the proof of
Theorem~\ref{thm:2D}
is completed.
\end{proof}

\subsection{SPDEs
with space-time white noise
and quadratic nonlinearities
in three space dimensions}
\label{sec:three_dim}

The aim of this subsection is
to prove local existence and
uniqueness of mild solutions
of SPDEs in three space
dimensions with quadratic nonlinearities
of the form
\begin{equation}
\label{eq:quadratic}
  d X_t =
  \left[
    \triangle X_t
    + \kappa_2(t) 
    : ( X_t )^2 :
    + \kappa_1(t) X_t
    + \kappa_0(t) 
  \right]
  dt
  + dW_t
\end{equation}
for $ t \in [0,\infty) $
with periodic boundary
conditions on
$ (0, 2 \pi)^3 $
where 
$ \kappa_0, \kappa_1, \kappa_2 \in C( [0,\infty) , \R ) $
are arbitrary continuous functions, 
where $ ( W_t )_{ t \geq 0 } $ is a cylindrical
$ I $-Wiener process 
and where $ : ( X_t )^2 : $
is a suitable renormalization
of $ ( X_t )^2 $
for $ t \in [0,\infty) $.
The precise result is formulated in the following theorem.

\begin{thm}[Quadratic
nonlinearities in three space dimensions]
\label{thm:quadratic}
Let 
$ 
\left( 
  \Omega, \mathcal{F}, \mathbb{P} 
\right) 
$
be a probability space,
let 
$
  t_0 \in \R
$,
$
  \kappa_0, \kappa_1, \kappa_2 \in 
  C( [t_0,\infty), \R )
$,
$ 
  \eta \in ( - 1, - \frac{ 1 }{ 2 } )
$,
let
$
  V 
  \colon 
  [t_0, \infty) 
  \times \Omega
  \to 
  \cap_{
    r \in ( - \infty, - 1 / 2 )
  }
  \mathcal{C}_{ \mathcal{P} }^r
  ( [0, 2 \pi]^3, \R )
$
and
$
  : \!\! ( V )^2 \!\! : \;
  \colon
  [t_0, \infty)
  \times \Omega
  \to 
  \cap_{
    r \in ( - \infty, - 1 )
  }
  \mathcal{C}_{ \mathcal{P} }^{
    r
  }(
    [0, 2 \pi]^3, \R
  )
$
be stochastic processes
with continuous sample paths
given by 
Propositions~\ref{prop:regularity_Wick}
and \ref{prop:regularity_OE}
and let 
$ 
  \xi \colon \Omega \to 
  \mathcal{C}^{ \eta }_{ 
    \mathcal{P}
  }
  ( [0, 2 \pi]^3, \R )
$
be a random variable.
Then there exists a unique random variable
$ 
  \tau \colon \Omega \to (t_0,\infty]
$
and a unique stochastic process
$
  X \colon [t_0,\infty) \times \Omega
  \to 
  \mathcal{C}^{ \eta }_{ 
    \mathcal{P}
  }
  ( [0, 2 \pi]^3, \R )
  \cup \{ \infty \}
$
such that for every 
$ \omega \in \Omega $ 
it holds that
$
  X_t( \omega ) = \infty
$
for all $ t \in [ \tau(\omega), \infty ) $,
that
\begin{equation}
  ( 
    X_s( \omega )
  )_{ 
    s \in 
    [ t_0, \tau(\omega) ) 
  }
  \in
  C( 
    [ t_0, \tau(\omega) ) ,
    \mathcal{C}^{ \eta }_{ \mathcal{P} }( [0, 2 \pi]^3, \R )
  ) 
  ,
\end{equation}
\begin{equation}
  ( 
    X_s(\omega) - V_s(\omega) 
  )_{ s \in (t_0,\infty) } 
  \in 
  C\big( 
    (t_0,\infty), 
    \cap_{ \nu \in ( \frac{ 1 }{ 2 } , 1) }
    \,
    [
      \mathcal{C}^{ \nu }_{ \mathcal{P} }( [0, 2 \pi]^3, \R ) \cup \{ \infty \}
    ]
  \big)
  ,
\end{equation}
\begin{equation}
  \sup_{
    s \in (t_0, t]
  }
  \,
  ( s - t_0 )^{ 
    \frac{ ( r - \eta ) }{ 2 }
  }
  \,
  \| 
    X_s( \omega ) - V_s( \omega )
  \|_{
    \mathcal{C}^{ r }_{ \mathcal{P} }( [0, 2 \pi]^3, \R ) 
  }
  < \infty
\end{equation}
for all
$
  r \in [ \eta , 1 )
$
and all
$
  t \in (t_0, \tau(\omega))
$
and that
\begin{multline}
\label{eq:3D}
  X_t( \omega ) 
  =
  e^{ 
    \mathcal{A}_3
    \left( t - t_0 \right) 
  }   
  \,
  \xi( \omega ) 
  +
  V_t( \omega )
  - e^{ 
    \mathcal{A}_3 
    \left( t - t_0 \right) 
  } 
  \, V_0( \omega )
  +
  \int_{ t_0 }^t
  e^{ 
    \mathcal{A}_3 
    (t - s ) 
  }
  \bigg[
    \kappa_2(t)
    \Big(
    \left( 
      X_t( \omega ) - V_t( \omega ) 
    \right)^2
\\
    +
    2 
    \left( 
      X_t( \omega ) - V_t( \omega ) 
    \right) V_t( \omega )
    +
    \left( 
      : \! ( V_t )^2 \! :
    \right)\!( \omega )
    \Big)
    + 
    \left( \kappa_1(t) + 1 \right)
    X_t( \omega )
    + 
    \kappa_0(t) 
  \bigg]
  ds
\end{multline}
for all $ t \in [t_0, \tau(\omega) ) $.
%
In that sense, the stochastic process 
$ X $ is a local mild solution 
of the SPDE~\eqref{eq:quadratic}.
\end{thm}

\begin{proof}[Proof
of Theorem~\ref{thm:quadratic}]
We show 
Theorem~\ref{thm:quadratic}
through an application of
Corollary~\ref{cor:continuity}.
For this application define
$ 
  \left( U, \left\| \cdot \right\|_U 
  \right) 
  := 
  \big( 
    \mathcal{C}^0_{ \mathcal{P} 
    }( [ 0, 2 \pi ]^3, \R ) ,
    \left\| \cdot 
    \right\|_{ 
      \mathcal{C}^0_{ \mathcal{P} 
      }( [ 0, 2 \pi ]^3, \R )    
    }
    \!
  \big)
$
and 
$ 
  \left( 
    U_r, 
    \left\| \cdot \right\|_{ U_r } 
  \right) 
  :=
  \left( 
    D( ( - \mathcal{A}_3 )^r ) ,
    \left\|
      ( - \mathcal{A}_3 )^r
      ( \cdot )
    \right\|_U
  \right) 
$
for all $ r \in \R $.
Moreover,
define
$ 
  r_0 := \frac{ \eta }{ 2 } \in 
  ( - \frac{ 1 }{ 2 }, - \frac{ 1 }{ 4 } )
$
and
let
$
  \varepsilon \in ( 0, \frac{ 1 }{ 4 } + \frac{ r_0 }{ 2 } )
$
be a real number.
Observe that this ensures that
$
  2 \varepsilon
  - r_0
  < \frac{ 1 }{ 2 }
$
and that
$
  \varepsilon
  < \frac{ 1 }{ 8 }
$.
Next define
$
  \alpha
  :=
  - \frac{ 1 }{ 2 }
  - \varepsilon 
$,
$
  \beta
  :=
  - \frac{ 1 }{ 4 } - \frac{ \varepsilon }{ 2 }
$,
$
  \gamma
  :=
  \frac{ 1 }{ 4 } + \varepsilon
$
and
$
  \delta
  :=
  1
$
and let
$
  F^{ \omega } 
  \colon 
  [t_0, \infty) 
  \times U_{ \max( \beta, \gamma ) } 
  \to 
  U_{ \alpha_i }
$,
$ 
  \omega \in \Omega 
$,
be functions defined 
through
\begin{equation}
\label{eq:3D_Fomega}
  F^{ \omega 
  }( t, y ) 
  :=
  \kappa_2(t) 
  \left(
    y^2 
    +
    2 \, V_t( \omega ) \,
    y
    +
    \left( 
      : \! ( V_t )^2 \! :
    \right)\!( \omega )
  \right)
  +
  \left(
    \kappa_1(t) + 1
  \right)
  \left( 
    y + V_t( \omega ) 
  \right)
  +
    \kappa_0(t)
\end{equation}
for all 
$ y \in U_{ \max( \beta, \gamma ) } $,
$ t \in [t_0, \infty) $,
$ \omega \in \Omega $.
Then note 
for every $ \omega \in \Omega $ 
that
$ 
  F^{ \omega }
  \in
  \mathcal{C}^1_{ 
    \alpha, \beta, \gamma, \delta 
  }( [t_0, \infty) )
$;
see \eqref{eq:Fspace_inf} for the definition of
$
  \mathcal{C}^1_{ 
    \alpha, \beta, \gamma, \delta 
  }( [t_0, \infty) )
$.
Next observe that
\begin{equation}
  \big[
    \max( \beta, \gamma ) ,
    1 +
    \alpha
  \big)
  =
  \big[
    \tfrac{ 1 }{ 4 } + \varepsilon ,
    \tfrac{ 1 }{ 2 } - \varepsilon
  \big)
  \neq
  \emptyset
\end{equation}
and
\begin{equation}
\begin{split}
&
    \gamma - \min( \alpha, r_0 )
    + 
    ( \beta - r_0 ) \delta
  =
    \gamma + \beta 
    - \alpha
    - r_0 
\\ & =
    \tfrac{ \varepsilon }{ 2 }
    +
    \tfrac{ 1 }{ 2 }
    +
    \varepsilon
    -
    r_0
    \leq
    \tfrac{ 1 }{ 2 }
    +
    2 \varepsilon 
    -
    r_0
  < 1 .
\end{split}
\end{equation}
We can thus apply Theorem~\ref{thm:continuity}
to obtain the existence of a unique lower
semicontinuous function
$ 
  \rho \colon 
  \mathcal{C}^1_{ \alpha, \beta, \gamma 
  }( [t_0, \infty) )
  \times U_{ r_0 }
  \to
  (t_0, \infty ]
$
and to obtain the existence
of a unique function
$
  y \colon 
  \mathcal{C}^1_{ 
    \alpha, \beta, \gamma 
  }( [t_0, \infty) )
  \times U_{ r_0 }
  \to
  \cup_{ 
    s \in (t_0, \infty] 
  }
  \,
  C( [t_0, s), U_{ r_0 } )
$
which satisfy
$
  y_{ G, v }
  \in
  C( 
    [t_0, \rho_{ G, v } ), U_{ r_0 } 
  )
$,
$
  y_{ G, v 
  }|_{ 
    ( t_0, \rho_{ G, v } ) 
  }
  \in
  C(
    (t_0, \rho_{ G, v } ), U_{ r_1 }
  )
$
and
\begin{equation}
  \sup_{
    s \in (t_0,t]
  }
  \left( s - t_0 \right)^{
    \left( r_1 - r_0 \right)
  }
  \| 
    y_{ G, v }(s)
  \|_{ U_{ r_1 } }
  <
  \infty
  =
  \lim_{ 
    s \nearrow \rho_{ G, v } 
  }
  \left[
    \rho_{ G, v }
    +
    \| 
      y_{ G, v }(s) 
    \|_{ 
      U_{ 
        \frac{ 1 }{ 4 } + \varepsilon
      } 
    }
  \right]
\end{equation}
and
\begin{equation}
\label{eq:use_y2}
  y_{ G, v }(t)
=
  e^{ 
    \mathcal{A}_3 ( t - t_0 ) 
  } \, v
+
  \int_{ t_0 }^t
  e^{ 
    \mathcal{A}_3 
    ( t - s ) 
  }
  \,
  G( s, y_{ G, v }(s) ) \, ds
\end{equation}
for all 
$ t \in ( t_0, \rho_{ G, v } ) $,
$ v \in U_{ r_0 } $,
$
  r_1
  \in
  \left[  
    \frac{ \eta }{ 2 } ,
    \frac{ 1 }{ 2 } - \varepsilon
  \right)
$
and all
$
  G 
  \in 
  \mathcal{C}^1_{ 
    \alpha, \beta, \gamma, \delta 
  }( [t_0, \infty) )
$.
Next we define functions
$ \tau \colon \Omega \to (t_0, \infty] $
and 
$ 
  X \colon [t_0, \infty) \times \Omega 
  \to U_{ r_0 } \cup \{ \infty \} 
$
through
$
  \tau( \omega ) :=
  \rho_{
    F^{ \omega }
    , 
    \,
    \xi( \omega ) - 
    V_{ t_0 }( \omega ) 
  }
$
for all 
$ \omega \in \Omega $
and through
\begin{equation}
  X_t( \omega ) :=
  \begin{cases}
     y_{
      F^{ \omega } ,
      \,
      \xi( \omega )        
      -
      V_{ t_0 }( \omega )
     }( t )
     +
     V_{ t }( \omega )
     & 
     \colon 
     t < \tau( \omega )
   \\ 
     \infty 
     &
     \colon
     t \geq \tau( \omega )
  \end{cases}
\end{equation}
for all $ t \in [t_0,\infty) $
and all $ \omega \in \Omega $.
This definition together
with \eqref{eq:use_y2}
ensures that
\begin{equation}
\begin{split}
  X_t( \omega )
  - 
  V_t( \omega )
  =
  e^{ 
    \mathcal{A}_3 
    ( t - t_0 ) 
  } 
  \big(
    \xi( \omega ) 
    -
    V_{ t_0 }( \omega )
  \big)
  +
  \int_{ t_0 }^t
  e^{ 
    \mathcal{A}_3 
   ( t - s ) 
  }
  \,
  F^{ \omega }( 
    s, 
    X_t( \omega ) -  
    V_t( \omega )
  ) \, 
  ds
\end{split}
\end{equation}
for all $ t \in ( t_0, \tau( \omega ) ) $
and all $ \omega \in \Omega $.
Combining this with \eqref{eq:3D_Fomega}
proves that $ X $ fulfills \eqref{eq:3D}.
In the next step we note that
\begin{equation}
\begin{split}
&
  \mathcal{B}\!\left(
    C\big( 
      [t_0, \infty), 
      \mathcal{C}^{ - ( 1 + \varepsilon ) / 2 }_{ \mathcal{P} }( [0, 2 \pi]^3, \R ) 
      \times
      \mathcal{C}^{ - ( 2 + \varepsilon ) / 2 }_{ \mathcal{P} }( [0, 2 \pi]^3, \R ) 
    \big)
  \right)
\\ & =
  \sigma_{ 
    C( 
      [t_0, \infty), 
      \mathcal{C}^{ - ( 1 + \varepsilon ) / 2 }_{ \mathcal{P} }( [0, 2 \pi]^3, \R ) 
      \times
      \mathcal{C}^{ - ( 2 + \varepsilon ) / 2 }_{ \mathcal{P} }( [0, 2 \pi]^3, \R ) 
    )
  }\!\Big(
\\ & 
    C\big( 
      [t_0, \infty), 
      \mathcal{C}^{ - ( 1 + \varepsilon ) / 2 }_{ \mathcal{P} }( [0, 2 \pi]^3, \R ) 
      \times
      \mathcal{C}^{ - ( 2 + \varepsilon ) / 2  }_{ \mathcal{P} }( [0, 2 \pi]^3, \R ) 
    \big)
\\ & 
    \ni
    f \mapsto
    f(t)
    \in
    \mathcal{C}^{ - ( 1 + \varepsilon ) / 2 }_{ \mathcal{P} }( [0, 2 \pi]^3, \R ) 
    \times
    \mathcal{C}^{ - ( 2 + \varepsilon ) / 2 }_{ \mathcal{P} }( [0, 2 \pi]^3, \R ) 
    \colon
    t \in [t_0, \infty)
  \Big) .
\end{split}
\end{equation}
This implies that the mapping
\begin{equation}
\begin{split}
  \Omega \ni
  \omega \mapsto
  \big( 
    V_t( \omega ), 
    ( : \! ( V_t )^2 \! : )(\omega)
  \big)_{
    t \in [t_0, \infty)
  }
  \in
  C\big( 
    [t_0, \infty), 
    \mathcal{C}^{ 
      - ( 1 + \varepsilon ) / 2 
    }_{ \mathcal{P} }( [0, 2 \pi]^3, \R ) 
    \times
    \mathcal{C}^{ - ( 2 + \varepsilon ) / 2 }_{ \mathcal{P} }( [0, 2 \pi]^3, \R ) 
  \big)
\end{split}
\end{equation}
is 
$ 
  \mathcal{F} 
$/$ 
  \mathcal{B}\big( 
    C( 
      [t_0, \infty), 
      \mathcal{C}^{ - ( 1 + \varepsilon ) / 2 }_{ \mathcal{P} }( [0, 2 \pi]^3, \R ) 
      \times
      \mathcal{C}^{ - ( 2 + \varepsilon ) / 2 }_{ \mathcal{P} }( [0, 2 \pi]^3, \R ) 
    )
  \big)
$-measurable
and this shows that
the mapping
\begin{equation}
  \Omega \ni 
  \omega
  \mapsto 
  \big(
    F^{ \omega }, 
    \xi( \omega )
    - V_{ t_0 }( \omega )
  \big)
  \in
  \mathcal{C}^1_{ 
    \alpha, \beta, \gamma, \delta 
  }( 
    [t_0, \infty) 
  )
  \times
  U_{ r_0 }
\end{equation}
is 
$ 
  \mathcal{F} 
$/$ 
  \mathcal{B}\big( 
    \mathcal{C}^1_{ 
      \alpha, \beta, \gamma, \delta 
    }( 
      [t_0, \infty) 
    )
    \times
    U_{ r_0 }
  \big) 
$-measurable.
Combining this with 
Corollary~\ref{cor:continuity}
proves that $ \tau $ is a random variable
and that $ X $ is a stochastic process
(see \eqref{eq:x_mess2} in Corollary~\ref{cor:continuity}
for details).
Since $ \varepsilon \in ( 0, \frac{ 1 }{ 4 } + \frac{ r_0 }{ 2 } ) $
was arbitrary,
the proof of Theorem~\ref{thm:quadratic}
is completed.
\end{proof}

\bibliographystyle{acm}
\bibliography{../Bib/bibfile}

\end{document}